\documentclass[10pt,journal,onecolumn,a4paper]{IEEEtran}
\usepackage{amsfonts,amsmath,amssymb,amstext,amsthm,cite,color,dsfont,enumerate,hyperref,makeidx,verbatim,xfrac,xr,url}
\usepackage{pict2e}
\usepackage[mathscr]{euscript}
\hypersetup{
pdftitle		={The Augustin Capacity and Center},
pdfauthor		={Bar{\i}s Nakiboglu},}
\title{The Augustin Capacity and Center}
\author{Bar\i\c{s} Nakibo\u{g}lu\\ \small{\href{mailto:bnakib@metu.edu.tr}{bnakib@metu.edu.tr}}
\thanks{This paper was presented in part at 
	the 2017 IEEE International Symposium on Information Theory
	\cite{nakiboglu17}.}
}
\theoremstyle{plain}
\newtheorem{lemma}{Lemma} 
\newtheorem{theorem}{Theorem}

\newtheorem*{conjecture*}{Conjecture}

\theoremstyle{definition}
\newtheorem{definition}{Definition} 

\newtheorem{example}{Example}
\newtheorem{remark}{Remark}
\newtheorem*{remark*}{Remark}
\definecolor{mygray}{gray}{0.4}
\newcommand{\set} [1]			{{\mathscr{{#1}}}}
\newcommand{\alg}[1]			{{\mathcal{{#1}}}}
\newcommand{\rndv}[1]      {{\mathsf{{#1}}}}
\newcommand{\oper}[1]      {{\mathtt{{#1}}}}
\newcommand{\msr}[1]       {{\it    {{#1}}}}
\newcommand{\cnst}[1]      {{\mathit{{#1}}}}

\newcommand{\sss}[1]		{{\mathit{2}^{{#1}}}}
\newcommand{\integers}[1]	{{\mathbb{Z}}_{^{{#1}}}}

\newcommand{\reals}[1]		{{\mathbb{R}}_{^{{#1}}}}
\newcommand{\inte}[1]      {{\mathtt{int}{{#1}}}}
\newcommand{\clos}[1]      {{\mathtt{cl}{{#1}}}}
\newcommand{\conv}[1]      {{\mathtt{ch}{{#1}}}}

\DeclareMathOperator*{\essup}{ess\,sup}

\newcommand{\dif}[1]       {{\mathrm{d}{#1}}}  
\newcommand{\der}[2]        {\tfrac{\dif{#1}}{\dif{#2}}}  
\newcommand{\pder}[2]       {\tfrac{\partial{#1}}{\partial{#2}}}  
\newcommand{\supp}[1]       {\mathtt{supp}({{#1}})}       

\newcommand{\DEF}[0]			{{\!\!~\triangleq\!~}}  
\newcommand{\mtimes}[0]			{{\circledast}}

\newcommand{\AC}[0]            {{\prec}}  
\newcommand{\UAC}[0]           {{\mathop{\prec}\nolimits^{uni}}}  
\newcommand{\NAC}[0]           {{\nprec}}  
\newcommand{\abs}[1]           {{\left\lvert{{#1}}\right\lvert}}
\newcommand{\abp}[1]           {{\left\lvert{{#1}}\right\lvert^{+}}}
\newcommand{\lon}[1]           {{{\left\lVert{{#1}}\right\lVert}}} 

\newcommand{\IND}[1]           {{\mathds{1}_{\{#1\}}}}    
\newcommand{\ind}[0]           {{\imath}}
\newcommand{\jnd}[0]           {{\jmath}}
\newcommand{\knd}[0]           {{\kappa}}
\newcommand{\tin}[0]           {{\cnst{t}}}
\newcommand{\blx}[0]           {{\cnst{n}}}

\newcommand{\tinS}[0]          {{\set{T}}}

\newcommand{\domtr}[1]         {{\set{Q}}_{{#1}}}

\newcommand{\EXS}[2]         {{\bf E}_{{#1}}\!\left[{#2}\right]}

\newcommand{\EX}[1]          {\EXS{\!}{{#1}}}                      %Expectation
\newcommand{\fX}[0]          {{\cnst{f}}}

\newcommand{\gX}[0]          {{\cnst{g}}}   
\newcommand{\GX}[0]          {{\cnst{G}}}   

\newcommand{\fXS}[0]         {{\set{F}}}

\newcommand{\bent}[1]			{{\cnst{\hslash}}			  \!\left(       			  {#1}								  	 \right)}

\newcommand{\RD}[3]				{{\cnst{D}}_{{#1}}            \!\left(\left.            \! {#2}\right\Vert {#3}                  \right)}

\newcommand{\CRD}[4]			{{\cnst{D}}_{{#1}}            \!\left(\left.\!\left.    \! {#2}\right\Vert {#3} \right\vert{{#4}}\right)}
\newcommand{\RDF}[4]			{{\cnst{D}}_{{#1}}^{{#2}}     \!\left(\left.            \! {#3}\right\Vert {#4}                  \right)}

\newcommand{\RMI}[3]			{{\cnst{I}}_{{#1}}            \!\left(                  \! {#2};         \!{#3}                \!\right)} 

\newcommand{\GMI}[3]			{{\cnst{I}}_{{#1}}^{{\scriptscriptstyle g}}\!\left(                  \! {#2};         \!{#3}                \!\right)} 

\newcommand{\RMIL}[4]			{{\cnst{I}}_{{#1}}^{{#4}}	  \!\left(                  \! {#2};         \!{#3}                \!\right)} 
\newcommand{\GMIL}[4]			{{\cnst{I}}_{{#1}}^{{\scriptscriptstyle g}{#4}}	  \!\left(                  \! {#2};         \!{#3}                \!\right)} 
 
\newcommand{\GCL}[3]			{{\cnst{C}}_{{#1},{#2}}^{{\scriptscriptstyle g}{#3}}}
\newcommand{\RCL}[3]			{{\cnst{C}}_{{#1},{#2}}^{{#3}}}
\newcommand{\RRL}[3]			{{\cnst{S}}_{{#1},{#2}}^{{#3}}}

\newcommand{\RC}[2]				{{\cnst{C}}_{{#1},{#2}}}

\newcommand{\CGC}[3]			{{\cnst{C}}_{{#1},{#2},{#3}}^{{\scriptscriptstyle g}}}
\newcommand{\CRC}[3]			{{\cnst{C}}_{{#1},{#2},{#3}}}

\newcommand{\RR}[2]				{{\cnst{S}}_{{#1},{#2}}}

\newcommand{\Aopi}[4]			{{\oper{T}}_{{#1},{#2}}^{#3}\left({#4}\right)} 
\newcommand{\Aop}[3]			{{\Aopi{#1}{#2}{}{#3}}}

\newcommand{\costf}[0]			{{\cnst{\rho}}}
\newcommand{\costc}[0]			{{\cnst{\varrho}}}
\newcommand{\lgm}[0]			{{\cnst{\lambda}}}
    
\newcommand{\rfm}[0]			{{{\msr{\nu}}}} 
 
\newcommand{\uc}[0]				{{\mathds{1}}}
\newcommand{\fcc}[1]			{{\cnst{\Gamma}_{{#1}}}}
\newcommand{\fccc}[1]			{{\cnst{\Gamma}_{{#1}}^{{\scriptscriptstyle ex}}}}
   
\newcommand{\csetA}[0]			{{\alg{A}}}
\newcommand{\cset}[0]			{{\set{A}}}

\newcommand{\rate}[0]			{{\cnst{R}}}
\newcommand{\cln}[1]          {{{\xi}_{{#1}}}}

\newcommand{\rnb}[0]          {{\cnst{\beta}}}
\newcommand{\rnf}[0]          {{\cnst{\phi}}}

\newcommand{\rno}[0]          {{\cnst{\alpha}}}
\newcommand{\rnt}[0]          {{\cnst{\eta}}}

\newcommand{\brl}[0]           {{\alg{B}}}
\newcommand{\rborel}[1]        {{\brl}({#1})}

\newcommand{\gsal}[1]          {{\sigma({{#1}})}}

\newcommand{\oev}[0]           {{\set{E}}}
\newcommand{\opa}[0]           {{\alg{E}}}
\newcommand{\cha}[0]           {{\set{W}}}
\newcommand{\chu}[0]           {{\set{U}}}

\newcommand{\GausDen}[1]		{{{{\cnst{\varphi}}}_{{#1}}}}
\newcommand{\smea}[1]          {{{\alg{M}}({#1})}}
\newcommand{\fmea}[1]          {{{\alg{M}}^{^{+}}\!({#1})}}
\newcommand{\zmea}[1]          {{{\alg{M}}_{{0}}^{^{+}}\!({#1})}}
\newcommand{\pmea}[1]          {{{\alg{P}}({#1})}}
\newcommand{\fdis}[1]          {{{\set{M}}^{^{+}}\!({#1})}}
\newcommand{\pdis}[1]          {{{\set{P}}({#1})}}

\newcommand{\dinp}[0]          {{\cnst{x}}}

\newcommand{\inpS}[0]          {{\set{X}}}
\newcommand{\inpA}[0]          {{\alg{X}}}

\newcommand{\dout}[0]          {{\cnst{y}}}
\newcommand{\out}[0]           {{\rndv{Y}}}
\newcommand{\outS}[0]          {{\set{Y}}}
\newcommand{\outA}[0]          {{\alg{Y}}}

\newcommand{\dsta}[0]          {{\cnst{z}}}

\newcommand{\mean}[0]        {{{\msr{\mu}}}}    
\newcommand{\mmn}[1]         {{{\mean}_{{#1}}}}
\newcommand{\mma}[2]         {{{\mean}_{{#1}}^{{#2}}}}    
\newcommand{\qgn}[1]         {{{\mQ}_{{#1}}^{{\scriptscriptstyle g}}}}
\newcommand{\qga}[2]         {{{\mQ}_{{#1}}^{{\scriptscriptstyle g}{#2}}}}

\newcommand{\mA}[0]				{{\msr{a}}}    
\newcommand{\amn}[1]			{{{\mA}_{{#1}}}}
\newcommand{\ama}[2]			{{{\mA}_{{#1}}^{{#2}}}}

\newcommand{\mB}[0]				{{\msr{b}}}

\newcommand{\mP}[0]				{{\msr{p}}}    
\newcommand{\pmn}[1]			{{{\mP}_{{#1}}}}
\newcommand{\pma}[2]			{{{\mP}_{{#1}}^{{#2}}}}

\newcommand{\mQ}[0]				{{\msr{q}}}    
\newcommand{\qmn}[1]			{{{\mQ}_{{#1}}}}
\newcommand{\qma}[2]			{{{\mQ}_{{#1}}^{{#2}}}}
\newcommand{\Qm}[0]				{{{\cnst{Q}}}}

\newcommand{\mS}[0]				{{\msr{s}}}    
\newcommand{\smn}[1]			{{{\mS}_{{#1}}}}

\newcommand{\mU}[0]				{{\msr{u}}}    
\newcommand{\umn}[1]			{{{\mU}_{{#1}}}}
\newcommand{\uma}[2]			{{{\mU}_{{#1}}^{{#2}}}}

\newcommand{\Umn}[1]			{{{\cnst{U}}_{{#1}}}}

\newcommand{\mV}[0]				{{\msr{v}}}

\newcommand{\Vm}[0]				{{{\cnst{V}}}}

\newcommand{\Vma}[2]			{{{\cnst{V}}_{{#1}}^{{#2}}}}

\newcommand{\mW}[0]				{{\msr{w}}}    
\newcommand{\wmn}[1]			{{{\mW}_{{#1}}}}
\newcommand{\wma}[2]			{{{\mW}_{{#1}}^{{#2}}}}
\newcommand{\Wm}[0]				{{{\cnst{W}}}}
\newcommand{\Wmn}[1]			{{{\cnst{W}}_{{#1}}}}
\newcommand{\Wma}[2]			{{{\cnst{W}}_{{#1}}^{{#2}}}}

\newcommand{\csiszar}[0]							{Csisz\'{a}r~}

\newcommand{\harremoes}[0]							{Harremo\"{e}s~}
\newcommand{\korner}[0]								{K\"{o}rner~}
\newcommand{\renyi}[0]								{R\'{e}nyi~}

\makeatletter
\DeclareRobustCommand{\bigplus}{%
	\mathop{\vphantom{\sum}\mathpalette\@bigplus\relax}\slimits@
}
\newcommand{\@bigplus}[2]{\vcenter{\hbox{\make@bigplus{#1}}}}
\newcommand{\make@bigplus}[1]{%
	\sbox\z@{$\m@th#1\sum$}%
	\setlength{\unitlength}{\wd\z@}%
	\begin{picture}(1.4,1.4)
	\linethickness{.17ex}
	\Line(.7,.14)(.7,1.26)
	\Line(.14,.7)(1.26,.7)
	\end{picture}%
}
\DeclareRobustCommand{\bigtimes}{%
	\mathop{\vphantom{\sum}\mathpalette\@bigtimes\relax}\slimits@
}
\newcommand{\@bigtimes}[2]{\vcenter{\hbox{\make@bigtimes{#1}}}}
\newcommand{\make@bigtimes}[1]{%
	\sbox\z@{$\m@th#1\sum$}%
	\setlength{\unitlength}{\wd\z@}%
	\begin{picture}(1,1)
	\linethickness{.17ex}
	\Line(.1,.1)(.9,.9)
	\Line(.1,.9)(.9,.1)
	\end{picture}%
}
\makeatother

\def\thesubsection{\thesection.\arabic{subsection}}
 
\def\thesubsectiondis{\thesectiondis\arabic{subsection}.} 
\begin{document}
\setcounter{tocdepth}{2}	
\pagestyle{plain}
\pagenumbering{arabic}
\hypersetup{hidelinks}
\maketitle 
\thispagestyle{empty}
\vspace{-.72cm}
%!TEX root=../main-C.tex
\begin{abstract}		
For any channel, the existence of a unique Augustin mean is established for any 
positive order and probability mass function on the input set.
The Augustin mean is shown to be the unique fixed point of an operator 
defined in terms of the order and the input distribution.
The Augustin information is shown to be continuously differentiable in the order.  
For any channel and convex constraint set with finite Augustin capacity, 
the existence of a unique Augustin center and the associated van Erven-Harremo\"{e}s 
bound are established.
The Augustin-Legendre (A-L) information, capacity, center, and radius are introduced and 
the latter three are proved to be equal to the corresponding  R\'{e}nyi-Gallager quantities.
The equality of the A-L capacity to the A-L radius for arbitrary channels
and the existence of a unique A-L center for channels with finite A-L capacity
are established.
For all interior points of the feasible set of cost constraints, 
the cost constrained Augustin capacity and center are expressed 
in terms of the A-L capacity and center. 
Certain shift invariant families of probabilities 
and certain Gaussian channels are analyzed as examples.
\end{abstract}
\vspace{-.55cm}
\tableofcontents
%!TEX root=../main-C.tex
\section{Introduction}\label{sec:introduction}
The mutual information, which is sometimes called the Shannon information, is a pivotal 
quantity in the analysis of various information transmission problems.
It is defined without referring to an optimization problem, but it satisfies the 
following two identities given in terms of the Kullback-Leibler divergence
\begin{align}
\label{eq:def:mutualinformation:R}
\RMI{}{\mP}{\Wm}
&=\inf\nolimits_{\mQ\in \pmea{\outA}} 
\RD{}{\mP \mtimes \Wm}{\mP\otimes\mQ}
\\
\label{eq:def:mutualinformation:A}
&=\inf\nolimits_{\mQ\in \pmea{\outA}}\sum\nolimits_{\dinp} \mP(\dinp) \RD{}{\Wm(\dinp)}{\mQ}
\end{align}
where 
\(\pmea{\outA}\) is the set of all probability measures on the output space
\((\outS,\outA)\),
\(\mP\) is a probability mass function that is positive only on a finite subset 
of the input set 
\(\inpS\),
and
\(\Wm\) is a function of the form \(\Wm:\inpS\to\pmea{\outA}\). 
Either of the expressions on the right hand side can be taken as 
the definition of the mutual information.
One can define the order \(\rno\) \renyi information via these expressions 
by replacing the Kullback-Leibler divergence with the order \(\rno\) \renyi 
divergence.
Since the order one \renyi divergence is the Kullback-Leibler divergence, the order one
\renyi information is equal to the mutual information for both definitions. 
For other orders, however,  these two definitions 
are not equivalent to the definition of the
mutual information or to one another, as pointed out by \csiszar \cite{csiszar95}. 
The generalization associated with the expression in \eqref{eq:def:mutualinformation:R} 
is called the order \(\rno\) \renyi information and denoted by \(\GMI{\rno}{\mP}{\Wm}\).
The generalization associated with the expression in \eqref{eq:def:mutualinformation:A} 
is called the order \(\rno\) Augustin information and denoted
by \(\RMI{\rno}{\mP}{\Wm}\).
Following the convention for the constrained Shannon capacity, the order 
\(\rno\) Augustin capacity for the constraint set \(\cset\) is defined as 
\(\sup_{\mP\in\cset} \RMI{\rno}{\mP}{\Wm}\).

For constant composition codes on the memoryless classical-quantum channels,
the Augustin information for orders less than one arises in 
the expression for the sphere packing exponent
and
the Augustin information for orders greater than one arises in 
the expression for the strong converse exponent,
as recently pointed out by
Dalai\cite{dalai17} and by Mosonyi and Ogawa \cite{mosonyiO18},
respectively.
For the constant composition codes on the discrete stationary product channels,
these observations were made implicitly 
by \csiszar and \korner in \cite[p. 172]{csiszarkorner}
and by \csiszar in \cite{csiszar95}.
For the cost constrained codes on (possibly non-stationary) product channels
with additive cost functions,
the cost constrained Augustin capacity plays an analogous role in the 
expressions for the sphere packing exponent and the strong converse exponent. 
The observations about the sphere packing exponent were also reported by
Augustin in \cite[Remark 36.7-(i) and \S36]{augustin78}
for quite general channel models.
Therefore Augustin's information measures do have operational significance, 
at the very least for the channel coding problem. 
Our main aim in the current manuscript, however, is to analyze 
the Augustin information and capacity as measure theoretic concepts.  
Throughout the manuscript, we will refrain from referring to the channel coding 
problem or the operational significance of Augustin's information measures because 
we believe the Augustin 
information and capacity can and should be understood as measure theoretic concepts 
first.
The operational significance of the Augustin information and capacity can be 
established afterward using information theoretic techniques together with the results 
of the measure theoretic analysis, as we do in \cite{nakiboglu18D}.

All of the previous works on the Augustin information or capacity,
except Augustin's \cite{augustin78}, assume the output set \(\outS\)
of the channel \(\Wm\) to be a finite set 
\cite{csiszar95,dalai17,ervenH14,shayevitz10,shayevitz11,verdu15}. 
This, however, is a major drawback because the finite output set 
assumption is violated by certain analytically interesting models 
that are also important because of their prominence in engineering
applications, such as the Gaussian and Poisson 
channel models.
We pursue our analysis on a more general model and
assume\footnote{We have additional hypotheses in \S\ref{sec:cost-TP},
	but those assumptions are satisfied by essentially all models
	of interest, as well.} 
the output space \((\outS,\outA)\) to be a measurable space composed of 
an output set \(\outS\) and a \(\sigma\)-algebra of its subsets \(\outA\). 
Our analysis of the Augustin information and capacity in this
general framework is built around two fundamental concepts:
the Augustin mean and the Augustin center.

Recall that the mutual information is defined as 
\(\RMI{}{\mP}{\Wm}\DEF\sum\nolimits_{\dinp} \mP(\dinp) \RD{}{\Wm(\dinp)}{\qmn{1,\mP}}\)
where \(\qmn{1,\mP}=\sum_{\dinp}\mP(\dinp)\Wm(\dinp)\).
Hence the infimum in \eqref{eq:def:mutualinformation:A} is achieved by \(\qmn{1,\mP}\).
Furthermore, one can confirm by substitution that
\begin{align}
\notag
\sum\nolimits_{\dinp} \mP(\dinp) \RD{}{\Wm(\dinp)}{\mQ}
&=\RMI{}{\mP}{\Wm}+\RD{}{\qmn{1,\mP}}{\mQ}
&
&\forall \mQ\in\pmea{\outA}.
\end{align}
Thus \(\qmn{1,\mP}\) is the only probability measure achieving the infimum 
in \eqref{eq:def:mutualinformation:A} because  the Kullback-Leibler divergence 
is positive for distinct probability measures.
A similar relation holds for other orders, as well:
for any \(\rno\) in \(\reals{+}\) there exists a unique probability measure 
\(\qmn{\rno,\mP}\) satisfying 
\(\RMI{\rno}{\mP}{\Wm}=\sum\nolimits_{\dinp} \mP(\dinp) \RD{\rno}{\Wm(\dinp)}{\qmn{\rno,\mP}}\).
We call the probability measure \(\qmn{\rno,\mP}\), \emph{the order \(\rno\) Augustin mean}.
In \cite[Lemma 34.2]{augustin78}, Augustin established the existence of a unique 
\(\qmn{\rno,\mP}\) for \(\rno\)'s in \((0,1]\) and derived certain important characteristics
of \(\qmn{\rno,\mP}\) that are the corner stones of the analysis of the Augustin information 
and capacity. 
We establish analogous relations for orders greater than one in \S\ref{sec:information},
see Lemma \ref{lem:information}-(\ref{information:oti}).

In \cite{kemperman74},
Kemperman proved the equality of the (unconstrained) Shannon capacity to the Shannon 
radius\footnote{Shannon radius is defined as 
	\(\inf_{\mQ\in\pmea{\outA}}\sup_{\dinp\in\inpS} \RD{}{\Wm(\dinp)}{\mQ}\).}
for any channel of the form \(\Wm:\inpS\to\pmea{\outA}\)
and the existence a unique Shannon center for channels with finite Shannon capacity. 
Using ideas that are already present in Kemperman's proof, one can establish a similar result
for the constrained Shannon capacity provided that the constrained set is convex, 
see \cite[Thm. \ref*{A-thm:Cminimax}]{nakiboglu19A}:
For any channel \(\Wm\) of the form \(\Wm:\inpS\to\pmea{\outA}\) and convex constraint set \(\cset\),
\begin{align}
\label{eq:shannon-minimax}
\sup\nolimits_{\mP\in\cset} \RMI{}{\mP}{\Wm}
&=\inf\nolimits_{\mQ\in\pmea{\outA}} \sup\nolimits_{\mP\in\cset} 
\sum\nolimits_{\dinp} \mP(\dinp) \RD{}{\Wm(\dinp)}{\mQ}.
\end{align}
Considering \eqref{eq:def:mutualinformation:A}, one can interpret \eqref{eq:shannon-minimax} 
as a minimax theorem. 
Furthermore, if the Shannon capacity for the constraint set \(\cset\) 
is finite, then there exists a unique probability measure \(\qmn{1,\Wm,\cset}\), 
called the Shannon center for the constraint set \(\cset\), such that
\begin{align}
%\label{eq:shannon-center}
\notag
\sup\nolimits_{\mP\in\cset} \RMI{}{\mP}{\Wm}
&=\sup\nolimits_{\mP\in\cset} 
\sum\nolimits_{\dinp} \mP(\dinp) \RD{}{\Wm(\dinp)}{\qmn{1,\Wm,\cset}}.
\end{align}
The name center is reminiscent of the name of the corresponding quantity in the unconstrained 
case, which is discussed in \cite{kemperman74}.
Augustin proved an analogous result for \(\RMI{\rno}{\mP}{\Wm}\) assuming \(\rno\) 
to be an order in \((0,1]\) and \(\cset\) to be a constraint set determined by cost 
constraints, see \cite[Lemma 34.7]{augustin78}.
We prove an analogous proposition for \(\RMI{\rno}{\mP}{\Wm}\) for any \(\rno\) in 
\(\reals{+}\) and convex constraint set \(\cset\)  in \S\ref{sec:capacity}, 
see Theorem \ref{thm:minimax}.
We call the corresponding probability measure \(\qmn{\rno,\Wm,\cset}\) \emph{the order 
\(\rno\) Augustin center for the constraint set \(\cset\)}.

Constraint sets determined by cost constraints are frequently encountered 
while employing the Augustin capacity to analyze channel coding problems. 
One can apply the convex conjugation techniques to provide an alternative 
characterization of the cost constrained Augustin capacity and center.
Augustin did so in \cite[\S35]{augustin78},
relying on a quantity that was previously 
employed 
in discrete channels by Gallager \cite[pp. 13-15]{gallager65}, \cite[\S7.3]{gallager} 
and
in various Gaussian channel 
models\footnote{Augustin assumed neither a specific noise model nor 
	the finiteness of the output set. Nevertheless, Gaussian channels are not subsumed 
	by Augustin's model in \cite[\S35]{augustin78} 
	because Augustin assumed a bounded cost function.} 
by Gallager \cite[pp. 15,16]{gallager65}, 
\cite[\S\S7.4,7.5]{gallager}, Ebert \cite{ebert66}, and Richters \cite{richters67}.
We call this quantity the \renyi\!\!-Gallager information and analyze it in \S\ref{sec:cost-RG}.
Compared to the application of convex conjugation techniques
to the cost constrained Shannon capacity provided by
\csiszar and \korner in \cite[Ch. 8]{csiszarkorner}, 
Augustin's analysis in \cite[\S 35]{augustin78} relying on the \renyi\!\!-Gallager 
information is rather convoluted. 
In \S\ref{sec:cost-AL}, we adhere to a more standard approach
and provide an analysis, 
which can be seen as a generalization of \cite[Ch. 8]{csiszarkorner}, 
relying on a new quantity, which we call the Augustin-Legendre information. 
We show the equivalence of these two approaches using minimax theorems
similar to the one described above for the constrained Augustin capacity.

Some of the most important observations we present in this paper have 
already been derived previously in \cite[\S\S33-35]{augustin78}, 
\cite{shayevitz11}, \cite{haroutunian68}, \cite{poltyrev82}.
In order to delineate our main contributions in the context of these works, 
we provide a tally in \S\ref{sec:introduction-contributions}.
Before doing that, we describe our notational conventions 
in \S\ref{sec:introduction-notation} and our model 
in \S\ref{sec:introduction-model}.

\subsection{Notational Conventions}\label{sec:introduction-notation}
 
The inner product of any two vectors \(\mean\) and \(\mQ\) in \(\reals{}^{\ell}\), 
i.e. \(\sum_{\ind=1}^{\ell} \mean^{\ind} \mQ^{\ind}\), is denoted by \(\mean \cdot \mQ\).
The \(\ell\) dimensional vector whose all entries are one
is denoted by \(\uc\) for any \(\ell\in\integers{+}\), 
the dimension \(\ell\) will be clear from the context.
We denote the closure, interior, and convex hull of a set \(\set{S}\) by 
\(\clos{\set{S}}\), \(\inte{\set{S}}\), and \(\conv{\set{S}}\), respectively;
the relevant topology or vector space structure will be evident from the context.

For any set \(\outS\), we denote 
the set of all subsets of \(\outS\) ---i.e. the power set of \(\outS\)--- 
by \(\sss{\outS}\),
the set of all probability measures on finite subsets of \(\outS\) 
by \(\pdis{\outS}\),
and 
the set of all non-zero finite measures with the same property by \(\fdis{\outS}\).
For any \(\mP\) in \(\fdis{\outS}\),  we call the set of all \(\dout\)'s 
satisfying \(\mP(\dout)>0\) the support of \(\mP\) and denote it by
\(\supp{p}\).

On a measurable space \((\outS,\outA)\), we denote 
the set of all finite signed measures by \(\smea{\outA}\),
the set of all finite measures by \(\zmea{\outA}\),
the set of all non-zero finite measures by \(\fmea{\outA}\),
and 
the set of all probability measures by  \(\pmea{\outA}\). 
Let \(\mean\) and \(\mQ\) be two measures on the measurable space \((\outS,\outA)\).
Then \(\mean\) is  absolutely continuous with respect to \(\mQ\), i.e. \(\mean\AC \mQ\), 
iff \(\mean(\oev)=0\) for any \(\oev \in \outA\) such that \(\mQ(\oev)=0\);
\(\mean\) and \(\mQ\) are equivalent, i.e. \(\mean\sim\mQ\),
iff \(\mean\AC\mQ\)  and \(\mQ\AC\mean\);
\(\mean\) and \(\mQ\) are singular, i.e. \(\mean\perp\mQ\),
iff \(\exists\oev \in \outA\) such that  \(\mean(\oev)=\mQ(\outS\setminus\oev)=0\).
Furthermore, a set of measures \(\cha\) on \((\outS,\outA)\) is absolutely continuous 
with respect to \(\mQ\), i.e. \(\cha\AC\mQ\), iff \(\mW\AC\mQ\) for all \(\mW\in \cha\)
and uniformly absolutely continuous with respect to \(\mQ\), i.e. \(\cha\UAC\mQ\), 
iff for every \(\epsilon>0\) there exists a \(\delta>0\) such that \(\mW(\oev)<\epsilon\)
for all \(\mW\in\cha\) provided that \(\mQ(\oev)<\delta\).

We denote the integral of a measurable function \(\fX\) with respect to the measure \(\mean\) 
by \(\int \fX \mean(\dif{\dout})\) or \(\int \fX(\dout) \mean(\dif{\dout})\).
If the integral is on the real line and if it is with respect to the Lebesgue measure, we 
denote it by \(\int\fX \dif{\dout}\) or \(\int \fX(\dout) \dif{\dout}\), as well.
If \(\mean\) is a probability measure, then we also call the integral of \(\fX\) with respect \(\mean\)
the expectation of \(\fX\) or the expected value of \(\fX\) and denote it by
\(\EXS{\mean}{\fX}\) or \(\EXS{\mean}{\fX(\out)}\).

Our notation will be overloaded for certain symbols; however, the relations represented 
by these symbols will be clear from the context.
We use \(\bent{\cdot}\) to denote both the Shannon entropy and the binary entropy:
\(\bent{\mP}\DEF\sum_{\dout} \mP(\dout)\ln\frac{1}{\mP(\dout)}\) for all \(\mP\in\pdis{\outS}\)
and \(\bent{\dsta}\DEF\dsta \ln \frac{1}{\dsta}+(1-\dsta)\ln\frac{1}{1-\dsta}\) for all \(\dsta\in[0,1]\).
We denote the product of topologies \cite[p. 38]{dudley}, 
\(\sigma\)-algebras \cite[p. 118]{dudley}, and measures \cite[Thm. 4.4.4]{dudley} by \(\otimes\).
We denote the Cartesian product  of sets \cite[p. 38]{dudley} by \(\times\).
We use the short hand 
\(\inpS_{1}^{\blx}\) for the Cartesian product of sets \(\inpS_{1},\ldots,\inpS_{\blx}\)
and 
\(\outA_{1}^{\blx}\) for the product of the \(\sigma\)-algebras  \(\outA_{1},\ldots,\outA_{\blx}\). 
We use \(\abs{\cdot}\) to denote the absolute value of real numbers and the size of sets. 
The sign \(\leq\) stands for the usual less than or equal to relation for real numbers
and the corresponding point-wise inequity for functions and vectors. 
For two measures \(\mean\) and \(\mQ\) on the measurable space \((\outS,\outA)\), 
\(\mean \leq \mQ\) iff \(\mean(\oev)\leq \mQ(\oev)\) for all \(\oev\in\outA\).

For \(\mA,\mB\in\reals{}\), \(\mA\wedge\mB\) is the minimum of \(\mA\) and \(\mB\).
For \(\fX:\outS\to\reals{}\) and \(\gX:\outS\to\reals{}\), the function \(\fX\wedge\gX\) is 
the pointwise minimum of \(\fX\) and \(\gX\).
For \(\mean,\mQ\in\smea{\outA}\), \(\mean\wedge\mQ\) is the unique measure 
satisfying \(\der{\mean\wedge\mQ}{\rfm}=\der{\mean}{\rfm}\wedge\der{\mQ}{\rfm}\)
\(\rfm\)-a.e. for any \(\rfm\) satisfying \(\mean\AC\rfm\) and \(\mQ\AC\rfm\).
For a  collection \(\fXS\) of real valued functions \(\wedge_{\fX\in\fXS}\fX\)
is the pointwise infimum of \(\fX\)'s in \(\fXS\), which is an extended real valued function. 
For a collection of measures \(\chu\subset\smea{\outA}\) satisfying \(\mW\leq\mU\)
for all \(\mU\in\chu\) for some \(\mW\in\pmea{\outA}\), 
\(\wedge_{\mU\in\chu}\mU\) is the infimum of
\(\chu\) with respect to the partial order \(\leq\).
There exists a unique infimum measure by \cite[Thm. 4.7.5]{bogachev}. 
We use the symbol \(\vee\) analogously to \(\wedge\) but we represent maxima and 
suprema with it, rather than minima and infima.

\subsection{Channel Model}\label{sec:introduction-model}
A \emph{channel} \(\Wm\) is a function from \emph{the input set} \(\inpS\) to  the set of all probability 
measures on \emph{the output space} \((\outS,\outA)\):
\begin{align}
\label{eq:def:channel}
\Wm:\inpS \to \pmea{\outA}.
\end{align}
\(\outS\) is called \emph{the output set} and \(\outA\) is called \emph{the \(\sigma\)-algebra  of the 
output events}. 
We denote the set of all channels from the input set \(\inpS\) to the output space \((\outS,\outA)\)
by \(\pmea{\outA|\inpS}\).
For any \(\mP\in\pdis{\inpS}\)  and \(\Wm\in\pmea{\outA|\inpS}\), 
the probability
measure whose marginal on \(\inpS\) is \(\mP\) and whose conditional distribution given \(\dinp\) is
\(\Wm(\dinp)\)
is denoted by \(\mP\mtimes\Wm\). 
Until \S\ref{sec:cost-TP}, we confine our discussion to the input distributions in \(\pdis{\inpS}\)
and avoid the subtleties related to measurability.
The more general case of input distributions in \(\pmea{\inpA}\) 
is considered\footnote{The structure described in \eqref{eq:def:channel} is not sufficient 
	on its own to ensure the existence of a unique \(\mP\mtimes\Wm\) with the desired properties for all  \(\mP\) in \(\pmea{\inpA}\).
	The existence of such a unique \(\mP\mtimes\Wm\) is guaranteed for all \(\mP\) in \(\pmea{\inpA}\),
	if  \(\Wm\) is a transition probability from \((\inpS,\inpA)\) to \((\outS,\outA)\), i.e. a
	member of \(\pmea{\outA|\inpA}\) rather than \(\pmea{\outA|\inpS}\).}  
in \S\ref{sec:cost-TP}.

A channel \(\Wm\) is called a \emph{discrete channel} if both \(\inpS\) and \(\outA\) are finite sets.
For any \(\blx\in\integers{+}\) and 
channels \(\Wmn{\tin}\!:\!\inpS_{\tin}\!\to\!\pmea{\outA_{\tin}}\) for \(\tin\in\{1,\ldots,\blx\}\),
the \emph{length \(\blx\) product channel}
\(\Wmn{[1,\blx]}\!:\!\inpS_{1}^{\blx}\!\to\!\pmea{\outA_{1}^{\blx}}\) is defined via the following relation:
\begin{align}
%\label{eq:def:product}
\notag
\Wmn{[1,\blx]}(\dinp_{1}^{\blx})
&=\bigotimes\nolimits_{\tin=1}^{\blx}\Wmn{\tin}(\dinp_{\tin})
&
&\forall \dinp_{1}^{\blx}\in\inpS_{1}^{\blx}.
\end{align}
A product channel is \emph{stationary} iff \(\Wmn{\tin}\!=\!\Wm\) 
for all \(\tin\!\in\!\{1,\ldots,\blx\}\) for some \(\Wm\!:\!\inpS\!\to\!\pmea{\outA}\).

For any \(\ell\in\integers{+}\), an \(\ell\) dimensional \emph{cost function} \(\costf\) 
is a function from the input set to \(\reals{}^{\ell}\) 
that is bounded from below, i.e. that is of the form
\(\costf:\inpS \to \reals{\geq \dsta}^{\ell}\)
for some \(\dsta\in \reals{}\). 
%%%%%%%%%For any \(\ell\in\integers{+}\), an \(\ell\) dimensional \emph{cost function} \(\costf\) 
%%%%%%%%%is a function from the input set to \(\reals{}^{\ell}\) 
%%%%%%%%%that is bounded from below:
%%%%%%%%%\begin{align}
%%%%%%%%%%\label{eq:def:costfunction}
%%%%%%%%%\notag
%%%%%%%%%&\costf:\inpS \to \reals{\geq \dsta}^{\ell}
%%%%%%%%%&
%%%%%%%%%&\dsta\in \reals{}.
%%%%%%%%%\end{align}
We assume without loss of generality that\footnote{Augustin \cite[\S33]{augustin78} has an additional 
	hypothesis, \(\bigvee_{\dinp\in \inpS} \costf(\dinp)\leq\uc\).
	This hypothesis, however, excludes certain important cases, 
	such as the Gaussian channels.}
\begin{align}
\notag
\inf\nolimits_{\dinp\in \inpS} \costf^{\ind}(\dinp) 
&\geq 0
&
&\forall \ind \in \{1,\ldots,\ell\}.
\end{align}
We denote the set of all cost constraints that can be satisfied by some member of \(\inpS\) by \(\fccc{\costf}\) 
and the set of all cost constraints that can be satisfied by some member of \(\pdis{\inpS}\) by \(\fcc{\costf}\):
\begin{align}
\label{eq:def:feasibleX}
\fccc{\costf}
&\DEF\{\costc\in\reals{\geq0}^{\ell}:\exists \dinp\in\inpS \mbox{~s.t.~}\costf(\dinp) \leq \costc\} 
\\
\label{eq:def:feasible}
\fcc{\costf}
&\DEF\{\costc\in\reals{\geq0}^{\ell}:\exists\mP\in\pdis{\inpS} \mbox{~s.t.~}\sum\nolimits_{\dinp} \mP(\dinp) \costf(\dinp) \leq \costc\}.
\end{align}
Then both \(\fccc{\costf}\) and \(\fcc{\costf}\) have non-empty interiors
and \(\fcc{\costf}\) is the convex hull of \(\fccc{\costf}\), i.e. \(\fcc{\costf}=\conv{\fccc{\costf}}\).

A cost function on a product channel is said to be additive iff 
it can be written as the sum of cost functions defined on the component channels.
Given \(\Wmn{\tin}\!:\!\inpS_{\tin}\!\to\!\pmea{\outA_{\tin}}\) and 
\(\costf_{\tin}\!:\!\inpS_{\tin}\!\to\!\reals{\geq0}^{\ell}\) for \(\tin\!\in\!\{1,\ldots,\blx\}\),
we denote the resulting additive cost function on \(\inpS_{1}^{\blx}\) for the channel
\(\Wmn{[1,\blx]}\)  by \(\costf_{[1,\blx]}\), i.e. 
{\color{white}\cite{fano,arimoto76,oohama17A,oohama17B,vazquezMF15,renyi61,csiszar67A,gilardoni10B,shiryaev,polyanskiyV10,kolmogorovfomin75,csiszar72,sibson69,blahut74,kostinaV15,nakiboglu19B,krantzparks,munkres,rudin,bertsekas,komiya88,sion58}}
\begin{align}
%\label{eq:additivecost}
\notag
\costf_{[1,\blx]}(\dinp_{1}^{\blx})
&=\sum\nolimits_{\tin=1}^{\blx} \costf_{\tin}(\dinp_{\tin})  
&
&\forall \dinp_{1}^{\blx}\in\inpS_{1}^{\blx}. 
\end{align}

\subsection{Previous Work and Main Contributions}\label{sec:introduction-contributions}
The following is a list of our contributions
that are important for a thorough understanding of the Augustin information 
measures and related results that have been reported before.
\begin{enumerate}[{\bf I.}]
		\item\label{contribution:mean-zto}
	For all \(\rno\) in \((0,1)\), \cite[Lemma 34.2]{augustin78} of Augustin asserts 
	the existence of a unique probability measure \(\qmn{\rno,\mP}\) satisfying  \(\RMI{\rno}{\mP}{\Wm}=\CRD{\rno}{\Wm}{\qmn{\rno,\mP}}{\mP}\)
	and characterizes \(\qmn{\rno,\mP}\) in terms of the 
	operator\footnote{The operator \(\Aop{\rno}{\mP}{\cdot}\), defined in \eqref{eq:def:Aoperator},
		is determined uniquely by \(\rno\) and \(\mP\) and well-defined for all \(\mQ\) with finite \(\CRD{\rno}{\Wm}{\mQ}{\mP}\).} 
	\(\Aop{\rno}{\mP}{\cdot}\) as follows:
		\begin{itemize}
		\item \(\Aop{\rno}{\mP}{\qmn{\rno,\mP}}=\qmn{\rno,\mP}\) and \(\qmn{\rno,\mP}\sim \qmn{1,\mP}\).
		\item  If \(\qmn{1,\mP}\AC \mQ\) and \(\Aop{\rno}{\mP}{\mQ}=\mQ\), then \(\qmn{\rno,\mP}=\mQ\).
		\item \(\lim\nolimits_{\jnd\to\infty}\lon{\qmn{\rno,\mP}\!-\Aopi{\rno}{\mP}{\jnd}{\qmn{1,\mP}}}=0\).
		\item \(\CRD{\rno}{\Wm}{\mQ}{\mP}\geq\RMI{\rno}{\mP}{\Wm}+\RD{\rno}{\qmn{\rno,\mP}}{\mQ}\)
		for\footnote{To be precise \cite[Lemma 34.2]{augustin78} asserts the inequality 
			\(\CRD{\rno}{\Wm}{\mQ}{\mP}\geq\RMI{\rno}{\mP}{\Wm}+\tfrac{\rno}{2}\lon{\qmn{\rno,\mP}-\mQ}^{2}\)
			rather than the one given above.
			But Augustin proves the inequality given above first and then uses 
			Pinsker's inequality to establish the one given in \cite[Lemma 34.2]{augustin78}.} 
		all \(\mQ\in\pmea{\outA}\).
	\end{itemize}
	We can not verify the correctness of the proof of \cite[Lemma 34.2]{augustin78};
	we discuss our reservations in 
    Appendix \ref{sec:augustinsproof}. 
	Lemma \ref{lem:information}-(\ref{information:zto}) is 
	proved\footnote{One can prove Lemma \ref{lem:information}-(\ref{information:zto}) using 
		the ideas employed in the proof of Lemma \ref{lem:information}-(\ref{information:oti}), 
		as well.}
	relying on the ideas employed in Augustin's proof of \cite[Lemma 34.2]{augustin78}.
	Lemma \ref{lem:information}-(\ref{information:zto}) implies all assertions of 
	\cite[Lemma 34.2]{augustin78}  except for \(\lim\nolimits_{\jnd\to\infty}\lon{\qmn{\rno,\mP}\!-\Aopi{\rno}{\mP}{\jnd}{\qmn{1,\mP}}}=0\);
	Lemma \ref{lem:information}-(\ref{information:zto}) establishes 
	\(\lim\nolimits_{\jnd\to\infty}\lon{\qmn{\rno,\mP}\!-\Aopi{\rno}{\mP}{\jnd}{\qgn{\rno,\mP}}}=0\)
	instead ---see \eqref{eq:lem:information-zto:iteration-general-new}
	and 
	Remark \ref{remark:information-zto:iteration-general}.
	Unlike \cite[Lemma 34.2]{augustin78}, Lemma \ref{lem:information}-(\ref{information:zto}) also 
	bounds \(\CRD{\rno}{\Wm}{\mQ}{\mP}\) from above. This bound is new to the best of our knowledge.
	The following inequality summarizes the upper and lower bounds on \(\CRD{\rno}{\Wm}{\mQ}{\mP}\)
	established in  Lemma \ref{lem:information}-(\ref{information:zto},\ref{information:oti}):
	\begin{align}
	\label{eq:augustinslaw}
	\RD{1\vee \rno}{\qmn{\rno,\mP}}{\mQ}
	\geq
	\CRD{\rno}{\Wm}{\mQ}{\mP}-\RMI{\rno}{\mP}{\Wm}
	&\geq \RD{1 \wedge \rno}{\qmn{\rno,\mP}}{\mQ}
	&
	&\forall \mQ\in\pmea{\outA}.
	\end{align}
	For finite \(\outS\) case, the existence of a \(\mQ\) in \(\pdis{\outS}\) satisfying both 
	\(\mQ\sim\qmn{1,\mP}\) and \(\Aop{\rno}{\mP}{\mQ}=\mQ\) has been discussed by other 
	authors. 
	We make a brief digression to point out the discussion of the aforementioned existence 
	result in these works. 
	\begin{itemize}
\item While deriving the sphere packing bound for the constant composition codes 
on discrete stationary product channels, Fano implicitly asserts
the existence of a fixed point that is equivalent to \(\qmn{1,\mP}\)
for each \(\rno\) in \((0,1)\), 
see  \cite[\S 9.2, (9.24) \!\&\! p.\! 292]{fano}.
Fano, however,  does not explain why such a fixed point must exist
and does not elaborate on its uniqueness or on its relation 
to \(\qmn{\rno,\mP}\) in \cite[\S 9.2]{fano}.
		
\item While establishing the equivalence of his expression for the sphere packing 
exponent in finite \(\outS\) case to the one provided by Fano in \cite{fano}, 
Haroutunian proved the existence of a fixed point that is equivalent to \(\qmn{1,\mP}\)
for each \(\rno\) in \((0,1)\), see  \cite[(16)-(19)]{haroutunian68}.

\item While discussing the random coding bounds for discrete stationary product channels, 
Poltyrev makes an observation that is equivalent to asserting 
for each \(\rno\) in \([\sfrac{1}{2},1)\)
the existence of a fixed point that is equivalent to \(\qmn{1,\mP}\),
see	\cite[(3.15), (3.16) and Thm. 3.2]{poltyrev82}.
		Poltyrev, however, does not formulate his observations as a fixed point property.
	\end{itemize}
	In our understanding, the main conceptual contribution of \cite[Lemma 34.2]{augustin78} 
	is the characterization of the Augustin mean as a fixed point of \(\Aop{\rno}{\mP}{\cdot}\) 
	that is equivalent to \(\qmn{1,\mP}\). 
	Bounds such as the one given in \eqref{eq:augustinslaw} follow from this observation 
	via Jensen's inequality. 
	
	\item\label{contribution:mean-oti}  
	For \(\rno\in(1,\infty)\), Lemma \ref{lem:information}-(\ref{information:oti})
	establishes the existence of a unique Augustin mean \(\qmn{\rno,\mP}\)
	and proves that it satisfies \eqref{eq:augustinslaw} as well as 
	the following two assertions:
	\begin{itemize}
		\item \(\Aop{\rno}{\mP}{\qmn{\rno,\mP}}=\qmn{\rno,\mP}\) and \(\qmn{\rno,\mP}\sim \qmn{1,\mP}\).
		\item  If \(\Aop{\rno}{\mP}{\mQ}=\mQ\), then \(\qmn{\rno,\mP}=\mQ\).
	\end{itemize}
Lemma \ref{lem:information}-(\ref{information:oti}) is new to the best of our knowledge.
For \(\rno\in(1,\infty)\) case, 
neither the characterization of \(\qmn{\rno,\mP}\) in terms of \(\Aop{\rno}{\mP}{\cdot}\), 
nor the inequalities given in \eqref{eq:augustinslaw} have been reported before, even for finite \(\outS\) case.

	\item\label{contribution:informationO:differentiability}
	\(\RMI{\rno}{\mP}{\Wm}\) is a continuously differentiable function of \(\rno\) from 
	\(\reals{+}\) to \([0,\bent{\mP}]\) by 
	Lemma \ref{lem:informationO}-(\ref{informationO:differentiability}).
	
	\item\label{contribution:capacity:minimax} 
	The following minimax identity is established in Theorem \ref{thm:minimax} 
	for any convex constraint set \(\cset\)
	\begin{align}
	\notag
	\sup\nolimits_{\mP \in \cset} \inf\nolimits_{\mQ \in \pmea{\outA}} 
	\CRD{\rno}{\Wm}{\mQ}{\mP}
	&=
	\inf\nolimits_{\mQ \in \pmea{\outA}} \sup\nolimits_{\mP \in \cset} 
	\CRD{\rno}{\Wm}{\mQ}{\mP}.
	\end{align}
	Theorem \ref{thm:minimax} establishes the existence of a unique Augustin center, \(\qmn{\rno,\Wm,\cset}\), 
	for any convex \(\cset\) with finite Augustin capacity
	and the convergence of \(\{\qmn{\rno,\pma{}{(\ind)}}\}_{\ind\in\integers{+}}\) 
	to \(\qmn{\rno,\Wm,\cset}\) in total variation topology 
	for any \(\{\pma{}{(\ind)}\}_{\ind\in\integers{+}}\subset\cset\) satisfying
	\(\lim_{\ind\to\infty}\RMI{\rno}{\pma{}{(\ind)}}{\Wm}=\CRC{\rno}{\Wm}{\cset}\). 
	Augustin proved this result only for \(\rno\)'s  in \((0,1]\) and the constraint sets determined by 
	cost constraints, see \cite[Lemma 34.7]{augustin78}.
	For \(\cset=\pdis{\inpS}\) case similar results were proved 
	by \csiszar \cite[Proposition 1]{csiszar95} assuming both \(\inpS\) and \(\outS\) are 
	finite sets
	and
	by van Erven and \harremoes \cite[Thm. 34]{ervenH14} assuming \(\outS\) is a finite set. 
	\item\label{contribution:capacity:EHB} 
	The following bound in terms of the Augustin capacity and center
	established in Lemma \ref{lem:EHB} is new to the best of our knowledge
	\begin{align}
	\notag
	\sup\nolimits_{\mP \in \cset} \CRD{\rno}{\Wm}{\mQ}{\mP}
	&\geq  \CRC{\rno}{\!\Wm\!}{\cset}+\RD{\rno\wedge 1}{\qmn{\rno,\!\Wm\!,\cset}}{\mQ}
	&
	&\forall \mQ \in \pmea{\outA}.
	\end{align} 	
	A similar  bound has been conjectured by van Erven and \harremoes in \cite{ervenH14}.
	For the \renyi capacity and center, 
	we have proved that conjecture and extended it to the constrained 
	case elsewhere, see \cite[Lemmas \ref*{A-lem:EHB} \& \ref*{A-lem:CEHB}]{nakiboglu19A}.
			
	\item\label{contribution:cost-AL} 
	The Augustin-Legendre information \(\RMIL{\rno}{\mP}{\Wm}{\lgm}\),
	defined as \(\RMI{\rno}{\mP}{\Wm}-\lgm\cdot\EXS{\mP}{\costf}\),
	as well as the resulting capacity, center, and radius are new 
	concepts that have not been studied before, except for \(\rno=1\) case.
	Thus, formally speaking, all of the propositions in \S\ref{sec:cost-AL}
	are new.
	The analysis presented in \S\ref{sec:cost-AL} is a standard application of 
	the convex conjugation techniques to characterize the cost constrained Augustin 
	capacity and center.  
	A similar analysis for \(\rno=1\) case is provided by \csiszar and \korner in 
	\cite[Ch. 8]{csiszarkorner} for discrete channels with a single cost constraint.
	The most important conclusions of the analysis presented in \S\ref{sec:cost-AL}
	are the followings:
	\begin{itemize}
		\item \(\RCL{\rno}{\Wm}{\lgm}\), defined as \(\sup_{\mP\in\pdis{\inpS}}\RMIL{\rno}{\mP}{\Wm}{\lgm}\),
		satisfies \(\RCL{\rno}{\Wm}{\lgm}=\sup\nolimits_{\costc\geq0} \CRC{\rno}{\Wm}{\costc}-\lgm\cdot\costc\)
		for all \(\lgm\in\reals{\geq0}^{\ell}\) by \eqref{eq:Lcapacity-astheconjugate}.

		\item \(\CRC{\rno}{\!\Wm\!}{\costc}=\inf\nolimits_{\lgm\geq0} \RCL{\rno}{\Wm}{\lgm}+\lgm\cdot\costc\)
		for all \(\costc\in\inte{\fcc{\costf}}\) and 
		the set of \(\lgm\)'s achieving this infimum form a non-empty 
		convex compact set whenever \(\CRC{\rno}{\!\Wm\!}{\costc}\) is finite by Lemma \ref{lem:Lcapacity}.
		
		\item \(\RCL{\rno}{\Wm}{\lgm}=\RRL{\rno}{\Wm}{\lgm}\) 
		where \(\RRL{\rno}{\Wm}{\lgm}\) is defined as 
		\(\inf_{\mQ\in\pmea{\outA}}\sup_{\dinp\in\inpS} \RD{\rno}{\Wm(\dinp)}{\mQ}-\lgm\cdot \costf(\dinp) \)
		by Theorem \ref{thm:Lminimax}.
		\item  If \(\RCL{\rno}{\Wm}{\lgm}<\infty\), then there exists a unique A-L center 
		\(\qma{\rno,\Wm}{\lgm}\) satisfying 
		\(\RCL{\rno}{\Wm}{\lgm}=\sup\nolimits_{\dinp \in \inpS}\RD{\rno}{\Wm(\dinp)}{\qma{\rno,\Wm}{\lgm}}-\lgm\cdot\costf(\dinp)\).
		Furthermore,
		\(\lim_{\ind \to\infty}\lon{\qmn{\rno,\mP}-\qma{\rno,\Wm}{\lgm}}=0\)
		for all \(\{\pma{}{(\ind)}\}_{\ind\in\integers{+}}\subset\pdis{\inpS}\)
		satisfying 
		\(\lim_{\ind \to\infty} \RMIL{\rno}{\pma{}{(\ind)}}{\Wm}{\lgm}=\RCL{\rno}{\Wm}{\lgm}\) 
		by Theorem \ref{thm:Lminimax}.
		
		\item  If \(\CRC{\rno}{\!\Wm\!}{\costc}=\RCL{\rno}{\Wm}{\lgm}+\lgm\cdot\costc<\infty\) 
		for a \(\lgm\in \reals{\geq0}^{\ell}\), then \(\qmn{\rno,\!\Wm\!,\costc}=\qma{\rno,\Wm}{\lgm}\)
		by Lemma \ref{lem:Lcenter}.
		
		\item  If \(\Wmn{[1,\blx]}\) is a product channel with an additive cost function, then 
		\(\RCL{\rno}{\Wmn{[1,\blx]}}{\lgm}=\sum\nolimits_{\tin=1}^{\blx} \RCL{\rno}{\Wmn{\tin}}{\lgm}\)
		for all \(\lgm \in \reals{\geq0}^{\ell},~\rno\in \reals{+}\) and
		whenever either of them 
%		\(\qma{\rno,\Wmn{[1,\blx]}}{\lgm}\) or \(\bigotimes\nolimits_{\tin=1}^{\blx} \qma{\rno,\Wmn{\tin}}{\lgm}\)
		exists \(\qma{\rno,\Wmn{[1,\blx]}}{\lgm}\) is equal to
		\(\bigotimes\nolimits_{\tin=1}^{\blx} \qma{\rno,\Wmn{\tin}}{\lgm}\) by Lemma \ref{lem:Lcapacityproduct}.
	\end{itemize}

	\item\label{contribution:cost-RG} 	
	The \renyi\!\!-Gallager information \(\GMIL{\rno}{\mP}{\Wm}{\lgm}\) is
	a generalization of the \renyi information \(\GMI{\rno}{\mP}{\Wm}\) with a Lagrange 
	multiplier because \(\GMIL{\rno}{\mP}{\Wm}{0}=\GMI{\rno}{\mP}{\Wm}\).
	This quantity was first employed by Gallager in \cite{gallager65} by a different 
	parametrization and scaling; later considered by 
	Arimoto \cite[\S IV]{arimoto76}, Augustin \cite{augustin78}, Ebert \cite{ebert66},
	Richters \cite{richters67},
	Oohama \cite{oohama17A}, \cite{oohama17B}, 
	and Vazquez-Vilar, Martinez, and {F{\`{a}}bregas} \cite{vazquezMF15}
	 with various parametrizations, scalings,
	and	names. 
	We chose the scaling and the parametrization so as to be compatible 
	with the ones for Augustin-Legendre information. 
	The most important conclusions of our analysis in \S\ref{sec:cost-RG} are 
	the followings:
	\begin{itemize}
		\item \(\GCL{\rno}{\Wm}{\lgm}=\RRL{\rno}{\Wm}{\lgm}\) by Theorem \ref{thm:Gminimax}, where
		\(\GCL{\rno}{\Wm}{\lgm}\) is defined as \(\sup_{\mP\in\pdis{\inpS}}\GMIL{\rno}{\mP}{\Wm}{\lgm}\).
		\item If \(\RCL{\rno}{\Wm}{\lgm}<\infty\) 
		and \(\lim_{\ind \to\infty} \GMI{\rno}{\pma{}{(\ind)}}{\Wm}=\RCL{\rno}{\Wm}{\lgm}\),
		then \(\lim_{\ind \to\infty}\lon{\qga{\rno,\mP}{\lgm}-\qma{\rno,\Wm}{\lgm}}=0\)
		by Theorem \ref{thm:Gminimax}.
		\item
		\(\sup\nolimits_{\dinp \in \inpS} \RD{\rno}{\Wm(\dinp)}{\mQ}-\lgm \cdot\costf(\dinp)
		\geq \RCL{\rno}{\Wm}{\lgm}+\RD{\rno}{\qma{\rno,\Wm}{\lgm}}{\mQ}\)
		for all \(\mQ \in \pmea{\outA}\) by Lemma \ref{lem:GEHB}.
	\end{itemize}
	Lemma \ref{lem:GEHB} is new to the best of our knowledge.
	For the case when both \(\rno\in(0,1)\) and \(\vee_{\dinp\in\inpS}\costf(\dinp)\leq \uc\), 
	Theorem \ref{thm:Gminimax} is implied by \cite[Lemma 35.2]{augustin78}.
	
	While pursuing a similar analysis in \cite[\S 35]{augustin78}, Augustin assumed 
	the cost function to be bounded. This hypothesis, however, excludes certain 
	important and interesting cases such as the Gaussian channels.
	The issue here is not a matter of rescaling: certain conclusions of Augustin's
	analysis,
	e.g. \cite[Lemma 35.3-(a)]{augustin78}, 
	are not correct when the cost function is unbounded.
	We do not assume the cost function to be bounded.
	Thus our model subsumes not only Augustin's model in \cite[\S 35]{augustin78}  
	but also other previously considered models,
	which were either discrete 
	\cite[pp. 13-15]{gallager65}, \cite[\S7.3]{gallager}, 
	\cite[\S IV]{arimoto76}, \cite{oohama17B}, \cite{vazquezMF15}
	or Gaussian 
	\cite[pp. 15,16]{gallager65}, \cite[\S\S7.4,7.5]{gallager},\cite{ebert66}, \cite{richters67}, \cite{oohama17A}.

	\item\label{contribution:cost-TP} 
	For channels with uncountable input sets the Shannon information and capacity 
	is often defined via the probability measures on the input space \((\inpS,\inpA)\), rather 
	than the probability mass functions on the input set \(\inpS\).
	In \S\ref{sec:cost-TP}, we discuss how and under which conditions one can make such a 
	generalization for Augustin's information measures.
	The most important conclusions of our analysis are the followings:
	\begin{itemize}
	\item If \(\Wm\) is a transition probability \((\inpS,\inpA)\) to \((\outS,\outA)\) 
	---i.e. \(\Wm\in\pmea{\outA|\inpA}\)--- and 
	\(\outA\) is countably generated, then
	\begin{itemize}
		\item \(\RMI{\rno}{\mP}{\Wm}\) is
		well defined for all \(\rno\in\reals{+}\) and \(\mP\in\pmea{\inpA}\)
		by \eqref{eq:def:general-conditionaldivergence}, \eqref{eq:def:general-information},
		and Lemma \ref{lem:divergence-measurability}  
		\item \(\RMIL{\rno}{\mP}{\Wm}{\lgm}\) is
		well defined for all \(\rno\in\reals{+}\), \(\mP\in\pmea{\inpA}\), and 
		\(\lgm\in\reals{\geq0}^{\ell}\) 
		by \eqref{eq:def:general-Linformation} provided that 
		\(\costf\) is \(\inpA\)-measurable.
	\end{itemize}
		
	\item If \(\Wm\in\pmea{\outA|\inpA}\), \(\inpA\) is countably separated, 
	\(\outA\) is countably generated, and \(\costf\) is \(\inpA\)-measurable, 
	then 
	\begin{itemize}
		\item \(\RCL{\rno}{\Wm}{\lgm}=\sup\nolimits_{\mP\in \csetA^{\lgm}} \RMIL{\rno}{\mP}{\Wm}{\lgm}\)
		for all \(\lgm\) in \(\reals{\geq0}^{\ell}\)
		by Theorem \ref{thm:general-Lminimax}
		where \(\csetA^{\lgm}\) is defined as \(\{\mP\in\pmea{\inpA}:\lgm\cdot\EXS{\mP}{\costf}<\infty \}\).
		\item If \(\RCL{\rno}{\Wm}{\lgm}<\infty\) for a \(\lgm\) in \(\reals{\geq0}^{\ell}\), 
		then
		\(\RCL{\rno}{\Wm}{\lgm}=\sup\nolimits_{\mP \in \csetA^{\lgm}} \CRD{\rno}{\Wm}{\qma{\rno,\Wm}{\lgm}}{\mP}-\lgm\cdot\EXS{\mP}{\costf}\)
		by Theorem \ref{thm:general-Lminimax}.
		\item \(\CRC{\rno}{\Wm}{\costc}=\sup\nolimits_{\mP\in \csetA(\costc)} \RMI{\rno}{\mP}{\Wm}\)
		for all \(\costc\) in \(\inte{\fcc{\costf}}\)
		by Theorem \ref{thm:general-minimax}
		where \(\csetA(\costc)\) is defined as \(\{\mP\!\in\!\pmea{\inpA}: \EXS{\mP}{\costf}\leq \costc\}\).
		\item If \(\CRC{\rno}{\Wm}{\costc}<\infty\) for a \(\costc\) in \(\inte{\fcc{\costf}}\), 
		then
		\(\CRC{\rno}{\Wm}{\costc}=\sup\nolimits_{\mP \in \csetA(\costc)} \CRD{\rno}{\Wm}{\qmn{\rno,\Wm,\costc}}{\mP}\)
		by Theorem \ref{thm:general-minimax}.		
	\end{itemize}
	Thus the A-L capacity and center as well as the cost constrained Augustin capacity and center 
	defined via probability mass functions are equal to the corresponding quantities that might 
	be defined via probability measures on \((\inpS,\inpA)\), provided that \(\inpA\) is countably 
	separated and \(\outA\) is countably generated. 
\end{itemize}
\end{enumerate}

\section{Preliminaries}\label{sec:preliminary}
\subsection{The \renyi Divergence}\label{sec:preliminary-divergence}
\begin{definition}\label{def:divergence}
For any \(\rno\in\reals{+}\) and \(\mW,\mQ\in\fmea{\outA}\)
\emph{the order \(\rno\) \renyi divergence between \(\mW\) and \(\mQ\)} is
\begin{align}
\label{eq:def:divergence}
\RD{\rno}{\mW}{\mQ}
&\DEF \begin{cases}
\tfrac{1}{\rno-1}\ln \int (\der{\mW}{\rfm})^{\rno} (\der{\mQ}{\rfm})^{1-\rno} \rfm(\dif{\dout})
&\rno\neq 1\\
\int  \der{\mW}{\rfm}\left[ \ln\der{\mW}{\rfm} -\ln \der{\mQ}{\rfm}\right] \rfm(\dif{\dout})
&\rno=1
\end{cases}
\end{align}
where \(\rfm\) is any measure satisfying \(\mW\AC\rfm\) and \(\mQ\AC\rfm\).
\end{definition}

Customarily, the \renyi divergence is defined for pairs of probability measures 
---see \cite{ervenH14} and \cite{renyi61} for example--- rather than pairs of 
non-zero finite measures. 
We adopt this slightly more general definition because it allows us to use 
the \renyi divergence to express certain observations more succinctly, 
see Lemma \ref{lem:divergence-RM} in the following and  \S\ref{sec:cost-RG}. 
For pairs of probability measures Definition \ref{def:divergence} is equivalent 
to usual definition employed in \cite{ervenH14} by \cite[Thm. 5]{ervenH14}. 

\begin{lemma}[{\!\!\cite[Lemma \ref*{A-lem:divergence-RM}]{nakiboglu19A}}]
	\label{lem:divergence-RM}
	Let \(\rno\) be a positive real number and	\(\mW\), \(\mQ\), \(\mV\) be 
	non-zero finite measures on \((\outS,\outA)\).
	\begin{itemize}
		\item If \(\mV\leq\mQ\), then \(\RD{\rno}{\mW}{\mQ}\leq \RD{\rno}{\mW}{\mV}\).
\item  If \(\mQ=\gamma\mV\) for some \(\gamma\in\reals{+}\)
and 
either \(\mW\) is a probability measure
or \(\rno\neq1\), then 
\(\RD{\rno}{\mW}{\mQ}=\RD{\rno}{\mW}{\mV}-\ln \gamma\).
	\end{itemize} 
\end{lemma}

If both arguments of the \renyi divergence are probability measures, then it is positive
unless the arguments are equal to one another by Lemma \ref{lem:divergence-pinsker}.
\begin{lemma}[{\!\!\cite[Thm. 3, Thm. 31]{ervenH14}}]\label{lem:divergence-pinsker}
	For any \(\rno\in\reals{+}\), probability measure \(\mW\) and \(\mQ\) on \((\outS,\outA)\)
	\begin{align}
\notag %\label{eq:lem:divergence-pinsker}
	\RD{\rno}{\mW}{\mQ}
	&\geq\tfrac{1\wedge\rno}{2} \lon{\mW-\mQ}^2.
	\end{align}
\end{lemma}
For orders in \((0,1]\) this inequality is called the Pinsker's inequality, 
\cite{csiszar67A}, \cite{gilardoni10B}.
%\cite[Lemma 26.5b]{augustin78}, \cite[Theorem 4.1]{csiszar67A}, \cite[Corollary 6]{gilardoni10B}.
For orders in \((0,1)\) it is possible to bound the \renyi divergence
from above in terms of the total variation distance.
For \(\rno=\sfrac{1}{2}\) case \cite[eq. (21), p. 364]{shiryaev} 
asserts
\begin{align}
\label{eq:shiryaev}
\RD{\sfrac{1}{2}}{\mW}{\mQ}
&\leq 2 \ln \tfrac{2}{2-\lon{\mW-\mQ}}.
\end{align}

As a function of its arguments, the order \(\rno\) \renyi divergence 
is continuous for the total variation topology provided that 
\(\rno\in(0,1)\). For arbitrary orders we only have lower semicontinuity,
but that holds even for the topology of setwise convergence.  
\begin{lemma}[{\!\!\cite[Thm. 15]{ervenH14}}]\label{lem:divergence:lsc}
	For any \(\rno\in\reals{+}\), \(\RD{\rno}{\mW}{\mQ}\) is a 
	lower semicontinuous function of the pair of probability measures 
	\((\mW,\mQ)\) in the topology of setwise convergence.
\end{lemma}

\begin{lemma}[{\!\!\cite[Thm. 17]{ervenH14}}]\label{lem:divergence:uc}
	For any \(\rno\in(0,1)\),  \(\RD{\rno}{\mW}{\mQ}\) is a 
	uniformly continuous function of the pair of probability measures
	\((\mW,\mQ)\) in the total variation topology.
\end{lemma}
The \renyi divergence is convex in its second argument for all positive orders, 
jointly convex in its arguments for positive orders that are not greater than one,
and jointly quasi-convex in its arguments for all positive orders. 

\begin{lemma}[{\!\!\cite[Thm.  12]{ervenH14}}]\label{lem:divergence-convexity}
For all \(\rno\in\reals{+}\), \(\mW,\qmn{0},\qmn{1}\in\pmea{\outA}\), \(\beta\in(0,1)\),
and \(\rfm\) satisfying \((\qmn{0}+\qmn{1})\AC\rfm\),
\begin{align}
\notag %\label{eq:lem:divergence-convexity-second}
\RD{\rno}{\mW}{\beta\qmn{1}+(1-\beta)\qmn{0}} 
&\leq \beta\RD{\rno}{\mW}{\qmn{1}}+(1-\beta)\RD{\rno}{\mW}{\qmn{0}}.
\end{align}
Furthermore, the equality holds iff 
\(\der{\qmn{1}}{\rfm}=\der{\qmn{0}}{\rfm}\) \(\mW\)-almost surely.
\end{lemma}

\begin{lemma}[{\!\!\cite[Thm.  11]{ervenH14}}]\label{lem:divergence-jointconvexity}
For all \(\rno\in(0,1]\), \(\wmn{0},\wmn{1},\qmn{0},\qmn{1}\in\pmea{\outA}\), 
\(\beta\in(0,1)\), and \(\rfm\) satisfying \((\wmn{0}+\wmn{1}+\qmn{0}+\qmn{1})\AC\rfm\),
\begin{align}
\label{eq:divergence-convexity-zo}
\RD{\rno}{\beta\wmn{1}+(1-\beta)\wmn{0}}{\beta\qmn{1}+(1-\beta)\qmn{0}} 
&\leq \beta\RD{\rno}{\wmn{1}}{\qmn{1}}+(1-\beta)\RD{\rno}{\wmn{0}}{\qmn{0}}.
\end{align}
Furthermore,
for \(\rno=1\) the equality holds iff 
\(\der{\wmn{0}}{\rfm}\der{\qmn{1}}{\rfm}=\der{\wmn{1}}{\rfm}\der{\qmn{0}}{\rfm}\) 
and
for \(\rno\in(0,1)\) the equality holds iff 
\(\der{\wmn{0}}{\rfm}\der{\qmn{1}}{\rfm}=\der{\wmn{1}}{\rfm}\der{\qmn{0}}{\rfm}\) 
and \(\RD{\rno}{\wmn{1}}{\qmn{1}}=\RD{\rno}{\wmn{0}}{\qmn{0}}\).
\end{lemma}

\begin{lemma}[{\!\!\cite[Thm. 13]{ervenH14}}]\label{lem:divergence-quasiconvexity}
For all \(\rno\in\reals{+}\), \(\wmn{0},\wmn{1},\qmn{0},\qmn{1}\in\pmea{\outA}\), 
and \(\beta\in(0,1)\)
\begin{align}
\notag %\label{eq:lem:divergence-quasiconvexity}
\RD{\rno}{\beta\wmn{1}+(1-\beta)\wmn{0}}{\beta\qmn{1}+(1-\beta)\qmn{0}} 
&\leq \RD{\rno}{\wmn{1}}{\qmn{1}}\vee \RD{\rno}{\wmn{0}}{\qmn{0}}.
\end{align}
\end{lemma}

\begin{lemma}[{\!\!\cite[Thm. 3, Thm. 7]{ervenH14}}]\label{lem:divergence-order}
	For all \(\mW,\mQ\in\pmea{\outA}\), \(\RD{\rno}{\mW}{\mQ}\) is a 
	nondecreasing and lower semicontinuous function of \(\rno\) on \(\reals{+}\)
	that is continuous on \((0,(1\vee\chi_{\mW,\mQ})]\) where
	\(\chi_{\mW,\mQ}\DEF\sup\{\rno:\RD{\rno}{\mW}{\mQ}<\infty\}\).
\end{lemma}
Since \(\RD{\rno}{\mW}{\mQ}=\tfrac{\rno}{1-\rno}\RD{1-\rno}{\mQ}{\mW}\) for 
all \(\rno\in(0,1)\),
Lemma \ref{lem:divergence-order} and \eqref{eq:shiryaev} imply
\begin{align}
\notag
\RD{\rno}{\mW}{\mQ}
&\leq \begin{cases}
\RD{\sfrac{1}{2}}{\mW}{\mQ}
&\mbox{if~}\rno\in(0,\sfrac{1}{2}]
\\
\tfrac{\rno}{1-\rno}\RD{\sfrac{1}{2}}{\mW}{\mQ}
&\mbox{if~}\rno\in(\sfrac{1}{2},1)
\end{cases}
\\
&\leq \tfrac{2}{1-\rno} \ln \tfrac{2}{2-\lon{\mW-\mQ}}
&
&\forall \rno\in(0,1).
\label{eq:shiryaev-new}
\end{align}
For a slightly tighter bound, see \cite[eq. (24), p. 365]{shiryaev}.

If \(\alg{G}\) is a sub-\(\sigma\)-algebra of \(\outA\),
then for any  \(\mW\) and \(\mQ\) in \(\pmea{\outA}\) 
the identities  \(\wmn{|\alg{G}}(\oev)\!=\!\mW(\oev)\) for all 
\(\oev\!\in\!\alg{G}\) 
and \(\qmn{|\alg{G}}(\oev)\!=\!\mQ(\oev)\) for all \(\oev\!\in\!\alg{G}\) uniquely define 
probability measures \(\wmn{|\alg{G}}\) and \(\qmn{|\alg{G}}\) on \((\outS,\alg{G})\).
We denote \(\RD{\rno}{\wmn{|\alg{G}}}{\qmn{|\alg{G}}}\) by 
\(\RDF{\rno}{\alg{G}}{\mW}{\mQ}\). 
\begin{lemma}[{\!\!\cite[Thm. 21]{ervenH14}}]\label{lem:divergence:ISAlimit}
	Let \(\outA_{1}\subset\outA_{2}\subset\cdots\subset \outA\) be an increasing family
	of \(\sigma\)-algebras, and let \(\outA_{\infty}=\gsal{\cup_{\ind=1}^{\infty}\outA_{\ind}}\)	
	be the smallest \(\sigma\)-algebra containing them. Then for any order \(\rno\in\reals{+}\)
	\begin{align}
	\notag
	\lim\nolimits_{\ind \to \infty}
	\RDF{\rno}{\outA_{\ind}}{\mW}{\mQ}
	&=\RDF{\rno}{\outA_{\infty}}{\mW}{\mQ}.
	\end{align}
\end{lemma}

\subsection{Tilted Probability Measure}\label{sec:preliminary-tiltedprobabilitymeasure}
\begin{definition}\label{def:tiltedprobabilitymeasure}
For any \(\rno\in\reals{+}\) and \(\mW,\mQ\in\pmea{\outA}\) satisfying 
\(\RD{\rno}{\mW}{\mQ}<\infty\), 
\emph{the order \(\rno\) tilted probability measure} \(\wma{\rno}{\mQ}\) is 
	\begin{align}
\notag %	\label{eq:def:tiltedprobabilitymeasure}
	\der{\wma{\rno}{\mQ}}{\rfm}
	&\DEF e^{(1-\rno)\RD{\rno}{\mW}{\mQ}}(\der{\mW}{\rfm})^{\rno} (\der{\mQ}{\rfm})^{1-\rno}.
	\end{align}
\end{definition}

Note that \(\wma{1}{\mQ}=\mW\) for any \(\mQ\) satisfying \(\RD{1}{\mW}{\mQ}<\infty\).
For other orders one can confirm the following identity by substitution:
if \(\RD{\rno}{\mW}{\mQ}<\infty\), then for any \(\mV\in\pmea{\outA}\) satisfying 
both \(\RD{1}{\mV}{\mW}<\infty\) and \(\RD{1}{\mV}{\mQ}<\infty\) also satisfies
\begin{align}
\notag
\tfrac{1}{1-\rno}\RD{1}{\mV}{\wma{\rno}{\mQ}}
+\RD{\rno}{\mW}{\mQ}
&=\tfrac{\rno}{1-\rno}\RD{1}{\mV}{\mW}+\RD{1}{\mV}{\mQ}.
\end{align}
This identity is used to derive the following variational characterization of
the \renyi divergence for orders other than one.
\begin{lemma}[{\!\!\cite[Thm. 30]{ervenH14}}]\label{lem:variational}
For any \(\mW,\mQ\in\pmea{\outA}\)	
	\begin{align}
\notag %	\label{eq:lem:variational}
	\RD{\rno}{\mW}{\mQ}
	&=\begin{cases}
	\inf\nolimits_{\mV\in\pmea{\outA}}
	\tfrac{\rno}{1-\rno}\RD{1}{\mV}{\mW}+\RD{1}{\mV}{\mQ}
	&\rno\in(0,1)
	\\
	\sup\nolimits_{\mV\in\pmea{\outA}}
	\tfrac{\rno}{1-\rno}\RD{1}{\mV}{\mW}+\RD{1}{\mV}{\mQ}
	&\rno\in(1,\infty)
	\end{cases}
	\end{align}
	where \(\tfrac{\rno}{1-\rno}\RD{1}{\mV}{\mW}+\RD{1}{\mV}{\mQ}\) stands for \(-\infty\) 
	when \(\rno\in(1,\infty)\) and \(\RD{1}{\mV}{\mW}=\RD{1}{\mV}{\mQ}=\infty\).
	Furthermore, if \(\RD{\rno}{\mW}{\mQ}\) is finite and 
	either \(\rno\in(0,1)\) or \(\RD{1}{\wma{\rno}{\mQ}}{\mW}<\infty\), then
	\begin{align}
	\label{eq:lem:variational-tilted}
	\RD{\rno}{\mW}{\mQ}
	&=\tfrac{\rno}{1-\rno}\RD{1}{\wma{\rno}{\mQ}}{\mW}+\RD{1}{\wma{\rno}{\mQ}}{\mQ}.
	\end{align}
\end{lemma}

We have observed in Lemma \ref{lem:divergence-order} that 
\(\RD{\rno}{\mW}{\mQ}\)
is continuous in \(\rno\) on the closure of the interval that it is finite. 
Lemma \ref{lem:analyticity}, in the following, establishes 
the analyticity of \(\RD{\rno}{\mW}{\mQ}\) in \(\rno\) on the interior of the interval 
that \(\RD{\rno}{\mW}{\mQ}\) is finite. 
Lemma \ref{lem:analyticity} also establishes the analyticity 
---and hence the finiteness---
of \(\RD{1}{\wma{\rno}{\mQ}}{\mW}\) and \(\RD{1}{\wma{\rno}{\mQ}}{\mQ}\) 
on the same interval.
This allows us to assert the validity 
of \eqref{eq:lem:variational-tilted} on the same interval:
\begin{align}
\notag %\label{eq:tilted}
\RD{\rno}{\mW}{\mQ}
&=\tfrac{\rno}{1-\rno}\RD{1}{\wma{\rno}{\mQ}}{\mW}+\RD{1}{\wma{\rno}{\mQ}}{\mQ}
&
&\forall \rno\in(0,\chi_{\mW,\mQ}).
\end{align}
 
\begin{lemma}\label{lem:analyticity}
For any \(\mW,\mQ\in\pmea{\outA}\) satisfying
 \(\chi_{\mW,\mQ}>0\), for \(\chi_{\mW,\mQ}\DEF\sup\{\rno:\RD{\rno}{\mW}{\mQ}<\infty\}\),
 \(\RD{\rno}{\mW}{\mQ}\), \(\RD{1}{\wma{\rno}{\mQ}}{\mW}\),
and
\(\RD{1}{\wma{\rno}{\mQ}}{\mQ}\) 
are analytic functions of \(\rno\)
on \((0,\chi_{\mW,\mQ})\).
Furthermore,	
\begin{align}
\label{eq:lem:analyticity-derivative}
\left.\pder{^{\knd}\RD{\rno}{\mW}{\mQ}}{\rno^{\knd}}\right\vert_{\rno=\rnf}
&=\begin{cases}
\knd!
\sum\limits_{\tin=0}^{\knd}
\tfrac{(-1)^{\knd-\tin}}{(\rnf-1)^{\knd-\tin+1}}
\GX_{\mW,\mQ}^{\tin}(\rnf)
&\rnf\neq 1 
\\
\knd!
\GX_{\mW,\mQ}^{\knd+1}(1)
&\rnf=1
\end{cases}
\end{align}
where \(\GX_{\mW,\mQ}^{\tin}(\rnf)\) is 
defined in terms of the set 
\(\set{J_{\tin}}\) as follows
\begin{align}
\label{eq:lem:def:set}
\set{J_{\tin}}
&\DEF\{(\jnd_{1},\jnd_{2},\ldots,\jnd_{\tin}): \jnd_{\ind}\in\integers{\geq0}~\forall \ind\mbox{~and~}
1\jnd_{1}+2\jnd_{2}+\ldots+\tin \jnd_{\tin}=\tin\},
\\
\label{eq:lem:def:G}
\GX_{\mW,\mQ}^{\tin}(\rnf)
&\DEF\begin{cases}
(\rnf-1)\RD{\rnf}{\mW}{\mQ}
&\tin=0
\\
\sum\nolimits_{\set{J}_{\tin}} 
\tfrac{-(\jnd_{1}+\jnd_{2}+\cdots +\jnd_{\tin}-1)!}{\jnd_{1}! \jnd_{2}!\ldots \jnd_{\tin}!}
\prod\nolimits_{\ind=1}^{\tin}
\left(\tfrac{(-1)}{\ind!}\EXS{\wma{\rnf}{\mQ}}{
	\left(\ln \der{\mW}{\rfm}-\ln \der{\mQ}{\rfm}\right)^{\ind}}\right)^{\jnd_{\ind}}
&\tin\in\integers{+}
\end{cases}.
\end{align}
\end{lemma}

Lemma \ref{lem:analyticity} is new to the best of our knowledge; it is proved in 
Appendix \ref{sec:analyticityproofs}
using standard results on the continuity and differentiability of parametric integrals and 
Fa\`{a} di Bruno formula for derivatives of compositions of smooth functions.

Note that 
\(\set{J}_{1}=\{(1)\}\),
\(\set{J}_{2}=\{(2,0),(0,1)\}\),
and
\(\set{J}_{3}=\{(3,0,0),(1,1,0),(0,0,1)\}\).
Thus one can confirm using \eqref{eq:lem:def:G} by substitution that
\begin{align}
\notag
\GX_{\mW,\mQ}^{1}(\rnf)
&=\EXS{\wma{\rnf}{\mQ}}{\cln{}}
\\
\notag
\GX_{\mW,\mQ}^{2}(\rnf)
&=\tfrac{1}{2}\left(\EXS{\wma{\rnf}{\mQ}}{\cln{}^{2}}-\EXS{\wma{\rnf}{\mQ}}{\cln{}}^{2}\right)
%\\
%\notag
&
&=\tfrac{1}{2}\EXS{\wma{\rnf}{\mQ}}{\left(\cln{}-\EXS{\wma{\rnf}{\mQ}}{\cln{}}\right)^{2}}
\\
\notag
\GX_{\mW,\mQ}^{3}(\rnf)
&=\tfrac{1}{3!}\EXS{\wma{\rnf}{\mQ}}{\cln{}^{3}} 
-\tfrac{1}{2}\EXS{\wma{\rnf}{\mQ}}{\cln{}^{2}}\EXS{\wma{\rnf}{\mQ}}{\cln{}} 
+\tfrac{1}{3}\EXS{\wma{\rnf}{\mQ}}{\cln{}}^{3} 
%\\
%\notag
&
&=\tfrac{1}{3!}\EXS{\wma{\rnf}{\mQ}}{\left(\cln{}-\EXS{\wma{\rnf}{\mQ}}{\cln{}}\right)^{3}}
\end{align}
where \(\cln{}=\ln\der{\mW}{\rfm}-\ln \der{\mQ}{\rfm}\).
If we substitute these expressions for 
\(\GX_{\mW,\mQ}^{1}(\rnf)\), \(\GX_{\mW,\mQ}^{2}(\rnf)\), and \(\GX_{\mW,\mQ}^{3}(\rnf)\)
in \eqref{eq:lem:analyticity-derivative}
and use the identity \(\cln{}=\tfrac{1}{\rnf-1}\left(\ln\der{\wma{\rnf}{\mQ}}{\mW}+
\GX_{\mW,\mQ}^{0}(\rnf)\right)\) which holds \(\wma{\rnf}{\mQ}\)-almost surely for
\(\rnf\in(0,\chi_{\mW,\mQ})\setminus\{1\}\),  
we get the following more succinct expressions for the first two derivatives
of \(\RD{\rno}{\mW}{\mQ}\) with respect to \(\rno\):
\begin{align}
\label{eq:divergence-derivative-first}
\left.\pder{}{\rno}\RD{\rno}{\mW}{\mQ}\right\vert_{\rno=\rnf}
&=\begin{cases}
\tfrac{1}{(\rnf-1)^2}\RD{1}{\wma{\rnf}{\mQ}}{\mW}
&\rnf\neq 1
\\
\tfrac{1}{2}
\EXS{\mW}{\left(\ln \der{\mW}{\mQ}-\RD{1}{\mW}{\mQ}\right)^{2}}
&\rnf= 1
\end{cases},
\\
\label{eq:divergence-derivative-second}
\left.\pder{^{2}}{\rno^{2}}\RD{\rno}{\mW}{\mQ}\right\vert_{\rno=\rnf}
&=\begin{cases}
\tfrac{1}{(\rnf-1)^{3}}\left(
\EXS{\wma{\rnf}{\mQ}}{\left(\ln\der{\wma{\rnf}{\mQ}}{\mW}\right)^{2}}
-2\RD{1}{\wma{\rnf}{\mQ}}{\mW}
-\left[\RD{1}{\wma{\rnf}{\mQ}}{\mW}\right]^{2}
\right)
&\rnf\neq 1
\\
\tfrac{1}{3}
\EXS{\mW}{\left(\ln \der{\mW}{\mQ}-\RD{1}{\mW}{\mQ}\right)^{3}}
&\rnf=1
\end{cases}.
\end{align}

Analyticity of \(\RD{\rno}{\mW}{\mQ}\) on \((0,\chi_{\mW,\mQ})\) implies that 
for any \(\rnf\in(0,\chi_{\mW,\mQ})\) there exists an open interval containing
\(\rnf\) on which \(\RD{\rno}{\mW}{\mQ}\) is equal to the power series determined 
by the derivatives of \(\RD{\rno}{\mW}{\mQ}\) at \(\rno=\rnf\).
If we have a finite collection of pairs of probability measures 
\(\{(\wmn{\ind},\qmn{\ind})\}_{\ind\in\set{I}}\), 
then for any \(\rnf\) that is in \((0,\chi_{\wmn{\ind},\qmn{\ind}})\) for all 
\(\ind\in\set{I}\) there exists an open interval containing \(\rnf\) on which each \(\RD{\rno}{\wmn{\ind}}{\qmn{\ind}}\) 
is equal to the power series determined by the derivatives of \(\RD{\rno}{\wmn{\ind}}{\qmn{\ind}}\) 
at \(\rno=\rnf\).
When the collection of pairs of probability measures is infinite,
then there might not be an open interval containing \(\rnf\) that is contained in all 
\((0,\chi_{\wmn{\ind},\qmn{\ind}})\)'s. 
Lemma \ref{lem:uniform-analyticity}, in the following, asserts the existence of such an interval 
when \(\RD{\rnb}{\wmn{\ind}}{\qmn{\ind}}\) is uniformly bounded for a \(\rnb>\rnf\) for all
\(\ind\in\set{I}\).
In addition, Lemma \ref{lem:uniform-analyticity} asserts uniform approximation error 
terms, over all \(\ind\in\set{I}\), for the power series on that interval.

\begin{lemma}\label{lem:uniform-analyticity}
For any \(\gamma,\rnf,\rnb\in\reals{+}\) satisfying \(\rnf\in(0,\rnb)\) 
and \(\mW,\mQ\in\pmea{\outA}\) satisfying \(\RD{\rnb}{\mW}{\mQ}\leq \gamma\),
\begin{align}
\label{eq:lem:uniform-analyticity-derivativebound}
\abs{\left.\pder{^{\knd}\RD{\rno}{\mW}{\mQ}}{\rno^{\knd}}\right\vert_{\rno=\rnf}} 
&\leq
\begin{cases}
\knd!\tau^{\knd+1}\knd
&\rnf\neq 1
\\
\knd!\tau^{\knd+1}
&\rnf=1
\end{cases}
\\
\label{eq:lem:uniform-analyticity-taylor}
\abs{\RD{\rnt}{\mW}{\mQ}-\sum\nolimits_{\ind=0}^{\knd-1}\tfrac{(\rnt-\rnf)^{\ind}}{\ind!}
\left.\pder{^{\ind}\RD{\rno}{\mW}{\mQ}}{\rno^{\ind}}\right\vert_{\rno=\rnf}}
&\leq
\begin{cases}
\tfrac{\tau^{\knd+1}\abs{\rnt-\rnf}^{\knd}}{1-\abs{\rnt-\rnf}\tau}\left[\knd-1+\tfrac{1}{1-\abs{\rnt-\rnf}\tau}\right]
&\rnf\neq1
\\
\tfrac{\tau^{\knd+1}\abs{\rnt-\rnf}^{\knd}}{1-\abs{\rnt-\rnf}\tau}
&\rnf=1
\end{cases}
&
&\forall \rnt:\abs{\rnt-\rnf}\leq \tfrac{1}{\tau}
\end{align}
where
\begin{align}
\label{eq:lem:uniform-analyticity-tau}
\tau
&\DEF\begin{cases}
\tfrac{1}{\abs{\rnf-1}}\vee 
\left[\tfrac{1+e^{(1\vee\rnb)\gamma}}{\rnf \wedge(\rnb-\rnf)}+\gamma\right]
&\rnf\neq 1
\\
\tfrac{1+e^{\rnb\gamma}}{1 \wedge(\rnb-1)}
&\rnf=1
\end{cases}.
\end{align}
\end{lemma}
Lemma \ref{lem:uniform-analyticity} is new to the best of our knowledge; it is proved in 
Appendix \ref{sec:analyticityproofs}
using \eqref{eq:lem:analyticity-derivative}
together with the elementary properties of the real analytic functions and power series.

\subsection{The Conditional \renyi Divergence and Tilted Channel}\label{sec:preliminary-conditionaldivergence}
The conditional \renyi divergence and the tilted channel
allows us to write certain frequently used expressions more succinctly.
\begin{definition}\label{def:conditionaldivergence}
	For any \(\rno\in\reals{+}\), \(\Wm:\inpS\to\pmea{\outA}\), \(\Qm:\inpS\to\pmea{\outA}\),
	and  \(\mP\in\pdis{\inpS}\) \emph{the order \(\rno\) conditional \renyi divergence for 
		the input distribution \(\mP\)} 
	is
	\begin{align}
	\label{eq:def:conditionaldivergence}
	\CRD{\rno}{\Wm}{\Qm}{\mP}
	&\DEF \sum\nolimits_{\dinp\in \inpS}  \mP(\dinp) \RD{\rno}{\Wm(\dinp)}{\Qm(\dinp)}.
	\end{align}
If \(\exists\mQ\in\pmea{\outA}\) such that  \(\Qm(\dinp)=\mQ\) for all \(\dinp\in\inpS\), then 
we denote \(\CRD{\rno}{\Wm}{\Qm}{\mP}\) by \(\CRD{\rno}{\Wm}{\mQ}{\mP}\).
\end{definition}
\begin{remark}
	In \cite{verdu15} and \cite{polyanskiyV10}, \(\CRD{\rno}{\Wm}{\Qm}{\mP}\) stands for
	\(\RD{\rno}{\mP\mtimes \Wm}{\mP\mtimes \Qm}\).
	For \(\rno=1\) case the convention used in \cite{verdu15} and \cite{polyanskiyV10} is equivalent to ours;
	for \(\rno\neq 1\) case, however, it is not. 
	If either \(\rno=1\) or 
	\(\RD{\rno}{\Wm(\dinp)}{\Qm(\dinp)}\) has the same value for all \(\dinp\)'s with positive 
	\(\mP(\dinp)\), then
	\(\RD{\rno}{\mP\mtimes \Wm}{\mP\mtimes \Qm}=\sum_{\dinp} \mP(\dinp)\RD{\rno}{\Wm(\dinp)}{\Qm(\dinp)}\),
	else \(\RD{\rno}{\mP\mtimes \Wm}{\mP\mtimes \Qm}<\sum_{\dinp} \mP(\dinp)\RD{\rno}{\Wm(\dinp)}{\Qm(\dinp)}\)
	for \(\rno\in(0,1)\) 
	and 
	\(\RD{\rno}{\mP\mtimes \Wm}{\mP\mtimes \Qm}>\sum_{\dinp} \mP(\dinp)\RD{\rno}{\Wm(\dinp)}{\Qm(\dinp)}\)
	for \(\rno\in (1,\infty)\).
	The inequalities follow from the Jensen's inequality and the 
	strict concavity of the natural logarithm 
	function.
\end{remark}

\begin{definition}\label{def:tiltedchannel}
	For any \(\rno\in\reals{+}\), \(\Wm:\inpS\to\pmea{\outA}\) and \(\Qm:\inpS\to\pmea{\outA}\),
	\emph{the order \(\rno\) tilted channel \(\Wma{\rno}{\Qm}\)} is a function
	from  \(\{\dinp:\RD{\rno}{\Wm(\dinp)}{\Qm(\dinp)}<\infty\}\) to \(\pmea{\outA}\)
	given by
	\begin{align}
	\label{eq:def:tiltedchannel}
	\der{\Wma{\rno}{\Qm}(\dinp)}{\rfm}
	&\DEF e^{(1-\rno)\RD{\rno}{\Wm(\dinp)}{\Qm(\dinp)}}
	\left[\der{\Wm(\dinp)}{\rfm}\right]^{\rno} 
	\left[\der{\Qm(\dinp)}{\rfm}\right]^{1-\rno}.
	\end{align}
	If \(\exists\mQ\in\pmea{\outA}\) such that  \(\Qm(\dinp)=\mQ\) for all \(\dinp\in\inpS\), then 
	we denote \(\Wma{\rno}{\Qm}\) by \(\Wma{\rno}{\mQ}\).
\end{definition}

%%%%%%%%As an immediate consequence of Definitions
%%%%%%%%\ref{def:divergence}, \ref{def:tiltedprobabilitymeasure}, \ref{def:conditionaldivergence}, \ref{def:tiltedchannel} 
%%%%%%%%and Lemma \ref{lem:variational} we get the following lemma. 
%%%%%%%%\begin{lemma}\label{lem:cvariational}
%%%%%%%%	For any \(\Wm:\inpS\to\pmea{\outA}\), \(\Qm:\inpS\to\pmea{\outA}\), and  \(\mP\in\pdis{\inpS}\) 
%%%%%%%%	\begin{align}
%%%%%%%%	\label{eq:lem:cvariational}
%%%%%%%%	\CRD{\rno}{\Wm}{\Qm}{\mP}
%%%%%%%%	&=\begin{cases}
%%%%%%%%	\inf\nolimits_{\Vm}
%%%%%%%%	\tfrac{\rno}{1-\rno}\CRD{1}{\Vm}{\Wm}{\mP}+\CRD{1}{\Vm}{\Qm}{\mP}
%%%%%%%%	&\rno\in(0,1)
%%%%%%%%	\\
%%%%%%%%	\sup\nolimits_{\Vm}
%%%%%%%%	\tfrac{\rno}{1-\rno}\CRD{1}{\Vm}{\Wm}{\mP}+\CRD{1}{\Vm}{\Qm}{\mP}
%%%%%%%%	&\rno\in(1,\infty)
%%%%%%%%	\end{cases}
%%%%%%%%	\end{align}
%%%%%%%%	where \(\tfrac{\rno}{1-\rno}\CRD{1}{\Vm}{\Wm}{\mP}+\CRD{1}{\Vm}{\Qm}{\mP}\)
%%%%%%%%	stands for \(-\infty\) when \(\rno\in(1,\infty)\) and 
%%%%%%%%	\(\CRD{1}{\Vm}{\Wm}{\mP}=\CRD{1}{\Vm}{\Qm}{\mP}=\infty\).
%%%%%%%%	Furthermore, if \(\CRD{\rno}{\Wm}{\Qm}{\mP}<\infty\) and either 
%%%%%%%%	\(\rno\in(0,1)\) or \(\CRD{1}{\Wma{\rno}{\Qm}}{\Wm}{\mP}<\infty\), 
%%%%%%%%	then
%%%%%%%%	\begin{align}
%%%%%%%%	\label{eq:lem:cvariational-tilted}
%%%%%%%%	\CRD{\rno}{\Wm}{\Qm}{\mP}
%%%%%%%%	&=\tfrac{\rno}{1-\rno}\CRD{1}{\Wma{\rno}{\Qm}}{\Wm}{\mP}+\CRD{1}{\Wma{\rno}{\Qm}}{\Qm}{\mP}.
%%%%%%%%	\end{align}
%%%%%%%%\end{lemma}

%!TEX root=../main-C.tex
\section{The Augustin Information}\label{sec:information}
The main aim of this section is to introduce the concepts of Augustin information and mean.
We define the order \(\rno\) Augustin information for the input distribution \(\mP\) and 
establish the existence of a unique Augustin mean for any input distribution 
\(\mP\) and positive finite order \(\rno\) in \S\ref{sec:information-mean}.
After that we analyze the Augustin information,
first as a function of the input distribution for a given order in \S\ref{sec:information-prior} 
and then as a function of the order for a given input distribution in \S\ref{sec:information-order}.
We conclude our discussion by comparing the Augustin information with the \renyi information 
and characterizing each quantity in terms of the other in \S\ref{sec:information-renyi}.
Some of the most important observations about the Augustin information and mean were first reported 
by Augustin in \cite[\S34]{augustin78} for orders not exceeding one. 
This is why we suggest naming these concepts after him. 
Proof of the lemmas presented in this section are presented in 
Appendix \ref{sec:informationproofs}.

\subsection{Existence of a Unique Augustin Mean}\label{sec:information-mean}
\begin{definition}\label{def:information}
For any \(\rno\in \reals{+}\), \(\Wm:\inpS\!\to\!\pmea{\outA}\), and \(\mP\in \pdis{\inpS}\)  
\emph{the order \(\rno\) Augustin information for the input distribution \(\mP\)} is
\begin{align}
\label{eq:def:information}
\RMI{\rno}{\mP}{\Wm}
&\DEF \inf\nolimits_{\mQ\in \pmea{\outA}} \CRD{\rno}{\Wm}{\mQ}{\mP}.
\end{align}
\end{definition}
One can confirm by substitution that
\begin{align}\label{eq:topsoe}
	\CRD{1}{\Wm}{\mQ}{\mP}
	&=\CRD{1}{\Wm}{\qmn{1,\mP}}{\mP}+\RD{1}{\qmn{1,\mP}}{\mQ}
	&
	&\forall \mQ\in\pmea{\outA}
\end{align}
where
\begin{align}
\label{eq:def:orderonemean}
\qmn{1,\mP}
&\DEF\sum\nolimits_{\dinp}\mP(\dinp)\Wm(\dinp).
\end{align}
Then Lemma \ref{lem:divergence-pinsker} and \eqref{eq:def:information} imply
\begin{align}
\notag
\RMI{1}{\mP}{\Wm}
&=\CRD{1}{\Wm}{\qmn{1,\mP}}{\mP}.
\end{align}
Thus the order one Augustin information has a closed form expression, which  
is equal to the mutual information.
For other orders, however, Augustin information does not have a closed form expression.
Nonetheless, Lemma \ref{lem:information}, presented in the following, establishes the existence 
of a unique probability measure \(\qmn{\rno,\mP}\) satisfying 
\(\RMI{\rno}{\mP}{\Wm}=\CRD{\rno}{\Wm}{\qmn{\rno,\mP}}{\mP}\) 
for\footnote{This is rather easy to prove when \(\outS\) is a finite set.
The uniqueness of \(\qmn{\rno,\mP}\) follows from the strict convexity of the \renyi divergence 
in its second argument described in Lemma \ref{lem:divergence-convexity}.
If \(\outS\) is finite, then \(\pmea{\outA}\) is  compact and the existence of \(\qmn{\rno,\mP}\) 
follows from the lower semicontinuity of the \renyi divergence in its second argument 
---which follows from Lemma \ref{lem:divergence:lsc}--- 
and  the extreme value theorem for the lower semicontinuous 
functions \cite[Ch3\S12.2]{kolmogorovfomin75}.
For channels with arbitrary output spaces, however, \(\pmea{\outA}\) is not compact;
thus we can not invoke the extreme value theorem to establish the existence of \(\qmn{\rno,\mP}\).}	
any positive order \(\rno\) and input distribution \(\mP\).
Furthermore, parts (\ref{information:zto}) and (\ref{information:oti}) of Lemma \ref{lem:information}
present an alternative characterization of \(\qmn{\rno,\mP}\) by showing that \(\qmn{\rno,\mP}\) 
is the unique 
fixed point of the operator \(\Aop{\rno}{\mP}{\cdot}\) satisfying \(\qmn{1,\mP}\AC\qmn{\rno,\mP}\).
Lemma \ref{lem:information}-(\ref{information:alternative})  provides an alternative characterization
of the Augustin information for orders other than 
one.\footnote{This alternative characterization is employed to prove the equivalence of two 
definitions of the sphere packing exponent and the strong converse exponent.}

\begin{definition}\label{def:preliminary}
Let \(\rno\) be a positive real number and
\(\Wm\) be a channel of the form
\(\Wm:\inpS\to \pmea{\outA}\).
\begin{itemize}
	\item For any \(\mP\in \fdis{\inpS}\), 
		\emph{the order \(\rno\) mean measure for the input distribution \(\mP\)} is given by
	\begin{align}
	\label{eq:def:meanmeasure}
	\der{\mmn{\rno,\mP}}{\rfm}
	&\DEF \left[\sum\nolimits_{\dinp} \mP(\dinp) \left(\der{\Wm{(\dinp)}}{\rfm}\right)^{\rno}  \right]^{\frac{1}{\rno}}
	\end{align}
	where \(\rfm\) is any measure for which \((\sum_{\dinp}\mP(\dinp) \Wm(\dinp))\AC\rfm\).
	\item For any \(\mP\in \pdis{\inpS}\), \emph{the order \(\rno\) \renyi mean for the input distribution \(\mP\)} is given by
	\begin{align}
	\label{eq:def:renyimean}
	\qgn{\rno,\mP}
	&\DEF\tfrac{\mmn{\rno,\mP}}{\lon{\mmn{\rno,\mP}}}.
	\end{align}
	\item For any \(\mP\in \pdis{\inpS}\), \emph{the order \(\rno\) Augustin operator for the input distribution \(\mP\)}, i.e. \(\Aop{\rno}{\mP}{\cdot}:\domtr{\rno,\mP}\to\pmea{\outA}\), is given by
	\begin{align}
	\label{eq:def:Aoperator}
	\Aop{\rno}{\mP}{\mQ}
	&\DEF\sum\nolimits_{\dinp}\mP(\dinp) \Wma{\rno}{\mQ}(\dinp)
	&
	&\forall \mQ\in  \domtr{\rno,\mP}
	\end{align}
	where 
\(\domtr{\rno,\mP}\DEF\{\mQ\in\pmea{\outA}:\CRD{\rno}{\Wm}{\mQ}{\mP}<\infty\}\) and
the tilted channel \(\Wma{\rno}{\mQ}\) is defined in \eqref{eq:def:tiltedchannel}.
Furthermore, \(\Aopi{\rno}{\mP}{0}{\mQ}=\mQ\) and 
\(\Aopi{\rno}{\mP}{\ind+1}{\mQ}\DEF\Aop{\rno}{\mP}{\Aopi{\rno}{\mP}{\ind}{\mQ}}\) 
for any non-negative integer \(\ind\).
\end{itemize}
\end{definition}
 
\begin{lemma}\label{lem:information}
Let \(\Wm\) be a channel of the form \(\Wm:\inpS\to \pmea{\outA}\) and \(\mP\) be 
an input distribution in \(\pdis{\inpS}\).
\begin{enumerate}[(a)]
\item\label{information:bounded}
\(\RMI{\rno}{\mP}{\Wm}\leq \CRD{\rno}{\Wm}{\qmn{1,\mP}}{\mP}\leq \bent{\mP}<\infty\) for all \(\rno\in\reals{+}\)
where \(\qmn{1,\mP}\) is defined in \eqref{eq:def:orderonemean}.
\item\label{information:one}
\(\RMI{1}{\mP}{\Wm}=\CRD{1}{\Wm}{\qmn{1,\mP}}{\mP}\).  Furthermore,
\begin{align}
\label{eq:lem:information-one:EHB}
\CRD{1}{\Wm}{\mQ}{\mP}-\RMI{1}{\mP}{\Wm}
&=\RD{1}{\qmn{1,\mP}}{\mQ}
&
&\forall \mQ\in\pmea{\outA}.
\end{align}
\item\label{information:zto} 
If \(\rno\in (0,1)\), then \(\exists! \qmn{\rno,\mP}\) such that \(\RMI{\rno}{\mP}{\Wm}=\CRD{\rno}{\Wm}{\qmn{\rno,\mP}}{\mP}\).
Furthermore,
\begin{align}
\label{eq:lem:information-zto:fixedpoint}
\Aop{\rno}{\mP}{\qmn{\rno,\mP}}
&=\qmn{\rno,\mP},
\\
\label{eq:lem:information-zto:iteration}
\lim\limits_{\jnd\to\infty}  
\lon{\qmn{\rno,\mP}\!-\Aopi{\rno}{\mP}{\jnd}{\qgn{\rno,\mP}}}
&=0,
\\
\label{eq:lem:information-zto:EHB}
\RD{1}{\qmn{\rno,\mP}}{\mQ}\geq 
\CRD{\rno}{\Wm}{\mQ}{\mP}-\RMI{\rno}{\mP}{\Wm}
&\geq \RD{\rno}{\qmn{\rno,\mP}}{\mQ}
&
&\forall \mQ\in\pmea{\outA},
\end{align}
and \(\qmn{\rno,\mP}\sim \qmn{1,\mP}\).
In addition,\footnote{Note that \(\Aop{\rno}{\mP}{\mQ}=\mQ\), on its own, 
does not imply \(\qmn{\rno,\mP}=\mQ\) for \(\rno\)'s in \((0,1)\). 
Consider for example a binary symmetric 
channel	and let \(\mQ\) be the probability measure that puts all of its probability to 
one of the output letters. Then \(\Aop{\rno}{\mP}{\mQ}=\mQ\), but \(\qmn{\rno,\mP}\neq\mQ\), 
for all \(\mP\in\pdis{\inpS}\) and \(\rno\in(0,1)\).} 
if \(\qmn{1,\mP}\AC \mQ\) and \(\Aop{\rno}{\mP}{\mQ}=\mQ\), then \(\qmn{\rno,\mP}=\mQ\).
\item\label{information:oti}
If \(\rno\in (1,\infty)\), then \(\exists! \qmn{\rno,\mP}\) such that \(\RMI{\rno}{\mP}{\Wm}=\CRD{\rno}{\Wm}{\qmn{\rno,\mP}}{\mP}\).
Furthermore, 
\begin{align}
\label{eq:lem:information-oti:fixedpoint}
\Aop{\rno}{\mP}{\qmn{\rno,\mP}}
&=\qmn{\rno,\mP},
\\
\label{eq:lem:information-oti:EHB}
\RD{\rno}{\qmn{\rno,\mP}}{\mQ}\geq 
\CRD{\rno}{\Wm}{\mQ}{\mP}-\RMI{\rno}{\mP}{\Wm}
&\geq \RD{1}{\qmn{\rno,\mP}}{\mQ}
&
&\forall \mQ\in\pmea{\outA},
\end{align}
and \(\qmn{\rno,\mP}\sim \qmn{1,\mP}\). 
In addition, if  \(\Aop{\rno}{\mP}{\mQ}=\mQ\), then \(\qmn{\rno,\mP}=\mQ\). 
\item\label{information:alternative}
If \(\rno\in\reals{+}\setminus\{1\}\), then 
\begin{align}
\label{eq:lem:information:alternative:opt}
\RMI{\rno}{\mP}{\Wm}
&=\tfrac{\rno}{1-\rno}\CRD{1}{\Wma{\rno}{\qmn{\rno,\mP}}}{\Wm}{\mP}+\RMI{1}{\mP}{\Wma{\rno}{\qmn{\rno,\mP}}}
\\
\label{eq:lem:information:alternative:def}
&=\begin{cases}
\inf\nolimits_{\Vm\in\pmea{\outA|\inpS}} \tfrac{\rno}{1-\rno}\CRD{1}{\Vm}{\Wm}{\mP}+\RMI{1}{\mP}{\Vm}
&\rno\in(0,1)
\\
\sup\nolimits_{\Vm\in\pmea{\outA|\inpS}} \tfrac{\rno}{1-\rno}\CRD{1}{\Vm}{\Wm}{\mP}+\RMI{1}{\mP}{\Vm}
&\rno\in(1,\infty)
\end{cases}
\\
\notag
&=\tfrac{\rno}{1-\rno}
\inf\nolimits_{\Vm\in\pmea{\outA|\inpS}}\left(
 \CRD{1}{\Vm}{\Wm}{\mP}+
\tfrac{1-\rno}{\rno}\RMI{1}{\mP}{\Vm}\right).
\end{align}
\end{enumerate}
\end{lemma}

The convergence described in \eqref{eq:lem:information-zto:iteration} 
holds not just for the \renyi mean \(\qgn{\rno,\mP}\) but also 
for certain other probability measures, 
as well. 
Remark \ref{remark:information-zto:iteration-general} in Appendix \ref{sec:informationproofs}
describes how one can establish
the following more general convergence result for any \(\rno\in(0,1)\)
and \(\mP\in\pdis{\inpS}\):
\begin{align}
\label{eq:lem:information-zto:iteration-general-new}
\lim\nolimits_{\jnd\to\infty}  
\lon{\qmn{\rno,\mP}\!-\Aopi{\rno}{\mP}{\jnd}{\mQ}}
&=0
&
&\mbox{if~}\mQ\sim\qmn{1,\mP}\mbox{~and~}
\essup\nolimits_{\qmn{1,\mP}}\abs{\ln \der{\mQ}{\qmn{1,\mP}}}<\infty.
\end{align}
Part (\ref{information:bounded}) is 
proved using Lemma \ref{lem:divergence-RM};
\(\RMI{\rno}{\mP}{\Wm}\leq \bent{\mP}\) was proved 
by \csiszar through a different argument in \cite[(24)]{csiszar95}.
Part (\ref{information:one}), which is well known, is proved by substitution.
Part (\ref{information:zto}) is due 
to\footnote{To be precise \cite[Lemma 34.2]{augustin78} does not include the assertion 
	\(\RD{1}{\qmn{\rno,\mP}}{\mQ}\geq \CRD{\rno}{\Wm}{\mQ}{\mP}-\RMI{\rno}{\mP}{\Wm}\)
	and claims \eqref{eq:lem:information-zto:iteration} for \(\qmn{1,\mP}\) 
	instead of \(\qgn{\rno,\mP}\). 
	We cannot verify the correctness of Augustin's proof of \cite[Lemma 34.2]{augustin78}, 
	see 
	Appendix \ref{sec:augustinsproof} 
	for a more detailed discussion.}
Augustin \cite[Lemma 34.2]{augustin78}.
Part (\ref{information:oti}) is new to the best our knowledge.
Part (\ref{information:alternative}) was proved for the finite \(\outS\) case 
by \csiszar \cite[(A24), (A27)]{csiszar95}.

\begin{definition}\label{def:mean}
For any \(\rno\in\reals{+}\), \(\Wm:\inpS\to \pmea{\outA}\), \(\mP\in \pdis{\inpS}\) the unique 
probability measure 
\(\qmn{\rno,\mP}\) on \((\outS,\outA)\) satisfying \(\RMI{\rno}{\mP}{\Wm}=\CRD{\rno}{\Wm}{\qmn{\rno,\mP}}{\mP}\) is 
called \emph{the order \(\rno\) Augustin mean for the input distribution \(\mP\)}.
\end{definition}

Lemma \ref{lem:divergence-pinsker} and 
%%%%%%%%Lemma\ref{lem:information}-(\ref{information:one},\ref{information:zto},\ref{information:oti}), i.e. 
\eqref{eq:lem:information-one:EHB}, \eqref{eq:lem:information-zto:EHB}, 
\eqref{eq:lem:information-oti:EHB},
imply the following bound, which is analogous to \cite[Thm. 3.1]{csiszar72} of \csiszar\!\!:
\begin{align}
\notag
\sqrt{2\tfrac{\CRD{\rno}{\Wm}{\mQ}{\mP}-\RMI{\rno}{\mP}{\Wm}}{\rno\wedge 1}}
&\geq  \lon{\qmn{\rno,\mP}-\mQ}
&
&\forall \mQ\in\pmea{\outA},\forall \rno\in\reals{+}.
\end{align}  

The Augustin information and mean have closed form expressions only for \(\rno=1\);
for other orders they do not have closed form expressions.
However, the fixed point property \(\Aop{\rno}{\mP}{\qmn{\rno,\mP}}=\qmn{\rno,\mP}\) 
established in Lemma \ref{lem:information}-(\ref{information:zto},\ref{information:oti}) 
and the definition of \(\Aop{\rno}{\mP}{\cdot}\) given in \eqref{eq:def:Aoperator}
imply the following identity for the Augustin mean:
\begin{align}
\label{eq:mean}
\der{\qmn{\rno,\mP}}{\rfm}
&=\left[\sum\nolimits_{\dinp}\mP(\dinp) \left(\der{\Wm(\dinp)}{\rfm}\right)^{\rno} e^{(1-\rno)\RD{\rno}{\Wm(\dinp)}{\qmn{\rno,\mP}}}\right]^{\frac{1}{\rno}}
&
&\forall \rfm:\qmn{1,\mP}\AC\rfm.
\end{align} 
In \S\ref{sec:information-order}, we use this identity in lieu of a closed form expression 
while analyzing \(\RMI{\rno}{\mP}{\Wm}\) and \(\qmn{\rno,\mP}\) as a function of \(\rno\).

\begin{lemma}\label{lem:information:product}
For any length \(\blx\) product channel \(\Wmn{[1,\blx]}:\inpS_{1}^{\blx}\to\pmea{\outA_{1}^{\blx}}\)
and input distribution \(\mP\in\pdis{\inpS_{1}^{\blx}}\) we have
\begin{align}
\label{eq:lem:information:product}
\RMI{\rno}{\mP}{\Wmn{[1,\blx]}}
&\leq\sum\nolimits_{\tin=1}^{\blx} \RMI{\rno}{\pmn{\tin}}{\Wmn{\tin}}
\end{align}
for all \(\rno\in\reals{+}\) 
where \(\pmn{\tin}\in \pdis{\inpS_{\tin}}\) is the marginal of \(\mP\) on \(\inpS_{\tin}\).
Furthermore, the inequality in \eqref{eq:lem:information:product} is an equality
for an \(\rno\in\reals{+}\)  iff \(\qmn{\rno,\mP}\) satisfies
	\begin{align}
	\label{eq:lem:information:product-mean}
	\qmn{\rno,\mP}
	&=\bigotimes\nolimits_{\tin=1}^{\blx}\qmn{\rno,\pmn{\tin}}.	
	\end{align} 
If  \(\mP=\bigotimes\nolimits_{\tin=1}^{\blx}\pmn{\tin}\),
then \eqref{eq:lem:information:product-mean} holds for all \(\rno\in\reals{+}\)
and consequently \eqref{eq:lem:information:product} holds as an equality 
for all \(\rno\in\reals{+}\).
\end{lemma}

\subsection{Augustin Information as a Function of the Input Distribution}\label{sec:information-prior}
The order \(\rno\) Augustin information for the input distribution \(\mP\) is defined as 
the infimum of  a family of conditional \renyi divergences, which are linear in \(\mP\).
Then the Augustin information is concave in \(\mP\), 
because pointwise infimum of a family of concave functions is concave. 
Lemma \ref{lem:informationP} strengthens this observation using Lemma \ref{lem:information}.

\begin{lemma}\label{lem:informationP}
	For any \(\rno\in\reals{+}\) and \(\Wm:\inpS\to \pmea{\outA}\),
	\(\RMI{\rno}{\mP}{\Wm}\) is a concave function of \(\mP\) satisfying
	\begin{align}
	\label{eq:lem:informationP-A}
	\RMI{\rno}{\pmn{\beta}}{\Wm}
	&\geq \beta\RMI{\rno}{\pmn{1}}{\Wm}+(1-\beta)\RMI{\rno}{\pmn{0}}{\Wm}
	+\beta\RD{\rno\wedge 1}{\qmn{\rno,\pmn{1}}}{\qmn{\rno,\pmn{\beta}}}
	+(1-\beta)\RD{\rno\wedge 1}{\qmn{\rno,\pmn{0}}}{\qmn{\rno,\pmn{\beta}}}
	&&
	\\
	\label{eq:lem:informationP-B}
	\RMI{\rno}{\pmn{\beta}}{\Wm}
	&\leq \beta\RMI{\rno}{\pmn{1}}{\Wm}+(1-\beta)\RMI{\rno}{\pmn{0}}{\Wm}
	+\beta\RD{\rno\vee 1}{\qmn{\rno,\pmn{1}}}{\qmn{\rno,\pmn{\beta}}}
	+(1-\beta)\RD{\rno\vee 1}{\qmn{\rno,\pmn{0}}}{\qmn{\rno,\pmn{\beta}}}
	&&
	\\
	\label{eq:lem:informationP-C}
	\RMI{\rno}{\pmn{\beta}}{\Wm}
	&\leq \beta\RMI{\rno}{\pmn{1}}{\Wm}+(1-\beta)\RMI{\rno}{\pmn{0}}{\Wm}+\bent{\beta}
	-\RD{\rno\wedge 1}{\qmn{\rno,\pmn{\beta}}}{\beta\qmn{\rno,\pmn{1}}+(1-\beta)\qmn{\rno,\pmn{0}}}
	&&
	\end{align}
	where \(\pmn{\beta}=\beta \pmn{1}+(1-\beta)\pmn{0}\) for all \(\pmn{0},\pmn{1}\in\pdis{\inpS}\)
	and \(\beta\in[0,1]\).
\end{lemma}

Lemma \ref{lem:informationP} implies that for any positive order \(\rno\) and channel \(\Wm\),
the order \(\rno\) Augustin information \(\RMI{\rno}{\mP}{\Wm}\) is a continuous function of 
the input distribution \(\mP\) iff \(\sup_{\mP\in\pdis{\inpS}}\RMI{\rno}{\mP}{\Wm}\)
is finite.\footnote{The \renyi information, discussed in \S\ref{sec:information-renyi}, 
	has already shown to satisfy analogous relations, see 
	\cite[{Lemma \ref*{A-lem:finitecapacity}-(\ref*{A-finitecapacity-d},\ref*{A-finitecapacity-uecP})}]{nakiboglu19A}. 
	The only substantial subtlety is that for orders in \((0,1)\) the \renyi information is a
	continuous function of \(\mP\) even when the corresponding capacity expression is 
	infinite because the \renyi information is quasi-concave rather than concave in \(\mP\)
	for orders in \((0,1)\), 
	see \cite[{Lemma \ref*{A-lem:informationP}-(\ref*{A-informationP-a})}]{nakiboglu19A}.} 
Furthermore, if \(\sup_{\mP\in\pdis{\inpS}}\RMI{\rnt}{\mP}{\Wm}\) is finite for 
an \(\rnt\in\reals{+}\) then \(\{\RMI{\rno}{\mP}{\Wm}\}_{\rno\in(0,\rnt]}\) is uniformly 
equicontinuous in \(\mP\) on \(\pdis{\inpS}\).

In order to see why the finiteness of \(\sup_{\mP\in\pdis{\inpS}}\RMI{\rno}{\mP}{\Wm}\) 
is necessary for the continuity, note that the non-negativity of the \renyi 
divergence for probability measures and \eqref{eq:lem:informationP-A} imply that
\begin{align}
\notag
\RMI{\rno}{\pmn{\beta}}{\Wm}-\RMI{\rno}{\pmn{0}}{\Wm} 
&\geq \beta (\RMI{\rno}{\pmn{1}}{\Wm}-\RMI{\rno}{\pmn{0}}{\Wm})
+\beta\RD{\rno\wedge 1}{\qmn{\rno,\pmn{1}}}{\qmn{\rno,\pmn{\beta}}}
+(1-\beta)\RD{\rno\wedge 1}{\qmn{\rno,\pmn{0}}}{\qmn{\rno,\pmn{\beta}}}
\\
\notag
&\geq \beta (\RMI{\rno}{\pmn{1}}{\Wm}-\RMI{\rno}{\pmn{0}}{\Wm}).
\end{align}
On the other hand \(\lon{\pmn{\beta}-\pmn{0}}\leq 2\beta\).
Thus if there exists a \(\{\pmn{\ind}\}_{\ind\in\integers{+}}\subset\pdis{\inpS}\) such that
\(\lim_{\ind \uparrow\infty}\RMI{\rno}{\pmn{\ind}}{\Wm}=\infty\) then \(\RMI{\rno}{\mP}{\Wm}\)
is discontinuous at every \(\mP\) in \(\pdis{\inpS}\).

The converse statement, i.e. the sufficiency, can be established together with the equicontinuity.
For any \(\pmn{0},\pmn{1} \in \pdis{\inpS}\) such that \(\pmn{0}\neq\pmn{1}\) let 
\(\smn{\wedge}\), \(\smn{1}\), and \(\smn{0}\) be
\begin{align}
\notag
\smn{\wedge}
&=\tfrac{\pmn{1}\wedge\pmn{0}}{\lon{\pmn{1}\wedge\pmn{0}}},
\\
\notag
\smn{1}
&=\tfrac{\pmn{1}-\pmn{1}\wedge\pmn{0}}{1-\lon{\pmn{1}\wedge \pmn{0}}},
\\
\notag
\smn{0}
&=\tfrac{\pmn{0}-\pmn{1}\wedge\pmn{0}}{1-\lon{\pmn{1}\wedge \pmn{0}}}.
\end{align}
Then \(\smn{\wedge},\smn{1},\smn{0}\in \pdis{\inpS}\) and 
\(\smn{1}\perp \smn{0}\).
On the other hand \(\lon{\pmn{1}-\pmn{0}}=2-2\lon{\pmn{1}\wedge \pmn{0}}\). 
Therefore,
\begin{align}
\notag
\pmn{1}
&=(\tfrac{2-\lon{\pmn{1}-\pmn{0}}}{2})\smn{\wedge} + \tfrac{\lon{\pmn{1}-\pmn{0}}}{2}\smn{1},
\\
\notag
\pmn{0}
&=(\tfrac{2-\lon{\pmn{1}-\pmn{0}}}{2})\smn{\wedge} + \tfrac{\lon{\pmn{1}-\pmn{0}}}{2}\smn{0}.
\end{align}
Thus as a result of Lemmas \ref{lem:divergence-pinsker} and \ref{lem:informationP} we have
\begin{align}
\notag
\RMI{\rno}{\pmn{0}}{\Wm}\!-\!\RMI{\rno}{\pmn{1}}{\Wm}
&\!\leq\!\bent{\tfrac{\lon{\pmn{1}-\pmn{0}}}{2}}+
\tfrac{\lon{\pmn{1}\!-\!\pmn{0}}}{2} (\RMI{\rno}{\mS_{0}}{\Wm}-\RMI{\rno}{\mS_{1}}{\Wm})
\\
\label{eq:concavitybound}
&\!\leq\!\bent{\tfrac{\lon{\pmn{1}-\pmn{0}}}{2}}+
\tfrac{\lon{\pmn{1}\!-\!\pmn{0}}}{2}\RMI{\rno}{\mS_{0}}{\Wm}
&
&\forall \pmn{1},\pmn{0}\in\pdis{\inpS},\rno\in\reals{+}.
\end{align}
Thus 
\begin{align}
\notag
\abs{\RMI{\rno}{\pmn{0}}{\Wm}\!-\!\RMI{\rno}{\pmn{1}}{\Wm}}
&\!\leq\!\bent{\tfrac{\lon{\pmn{1}-\pmn{0}}}{2}}+
\tfrac{\lon{\pmn{1}\!-\!\pmn{0}}}{2}
\sup\nolimits_{\mP\in\pdis{\inpS}} \RMI{\rnt}{\mP}{\Wm}
&
&\forall \pmn{1},\pmn{0}\in\pdis{\inpS}, \rno\in(0,\rnt].
\end{align}

 \subsection{Augustin Information as a Function of the Order}\label{sec:information-order}
 	The main goal of this subsection is to characterize the behavior of 
 	the Augustin information as 
 	a function of the order for a given input distribution.
 	Lemma \ref{lem:meanO} presents preliminary observations that facilitate the analysis of 
 	Augustin information as a function of the order;
 	results of this analysis are presented in Lemma \ref{lem:informationO}. 
 	
\begin{lemma}\label{lem:meanO}
	For any channel \(\Wm\) of the form \(\Wm:\inpS\to \pmea{\outA}\) and 
 	input distribution \(\mP\in\pdis{\inpS}\),
 	\begin{enumerate}[(a)]
 	\item\label{meanO:Rdivergencebound} 
 	\(\RD{\rno}{\Wm(\dinp)}{\qmn{\rno,\mP}}\leq \ln \tfrac{1}{\mP(\dinp)}\),
 	\item\label{meanO:measurebound} 
 	\([\mP(\dinp)]^{\frac{1}{\rno\wedge 1}}\Wm(\dinp)\leq \qmn{\rno,\mP}\),
 	\item\label{meanO:boundedlogRND}  
 	 \(\abs{\ln \der{\qmn{\rno,\mP}}{\qmn{1,\mP}}}\leq \tfrac{\abs{\rno-1}}{\rno}\ln \tfrac{1}{\min\nolimits_{\dinp:\mP(\dinp)>0} \mP(\dinp)}\).
  \end{enumerate}
\end{lemma}
Bounds given in Lemma \ref{lem:meanO} follow from \eqref{eq:mean} via elementary manipulations.

\begin{lemma}\label{lem:informationO} 
For any channel \(\Wm\) of the form \(\Wm:\inpS\to \pmea{\outA}\) and input distribution 
\(\mP\in\pdis{\inpS}\),
\begin{enumerate}[(a)]
\item\label{informationO:strictconvexity}
Either 
\((\rno-1)\RMI{\rno}{\mP}{\Wm}\) is a strictly convex function 
of \(\rno\) from \(\reals{+}\) to  \([-\bent{\mP},\infty)\)
or 
\(\RMI{\rno}{\mP}{\Wm}=\sum_{\dinp}\mP(\dinp)\ln \gamma(\dinp)\) 
for some \(\gamma:\inpS\to[1,\infty)\) satisfying
\(\der{\Wm(\dinp)}{\qmn{1,\mP}}=\gamma(\dinp)\)  \(\Wm(\dinp)\)-a.s. for all 
\(\dinp\in\supp{\mP}\)
and \(\qmn{\rno,\mP}=\qmn{1,\mP}\) for all \(\rno\in\reals{+}\).

\item\label{informationO:decreasing}
\(\frac{1-\rno}{\rno}\RMI{\rno}{\mP}{\Wm}\) is a nonincreasing and continuous function of \(\rno\)
from \(\reals{+}\) to \(\reals{}\).
\item\label{informationO:continuity}
 \(\RMI{\rno}{\mP}{\Wm}\) is a nondecreasing and continuous function of \(\rno\)
 from \(\reals{+}\) to \([0,\bent{\mP}]\). 
\item\label{informationO:mean-uec} 
 \(\{\ln\der{\qmn{\rno,\mP}}{\qmn{1,\mP}}\}_{\dout\in\outS}\) is an
equicontinuous family of functions of \(\rno\) on \(\reals{+}\).
\item\label{informationO:differentiability}
\(\RMI{\rno}{\mP}{\Wm}\) is a continuously differentiable function of \(\rno\) 
from \(\reals{+}\) to \([0,\bent{\mP}]\) such that
\begin{align}
\label{eq:lem:informationO:differentiability}
\left.\pder{}{\rno}\RMI{\rno}{\mP}{\Wm}\right\vert_{\rno=\rnf}
&=\left.\pder{}{\rno}\CRD{\rno}{\Wm}{\qmn{\rnf,\mP}}{\mP}\right\vert_{\rno=\rnf}
\\
\label{eq:lem:informationO:differentiability-alt}
&=\begin{cases}
\tfrac{1}{(\rnf-1)^2}\CRD{1}{\Wma{\rnf}{\qmn{\rnf,\mP}}}{\Wm}{\mP}
&\rnf\neq 1
\\
\sum\nolimits_{\dinp}\tfrac{\mP(\dinp)}{2}
 \EXS{\Wm(\dinp)}{\left(\ln \der{\Wm(\dinp)}{\qmn{1,\mP}}-\RD{1}{\Wm(\dinp)}{\qmn{1,\mP}}\right)^{2}}
 &\rnf= 1
 \end{cases}.
 \end{align}
\item\label{informationO:monotonicityofharoutunianinformation}
If \((\rno-1)\RMI{\rno}{\mP}{\Wm}\) is strictly convex in \(\rno\), then 
\(\RMI{1}{\mP}{\Wma{\rno}{\qmn{\rno,\mP}}}\) ---i.e. \(\CRD{1}{\Wma{\rno}{\qmn{\rno,\mP}}}{\qmn{\rno,\mP}}{\mP}\)---
is a monotonically increasing continuous function of \(\rno\) on \(\reals{+}\);
else  \(\RMI{1}{\mP}{\Wma{\rno}{\qmn{\rno,\mP}}}=\sum_{\dinp}\mP(\dinp)\ln \gamma(\dinp)\)
---i.e. \(\CRD{1}{\Wma{\rno}{\qmn{\rno,\mP}}}{\qmn{\rno,\mP}}{\mP}=\sum_{\dinp}\mP(\dinp)\ln \gamma(\dinp)\)---
for some \(\gamma:\inpS\to[1,\infty)\) satisfying
\(\der{\Wm(\dinp)}{\qmn{1,\mP}}=\gamma(\dinp)\)  \(\Wm(\dinp)\)-a.s. for all 
\(\dinp\in\supp{\mP}\)
and \(\qmn{\rno,\mP}=\qmn{1,\mP}\) for all \(\rno\in\reals{+}\). 
\item\label{informationO:limitofharoutunianinformation} 
\(\lim_{\rno\downarrow0}\RMI{1}{\mP}{\Wma{\rno}{\qmn{\rno,\mP}}}=
\lim_{\rno\downarrow0}\RMI{\rno}{\mP}{\Wm}\).
\end{enumerate}
\end{lemma}
The (strict) convexity of \((\rno-1)\RMI{\rno}{\mP}{\Wm}\) in \(\rno\) on  \(\reals{+}\) 
is equivalent to the (strict) concavity of the function \(\mS\RMI{\frac{1}{1+\mS}}{\mP}{\Wm}\)
in \(\mS\) on \((-1,\infty)\), 
see the proof of part (\ref{informationO:monotonicityofharoutunianinformation})
for a proof.
The concavity of the function \(\mS\RMI{\frac{1}{1+\mS}}{\mP}{\Wm}\)
in \(\mS\) on \((-1,\infty)\) and 
parts (\ref{informationO:decreasing}) and (\ref{informationO:continuity}) 
of Lemma \ref{lem:informationO}
have been reported by Augustin in \cite[Lemma 34.3]{augustin78} for orders between
zero and one. Parts (\ref{informationO:strictconvexity}), (\ref{informationO:mean-uec}), 
  (\ref{informationO:differentiability}), (\ref{informationO:monotonicityofharoutunianinformation}),
  and (\ref{informationO:limitofharoutunianinformation}) 
  of Lemma \ref{lem:informationO} are new to the best of our 
 knowledge.
 Lemma \ref{lem:informationO} is primarily about the Augustin information as a function of the order for a given 
 	input distribution.
 	Part (\ref{informationO:mean-uec}), i.e. 
 	the equicontinuity of \(\{\ln\der{\qmn{\rno,\mP}}{\qmn{1,\mP}}\}_{\dout\in\outS}\) as a 
 	family of functions of the order \(\rno\),
 	  is derived as a necessary tool for establishing the continuity 
 	of the derivative of the Augustin information, i.e. part (\ref{informationO:differentiability}). 
 	Note that Lemma \ref{lem:meanO}-(\ref{meanO:boundedlogRND}) has already established this equicontinuity at \(\rno=1\).

\subsection{Augustin Information vs \renyi Information}\label{sec:information-renyi}
The Augustin information is not the only information that has been 
defined in terms of the \renyi divergence; 
there are others. 
The \renyi information, defined first by 
Gallager\footnote{Gallager uses a different parametrization
and confines his discussion to \(\rno\in(0,1)\) case. } 
\cite{gallager65} and then by Sibson \cite{sibson69}, 
is arguably the most prominent one among them 
because of its operational significance 
established by Gallager \cite{gallager65}.
\begin{definition}\label{def:renyiinformation}
	For any \(\rno\in \reals{+}\), \(\Wm:\inpS\to \pmea{\outA}\), and \(\mP\in \pdis{\inpS}\)  
	\emph{the order \(\rno\) \renyi information for the input distribution \(\mP\)} is
	\begin{align}
	\label{eq:def:renyiinformation}
	\GMI{\rno}{\mP}{\Wm}
	&\DEF \inf\nolimits_{\mQ\in \pmea{\outA}} \RD{\rno}{\mP\mtimes \Wm}{\mP\otimes\mQ}.
	\end{align}
\end{definition}

As noted by Sibson \cite{sibson69}, one can confirm by substitution that  
\begin{align}
\notag  %\label{eq:sibson}
\RD{\rno}{\mP\mtimes \Wm}{\mP\otimes\mQ}
&=\RD{\rno}{\mP\mtimes \Wm}{\mP\otimes\qgn{\rno,\mP}}+\RD{\rno}{\qgn{\rno,\mP}}{\mQ}
&
&\forall \mP\in\pdis{\inpS},\mQ\in\pmea{\outA},\rno\in\reals{+}
\end{align}
where \(\qgn{\rno,\mP}\) is the \renyi mean defined in \eqref{eq:def:renyimean}.
Then using Lemma \ref{lem:divergence-pinsker} we can conclude
that
\begin{align}
\label{eq:renyiinformation}
\GMI{\rno}{\mP}{\Wm}
&=\RD{\rno}{\mP\mtimes \Wm}{\mP\otimes\qgn{\rno,\mP}}
&
&\forall \mP\in\pdis{\inpS},\rno\in\reals{+},
\\
\label{eq:renyiinformation:EHB}
\RD{\rno}{\mP\mtimes \Wm}{\mP\otimes\mQ}
&=\GMI{\rno}{\mP}{\Wm}+\RD{\rno}{\qgn{\rno,\mP}}{\mQ}
&
&\forall \mP\in\pdis{\inpS},\mQ\in\pmea{\outA},\rno\in\reals{+}.
\end{align}
For orders other than one the closed form expression given in \eqref{eq:renyiinformation}
is equal to the following expression, which is sometimes taken as the definition of
the \renyi information,
\begin{align}
\notag %\label{eq:renyiinformation-neq-alternative}
\GMI{\rno}{\mP}{\Wm}
&=\tfrac{\rno}{\rno-1}\ln \lon{\mmn{\rno,\mP}}.
&
&\rno\in\reals{+}\setminus\{1\}.
\end{align}

Note that unlike the order \(\rno\) Augustin mean, the order \(\rno\) \renyi mean 
has a closed form expression for orders other than one,
as well. Furthermore, the inequalities given in equations 
\eqref{eq:lem:information-one:EHB},
\eqref{eq:lem:information-zto:EHB},
\eqref{eq:lem:information-oti:EHB}
of Lemma \ref{lem:information}
are replaced by the equality given in \eqref{eq:renyiinformation:EHB}. 
A discussion of the \renyi information similar to the one we have presented in this section 
for the Augustin information can be found in \cite{nakiboglu19A}. 

The order one \renyi information is equal to the order one Augustin information for all input 
distributions. For other orders such an equality does not hold for arbitrary input distributions. 
However, it is possible to characterize the Augustin information and the \renyi information in 
terms of one another through appropriate variational forms. 
Characterizing the Augustin information in a variational form in terms of the \renyi information 
is especially useful, because the Augustin information does not have a closed form expression 
whereas the \renyi information does.
This characterization also implies another variational characterization of 
the Augustin information.
\begin{lemma}\label{lem:polsha}
	Let \(\Wm\) be a channel of the form \(\Wm:\inpS\to \pmea{\outA}\) and \(\mP\) be 
	an input distribution in \(\pdis{\inpS}\).
	\begin{enumerate}[(a)]
		\item\label{polsha:poltyrev}
		Let \(\umn{\rno,\mP}\in\pdis{\inpS}\) be
		\(\umn{\rno,\mP}(\dinp)=\tfrac{\mP(\dinp)e^{(1-\rno)\RD{\rno}{\Wm(\dinp)}{\qmn{\rno,\mP}}}}{\sum_{\tilde{\dinp}}\mP(\tilde{\dinp})e^{(1-\rno)\RD{\rno}{\Wm(\tilde{\dinp})}{\qmn{\rno,\mP}}}}\) for all \(\dinp\); then
		\begin{align}
		\label{eq:lem:polsha:poltyrev}
		\RMI{\rno}{\mP}{\Wm}
		&=\GMI{\rno}{\umn{\rno,\mP}}{\Wm}+\tfrac{1}{\rno-1}\RD{1}{\mP}{\umn{\rno,\mP}}
		&
		&
		\\
		\label{eq:lem:polsha:poltyrev-variational}
		&=\begin{cases}
		\sup_{\mU\in\pdis{\inpS}} \GMI{\rno}{\mU}{\Wm}+\tfrac{1}{\rno-1}\RD{1}{\mP}{\mU}
		&\rno\in(0,1)
		\\
		\inf_{\mU\in\pdis{\inpS}} \GMI{\rno}{\mU}{\Wm}+\tfrac{1}{\rno-1}\RD{1}{\mP}{\mU}
		&\rno\in(1,\infty)
		\end{cases}.
		&
		&
		\end{align}
		\item\label{polsha:shayevitz}
		Let \(\amn{\rno,\mP}\in\pdis{\inpS}\) be
		\(\amn{\rno,\mP}(\dinp)=\tfrac{\mP(\dinp)e^{(\rno-1)\RD{\rno}{\Wm(\dinp)}{\qgn{\rno,\mP}}}}{\sum_{\tilde{\dinp}}\mP(\tilde{\dinp})e^{(\rno-1)\RD{\rno}{\Wm(\tilde{\dinp})}{\qgn{\rno,\mP}}}}\) for all \(\dinp\); 
		then
		\begin{align}
		\label{eq:lem:polsha:shayevitz}
		\GMI{\rno}{\mP}{\Wm}
		&=\RMI{\rno}{\amn{\rno,\mP}}{\Wm}-\tfrac{1}{\rno-1}\RD{1}{\amn{\rno,\mP}}{\mP}
		&
		&
		\\
		\label{eq:lem:polsha:shayevitz-variational}
		&=\begin{cases}
		\inf_{\mA\in\pdis{\inpS}} \RMI{\rno}{\mA}{\Wm}-\tfrac{1}{\rno-1}\RD{1}{\mA}{\mP}
		&\rno\in(0,1)
		\\
		\sup_{\mA\in\pdis{\inpS}} \RMI{\rno}{\mA}{\Wm}-\tfrac{1}{\rno-1}\RD{1}{\mA}{\mP}
		&\rno\in(1,\infty)
		\end{cases}.
		&
		&
		\end{align}
		\item\label{polsha:augustin}
		Let \(\fX_{\rno,\mP}:\inpS\to\reals{}\) be 
		\(\fX_{\rno,\mP}(\dinp)=[\RD{\rno}{\Wm(\dinp)}{\qmn{\rno,\mP}}-\RMI{\rno}{\mP}{\Wm}]\IND{\mP(\dinp)>0}\) 
		for all \(\dinp\); then
		\begin{align}
		\label{eq:lem:polsha:augustin}
		\RMI{\rno}{\mP}{\Wm}
		&=\tfrac{\rno}{\rno-1}\ln\EXS{\rfm}{\left(
			\sum\nolimits_{\dinp}\mP(\dinp)e^{(1-\rno)\fX_{\rno,\mP}(\dinp)}
			\left[\der{\Wm(\dinp)}{\rfm}\right]^{\rno}
			\right)^{\sfrac{1}{\rno}}}
		&
		&
		\\
		\label{eq:lem:polsha:augustin-variational}
		&=\tfrac{\rno}{\rno-1}\ln \inf\nolimits_{\fX:\EXS{\mP}{\fX}=0}\EXS{\rfm}{\left(
			\sum\nolimits_{\dinp}\mP(\dinp)e^{(1-\rno)\fX(\dinp)}
			\left[\der{\Wm(\dinp)}{\rfm}\right]^{\rno}
			\right)^{\sfrac{1}{\rno}}}.		
		&
		&
		\end{align}		
	\end{enumerate}	
\end{lemma} 
Lemma \ref{lem:polsha}-(\ref{polsha:poltyrev}) was first proved by Poltyrev,
\cite[Thm. 3.4]{poltyrev82}, in a slightly different form 
for \(\rno\in[\sfrac{1}{2},1)\) case 
assuming that \(\outS\) is finite.
Equation \eqref{eq:lem:polsha:shayevitz-variational} of Lemma \ref{lem:polsha}-(\ref{polsha:shayevitz}) 
was first proved by Shayevitz, \cite[Thm. 1]{shayevitz11}, for finite \(\outS\) case.
Shayevitz, however, neither gave the expression for the optimal \(\amn{\rno,\mP}\), 
nor asserted its existence in \cite{shayevitz11}. 
Lemma \ref{lem:polsha}-(\ref{polsha:augustin}) was first proved by 
Augustin, \cite[Lemma 35.7]{augustin78} for orders less than 
one.\footnote{\!\cite[Lemma 35.7-(d)]{augustin78} is implied by the stronger inequalities 
	established using \eqref{eq:lem:information-zto:EHB} and 
	Lemma \ref{lem:polsha}-(\ref{polsha:augustin}).} 

The following inequalities are implied by both
\(\mU=\mP\) point in the variational characterization given in Lemma \ref{lem:polsha}-(\ref{polsha:poltyrev})
and
 \(\mA=\mP\) point in the variational characterization given in Lemma \ref{lem:polsha}-(\ref{polsha:shayevitz}).
These inequalities can also be 
obtained using the Jensen's inequality and the concavity of the natural logarithm function.
%On the other hand
%\(\RD{\rno}{\mP\mtimes \Wm}{\mP\otimes\mQ}\leq \CRD{\rno}{\Wm}{\mQ}{\mP}\)
%for \(\rno\in(0,1)\)
%and
%\(\RD{\rno}{\mP\mtimes \Wm}{\mP\otimes\mQ}\geq \CRD{\rno}{\Wm}{\mQ}{\mP}\)
%for \(\rno \in (1,\infty)\)
%by Jensen's inequality and the concavity of the natural logarithm function.
%Hence
\begin{align}
\label{eq:renyiaugustin-information-zto}
\RMI{\rno}{\mP}{\Wm}
&\geq \GMI{\rno}{\mP}{\Wm}
&
&\rno \in (0,1]
\\
\label{eq:renyiaugustin-information-oti}
\RMI{\rno}{\mP}{\Wm}
&\leq \GMI{\rno}{\mP}{\Wm}
&
&\rno \in [1,\infty)
\end{align}

%!TEX root=../main-C.tex
\section{The Augustin Capacity}\label{sec:capacity}
In the previous section we have defined and analyzed the Augustin information and mean;
our main aim in this section is doing the same for the Augustin capacity and center. 
In \S\ref{sec:capacity-center}, we establish the existence of a unique Augustin center 
for all convex constraint sets with finite Augustin capacity and investigate the 
implications of the existence of an Augustin center for a given order and constraint set.  
In \S\ref{sec:capacity-order}, we analyze the Augustin capacity and center as a function 
of the order for a given constraint set.
In \S\ref{sec:miscellaneous}, we bound the Augustin capacity of the convex hull of
a collection of constraint sets on a given channel in terms of the Augustin capacities 
of individual constraint sets and determine the Augustin capacity of products of 
constraint sets on the product channels. Proofs of the propositions presented 
in this section can be found in 
Appendix \ref{sec:capacityproofs}.
 
Augustin provided a presentation similar to the current section in \cite[\S\S33,34]{augustin78}
and derived many of the key results ---such as the existence of unique Augustin center and 
its continuity as a function of order, see \cite[Lemmas 34.6, 34.7, 34.8]{augustin78}---
for orders not exceeding one.
Augustin, however, defines capacity and center only for the subsets of \(\pdis{\inpS}\) 
defined through cost constraints. 
We investigate that important special case more closely in \S\ref{sec:cost}.

\subsection{Existence of a Unique Augustin Center}\label{sec:capacity-center}
\begin{definition}\label{def:capacity}
	For any \(\rno\!\in\!\reals{+}\), \(\Wm\!:\!\inpS\!\to\!\pmea{\outA}\), and \(\cset\!\subset\!\pdis{\inpS}\),
	\emph{the order \(\rno\) Augustin capacity of \(\Wm\) for constraint set \(\cset\)} is 
	\begin{align}
\notag	%\label{eq:def:capacity}
	\CRC{\rno}{\!\Wm\!}{\cset}
	&\DEF \sup\nolimits_{\mP \in \cset}  \RMI{\rno}{\mP}{\Wm}.
	\end{align}
	When the constraint set \(\cset\) is the whole \(\pdis{\inpS}\), we denote the order \(\rno\) 
	Augustin capacity by \(\RC{\rno}{\Wm}\), i.e. 
	\(\RC{\rno}{\!\Wm\!}\DEF\CRC{\rno}{\!\Wm\!}{\pdis{\inpS}}\).
\end{definition}

Using the definition of the Augustin information  \(\RMI{\rno}{\mP}{\Wm}\) given in
\eqref{eq:def:information} we get  the following expression for \(\CRC{\rno}{\!\Wm\!}{\cset}\)
\begin{align}
\label{eq:capacity}
\CRC{\rno}{\!\Wm\!}{\cset}
&=\sup\nolimits_{\mP \in \cset}\inf\nolimits_{\mQ\in\pmea{\outA}} \CRD{\rno}{\Wm}{\mQ}{\mP}.
\end{align}
Theorem \ref{thm:minimax} in the following demonstrates that at least for convex \(\cset\)'s 
one can exchange the order of the supremum and infimum without changing the value in the 
above expression. 

\begin{theorem}\label{thm:minimax}
	For any order \(\rno\in \reals{+}\),
	channel \(\Wm\) of the form \(\Wm:\inpS\to \pmea{\outA}\), 
	and convex constraint set \(\cset\subset \pdis{\inpS}\)
	\begin{align}
	\label{eq:thm:minimax}
	\sup\nolimits_{\mP \in \cset} \inf\nolimits_{\mQ \in \pmea{\outA}} 
	\CRD{\rno}{\Wm}{\mQ}{\mP}
	&=
	\inf\nolimits_{\mQ \in \pmea{\outA}} \sup\nolimits_{\mP \in \cset} 
	\CRD{\rno}{\Wm}{\mQ}{\mP}.
	\end{align}
	If the expression on the left hand side of \eqref{eq:thm:minimax} is finite,
	i.e. if \(\CRC{\rno}{\Wm}{\cset}\!\in\!\reals{\geq0}\), 
	then \(\exists!\qmn{\rno,\!\Wm\!,\cset}\!\in\!\pmea{\outA}\),
	called the order \(\rno\) Augustin center of \(\Wm\) for the constraint set \(\cset\), 
	satisfying
	\begin{align}
	\label{eq:thm:minimaxcenter}
	\CRC{\rno}{\!\Wm\!}{\cset}
	&=\sup\nolimits_{\mP \in \cset} \CRD{\rno}{\Wm}{\qmn{\rno,\!\Wm\!,\cset}}{\mP}.
	\end{align}
	Furthermore, for every sequence of input distributions \(\{\pma{}{(\ind)}\}_{\ind\in\integers{+}}\subset\cset\)
	satisfying 
	\(\lim_{\ind \to \infty} \RMI{\rno}{\pma{}{(\ind)}}{\Wm}=\CRC{\rno}{\!\Wm\!}{\cset}\),
	the  corresponding sequence of order \(\rno\) Augustin means \(\{\qmn{\rno,\pma{}{(\ind)}}\}_{\ind\in\integers{+}}\)  
	is a Cauchy sequence  for the total variation metric on \(\pmea{\outA}\) 
	and \(\qmn{\rno,\!\Wm\!,\cset}\) is the unique limit point of that Cauchy sequence.
\end{theorem}

In order to prove Theorem \ref{thm:minimax}, we follow the program put forward 
by Kemperman \cite{kemperman74} for establishing a similar result for \(\rno=1\) 
and \(\cset=\pdis{\inpS}\) case. 
We first state and prove Theorem \ref{thm:minimax} assuming that the input set 
is finite. 
Then we generalize the result to the case with arbitrary input sets.
In the case when \(\inpS\) is a finite set, we can also assert the existence of
an optimal input distribution for which the Augustin information is equal to 
the Augustin capacity.

\begin{lemma}\label{lem:capacityFLB}
	For any order \(\rno\in\reals{+}\),
	channel \(\Wm\) of the form \(\Wm:\inpS\to \pmea{\outA}\) with a finite input set \(\inpS\), 
	and closed convex constraint set \(\cset\subset \pdis{\inpS}\), there exists
	\(\widetilde{\mP} \in \cset\) such that \(\RMI{\rno}{\widetilde{\mP}}{\Wm}=\CRC{\rno}{\!\Wm\!}{\cset}\)
	and \(\exists!\qmn{\rno,\!\Wm\!,\cset}\in\pmea{\outA}\) satisfying  
	\begin{align}
	\label{eq:lem:capacityFLB}
	\CRD{\rno}{\Wm}{\qmn{\rno,\!\Wm\!,\cset}}{\mP}
	&\leq  \CRC{\rno}{\!\Wm\!}{\cset}
	&
	&\forall \mP \in \cset.   
	\end{align}
	Furthermore, \(\qmn{\rno,\widetilde{\mP}}=\qmn{\rno,\!\Wm\!,\cset}\) for all \(\widetilde{\mP}\in \cset\) 
	such that \(\RMI{\rno}{\widetilde{\mP}}{\Wm}=\CRC{\rno}{\!\Wm\!}{\cset}\).
\end{lemma}
If \(\cset\) is \(\pdis{\inpS}\), then the expression on the right hand side of 
\eqref{eq:thm:minimaxcenter}, is equal to the \renyi radius \(\RR{\rno}{\!\Wm\!}\) 
defined in the following. Thus Theorem \ref{thm:minimax} implies
\(\RC{\rno}{\Wm}=\RR{\rno}{\Wm}\). 
\begin{definition}\label{def:renyiradius}
	For any \(\rno\in\reals{+}\) and \(\Wm:\inpS\to\pmea{\outA}\), 
	\emph{the order \(\rno\) \renyi radius of \(\Wm\)} is 
	\begin{align}
\notag	%\label{eq:def:renyiradius}
	\RR{\rno}{\!\Wm\!}
	&\DEF \inf\nolimits_{\mQ\in\pmea{\outA}} \sup\nolimits_{\dinp\in\inpS}  \RD{\rno}{\Wm(\dinp)}{\mQ}.
	\end{align}
\end{definition}

Theorem \ref{thm:minimax} asserts the existence of a unique order \(\rno\) Augustin center 
for convex constraint sets with finite Augustin capacity.
However, a probability measure \(\qmn{\rno,\!\Wm\!,\cset}\) satisfying \eqref{eq:thm:minimaxcenter}, 
i.e. an order \(\rno\) Augustin center, can in principle exist even for non-convex constraint sets.

\begin{definition}
	A constraint set \(\cset\) for the channel \(\Wm:\inpS\to\pmea{\outA}\) 
	has an order \(\rno\) Augustin center iff 
	\(\exists\mQ\in\pmea{\outA}\) such that
	\begin{align}
	\label{eq:def:center}
	\sup\nolimits_{\mP\in\cset} \CRD{\rno}{\Wm}{\mQ}{\mP}
	&=\CRC{\rno}{\!\Wm\!}{\cset}. 
	\end{align}
\end{definition}

If \(\CRC{\rno}{\!\Wm\!}{\cset}\) is infinite, then all probability measures on 
the output space satisfy \eqref{eq:def:center}  as a result of \eqref{eq:capacity} 
and the max-min inequality. 
Thus for constraint sets with infinite order \(\rno\) Augustin capacity all probability 
measures on the output space are order \(\rno\) Augustin centers.
On the other hand, some constraint sets do not have any order \(\rno\) Augustin center.
Consider for example \(\pmn{1}\) and \(\pmn{2}\) satisfying 
\(\qmn{\rno,\pmn{1}}\neq\qmn{\rno,\pmn{2}}\)
and
\(\RMI{\rno}{\pmn{1}}{\Wm}=\RMI{\rno}{\pmn{2}}{\Wm}\). 
Then \eqref{eq:def:center} is not satisfied by any probability measure for
\(\cset=\{\pmn{1},\pmn{2}\} \) and \(\cset\) does not have an order \(\rno\)
Augustin center. 
Lemma \ref{lem:center} asserts that if Augustin center exists for a 
constraint set with finite Augustin capacity, then the Augustin center
is unique.

\begin{lemma}\label{lem:center}
	Let \(\cset\subset\pdis{\inpS}\) be a constraint set satisfying \(\CRC{\rno}{\!\Wm\!}{\cset}\in\reals{\geq0}\),
	and \(\qmn{\rno,\!\Wm\!,\cset}\) be a probability measure satisfying \eqref{eq:def:center}.
	Then for every \(\{\pma{}{(\ind)}\}_{\ind\in\integers{+}}\subset \cset\) satisfying
	\(\lim_{\ind \to \infty} \RMI{\rno}{\pma{}{(\ind)}}{\Wm}=\CRC{\rno}{\!\Wm\!}{\cset}\)
	the sequence of order \(\rno\) Augustin means \(\{\qmn{\rno,\pma{}{(\ind)}}\}_{\ind\in\integers{+}}\) 
	is a Cauchy sequence with the limit point  \(\qmn{\rno,\!\Wm\!,\cset}\) 
	and the order \(\rno\) Augustin center \(\qmn{\rno,\!\Wm\!,\cset}\) is unique.  
\end{lemma}

For any \(\cset\) that has an order \(\rno\) Augustin center and a finite \(\CRC{\rno}{\!\Wm\!}{\cset}\),
Lemma \ref{lem:information}-(\ref{information:one},\ref{information:zto},\ref{information:oti}) and 
Lemma \ref{lem:center} imply that
\begin{align}
\notag%\label{eq:capacityLB}
\CRC{\rno}{\!\Wm\!}{\cset}-\RMI{\rno}{\mP}{\Wm}
&\geq  \RD{\rno\wedge 1}{\qmn{\rno,\mP}}{\qmn{\rno,\!\Wm\!,\cset}}
&
&\forall \mP \in \cset.
\end{align}
Lemma \ref{lem:information}-(\ref{information:one},\ref{information:zto},\ref{information:oti}) 
and 
Lemma \ref{lem:center} can also be used establish a lower bound on  
\(\sup\nolimits_{\mP \in \cset} \CRD{\rno}{\Wm}{\mQ}{\mP}\) 
in terms of the Augustin capacity and 
center.
\begin{lemma}\label{lem:EHB}
	For any constraint set \(\cset\) that has an order \(\rno\) Augustin center 
	and a finite \(\CRC{\rno}{\!\Wm\!}{\cset}\) we have
	\begin{align}
	\label{eq:lem:EHB}
	\sup\nolimits_{\mP \in \cset} \CRD{\rno}{\Wm}{\mQ}{\mP}
	&\geq  \CRC{\rno}{\!\Wm\!}{\cset}+\RD{\rno\wedge 1}{\qmn{\rno,\!\Wm\!,\cset}}{\mQ}
	&
	&\forall \mQ \in \pmea{\outA}.
	\end{align}
\end{lemma}

Note that the form of the lower bound given in \eqref{eq:lem:EHB} is, in a sense, 
analogous to the ones given in \eqref{eq:lem:information-one:EHB},  
\eqref{eq:lem:information-zto:EHB}, \eqref{eq:lem:information-oti:EHB}. 
The bound given in \eqref{eq:lem:EHB} is a van Erven-\harremoes 
bound\footnote{In \cite{ervenH14} van Erven and 
	\harremoes have conjectured that the inequality  
	\(\sup_{\dinp\in \inpS} \RD{\rno}{\Wm(\dinp)}{\mQ}\geq \RC{\rno}{\Wm}+\RD{\rno}{\qmn{\rno,\Wm}}{\mQ}\)
	holds for all \(\mQ\in\pmea{\outA}\).
	Van Erven and \harremoes have also proved the bound for the case when  \(\rno=\infty\),  
	assuming that \(\outS\) is countable \cite[Thm. 37]{ervenH14}. 
	We have confirmed van Erven-\harremoes conjecture in \cite[Lemma \ref*{A-lem:EHB}]{nakiboglu19A}
	and generalized it to the convex constrained case for the \renyi capacity and center in 
	\cite[Lemma \ref*{A-lem:CEHB}]{nakiboglu19A}. See \S\ref{sec:capacity-renyi} for a 
	brief discussion of the \renyi capacity and center; a more comprehensive discussion can 
	be found in \cite{nakiboglu19A}.} 
for \(\rno\in(0,1]\), but it is not a van Erven-\harremoes bound for \(\rno\in(1,\infty)\) 
because we have a \(\RD{1}{\qmn{\rno,\!\Wm\!,\cset}}{\mQ}\) term rather than a
\(\RD{\rno}{\qmn{\rno,\!\Wm\!,\cset}}{\mQ}\) term for \(\rno\in(1,\infty)\).

For orders other than one, using \csiszar\!\!'s form for the Augustin information 
given in \eqref{eq:lem:information:alternative:def} and the definition of the 
Augustin capacity, we obtain the following expressions:
\begin{align}
\label{eq:capacity:altdef}
\CRC{\rno}{\!\Wm\!}{\cset}
&=
\begin{cases}
\sup\nolimits_{\mP \in \cset}\inf\nolimits_{\Vm\in\pmea{\outA|\inpS}} 
\tfrac{\rno}{1-\rno}\CRD{1}{\Vm}{\Wm}{\mP}+\RMI{1}{\mP}{\Vm}
&\rno\in(0,1)
\\
\sup\nolimits_{\mP \in \cset}\sup\nolimits_{\Vm\in\pmea{\outA|\inpS}} 
\tfrac{\rno}{1-\rno}\CRD{1}{\Vm}{\Wm}{\mP}+\RMI{1}{\mP}{\Vm}
&\rno\in(1,\infty)
\end{cases}.
\end{align}
Then
\begin{align}
\notag
\CRC{\rno}{\!\Wm\!}{\cset}
&=\sup\nolimits_{\Vm\in\pmea{\outA|\inpS}}\sup\nolimits_{\mP \in \cset} 
\tfrac{\rno}{1-\rno}\CRD{1}{\Vm}{\Wm}{\mP}+\RMI{1}{\mP}{\Vm}
&
&\forall \rno\in(1,\infty).
\end{align}
For \(\rno\in(0,1)\), if the constraint set \(\cset\) has an order \(\rno\)
Augustin center, e.g. when \(\cset\) is convex, then one can exchange 
the order of the supremum and the infimum and  replace the infimum with 
a minimum whenever the Augustin capacity is finite
by Lemma \ref{lem:blahutidentity}, given in the following.

\begin{lemma}\label{lem:blahutidentity}
	For any \(\rno\in(0,1)\), 
	if the constraint set \(\cset\) for the channel \(\Wm:\inpS\to \pmea{\outA}\)
	has an order \(\rno\) Augustin center, then 
	\begin{align}
	\label{eq:lem:blahutminimax}
	\CRC{\rno}{\!\Wm\!}{\cset}
	&=\inf\nolimits_{\Vm\in\pmea{\outA|\inpS}} \sup\nolimits_{\mP \in \cset} 
	\tfrac{\rno}{1-\rno}\CRD{1}{\Vm}{\Wm}{\mP}+\RMI{1}{\mP}{\Vm}.
	\end{align}
	If \(\CRC{\rno}{\!\Wm\!}{\cset}\) is finite, then \(\Wma{\rno}{\qmn{\rno,\!\Wm\!,\cset}}\)
	satisfies
	\begin{align}
	\label{eq:lem:blahutidentity}
	\CRC{\rno}{\!\Wm\!}{\cset}
	&=
	\sup\nolimits_{\mP \in \cset}
	\tfrac{\rno}{1-\rno}\CRD{1}	{\Wma{\rno}{\qmn{\rno,\!\Wm\!,\cset}}}{\Wm}{\mP}
	+\RMI{1}{\mP}{\Wma{\rno}{\qmn{\rno,\!\Wm\!,\cset}}}.
	\end{align}
\end{lemma}
Lemma \ref{lem:blahutidentity} is proved using 
\csiszar\!\!'s form for the Augustin information, given in 
Lemma \ref{lem:information}-(\ref{information:alternative}),
and Lemma \ref{lem:center}.
In \cite{blahut74}, Blahut proved a similar result assuming 
both \(\inpS\) and \(\outS\) are finite sets and \(\cset=\pdis{\inpS}\).
Even under those assumptions Blahut's result \cite[Thm. 16]{blahut74} imply 
\eqref{eq:lem:blahutminimax} and \eqref{eq:lem:blahutidentity} for all orders 
in \((0,1)\) only when \(\RC{\rno}{\Wm}\) is a differentiable function of the
order \(\rno\). 
Blahut was motivated by the expression for the sphere packing exponent; 
consequently, \cite[Thm. 16]{blahut74} is stated in terms of an optimal input distribution 
at a given rate \(\rate\in (\RC{0}{\Wm},\RC{1}{\Wm})\) 
and the corresponding optimal order \(\rno^{*}(\rate)\).

\subsection{Augustin Capacity and Center as a Function of the Order}\label{sec:capacity-order}
\begin{lemma}\label{lem:capacityO}
For any channel \(\Wm\) of the form \(\Wm:\inpS\to\pmea{\outA}\) and constraint set \(\cset\subset\pdis{\inpS}\),
\begin{enumerate}[(a)]
\item\label{capacityO-ilsc}
\(\CRC{\rno}{\!\Wm\!}{\cset}\) is a nondecreasing and lower semicontinuous function of \(\rno\) on \(\reals{+}\).
\item\label{capacityO-decreasing}
\(\tfrac{1-\rno}{\rno}\CRC{\rno}{\!\Wm\!}{\cset}\) 
is a nonincreasing and continuous function of \(\rno\) 
on\footnote{We exclude \(\rno=1\) case because we do not want to assume  \(\CRC{1}{\!\Wm\!}{\cset}\) to be finite.}
\((0,1)\).
\item\label{capacityO-convexity} \((\rno-1)\CRC{\rno}{\!\Wm\!}{\cset}\) is a convex function of \(\rno\) on \((1,\infty)\).
\item\label{capacityO-continuity} 
\(\CRC{\rno}{\!\Wm\!}{\cset}\) is nondecreasing and continuous in \(\rno\) 
on \((0,1]\) and \((1,\chi_{\Wm,\cset}]\) where 
\(\chi_{\Wm,\cset}\DEF\sup\{\rnf:\CRC{\rnf}{\!\Wm\!}{\cset}\in\reals{\geq0}\}\).
\item\label{capacityO-continuity-extension} 
If \(\sup_{\mP\in\cset} \GMI{\rnf}{\mP}{\Wm}\in\reals{\geq0}\) for a \(\rnf>1\), 
then \(\CRC{\rno}{\!\Wm\!}{\cset}\) is nondecreasing and continuous in \(\rno\) 
on \((0,(1\vee\chi_{\Wm,\cset})]\). 
\end{enumerate}
\end{lemma}

The continuity results presented in parts (\ref{capacityO-continuity}) and (\ref{capacityO-continuity-extension}) 
are somewhat unsatisfactory. One would like to
either establish the continuity of \(\CRC{\rno}{\!\Wm\!}{\cset}\) from the right at \(\rno=1\) whenever
\(\CRC{\rnf}{\!\Wm\!}{\cset}\) is finite for a \(\rnf>1\)
or provide a channel \(\Wm\) and a constraint set \(\cset\)  for which \(\CRC{\rnf}{\!\Wm\!}{\cset}\) is 
finite for a \(\rnf>1\) and \(\lim_{\rno \downarrow 1} \CRC{\rno}{\!\Wm\!}{\cset}>\CRC{1}{\!\Wm\!}{\cset}\).
We could not do either. Instead we establish the continuity of \(\CRC{\rno}{\!\Wm\!}{\cset}\) from the right 
at \(\rno=1\) assuming  that \(\sup_{\mP\in\cset} \GMI{\rnf}{\mP}{\Wm}\) is finite for a \(\rnf>1\).

%%%%If \(\RC{\rnf}{\!\Wm\!}\) is finite, then \(\sup_{\mP\in\cset} \GMI{\rnf}{\mP}{\Wm}\) is
%%%%finite for all \(\cset\subset\pdis{\inpS}\) by 
%%%% \eqref{eq:def:renyiinformation} and Theorem \ref{thm:minimax}. 
%%%%Thus \(\CRC{\rno}{\!\Wm\!}{\cset}\) is nondecreasing and continuous in \(\rno\) on 
%%%%\((0,\chi_{\Wm,\cset}]\) for all \(\cset\subset\pdis{\inpS}\),
%%%%provided that \(\RC{\rnf}{\Wm}\) is finite for a \(\rnf>1\).

Since \(\RC{\rnf}{\!\Wm\!}=\RR{\rnf}{\!\Wm\!}\) by Theorem \ref{thm:minimax}
and \(\GMI{\rnf}{\mP}{\Wm}\leq \RR{\rnf}{\!\Wm\!}\) for all
\(\mP\in\pdis{\inpS}\) by \eqref{eq:def:renyiinformation},
\(\sup_{\mP\in\cset} \GMI{\rnf}{\mP}{\Wm}\) is
finite for all \(\cset\subset\pdis{\inpS}\) whenever 
\(\RC{\rnf}{\!\Wm\!}\) is finite.
Thus  \(\CRC{\rno}{\!\Wm\!}{\cset}\) is nondecreasing and continuous in \(\rno\) on 
\((0,\chi_{\Wm,\cset}]\) for all \(\cset\subset\pdis{\inpS}\),
provided that \(\RC{\rnf}{\Wm}\) is finite for a \(\rnf>1\).

Lemma \ref{lem:EHB} allows us to use  the continuity of \(\CRC{\rno}{\!\Wm\!}{\cset}\) in \(\rno\) and  
Lemma \ref{lem:divergence-pinsker} to establish the continuity of \(\qmn{\rno,\!\Wm\!,\cset}\) in  \(\rno\) for the 
total variation topology  on \(\pmea{\outA}\).

\begin{lemma}\label{lem:centercontinuity} 
For any \(\rnt\in \reals{+}\), \(\Wm:\inpS\to \pmea{\outA}\), 
and convex \(\cset\subset \pdis{\inpS}\) such that \(\CRC{\rnt}{\!\Wm\!}{\cset}\in\reals{+}\),
\begin{align}
\label{eq:lem:centercontinuity}
\RD{\rno\wedge 1}{\qmn{\rno,\!\Wm\!,\cset}}{\qmn{\rnf,\Wm,\cset}} 
&\leq \CRC{\rnf}{\Wm}{\cset}-\CRC{\rno}{\!\Wm\!}{\cset}
&
&\forall \rno,\rnf \mbox{~such that~}0<\rno<\rnf\leq\rnt.
\end{align}
Consequently, if \(\CRC{\rno}{\!\Wm\!}{\cset}\) is continuous in \(\rno\) on  
\(\alg{I}\) for some \(\alg{I}\subset (0,\rnt]\), 
then \(\qmn{\rno,\!\Wm\!,\cset}:\alg{I}\to \pmea{\outA}\) 
is continuous  in \(\rno\) on \(\alg{I}\) 
for the total variation topology on \(\pmea{\outA}\).
\end{lemma}

\subsection{Convex Hulls of Constraints and Product Constraints}\label{sec:miscellaneous}
In the following we consider two kinds of frequently encountered constraint sets 
that are described in terms of simpler constraint sets.
Lemma \ref{lem:capacityunion} considers convex hull of a family constraint sets
and bounds the Augustin capacity for the convex hull in terms of the Augustin 
capacities of the individual constraint sets.
Lemma \ref{lem:capacityproduct} considers a product channel 
for the constraint set that is the product of
convex hulls of the constraint sets on the component channels 
that have Augustin centers
and shows that Augustin capacity has an additive form and Augustin center has a 
product form. 

\begin{lemma}\label{lem:capacityunion}
Let \(\rno\) be a positive real, \(\Wm\) be a channel of the form
\(\Wm:\inpS\to \pmea{\outA}\), and 
\(\cset^{(\ind)}\) be a constraint set that has an order \(\rno\) Augustin center 
and a finite \(\CRC{\rno}{\Wm}{\cset^{(\ind)}}\) for all \(\ind\in\tinS\).
Then
\begin{align}
\notag %\label{eq:lem:capacityunion}
\sup\nolimits_{\ind \in \tinS} \CRC{\rno}{\!\Wm\!}{\cset^{(\ind)}}
\leq 
\CRC{\rno}{\!\Wm\!}{\cset}
&\leq \ln \sum\nolimits_{\ind\in \tinS} e^{\CRC{\rno}{\!\Wm\!}{\cset^{(\ind)}}}
\end{align}
where \(\cset\) is the convex hull of the union, i.e. \(\cset=\conv{(\cup_{\ind\in \tinS} \cset^{(\ind)})}\). 
Furthermore, 
\begin{itemize}
	\item \(\CRC{\rno}{\!\Wm\!}{\cset^{(\ind)}}=\CRC{\rno}{\!\Wm\!}{\cset}<\infty\) \(\Leftrightarrow\) 
	\(\sup_{\mP\in\cset} \CRD{\rno}{\Wm}{\qmn{\rno,\!\Wm\!,\cset^{(\ind)}}}{\mP}\leq \CRC{\rno}{\!\Wm\!}{\cset^{(\ind)}}\)
	\(\Rightarrow\) 
	\(\qmn{\rno,\!\Wm\!,\cset}=\qmn{\rno,\!\Wm\!,\cset^{(\ind)}}\). 
	\item \(\CRC{\rno}{\!\Wm\!}{\cset}=\ln \sum\nolimits_{\ind\in\tinS} e^{\CRC{\rno}{\!\Wm\!}{\cset^{(\ind)}}}<\infty\)
	\(\Leftrightarrow\) 
	\(\qmn{\rno,\!\Wm\!,\cset^{(\ind)}}\perp\qmn{\rno,\!\Wm\!,\cset^{(\jnd)}}~~\forall\ind \neq \jnd\)
	and \(\abs{\tinS}<\infty\)
	\(\Rightarrow\)
	\(\qmn{\rno,\!\Wm\!,\cset}=\sum\nolimits_{\ind\in\tinS}
	\tfrac{e^{\CRC{\rno}{\!\Wm\!}{\cset^{(\ind)}}}}{e^{\CRC{\rno}{\!\Wm\!}{\cset}}} 
	\qmn{\rno,\!\Wm\!,\cset^{(\ind)}}\).
\end{itemize}
\end{lemma}
Note that if \(\cset^{(\ind)}\) is convex and \(\CRC{\rno}{\Wm}{\cset^{(\ind)}}\) is finite, 
then 
\(\cset^{(\ind)}\) has a unique order \(\rno\) Augustin center by Theorem \ref{thm:minimax}.
 
\begin{lemma}\label{lem:capacityproduct}
For any \(\rno\in\reals{+}\),
length \(\blx\) product channel
\(\Wmn{[1,\blx]}:\inpS_{1}^{\blx}\to\pmea{\outA_{1}^{\blx}}\),
and  constraint sets \(\cset_{\tin}\subset\pdis{\inpS_{\tin}}\) 
that have order \(\rno\) Augustin centers
	\begin{align}
\notag  %	\label{eq:lem:capacityproduct}
	\CRC{\rno}{\!\Wmn{[1,\blx]}\!}{\cset}
	&=\CRC{\rno}{\!\Wmn{[1,\blx]}\!}{\cset_{1}^{\blx}}
	=\sum\nolimits_{\tin=1}^{\blx}\CRC{\rno}{\Wmn{\tin}}{\cset_{\tin}}
	\end{align}
	where \(\cset\!=\!\{\mP\!\in\!\pdis{\inpS_{1}^{\blx}}:\pmn{\tin}\!\in\!\conv{\cset_{\tin}}~\forall\tin\in \{1,\ldots,\blx\}\}\),
	i.e. a \(\mP\in\pdis{\inpS_{1}^{\blx}}\) is in \(\cset\) iff 
	for all  \(\tin\in\{1,\ldots,\blx\}\) 	its 
	\(\inpS_{\tin}\) marginal \(\pmn{\tin}\) is in the convex hull of \(\cset_{\tin}\).
	Furthermore, if \(\CRC{\rno}{\Wmn{\tin}}{\cset_{\tin}}\) is finite for all \(\tin\in\{1,\ldots,\blx\}\), then 
	\(\qmn{\rno,\!\Wmn{[1,\blx]}\!,\cset}=\qmn{\rno,\!\Wmn{[1,\blx]}\!,\cset_{1}^{\blx}}=\bigotimes\nolimits_{\tin=1}^{\blx}\qmn{\rno,\Wmn{\tin},\cset_{\tin}}\).
\end{lemma}
\begin{remark}
Note that the convex hull of any subset of \(\cset\) is a subset of \(\cset\)
because \(\cset\) is convex by definition. In particular, \(\cset_{1}^{\blx}\!\subset\!\conv{\cset_{1}^{\blx}}\!\subset\!\cset\).
Then \(\CRC{\rno}{\!\Wmn{[1,\blx]}\!}{\conv{\cset_{1}^{\blx}}}\!=\!\sum\nolimits_{\tin=1}^{\blx}\CRC{\rno}{\Wmn{\tin}}{\cset_{\tin}}\)
by Lemma \ref{lem:capacityproduct}.
Furthermore,
if \(\CRC{\rno}{\Wmn{\tin}}{\cset_{\tin}}\) is finite for all \(\tin\in\{1,\ldots,\blx\}\), 
then \(\qmn{\rno,\!\Wmn{[1,\blx]}\!,\conv{\cset_{1}^{\blx}}}=\bigotimes\nolimits_{\tin=1}^{\blx}\qmn{\rno,\Wmn{\tin},\cset_{\tin}}\)
by Lemma \ref{lem:capacityunion}.
\end{remark}
\begin{remark}
The constraint set \(\cset_{1}^{\blx}\) described in 
Lemma \ref{lem:capacityproduct} may not be convex, yet
\(\cset_{1}^{\blx}\) is guaranteed 
to have an order \(\rno\) Augustin center.
\end{remark}

\subsection{Augustin Capacity vs \renyi Capacity}\label{sec:capacity-renyi}
Using the \renyi information instead of the Augustin information, one can define the \renyi capacity,
as follows.
\begin{definition}\label{def:renyicapacity}
	For any \(\rno\in\reals{+}\), \(\Wm:\inpS\to\pmea{\outA}\), and \(\cset\subset\pdis{\inpS}\)
	\emph{the order \(\rno\) \renyi capacity of \(\Wm\) for constraint set \(\cset\)} is 
	\begin{align}
\notag	%\label{eq:def:renyicapacity}
	\CGC{\rno}{\!\Wm\!}{\cset}
	&\DEF \sup\nolimits_{\mP \in \cset}  \GMI{\rno}{\mP}{\Wm}.
	\end{align}
	When the constraint set \(\cset\) is the whole \(\pdis{\inpS}\), we denote the order \(\rno\) 
	\renyi capacity by \(\GCL{\rno}{\Wm}{}\), i.e. 
	\(\GCL{\rno}{\!\Wm\!}{}\DEF\CGC{\rno}{\!\Wm\!}{\pdis{\inpS}}\).
\end{definition}

Since \(\RMI{1}{\mP}{\Wm}\!=\!\GMI{1}{\mP}{\Wm}\), \(\CGC{1}{\Wm}{\cset}\!=\!\CRC{1}{\Wm}{\cset}\) 
by definition. 
We cannot say the same for other orders;
by \eqref{eq:renyiaugustin-information-zto},
\eqref{eq:renyiaugustin-information-oti} we have
\begin{align}
\notag   %\label{eq:renyiaugustin-capacity-zto}
\CGC{\rno}{\Wm}{\cset}
&\leq \CRC{\rno}{\!\Wm\!}{\cset}
&
&\rno \in (0,1],
\\
\notag   %\label{eq:renyiaugustin-capacity-oti}
\CGC{\rno}{\Wm}{\cset}
&\geq \CRC{\rno}{\!\Wm\!}{\cset}
&
&\rno \in [1,\infty).
\end{align}
As a result of definitions of the \renyi information and capacity we have
\begin{align}
\notag
%\label{eq:capacity-renyi}
\CGC{\rno}{\!\Wm\!}{\cset}
&=\sup\nolimits_{\mP \in \cset}\inf\nolimits_{\mQ\in\pmea{\outA}} \RD{\rno}{\mP\mtimes\Wm}{\mP\otimes\mQ}.
\end{align}
The \renyi capacity satisfies a minimax theorem, \cite[Thm. \ref*{A-thm:Cminimax}]{nakiboglu19A}, 
similar to Theorem \ref{thm:minimax}: For any convex constraint set \(\cset\subset\pdis{\inpS}\)
\begin{align}
\notag
\sup\nolimits_{\mP \in \cset}\inf\nolimits_{\mQ\in\pmea{\outA}} \RD{\rno}{\mP\mtimes\Wm}{\mP\otimes\mQ}
&=
\inf\nolimits_{\mQ\in\pmea{\outA}}\sup\nolimits_{\mP \in \cset} \RD{\rno}{\mP\mtimes\Wm}{\mP\otimes\mQ}.
\end{align}
If \(\CGC{\rno}{\Wm}{\cset}\) is finite, then \(\exists!\qgn{\rno,\Wm,\cset}\in\pmea{\outA}\), 
\emph{the order \(\rno\) \renyi center  \(\Wm\) for the constraint set \(\cset\)}, 
satisfying
\begin{align}
\notag
\CGC{\rno}{\!\Wm\!}{\cset}
&=\sup\nolimits_{\mP \in \cset} \RD{\rno}{\mP\mtimes\Wm}{\mP\otimes\qgn{\rno,\Wm,\cset}}.
\end{align}
Consequently, the \renyi capacity equals to the \renyi radius provided that \(\cset\!=\!\pdis{\inpS}\).
Hence \(\GCL{\rno}{\!\Wm\!}{}\!=\!\RC{\rno}{\!\Wm\!}\) and \(\qgn{\rno,\Wm}\!=\!\qmn{\rno,\Wm}\)
by Theorem \ref{thm:minimax}.
The other observations presented in this section have their counter parts for the \renyi capacity and center;
compare for example 
Lemma \ref{lem:EHB} and \cite[Lemma \ref*{A-lem:CEHB}]{nakiboglu19A}.

%%%%%%Using the above described minimax theorem for the \renyi capacity, 
%%%%%%i.e. \cite[Thm. \ref*{A-thm:Cminimax}]{nakiboglu19A},
%%%%%%we have established the following van Erven-\harremoes bound
%%%%%%in  \cite[Lemma \ref*{A-lem:CEHB}]{nakiboglu19A}:
%%%%%%for any convex constraint set \(\cset\subset \pdis{\inpS}\) satisfying 
%%%%%%\(\CGC{\rno}{\Wm}{\cset}<\infty\)
%%%%%%\begin{align}
%%%%%%\notag
%%%%%%\sup\nolimits_{\mP \in \cset} \RD{\rno}{\mP\mtimes \Wm}{\mP\otimes\mQ}
%%%%%%&\geq  \CGC{\rno}{\!\Wm\!}{\cset}+\RD{\rno}{\qgn{\rno,\!\Wm\!,\cset}}{\mQ}
%%%%%%&
%%%%%%&\forall \mQ \in \pmea{\outA}.
%%%%%%\end{align} 
%%%%%%Note that this bound implies that Lemma \ref{lem:EHB} can be strengthened
%%%%%%as follows, for \(\cset=\pdis{\inpS}\) case 
%%%%%%\begin{align}
%%%%%%\notag
%%%%%%\sup\nolimits_{\mP \in \pdis{\inpS}} \CRD{\rno}{\Wm}{\mQ}{\mP}
%%%%%%&\geq  \RC{\rno}{\Wm}+\RD{\rno}{\qmn{\rno,\Wm}}{\mQ}
%%%%%%&
%%%%%%&\forall \mQ \in \pmea{\outA}.
%%%%%%\end{align}
%%%%%%Other observations presented in this section have their counter parts 
%%%%%%for the \renyi capacity and center. 
%%%%%%Compare for example 
%%%%%%Lemma \ref{lem:capacityunion} and \cite[Lemma \ref*{A-lem:constrainedcapacityUnion}]{nakiboglu19A}
%%%%%%or
%%%%%%Lemma \ref{lem:capacityproduct} and \cite[Lemma \ref*{A-lem:constrainedcapacityProduct}]{nakiboglu19A}.

%!TEX root=../main-C.tex
\section{The Cost Constrained Problem}\label{sec:cost}
In the previous section, we have defined the Augustin capacity for arbitrary constraint sets 
and proved the existence of a unique Augustin center for any convex constraint set with
finite Augustin capacity. 
The convex constraint sets of interest are often defined via the cost constraints;
the main aim of this section is to investigate this important special case more closely.
In \S\ref{sec:cost-CC} we investigate the immediate consequences of the definition 
of the cost constrained Augustin capacity and ramifications of the analysis presented in the previous section. 
In \S\ref{sec:cost-AL} we define and analyze the Augustin-Legendre (A-L) 
information, capacity, radius, and center. 
The discussion in \S\ref{sec:cost-AL} is a generalization of certain parts of 
the analysis presented by \csiszar and \korner in \cite[Ch. 8]{csiszarkorner} 
for the supremum of the mutual information for discrete channels with single cost constraint, 
i.e. \(\rno=1\), \(\abs{\inpS}<\infty\), \(\abs{\outS}<\infty\), \(\ell=1\) case.
In \S\ref{sec:cost-RG} we define and analyze 
the \renyi\!\!-Gallager (R-G) information, mean, capacity,
radius, and center.
The most important conclusion of our analysis in \S\ref{sec:cost-RG} is the equality of 
the A-L capacity and center to the R-G capacity and center. 
In \S\ref{sec:cost-TP}, we demonstrate how the results presented in 
\S\ref{sec:cost-CC}, \S\ref{sec:cost-AL}, and \S\ref{sec:cost-RG} 
can be used to determine the Augustin capacity and center of a 
transition probability with cost constraints.
Proofs of the propositions presented in \S\ref{sec:cost-CC}, 
\S\ref{sec:cost-AL}, and \S\ref{sec:cost-RG} can be found is 
Appendix \ref{sec:costproofs}.

Augustin presented a discussion of the cost constrained capacity 
\(\CRC{\rno}{\!\Wm\!}{\costc}\) in \cite[\S34]{augustin78}
for the case when the cost function \(\costf\) is a bounded function of the form \(\costf:\inpS\to [0,1]^{\ell}\)
and the order \(\rno\) is in \((0,1]\). 
In \cite[\S35]{augustin78}, Augustin also analyzed quantities closely related to 
the R-G information and capacity.
The quantities analyzed by Augustin in \cite[\S35]{augustin78} have first appeared in Gallager's error 
exponents analysis for cost constrained channels \cite[\S6]{gallager65}, 
\cite[\S7.3,\S7.4,\S7.5]{gallager}.
Unlike Augustin, Gallager did not assume \(\costf\) to be bounded;
but Gallager confined his analysis to the case when there is a single cost constraint,
i.e. \(\ell=1\) case, and refrained from defining the R-G capacity as a quantity that is of 
interest on its own 
right.
Other authors studying cost constrained problems,
\cite[\S IV]{arimoto76},  \cite{oohama17A,oohama17B,vazquezMF15},
have considered the R-G information and capacity, as well.
Yet to the best of our knowledge for orders other than one 
the A-L information measures, 
which are obtained through a more direct application of convex conjugation,
have not been studied before. 

\subsection{The Cost Constrained Augustin Capacity and Center}\label{sec:cost-CC}
We denote the set of all probability mass functions satisfying a cost constraint \(\costc\) by \(\cset(\costc)\), i.e.
\begin{align}
\notag %\label{eq:def:costconstraint}
\cset(\costc)
&\DEF \{\mP\in\pdis{\inpS}:\EXS{\mP}{\costf}\leq\costc\}.
\end{align}
\(\cset(\costc)\!\neq\!\emptyset\) iff \(\costc\!\in\!\fcc{\costf}\) 
where \(\fcc{\costf}\) is defined  in \eqref{eq:def:feasible} 
as the set of all feasible cost constraints for the cost function \(\costf\). 
\(\cset(\costc)\) is nondecreasing in \(\costc\), i.e. \(\costc_{1}\!\leq\!\costc_{2}\) implies
 \(\cset(\costc_{1})\!\subset\!\cset(\costc_{2})\).
We define the order \(\rno\) Augustin capacity of \(\Wm\) for the cost constraint \(\costc\) as
\begin{align}
\label{eq:def:costcapacity}
\CRC{\rno}{\!\Wm\!}{\costc}
&\DEF 
\begin{cases}
\sup\nolimits_{\mP\in\cset(\costc)} \RMI{\rno}{\mP}{\Wm}
&\mbox{if~}\costc\in \fcc{\costf}
\\
-\infty
&\mbox{if~}\costc\in \reals{\geq0}^{\ell}\setminus\fcc{\costf}
\end{cases}
&
&\forall\rno\in\reals{+}.
\end{align}
We defined \(\CRC{\rno}{\!\Wm\!}{\costc}\) for \(\costc\)'s that are not feasible
in order to be able to use standard results without modifications.
Since \(\cset(\costc)\) is a convex set, Theorem \ref{thm:minimax} holds for \(\cset(\costc)\).
We denote\footnote{This slight abuse of notation ---which can be avoided by using 
\(\CRC{\rno}{\!\Wm\!}{\cset(\costc)}\) and \(\qmn{\rno,\!\Wm\!,\cset(\costc)}\) instead of 
\(\CRC{\rno}{\!\Wm\!}{\costc}\) and \(\qmn{\rno,\!\Wm\!,\costc}\)--- 
provides brevity without leading to any notational ambiguity.} 
the order \(\rno\) Augustin center of \(\Wm\) for the cost constraint \(\costc\) by 
\(\qmn{\rno,\!\Wm\!,\costc}\).

For a given order \(\rno\), the Augustin capacity \(\CRC{\rno}{\!\Wm\!}{\costc}\) is a 
concave function of the cost constraint \(\costc\). Hence,
if it is finite at an interior point of \(\fcc{\costf}\), 
then it is a continuous function of the cost constraint \(\costc\) 
that lies below its tangent planes drawn at interior points of \(\fcc{\costc}\). 
Lemma \ref{lem:CCcapacity}, in the following, summarizes these observations. 
\begin{lemma}\label{lem:CCcapacity}
Let \(\Wm\) be a channel of the form \(\Wm:\inpS\to \pmea{\outA}\) 
with the cost function \(\costf\) of the form \(\costf:\inpS\to \reals{\geq0}^{\ell}\).
\begin{enumerate}[(a)]
\item\label{CCcapacity:function} 
For any \(\rno\in\reals{+}\),
\(\CRC{\rno}{\!\Wm\!}{\costc}\) is a nondecreasing and concave function of \(\costc\) 
on \(\reals{\geq0}^{\ell}\),
which is either infinite on every point in \(\inte{\fcc{\costf}}\)
or finite and continuous on \(\inte{\fcc{\costf}}\). 
\item\label{CCcapacity:interior}
If \(\CRC{\rno}{\!\Wm\!}{\costc}\) is finite on \(\inte{\fcc{\costf}}\) for an \(\rno\in\reals{+}\), 
then 
for every \(\costc\in\inte{\fcc{\costf}}\) there exists a \(\lgm_{\rno,\!\Wm\!,\costc}\in \reals{\geq0}^{\ell}\) such that
\begin{align}
\label{eq:CCcapacity:interior}
\CRC{\rno}{\!\Wm\!}{\costc}+\lgm_{\rno,\!\Wm\!,\costc}\cdot(\tilde{\costc}-\costc)
&\geq \CRC{\rno}{\!\Wm\!}{\tilde{\costc}}
&
&\forall \tilde{\costc}\in\reals{\geq0}^{\ell}
\end{align}
Furthermore, the set of all such \(\lgm_{\rno,\!\Wm\!,\costc}\)'s is convex and compact.
\item\label{CCcapacity:zeroonedichotomy}
Either \(\CRC{\rno}{\!\Wm\!}{\costc}\!=\!\infty\) for all \((\rno,\costc)\!\in\!(0,1)\!\times\!\inte{\fcc{\costf}}\)
or \(\CRC{\rno}{\!\Wm\!}{\costc}\) 
and \(\qmn{\rno,\!\Wm\!,\costc}\)
are continuous in \((\rno,\costc)\) on \((0,1)\!\times\!\inte{\fcc{\costf}}\)
for the total variation topology on \(\pmea{\outA}\).
\end{enumerate}
\end{lemma}

If the cost function for a product channel is additive, 
then the cost constrained Agustin capacity 
of the product channel is equal to the supremum of the sum of the cost constrained Augustin capacities
of the component channels over all feasible cost allocations. Furthermore, if there exists an 
optimal cost allocation, 
then the Augustin center of the product channel is a product measure.
Lemma \ref{lem:CCcapacityproduct}, given in the following, states these observations formally.

\begin{lemma}\label{lem:CCcapacityproduct}
For any length \(\blx\) product channel
\(\Wmn{[1,\blx]}:\inpS_{1}^{\blx}\to\pmea{\outA_{1}^{\blx}}\)
and additive cost function 
\(\costf_{[1,\blx]}:\inpS_{1}^{\blx}\to\reals{\geq0}^{\ell}\)
%satisfying 
%\(\costf_{[1,\blx]}(\dinp_{1}^{\blx})\!=\!\sum_{\tin=1}^{\blx}\!\costf_{\tin}(\dinp_{\tin})\) 
%for some \(\costf_{\tin}:\inpS_{\tin}\to\reals{\geq0}^{\ell}\)
we have\footnote{If \(\CRC{\rno}{\!\Wmn{\tin}}{\costc_{\tin}}=-\infty\) for any 
	\(\tin\in\{1,\ldots,\blx\}\), then 
	\(\sum\nolimits_{\tin=1}^{\blx} \CRC{\rno}{\!\Wmn{\tin}}{\costc_{\tin}}\) 
	stands for \(-\infty\); even if one or more of other
	\(\CRC{\rno}{\!\Wmn{\tin}}{\costc_{\tin}}\)'s are equal to \(\infty\).}
\begin{align}
\label{eq:lem:CCcapacityproduct}
\CRC{\rno}{\!\Wmn{[1,\blx]}}{\costc}
&=\sup\left\{\sum\nolimits_{\tin=1}^{\blx} \CRC{\rno}{\!\Wmn{\tin}}{\costc_{\tin}}:
\sum\nolimits_{\tin=1}^{\blx} \costc_{\tin}\leq \costc,~
\costc_{\tin}\in \reals{\geq0}^{\ell} \right\}
&
&\forall \costc\in \reals{\geq0}^{\ell}, \rno\in\reals{+}. 
\end{align}
If 
\(\CRC{\rno}{\!\Wmn{[1,\blx]}}{\costc}\in\reals{\geq0}\) for an \(\rno\in\reals{+}\)
and \(\exists(\costc_{1},\ldots,\costc_{\blx})\) such that
\(\CRC{\rno}{\!\Wmn{[1,\blx]}}{\costc}=\sum\nolimits_{\tin=1}^{\blx} \CRC{\rno}{\!\Wmn{\tin}}{\costc_{\tin}}\),
then \(\qmn{\rno,\!\Wmn{[1,\blx]},\costc}=\bigotimes\nolimits_{\tin=1}^{\blx}\qmn{\rno,\Wmn{\tin},\costc_{\tin}}\).
\end{lemma}

Since the Augustin capacity is concave in the cost constraint by Lemma \ref{lem:CCcapacity}-(\ref{CCcapacity:function}),
\(\CRC{\rno}{\!\Wmn{[1,\blx]}}{\costc}=\sum\nolimits_{\tin=1}^{\blx} \CRC{\rno}{\Wmn{\tin}}{\frac{\costc}{\blx}}\)
whenever \(\Wmn{[1,\blx]}\) is stationary and \(\costf_{\tin}=\costf_{1}\) for all \(\tin\in\{1,\ldots,\blx\}\).
Alternatively, if \(\fcc{\costf_{\tin}}\)'s are closed and 
\(\CRC{\rno}{\Wmn{\tin}}{\costc}\)'s are upper semicontinuous functions of \(\costc\)
on \(\fcc{\costf_{\tin}}\)'s, then we can use the extreme value 
theorem\footnote{Consider the function \(\fX(\costc_{1},\ldots,\costc_{\blx})\) which is equal to 
\(\sum\nolimits_{\tin=1}^{\blx} \CRC{\rno}{\!\Wmn{\tin}}{\costc_{\tin}}\)
if \(\sum\nolimits_{\tin=1}^{\blx} \costc_{\tin}\leq \costc\) and
\(\costc_{\tin}\in\fcc{\costf_{\tin}}\) for all \(\tin\in\{1,\ldots,\blx\}\)
and which is equal to \(-\infty\) otherwise. 
We choose a large enough but bounded set using the vector \(\costc\)
to obtain a compact set for the supremum.} 
for the upper semicontinuous functions to establish the 
existence of a \((\costc_{1},\ldots,\costc_{\blx})\) satisfying
both
\(\CRC{\rno}{\!\Wmn{[1,\blx]}}{\costc}=\sum\nolimits_{\tin=1}^{\blx} \CRC{\rno}{\!\Wmn{\tin}}{\costc_{\tin}}\)
and \(\sum\nolimits_{\tin=1}^{\blx} \costc_{\tin}\leq\costc\).
However, such an existence assertion does not hold in general,
see Example \ref{eg:non-additive-product}.

\subsection{The Augustin-Legendre Information Measures}\label{sec:cost-AL}
The cost constrained Augustin capacity \(\CRC{\rno}{\!\Wm\!}{\costc}\) and center \(\qmn{\rno,\!\Wm\!,\costc}\) 
can be characterized using convex conjugation, as well. 
In this part of the paper, we introduce and analyze the concepts of the Augustin-Legendre information, 
capacity, center, and radius in order to obtain a more complete understanding of this characterization. 
The current method seems to us to be the standard application of the convex conjugation technique 
to characterize the cost constrained Augustin capacity. 
Yet, it is not the customary method.
Starting with the seminal work of Gallager \cite{gallager65}, 
a more ad hoc method based on the \renyi information became the customary way to 
apply Lagrange multipliers techniques to characterize the Augustin capacity,
see \cite[\S 35]{augustin78}, \cite{oohama17A,oohama17B}.
We discuss that approach in \S\ref{sec:cost-RG}. 
Theorem \ref{thm:Lminimax} presented in the following
and 
Theorem \ref{thm:Gminimax} presented in \S\ref{sec:cost-RG}
establish the equivalence of these two approaches by establishing 
the equality of 
the Augustin-Legendre capacity and center to the \renyi\!\!-Gallager 
capacity and center.

\begin{definition}\label{def:Linformation}
For any \(\rno\in\reals{+}\), channel \(\Wm\) of the form \(\Wm:\inpS\to \pmea{\outA}\) with 
a cost function \(\costf:\inpS\to \reals{\geq0}^{\ell}\), \(\mP\in \pdis{\inpS}\),  
and \(\lgm \in \reals{\geq0}^{\ell}\), 
\emph{the order \(\rno\) Augustin-Legendre information for
	the input distribution \(\mP\) 
	and the Lagrange multiplier \(\lgm\)} is
\begin{align}
\label{eq:def:Linformation}
\RMIL{\rno}{\mP}{\Wm}{\lgm}
&\DEF \RMI{\rno}{\mP}{\Wm}-\lgm\cdot \EXS{\mP}{\costf}.
\end{align}
\end{definition}
Note that as an immediate consequence of the definition of the A-L information we 
have
\begin{align}
\label{eq:information-constrained}
\inf\nolimits_{\lgm\geq 0} \RMIL{\rno}{\mP}{\Wm}{\lgm}+\lgm\cdot\costc
&=\xi_{\rno,\mP}(\costc)
\end{align}
where \(\xi_{\rno,\mP}(\cdot):\reals{\geq0}^{\ell}\to[-\infty,\infty)\) is defined as
\begin{align}
\label{eq:def:information-constrained}
\xi_{\rno,\mP}(\costc)
&\DEF \begin{cases}
\RMI{\rno}{\mP}{\Wm}		&\costc\geq \EXS{\mP}{\costf}
\\
-\infty						&\mbox{else}
\end{cases}.
\end{align}
Then the Augustin-Legendre information \(\RMIL{\rno}{\mP}{\Wm}{\lgm}\) can also be expressed 
as
\begin{align}
\label{eq:Linformation-conjugate}
\RMIL{\rno}{\mP}{\Wm}{\lgm}
&=\sup\nolimits_{\costc\geq0} \xi_{\rno,\mP}(\costc)-\lgm\cdot\costc
\end{align}
\begin{remark}
Note that if \(\fX:\reals{\geq0}^{\ell}\to (-\infty,\infty]\)
		and \(\fX^{*}:(-\infty,0]^{\ell}\to \reals{}\) are defined as 
		\(\fX(\costc)\DEF-\xi_{\rno,\mP}(\costc)\)
		and 
		\(\fX^{*}(\gamma)\DEF\RMIL{\rno}{\mP}{\Wm}{-\gamma}\),
		then \(\fX^{*}\) is the convex conjugate, i.e. Legendre transform, of the convex function
		\(\fX\). This is why we call \(\RMIL{\rno}{\mP}{\Wm}{\lgm}\) the 
		Augustin-Legendre  information.
	\end{remark}
\begin{definition}\label{def:Lcapacity}
For any \(\rno\in\reals{+}\), channel \(\Wm\) of the form \(\Wm:\inpS\to \pmea{\outA}\) with 
a cost function \(\costf:\inpS\to \reals{\geq0}^{\ell}\), and \(\lgm \in \reals{\geq0}^{\ell}\)
\emph{the order \(\rno\) Augustin-Legendre (A-L) capacity for the Lagrange multiplier \(\lgm\)} is
\begin{align}
\label{eq:def:Lcapacity}
\RCL{\rno}{\Wm}{\lgm}
&\DEF \sup\nolimits_{\mP\in \pdis{\inpS}} \RMIL{\rno}{\mP}{\Wm}{\lgm}.
\end{align}
\end{definition}
Then as a result of \eqref{eq:def:information-constrained} and \eqref{eq:Linformation-conjugate}
we have
\begin{align}
\label{eq:Lcapacity-astheconjugate}
\RCL{\rno}{\Wm}{\lgm}
&=\sup\nolimits_{\costc\geq0} \CRC{\rno}{\Wm}{\costc}-\lgm\cdot\costc
&
&\forall \lgm\in\reals{\geq0}^{\ell}.
\end{align}
Hence, using the max-min inequality we can conclude that
\begin{align}
\label{eq:CCcapacity:dualitybound}
\CRC{\rno}{\!\Wm\!}{\costc} 
&\leq \inf\nolimits_{\lgm\geq 0} \RCL{\rno}{\Wm}{\lgm}+\lgm\cdot \costc
&
&\forall \costc\in\reals{\geq0}^{\ell}.
\end{align}
Then \(\CRC{\rno}{\!\Wm\!}{\costc}<\infty\) for all \(\costc\in\reals{\geq0}^{\ell}\)
provided that \(\RCL{\rno}{\Wm}{\lgm}<\infty\) for a \(\lgm\in\reals{\geq0}\).
But \(\RCL{\rno}{\Wm}{\lgm}=\infty\) might hold for \(\lgm\) small 
enough even when \(\CRC{\rno}{\!\Wm\!}{\costc}<\infty\) 
for all \(\costc\in\reals{\geq0}^{\ell}\),
see
Example \ref{eg:affine}.
\begin{remark}
In \cite[\S33-\S35]{augustin78}, Augustin  considered the case when the cost function
\(\costf\) is a bounded function of the form \(\costf:\inpS\to[0,1]^{\ell}\). 
In that case 	\(\RCL{\rno}{\!\Wm\!}{\lgm}<\infty\) for all \(\lgm\in\reals{\geq0}^{\ell}\) provided that  
\(\CRC{\rno}{\!\Wm\!}{\costc}<\infty\) for a \(\costc\in\inte{\fcc{\costf}}\)
because \(\CRC{\rno}{\Wm}{\uc}<\infty\) by Lemma \ref{lem:CCcapacity}-(\ref{CCcapacity:interior}) and \(\CRC{\rno}{\Wm}{\uc}=\RC{\rno}{\!\Wm\!}\)
and \(\RCL{\rno}{\Wm}{\lgm}\leq \RCL{\rno}{\Wm}{0}=\RC{\rno}{\!\Wm\!}\) for all \(\lgm\in\reals{\geq0}^{\ell}\) by definition.	
	\end{remark}

The inequality given in \eqref{eq:CCcapacity:dualitybound} is an equality 
for many cases of interest as demonstrated by the following lemma.
However, the inequality given in \eqref{eq:CCcapacity:dualitybound}
is not an equality in general, see Example \ref{eg:non-uppersemicontinuous}. 

\begin{lemma}\label{lem:Lcapacity}
Let \(\rno\in\reals{+}\) and \(\Wm\) be a channel of the form \(\Wm:\inpS\to \pmea{\outA}\) with 
a cost function \(\costf:\inpS\to \reals{\geq0}^{\ell}\). Then
\begin{enumerate}[(a)]
\item\label{Lcapacity:function}
\(\RCL{\rno}{\Wm}{\lgm}\) is convex, nonincreasing, and lower semicontinuous in \(\lgm\) 
on \(\reals{\geq0}^{\ell}\) and continuous in \(\lgm\) on 
\(\{\lgm\!:\exists\epsilon>0~s.t.~ \RCL{\rno}{\Wm}{\lgm-\epsilon\uc}\!<\!\infty\!\}\).

\item\label{Lcapacity:minimax} 
If \(\inpS\) is a finite set, then
\(\CRC{\rno}{\!\Wm\!}{\costc}=\inf\nolimits_{\lgm\geq0} \RCL{\rno}{\Wm}{\lgm}+\lgm\cdot\costc\).
\item\label{Lcapacity:interior} 
If \(\costc\in\inte{\fcc{\costf}}\), then
\(\CRC{\rno}{\!\Wm\!}{\costc}=\inf\nolimits_{\lgm\geq0} \RCL{\rno}{\Wm}{\lgm}+\lgm\cdot\costc\).
If in addition \(\CRC{\rno}{\!\Wm\!}{\costc}<\infty\), 
then there exists a non-empty convex, compact  set of 
\(\lgm_{\rno,\!\Wm\!,\costc}\)'s satisfying
both  \eqref{eq:CCcapacity:interior}
and
\(\CRC{\rno}{\!\Wm\!}{\costc}=\RCL{\rno}{\Wm}{\lgm_{\rno,\!\Wm\!,\costc}}+\lgm_{\rno,\!\Wm\!,\costc}\cdot \costc\).
\item\label{Lcapacity:optimal}
If \(\CRC{\rno}{\!\Wm\!}{\costc}\) is finite and \(\CRC{\rno}{\!\Wm\!}{\costc}=\RCL{\rno}{\Wm}{\lgm}+\lgm\cdot\costc\) 
for some \(\costc\in \fcc{\costf}\) and \(\lgm\in\reals{\geq0}^{\ell}\), then
\(\lim\nolimits_{\ind \to \infty} \RMIL{\rno}{\pma{}{(\ind)}}{\Wm}{\lgm}=\RCL{\rno}{\Wm}{\lgm}\)
 for all \(\{\pma{}{(\ind)}\}_{\ind\in\integers{+}}\in \cset(\costc)\) s.t.
\(\lim\nolimits_{\ind \to \infty} \RMI{\rno}{\pma{}{(\ind)}}{\Wm}=\CRC{\rno}{\!\Wm\!}{\costc}\).
\end{enumerate}
\end{lemma}

Using the definitions of \(\RMI{\rno}{\mP}{\Wm}\), \(\RMIL{\rno}{\mP}{\Wm}{\lgm}\), 
and \(\RCL{\rno}{\Wm}{\lgm}\) given in  
\eqref{eq:def:information}, \eqref{eq:def:Linformation}, \eqref{eq:def:Lcapacity}
we get the following expression for \(\RCL{\rno}{\Wm}{\lgm}\).
\begin{align}
\label{eq:Lcapacity}
\RCL{\rno}{\Wm}{\lgm}
&=\sup\nolimits_{\mP \in \pdis{\inpS}}\inf\nolimits_{\mQ\in\pmea{\outA}} \CRD{\rno}{\Wm}{\mQ}{\mP}-\lgm\cdot \EXS{\mP}{\costf}.
\end{align}

The A-L capacity satisfies a minimax theorem similar to the one satisfied by the Augustin capacity,
which allows us to assert the existence of a unique A-L center whenever the A-L capacity 
is finite. 
\begin{theorem}\label{thm:Lminimax}
For any \(\rno\in \reals{+}\), channel \(\Wm:\inpS\to \pmea{\outA}\) 
with a cost function \(\costf:\inpS\to \reals{\geq0}^{\ell}\), and Lagrange multiplier 
\(\lgm \in \reals{\geq0}^{\ell}\)
\begin{align}
\label{eq:thm:Lminimax}
\sup\nolimits_{\mP\in \pdis{\inpS}}\inf\nolimits_{\mQ \in \pmea{\outA}} 
\CRD{\rno}{\Wm}{\mQ}{\mP }-\lgm\cdot\EXS{\mP}{\costf}
&=
\inf\nolimits_{\mQ \in \pmea{\outA}} \sup\nolimits_{\mP\in \pdis{\inpS}}
\CRD{\rno}{\Wm}{\mQ}{\mP }-\lgm\cdot\EXS{\mP}{\costf}
\\
\label{eq:thm:Lminimaxradius}
&=\inf\nolimits_{\mQ\in\pmea{\outA}}\sup\nolimits_{\dinp \in \inpS} 
\RD{\rno}{\Wm(\dinp)}{\mQ}-\lgm\cdot \costf(\dinp).
\end{align}
If the expression on the left hand side of \eqref{eq:thm:Lminimax} is finite,
i.e. if \(\RCL{\rno}{\Wm}{\lgm}<\infty\), then
\(\exists!\qma{\rno,\Wm}{\lgm}\!\in\!\pmea{\outA}\),
called the order \(\rno\) Augustin-Legendre center of \(\Wm\) for the Lagrange multiplier \(\lgm\),
satisfying
\begin{align}
\label{eq:thm:Lminimaxcenter}
\RCL{\rno}{\Wm}{\lgm}
&=\sup\nolimits_{\mP \in \pdis{\inpS}} 
\CRD{\rno}{\Wm}{\qma{\rno,\Wm}{\lgm}}{\mP}-\lgm\cdot\EXS{\mP}{\costf}
\\
\label{eq:thm:Lminimaxradiuscenter}
&=\sup\nolimits_{\dinp \in \inpS} \RD{\rno}{\Wm(\dinp)}{\qma{\rno,\Wm}{\lgm}}-\lgm\cdot\costf(\dinp).
\end{align}
Furthermore, for every sequence of input distributions \(\{\pma{}{(\ind)}\}_{\ind\in\integers{+}} \subset \pdis{\inpS}\) such 
that \(\lim_{\ind \to \infty} \RMIL{\rno}{\pma{}{(\ind)}}{\Wm}{\lgm}=\RCL{\rno}{\Wm}{\lgm}\), 
the corresponding sequence of order \(\rno\) Augustin means \(\{\qmn{\rno,\pma{}{(\ind)}}\}_{\ind\in\integers{+}}\)is a Cauchy 
sequence  for the total variation metric on \(\pmea{\outA}\) and \(\qma{\rno,\Wm}{\lgm}\) is the unique 
limit point of that Cauchy sequence.
\end{theorem}

Note that Theorem \ref{thm:Lminimax} for \(\lgm=0\) is nothing but 
Theorem \ref{thm:minimax} for \(\cset=\pdis{\inpS}\).
The proof of Theorem \ref{thm:Lminimax} is very similar to that of Theorem \ref{thm:minimax},
as well;
it employs Lemma \ref{lem:LcapacityFLB}, presented in the following, 
instead of Lemma \ref{lem:capacityFLB}. 
Note that, Lemma \ref{lem:LcapacityFLB} for \(\lgm=0\) is nothing but Lemma \ref{lem:capacityFLB}
for \(\cset=\pdis{\inpS}\), as well.
\begin{lemma}\label{lem:LcapacityFLB}
For any \(\rno\in\reals{+}\), channel \(\Wm:\inpS\to\pmea{\outA}\) 
with a cost function \(\costf:\inpS\to \reals{\geq0}^{\ell}\) for a finite input set \(\inpS\), 
and Lagrange multiplier \(\lgm \in \reals{\geq0}^{\ell}\), 
there exists a \(\widetilde{\mP}\in\pdis{\inpS}\) such that 
\(\RMIL{\rno}{\widetilde{\mP}}{\Wm}{\lgm}=\RCL{\rno}{\Wm}{\lgm}\)
and \(\exists!\qma{\rno,\Wm}{\lgm}\in\pmea{\outA}\) satisfying 
	\begin{align}
	\label{eq:lem:LcapacityFLB}
	\CRD{\rno}{\Wm}{\qma{\rno,\Wm}{\lgm}}{\mP}-\lgm\cdot\EXS{\mP}{\costf}
	&\leq  \RCL{\rno}{\Wm}{\lgm}
	&
	&\forall \mP \in \pdis{\inpS}.   
	\end{align}
Furthermore, \(\qmn{\rno,\widetilde{\mP}}=\qma{\rno,\Wm}{\lgm}\) for all 
\(\widetilde{\mP}\in\pdis{\inpS}\) such that \(\RMIL{\rno}{\widetilde{\mP}}{\Wm}{\lgm}=\RCL{\rno}{\Wm}{\lgm}\).
\end{lemma}

Note that the expression on the left hand side of equation \eqref{eq:thm:Lminimax} is nothing but
the A-L capacity. Thus Theorem \ref{thm:Lminimax} is establishes the equality of the A-L capacity to the A-L radius
defined in the following. 
\begin{definition}\label{def:Lradius}
For any \(\rno\in\reals{+}\), channel \(\Wm:\inpS\to \pmea{\outA}\) with 
a cost function \(\costf:\inpS\to \reals{\geq0}^{\ell}\), and \(\lgm \in \reals{\geq0}^{\ell}\),
\emph{the order \(\rno\) Augustin-Legendre radius of \(\Wm\) for the Lagrange multiplier \(\lgm\)} is
\begin{align}
\label{eq:def:Lradius}
\RRL{\rno}{\Wm}{\lgm}
&\DEF \inf\nolimits_{\mQ\in\pmea{\outA}} \sup\nolimits_{\dinp\in \inpS} \RD{\rno}{\Wm(\dinp)}{\mQ}-\lgm\cdot\costf(\dinp). 
\end{align}
\end{definition}

If \(\RCL{\rno}{\Wm}{\lgm}\) is finite, then Lemma \ref{lem:information}-(\ref{information:one},\ref{information:zto},\ref{information:oti}),
Theorem \ref{thm:Lminimax}, and the definition of \(\RMIL{\rno}{\mP}{\Wm}{\lgm}\) given in 
\eqref{eq:def:Linformation} imply that 
\begin{align}
\notag %\label{eq:LcapacityLB}
\RCL{\rno}{\Wm}{\lgm}-\RMIL{\rno}{\mP}{\Wm}{\lgm}
&\geq  \RD{\rno\wedge 1}{\qmn{\rno,\mP}}{\qma{\rno,\Wm}{\lgm}}
&
&\forall \mP \in \pdis{\inpS}.
\end{align}
Using Lemma \ref{lem:information} and Theorem \ref{thm:Lminimax} one can also 
establish a bound similar to the one given in Lemma \ref{lem:EHB}.
However, we will not do so here because one can obtain a slightly stronger results, 
using the characterization of the A-L capacity and center 
via R-G capacity and center presented in \S\ref{sec:cost-RG}, 
see Lemma \ref{lem:GEHB} and the ensuing discussion. 

As a result of Lemma \ref{lem:Lcapacity}-(\ref{Lcapacity:interior}), we know that 
if \(\CRC{\rno}{\Wm}{\costc}\) is finite for a \(\costc\in\inte{\fcc{\costf}}\),
then there exists at least one \(\lgm_{\rno,\Wm,\costc}\) for which 
\(\CRC{\rno}{\Wm}{\costc}=\RCL{\rno}{\Wm}{\lgm_{\rno,\Wm,\costc}}+\lgm_{\rno,\Wm,\costc}\cdot\costc\)
holds.
Lemma \ref{lem:Lcenter}, given in the following, asserts that for any such 
Lagrange multiplier the corresponding order \(\rno\) A-L center should 
be equal to the order \(\rno\) Augustin center for the cost constraint \(\costc\).
Thus if there are multiple \(\lgm_{\rno,\Wm,\costc}\)'s 
satisfying 
\(\CRC{\rno}{\Wm}{\costc}=\RCL{\rno}{\Wm}{\lgm_{\rno,\Wm,\costc}}+\lgm_{\rno,\Wm,\costc}\cdot\costc\),
then they all have the same order \(\rno\) A-L center.

\begin{lemma}\label{lem:Lcenter} 
For any \(\rno\in\reals{+}\), channel \(\Wm:\inpS\to \pmea{\outA}\) 
with a cost function \(\costf:\inpS\to \reals{\geq0}^{\ell}\), and 
a cost constraint \(\costc\in\fcc{\costf}\) such that \(\CRC{\rno}{\!\Wm\!}{\costc}<\infty\),
if \(\CRC{\rno}{\!\Wm\!}{\costc}=\RCL{\rno}{\Wm}{\lgm}+\lgm\cdot\costc\) for a \(\lgm\in \reals{\geq0}^{\ell}\),
then \(\qmn{\rno,\!\Wm\!,\costc}=\qma{\rno,\Wm}{\lgm}\).
\end{lemma}

For product constraints on product channels, 
the Augustin capacity has an additive form and the Augustin center has a multiplicative form
---whenever it exists--- by Lemma \ref{lem:capacityproduct}.
The cost constraints for additive cost functions, however, are not product constraints. 
In order to calculate the cost constrained Augustin capacity for product channels with 
additive cost functions, we need to optimize over the feasible allocations of the cost over 
the component channels by Lemma \ref{lem:CCcapacityproduct}.
In addition, we can express the cost constrained Augustin center of the product channel 
as the product of the cost constrained Augustin centers of the components channels
---using Lemma \ref{lem:CCcapacityproduct}--- 
only when there exists a feasible allocation of the cost that achieves the optimum value. 
For the A-L capacity and center, on the other hand, we have a considerably neater picture:
For product channels with additive cost functions the A-L capacity is additive and 
the A-L center is multiplicative, whenever it exists.
\begin{lemma}\label{lem:Lcapacityproduct}
For any length \(\blx\) product channel
\(\Wmn{[1,\blx]}:\inpS_{1}^{\blx}\to\pmea{\outA_{1}^{\blx}}\)
and additive cost function 
\(\costf_{[1,\blx]}:\inpS_{1}^{\blx}\to\reals{\geq0}^{\ell}\)
%satisfying 
%\(\costf_{[1,\blx]}(\dinp_{1}^{\blx})\!=\!\sum_{\tin=1}^{\blx}\!\costf_{\tin}(\dinp_{\tin})\) 
%for some \(\costf_{\tin}:\inpS_{\tin}\to\reals{\geq0}^{\ell}\)
we have	
\begin{align}
\label{eq:lem:Lcapacityproduct}
\RCL{\rno}{\Wmn{[1,\blx]}}{\lgm}
&=\sum\nolimits_{\tin=1}^{\blx} \RCL{\rno}{\Wmn{\tin}}{\lgm}
&
&\forall \lgm \in \reals{\geq0}^{\ell},~\rno\in \reals{+}. 
\end{align}
Furthermore, if \(\RCL{\rno}{\Wmn{[1,\blx]}}{\lgm}<\infty\), 
then \(\qma{\rno,\Wmn{[1,\blx]}}{\lgm}=\bigotimes\nolimits_{\tin=1}^{\blx} \qma{\rno,\Wmn{\tin}}{\lgm}\).
\end{lemma}

The additivity of the cost function \(\costf_{[1,\blx]}\) 
implies for any \(\mP\) in \(\pdis{\inpS_{1}^{\blx}}\)
\begin{align}
\notag
\EXS{\mP}{\costf_{[1,\blx]}}
=\sum\nolimits_{\tin=1}^{\blx}\EXS{\pmn{\tin}}{\costf_{\tin}}
\end{align}
where \(\pmn{\tin}\in\pdis{\inpS_{\tin}}\) is  the \(\inpS_{\tin}\) marginal of \(\mP\).
Thus  Lemma \ref{lem:information:product}
and the definition of the A-L information
imply
	\begin{align}
\notag
\RMIL{\rno}{\mP}{\Wmn{[1,\blx]}}{\lgm}
&\leq \RMIL{\rno}{\pmn{1}\otimes\cdots\otimes\pmn{\blx}}{\Wmn{[1,\blx]}}{\lgm}
\\
\label{eq:lem:A-L-information:product}
&=\sum\nolimits_{\tin=1}^{\blx}\RMIL{\rno}{\pmn{\tin}}{\Wmn{\tin}}{\lgm}.
\end{align}
Lemma \ref{lem:Lcapacityproduct} is proved using 
\eqref{eq:lem:A-L-information:product} together with Theorem \ref{thm:Lminimax}.

\subsection{The \renyi\!\!-Gallager Information Measures}\label{sec:cost-RG}
In \S\ref{sec:cost-AL}, we have characterized the cost constrained Augustin capacity and center 
in terms of the A-L capacity and center.
The A-L capacity is defined as the supremum of the A-L information. 
Gallager ---implicitly--- proposed another information
with a Lagrange multiplier in \cite[(103) and (116)]{gallager65}. 
Augustin characterized the cost constrained Augustin capacity 
in terms of the supremum of this information, 
assuming that the cost function is bounded, in
 \cite[Lemmas 35.4-(b) and 35.8-(b)]{augustin78}.
We call this supremum the R-G capacity. 
The main aim of this subsection is establishing the equality of the A-L capacity and center 
to the R-G capacity and center.
We will also derive a van Erven-\harremoes bound for the A-L capacity and center and use it 
to derive the continuity of the A-L center as a function of the Lagrange multiplier \(\lgm\).

\begin{definition}\label{def:Ginformation}
For any \(\rno\in\reals{+}\), channel \(\Wm:\inpS\to \pmea{\outA}\) with 
a cost function \(\costf:\inpS\to \reals{\geq0}^{\ell}\), \(\mP\in \pdis{\inpS}\),  
and \(\lgm \in \reals{\geq0}^{\ell}\) 
\emph{the order \(\rno\) \renyi\!\!-Gallager (R-G) information for
	the input distribution \(\mP\) 
	and the Lagrange multiplier \(\lgm\)} is
\begin{align}
\label{eq:def:Ginformation}
\GMIL{\rno}{\mP}{\Wm}{\lgm}
&\DEF 
\begin{cases}
\inf\nolimits_{\mQ\in \pmea{\outA}} \RD{\rno}{\mP \mtimes  \Wm e^{\frac{1-\rno}{\rno}\lgm\cdot\costf}}{\mP\otimes \mQ}
&\rno\in\reals{+}\setminus\{1\}
\\
\inf\nolimits_{\mQ\in \pmea{\outA}} \RD{1}{\mP \mtimes  \Wm}{\mP\otimes \mQ}
-\lgm\cdot\EXS{\mP}{\costf}
&\rno=1
\end{cases}.
\end{align}
\end{definition}
If \(\lgm\) is a vector of zeros, then the R-G information is the \renyi information.
Similar to the \renyi information, the R-G information has a closed form expression, described
in terms of the probability measure achieving the infimum in its definition.
\begin{definition}\label{def:Gmean}
For any \(\rno\in\reals{+}\), channel \(\Wm:\inpS\to \pmea{\outA}\) with 
a cost function \(\costf:\inpS\to \reals{\geq0}^{\ell}\), \(\mP\in \pdis{\inpS}\),  
and \(\lgm \in \reals{\geq0}^{\ell}\),
\emph{the order \(\rno\) mean measure for the input distribution \(\mP\) and  
the Lagrange multiplier \(\lgm\)} is
\begin{align}
\label{eq:def:Gmeanmeasure}
\der{\mma{\rno,\mP}{\lgm}}{\rfm}
&\DEF\left[\sum\nolimits_{\dinp} \mP(\dinp) e^{(1-\rno)\lgm \cdot \costf(\dinp)}\left(\der{\Wm(\dinp)}{\rfm}\right)^{\rno}\right]^{\frac{1}{\rno}}.
\end{align}
\emph{The order \(\rno\) \renyi\!\!-Gallager (R-G) mean for the input distribution \(\mP\) 
and the  Lagrange multiplier \(\lgm\)} is
\begin{align}
\label{eq:def:Gmean}
\qga{\rno,\mP}{\lgm}
&\DEF \tfrac{\mma{\rno,\mP}{\lgm}}{\lon{\mma{\rno,\mP}{\lgm}}}.
\end{align}
\end{definition}
Both \(\mma{\rno,\mP}{\lgm}\) and \(\qga{\rno,\mP}{\lgm}\) depend on the Lagrange multiplier 
\(\lgm\) for \(\rno\!\in\!\reals{+}\!\setminus\!\{1\}\).
Furthermore, one can confirm by substitution that 
\begin{align}
\label{eq:Ginformation-neq-A}
\RD{\rno}{\mP \mtimes  \Wm e^{\frac{1-\rno}{\rno}\lgm\cdot\costf}}{\mP\otimes \mQ}
&=\RD{\rno}{\mP \mtimes  \Wm e^{\frac{1-\rno}{\rno}\lgm\cdot\costf}}{\mP\otimes \qga{\rno,\mP}{\lgm}}
+\RD{\rno}{\qga{\rno,\mP}{\lgm}}{\mQ}
&
&\rno\in\reals{+}\setminus\{1\}.
\end{align}
Then as a result of  Lemma \ref{lem:divergence-pinsker}, we have
\begin{align}
\label{eq:Ginformation-neq}
\GMIL{\rno}{\mP}{\Wm}{\lgm}
&=\RD{\rno}{\mP \mtimes  \Wm e^{\frac{1-\rno}{\rno}\lgm\cdot\costf}}{\mP\otimes \qga{\rno,\mP}{\lgm}}
&
&
\\
\label{eq:Ginformation-neq-alternative}
&=\tfrac{\rno}{\rno-1}\ln \lon{\mma{\rno,\mP}{\lgm}}
&
&\rno\in\reals{+}\!\setminus\!\{1\}.
\end{align}
Neither \(\mma{1,\mP}{\lgm}\), nor \(\qga{1,\mP}{\lgm}\) depends on the Lagrange multiplier \(\lgm\).
In addition, one can confirm by substitution that 
\begin{align}
\label{eq:Ginformation-one-A}
\RD{1}{\mP \mtimes  \Wm}{\mP\otimes \mQ}-\lgm\cdot\EXS{\mP}{\costf}
&=\RD{1}{\mP \mtimes  \Wm}{\mP\otimes \qga{1,\mP}{\lgm}}-\lgm\cdot\EXS{\mP}{\costf}
+\RD{1}{\qga{1,\mP}{\lgm}}{\mQ}.
&
&
\end{align}
Then as a result of  Lemma \ref{lem:divergence-pinsker}, we have
\begin{align}
\label{eq:Ginformation-one}
\GMIL{1}{\mP}{\Wm}{\lgm}
&=\RD{1}{\mP \mtimes  \Wm}{\mP\otimes \qga{1,\mP}{\lgm}}-\lgm \cdot \EXS{\mP}{\costf}.
&
&
\end{align}
Using the definitions of the A-L information and the R-G information given 
in \eqref{eq:def:Linformation} and \eqref{eq:def:Ginformation} together
with the Jensen's inequality and the concavity of the natural logarithm 
function we get
\begin{align}
\notag %\label{eq:RGAL-information-zto}
\RMIL{\rno}{\mP}{\Wm}{\lgm}
&\geq \GMIL{\rno}{\mP}{\Wm}{\lgm}
&
&\rno \in (0,1]
\\
\notag %\label{eq:RGAL-information-oti}
\RMIL{\rno}{\mP}{\Wm}{\lgm}
&\leq \GMIL{\rno}{\mP}{\Wm}{\lgm}
&
&\rno \in [1,\infty).
\end{align}
It is possible to strengthen these relations by expressing the
A-L information and the R-G information in terms of one another as follows.

\begin{lemma}\label{lem:Lpolsha}
	Let \(\Wm\) be a channel of the form \(\Wm:\inpS\to \pmea{\outA}\) with a cost function
	\(\costf:\inpS\to \reals{\geq0}^{\ell}\), \(\mP\) be an input distribution in \(\pdis{\inpS}\) 
	and \(\lgm\)  be a Lagrange multiplier in \(\reals{\geq0}^{\ell}\).
	\begin{enumerate}[(a)]
		\item\label{Lpolsha:poltyrev}
		Let \(\uma{\rno,\mP}{\lgm}\in\pdis{\inpS}\) be
		\(\uma{\rno,\mP}{\lgm}(\dinp)=\tfrac{\mP(\dinp)e^{(1-\rno)\RD{\rno}{\Wm(\dinp)}{\qmn{\rno,\mP}}+(\rno-1)\lgm \cdot \costf(\dinp)} }{\sum_{\tilde{\dinp}}\mP(\tilde{\dinp})e^{(1-\rno)\RD{\rno}{\Wm(\tilde{\dinp})}{\qmn{\rno,\mP}}+(\rno-1)\lgm \cdot \costf(\dinp)}}\) 
		for all \(\dinp\); then
		\begin{align}
		\label{eq:lem:Lpolsha:poltyrev}
		\RMIL{\rno}{\mP}{\Wm}{\lgm}
		&=\GMIL{\rno}{\umn{\rno,\mP}}{\Wm}{\lgm}+\tfrac{1}{\rno-1}\RD{1}{\mP}{\umn{\rno,\mP}}
		&
		&
		\\
		\label{eq:lem:Lpolsha:poltyrev-variational}
		&=\begin{cases}
		\sup_{\mU\in\pdis{\inpS}} \GMIL{\rno}{\mU}{\Wm}{\lgm}+\tfrac{1}{\rno-1}\RD{1}{\mP}{\mU}
		&\rno\in(0,1)
		\\
		\inf_{\mU\in\pdis{\inpS}} \GMIL{\rno}{\mU}{\Wm}{\lgm}+\tfrac{1}{\rno-1}\RD{1}{\mP}{\mU}
		&\rno\in(1,\infty)
		\end{cases}.
		&
		&
		\end{align}
		\item\label{Lpolsha:shayevitz}
		Let \(\ama{\rno,\mP}{\lgm}\in\pdis{\inpS}\) be
		\(\ama{\rno,\mP}{\lgm}(\dinp)=\tfrac{\mP(\dinp)e^{(\rno-1)\RD{\rno}{\Wm(\dinp)}{\qga{\rno,\mP}{\lgm}}+(1-\rno)\lgm \cdot \costf(\dinp) }}{\sum_{\tilde{\dinp}}\mP(\tilde{\dinp})e^{(\rno-1)\RD{\rno}{\Wm(\tilde{\dinp})}{\qga{\rno,\mP}{\lgm}}+(1-\rno)\lgm \cdot \costf(\dinp)}}\) 
		for all \(\dinp\); then
		\begin{align}
		\label{eq:lem:Lpolsha:shayevitz}
		\GMIL{\rno}{\mP}{\Wm}{\lgm}
		&=\RMIL{\rno}{\ama{\rno,\mP}{\lgm}}{\Wm}{\lgm}-\tfrac{1}{\rno-1}\RD{1}{\ama{\rno,\mP}{\lgm}}{\mP}
		&
		&
		\\
		\label{eq:lem:Lpolsha:shayevitz-variational}
		&=\begin{cases}
		\inf_{\mA\in\pdis{\inpS}} \RMIL{\rno}{\mA}{\Wm}{\lgm}-\tfrac{1}{\rno-1}\RD{1}{\mA}{\mP}
		&\rno\in(0,1)
		\\
		\sup_{\mA\in\pdis{\inpS}} \RMIL{\rno}{\mA}{\Wm}{\lgm}-\tfrac{1}{\rno-1}\RD{1}{\mA}{\mP}
		&\rno\in(1,\infty)
		\end{cases}.
		&
		&
		\end{align}
		\item\label{Lpolsha:augustin}
		Let \(\fX_{\rno,\mP}^{\lgm}:\inpS\to\reals{}\) be 
		\(\fX_{\rno,\mP}^{\lgm}(\dinp)=
		[\RD{\rno}{\Wm(\dinp)}{\qmn{\rno,\mP}}-\lgm\cdot\costf(\dinp)-\RMIL{\rno}{\mP}{\Wm}{\lgm}]\IND{\mP(\dinp)>0}\) 
		for all \(\dinp\); then
		\begin{align}
\label{eq:lem:Lpolsha:augustin}
\RMIL{\rno}{\mP}{\Wm}{\lgm}
&=\tfrac{\rno}{\rno-1}\ln\EXS{\rfm}{\left(\sum\nolimits_{\dinp}\mP(\dinp)
	e^{(1-\rno)(\fX_{\rno,\mP}^{\lgm}(\dinp)+\lgm\cdot\costf(\dinp))}
	\left[\der{\Wm(\dinp)}{\rfm}\right]^{\rno}
	\right)^{\sfrac{1}{\rno}}}
&
&
\\
\label{eq:lem:Lpolsha:augustin-variational}
&=\tfrac{\rno}{\rno-1}\ln\inf\nolimits_{\fX:\EXS{\mP}{\fX}=0}\EXS{\rfm}{\left(\sum\nolimits_{\dinp}\mP(\dinp)
	e^{(1-\rno)(\fX(\dinp)+\lgm\cdot\costf(\dinp))}
	\left[\der{\Wm(\dinp)}{\rfm}\right]^{\rno}
	\right)^{\sfrac{1}{\rno}}}.		
&
&
\end{align}		
	\end{enumerate}
\end{lemma} 
Lemma \ref{lem:Lpolsha} for \(\lgm=0\) is Lemma \ref{lem:polsha}, which was previously discussed 
by Poltyrev \cite{poltyrev82}, Shayevitz \cite{shayevitz11}, and  Augustin \cite{augustin78}. 

\begin{definition}
For any \(\rno\in\reals{+}\), channel  \(\Wm:\inpS\to \pmea{\outA}\) with 
a cost function \(\costf:\inpS\to \reals{\geq0}^{\ell}\), and \(\lgm \in \reals{\geq0}^{\ell}\), 
\emph{the order \(\rno\) \renyi\!\!-Gallager (R-G) capacity for 
	the Lagrange multiplier \(\lgm\)} is
\begin{align}
\notag %\label{eq:def:Gcapacity}
\GCL{\rno}{\Wm}{\lgm}
&\DEF \sup\nolimits_{\mP\in \pdis{\inpS}} \GMIL{\rno}{\mP}{\Wm}{\lgm}.
\end{align}
\end{definition} 
Using the definition of \(\GMIL{\rno}{\mP}{\Wm}{\lgm}\), given in  \eqref{eq:def:Ginformation}, 
we get  the following expression for \(\GCL{\rno}{\Wm}{\lgm}\).
\begin{align}
\label{eq:Gcapacity}
\GCL{\rno}{\Wm}{\lgm}
&=
\begin{cases}
\sup\nolimits_{\mP \in \pdis{\inpS}}\inf\nolimits_{\mQ\in \pmea{\outA}} 
\RD{\rno}{\mP \mtimes  \Wm e^{\frac{1-\rno}{\rno}\lgm\cdot\costf}}{\mP\otimes \mQ}
&\rno\in\reals{+}\setminus\{1\}
\\
\sup\nolimits_{\mP \in \pdis{\inpS}}\inf\nolimits_{\mQ\in \pmea{\outA}} 
\RD{\rno}{\mP \mtimes  \Wm}{\mP\otimes \mQ}-\lgm\cdot\EXS{\mP}{\costf}
&\rno=1
\end{cases}
\end{align}
The R-G capacity satisfies a minimax theorem similar to the one satisfied by 
the A-L capacity, i.e. Theorem \ref{thm:Lminimax}.
Since both the statement and the proof of the minimax theorems are identical 
for the order one A-L capacity and the order one R-G capacity,
we state the minimax theorem for the R-G capacity 
only for finite positive orders other than one.

\begin{theorem}\label{thm:Gminimax}
For any \(\rno\!\in\!\reals{+}\!\setminus\!\{1\}\), channel \(\Wm\!:\!\inpS\to \pmea{\outA}\) with 
a cost function \(\costf\!:\!\inpS\to \reals{\geq0}^{\ell}\), and Lagrange multiplier 
\(\lgm\!\in\!\reals{\geq0}^{\ell}\)
\begin{align}
\label{eq:thm:Gminimax}
\sup\nolimits_{\mP\in \pdis{\inpS}}\inf\nolimits_{\mQ \in \pmea{\outA}} 
\RD{\rno}{\mP \mtimes  \Wm e^{\frac{1-\rno}{\rno}\lgm\cdot\costf}}{\mP\otimes \mQ}
&=
\inf\nolimits_{\mQ \in \pmea{\outA}} \sup\nolimits_{\mP\in \pdis{\inpS}}
\RD{\rno}{\mP \mtimes  \Wm e^{\frac{1-\rno}{\rno}\lgm\cdot\costf}}{\mP\otimes \mQ}
\\
\label{eq:thm:Gminimaxradius}
&=\inf\nolimits_{\mQ\in\pmea{\outA}}\sup\nolimits_{\dinp \in \inpS} 
\RD{\rno}{\Wm(\dinp)}{\mQ}-\lgm\cdot \costf(\dinp).
\end{align}
If the expression on the left hand side of \eqref{eq:thm:Gminimax} is finite,
i.e. if \(\GCL{\rno}{\Wm}{\lgm}<\infty\), then 
\(\exists!\qga{\rno,\Wm}{\lgm}\in\pmea{\outA}\),
called the order \(\rno\) \renyi\!\!-Gallager center of \(\Wm\) for the Lagrange multiplier \(\lgm\), 
satisfying
\begin{align}
\label{eq:thm:Gminimaxcenter}
\GCL{\rno}{\Wm}{\lgm}
&=\sup\nolimits_{\mP \in \pdis{\inpS}} 
\RD{\rno}{\mP \mtimes  \Wm e^{\frac{1-\rno}{\rno}\lgm\cdot\costf}}{\mP\otimes \qga{\rno,\Wm}{\lgm}}
\\
\label{eq:thm:Gminimaxradiuscenter}
&=\sup\nolimits_{\dinp \in \inpS} 
\RD{\rno}{\Wm(\dinp)}{\qga{\rno,\Wm}{\lgm}}-\lgm\cdot\costf(\dinp).
\end{align}
Furthermore, for every sequence of input distributions 
\(\{\pma{}{(\ind)}\}_{\ind\in\integers{+}}\subset\pdis{\inpS}\) such that  
\(\lim_{\ind \to \infty} \GMIL{\rno}{\pma{}{(\ind)}}{\Wm}{\lgm}=\GCL{\rno}{\Wm}{\lgm}\),
corresponding sequence of the order \(\rno\) \renyi\!\!-Gallager means \(\{\qga{\rno,\pma{}{(\ind)}}{\lgm}\}_{\ind\in\integers{+}}\)  
is a Cauchy sequence  for the total variation metric on \(\pmea{\outA}\) 
and \(\qga{\rno,\Wm}{\lgm}\) is the unique limit point of that Cauchy sequence.
\end{theorem}

Proof of Theorem \ref{thm:Gminimax} is very similar to the proofs of Theorem \ref{thm:minimax} 
and Theorem \ref{thm:Lminimax}. It relies on Lemma \ref{lem:GcapacityFLB}, given in the following, 
instead of Lemma \ref{lem:capacityFLB}  or Lemma \ref{lem:LcapacityFLB}. 
\begin{lemma}\label{lem:GcapacityFLB}
For any \(\rno\in \reals{+}\setminus\{1\}\), 
channel \(\Wm:\inpS\to \pmea{\outA}\) 
with cost function \(\costf:\inpS\to \reals{\geq0}^{\ell}\) for a finite input set \(\inpS\), 
and Lagrange multiplier \(\lgm \in \reals{\geq0}^{\ell}\),
there exists a \(\widetilde{\mP} \in \pdis{\inpS}\) 
such that \(\RMIL{\rno}{\widetilde{\mP}}{\Wm}{\lgm}=\RCL{\rno}{\Wm}{\lgm}\)
and \(\exists!\qma{\rno,\Wm}{\lgm}\in\pmea{\outA}\) satisfying 
\begin{align}
\label{eq:lem:GcapacityFLB}
\RD{\rno}{\mP \mtimes  \Wm e^{\frac{1-\rno}{\rno}\lgm \cdot \costf}}{\mP \otimes\qga{\rno,\Wm}{\lgm}}
&\leq  \GCL{\rno}{\Wm}{\lgm}
&
&\forall \mP \in \pdis{\inpS}.   
\end{align}
Furthermore,  \(\qga{\rno,\widetilde{\mP}}{\lgm}=\qga{\rno,\Wm}{\lgm}\) for all 
\(\widetilde{\mP}\in\pdis{\inpS}\)
such that \(\GMIL{\rno}{\widetilde{\mP}}{\Wm}{\lgm}=\GCL{\rno}{\Wm}{\lgm}\).
\end{lemma}

The expression on the left hand side of \eqref{eq:thm:Gminimax} 
is the R-G capacity, whereas the expression 
in \eqref{eq:thm:Gminimaxradius}  
is the A-L radius defined in \eqref{eq:def:Lradius}.
Thus Theorems \ref{thm:Lminimax} and \ref{thm:Gminimax} imply that
\begin{align}
\label{eq:equalityofcapacities}
\RCL{\rno}{\Wm}{\lgm}
&=\RRL{\rno}{\Wm}{\lgm}=\GCL{\rno}{\Wm}{\lgm}
&
&\forall \rno\in\reals{+},\lgm\in\reals{\geq 0}^{\ell}.
\end{align}
Furthermore, whenever \(\RCL{\rno}{\Wm}{\lgm}\) is finite 
the unique A-L center described in \eqref{eq:thm:Lminimaxradiuscenter}
is equal to 
the unique R-G center described in  \eqref{eq:thm:Gminimaxradiuscenter}
by  Theorems \ref{thm:Lminimax} and \ref{thm:Gminimax}, as well.
\begin{align}
\label{eq:equalityofcenters}
&
&\qma{\rno,\Wm}{\lgm}
&=\qga{\rno,\Wm}{\lgm}
&
&&
&\forall \rno\in\reals{+},\lgm\in\reals{\geq 0}^{\ell} \mbox{~s.t.~}\RCL{\rno}{\Wm}{\lgm}<\infty.
\end{align}
In order to avoid using multiple names for the same quantity, we will state our propositions
in terms of the A-L capacity and center in the rest of the paper. 

If \(\RCL{\rno}{\Wm}{\lgm}\) is finite, then 
\eqref{eq:Ginformation-neq-A},
\eqref{eq:Ginformation-neq},
and
Theorem \ref{thm:Gminimax} 
for \(\rno\!\in\!\reals{+}\!\setminus\!\{1\}\)
and
\eqref{eq:Ginformation-one-A},
\eqref{eq:Ginformation-one}
and
Theorem \ref{thm:Lminimax} 
for \(\rno=1\)
imply that
\begin{align}
\notag %\label{eq:GcapacityLB}
\RCL{\rno}{\Wm}{\lgm}-\GMIL{\rno}{\mP}{\Wm}{\lgm}
&\geq  \RD{\rno}{\qga{\rno,\mP}{\lgm}}{\qma{\rno,\Wm}{\lgm}}
&
&\forall \mP \in \pdis{\inpS}.
\end{align}

Using the same observations, we can prove a van  Erven-\harremoes bound for 
the A-L capacity, as well.
\begin{lemma}\label{lem:GEHB}
For any \(\rno\in \reals{+}\), channel \(\Wm:\inpS\to \pmea{\outA}\) with 
a cost function \(\costf:\inpS\to \reals{\geq0}^{\ell}\), and Lagrange multiplier 
\(\lgm \in \reals{\geq0}^{\ell}\) satisfying \(\RCL{\rno}{\!\Wm\!}{\lgm}<\infty\)
\begin{align}
\label{eq:lem:GEHB}
\sup\nolimits_{\dinp \in \inpS} \RD{\rno}{\Wm(\dinp)}{\mQ}-\lgm \cdot\costf(\dinp)
&\geq \RCL{\rno}{\Wm}{\lgm}+\RD{\rno}{\qma{\rno,\Wm}{\lgm}}{\mQ}
&
&\forall \mQ \in \pmea{\outA}.
\end{align}
\end{lemma}
One can prove a similar, but weaker, result using Lemma \ref{lem:information} and Theorem \ref{thm:Lminimax}.
The right most term of the resulting bound is  
\(\RD{\rno\wedge 1}{\qma{\rno,\Wm}{\lgm}}{\mQ}\) rather than
\(\RD{\rno}{\qma{\rno,\Wm}{\lgm}}{\mQ}\). 

Lemma \ref{lem:GEHB} and the continuity of the A-L capacity \(\RCL{\rno}{\Wm}{\lgm}\) as a 
function of \(\lgm\), established in Lemma \ref{lem:Lcapacity}-(\ref{Lcapacity:function}), 
imply the continuity of the A-L center \(\qma{\rno,\Wm}{\lgm}\) in  \(\lgm\) for the total 
variation topology  on \(\pmea{\outA}\) via Lemma \ref{lem:divergence-pinsker}. 

\begin{lemma}\label{lem:Lcentercontinuity} 
For any \(\rno\in \reals{+}\), channel \(\Wm:\inpS\to \pmea{\outA}\) with 
a cost function \(\costf:\inpS\to \reals{\geq0}^{\ell}\), 
and Lagrange multiplier \(\lgm_{0} \in \reals{\geq0}^{\ell}\) satisfying \(\RCL{\rno}{\Wm}{\lgm_{0}}<\infty\),
\begin{align}
\label{eq:lem:Lcentercontinuity}
\RD{\rno}{\qma{\rno,\Wm}{\lgm_{2}}}{\qma{\rno,\Wm}{\lgm_{1}}} 
&\leq \RCL{\rno}{\Wm}{\lgm_{1}}-\RCL{\rno}{\Wm}{\lgm_{2}}
&
&\forall \lgm_{1},\lgm_{2}\in \reals{\geq0}^{\ell} \mbox{~such that~}\lgm_{0}\leq \lgm_{1} \leq \lgm_{2}.
\end{align}
Furthermore \(\qma{\rno,\Wm}{\lgm}\) is continuous in \(\lgm\) on  
\(\{\lgm:\exists\epsilon>0~s.t.~ \RCL{\rno}{\Wm}{\lgm-\epsilon\uc}\!<\!\infty\!\}\) 
for the total variation topology on \(\pmea{\outA}\).
\end{lemma}

\subsection{Information Measures for Transition Probabilities}\label{sec:cost-TP}
We have defined the conditional \renyi divergence, the Augustin information, 
the A-L information, and the R-G information, only for input distributions 
in \(\pdis{\inpS}\), 
i.e. for probability mass functions that are zero in all but finite number 
of elements of \(\inpS\). 
In many practically relevant and analytically interesting models, however, 
the input set \(\inpS\) is an uncountably infinite set equipped with a 
\(\sigma\)-algebra \(\inpA\).
The Gaussian channels ---possibly with multiple input and output antennas 
and fading--- and the Poisson channels are among the most prominent 
examples of such models.  
For such models, it is often desirable to extend the definitions of 
the Augustin information and the A-L information from \(\pdis{\inpS}\) 
to \(\pmea{\inpA}\).
For instance, in the additive Gaussian channels described in Examples
\ref{eg:SGauss} and \ref{eg:PGauss},
the equality \(\RMI{\rno}{\mP}{\Wm}=\CRC{\rno}{\Wm}{\costc}\) is 
not satisfied by any probability mass function \(\mP\) satisfying the cost
constraint; but it is satisfied by the zero mean Gaussian distribution 
with variance \(\costc\).

In the following, we will first show that if \(\outA\) is a countably generated
\(\sigma\)-algebra, then one can generalize the definitions of the conditional 
\renyi divergence, the Augustin information, and the A-L information from 
\(\pdis{\inpS}\) to \(\pmea{\inpA}\) provided that \(\Wm\) and \(\Qm\)
are not only functions from \(\inpS\) to \(\pmea{\outA}\), but also 
transition probabilities from \((\inpS,\inpA)\) to \((\outS,\outA)\). 
After that we will show that if in addition \(\inpA\) is countably separated,
then the supremum of A-L information \(\RMIL{\rno}{\mP}{\Wm}{\lgm}\) over 
\(\pmea{\inpA}\) is equal to the A-L radius \(\RRL{\rno}{\Wm}{\lgm}\), 
see Theorem \ref{thm:general-Lminimax}.
This will imply that the cost constrained Augustin capacity \(\CRC{\rno}{\Wm}{\costc}\)
---defined in \eqref{eq:def:costcapacity}---
is equal to the supremum of the Augustin information \(\RMI{\rno}{\mP}{\Wm}\)
over members of \(\pmea{\inpA}\)  satisfying \(\EXS{\mP}{\costf}\leq\costc\),
as well, at least for the cost constraints that are in the interior of the set 
of all feasible constraints, see Theorem \ref{thm:general-minimax}.

Let us first recall the definition of transition 
probability.
We adopt the definition provided by Bogachev \cite[10.7.1]{bogachev}
with a minor modification: we use \(\Wm(\oev|\dinp)\) instead of \(\Wm(\dinp|\oev)\).
\begin{definition}\label{def:transitionprobability}  
	Let \((\inpS,\inpA)\) and \((\outS,\outA)\) be measurable  spaces. Then a function 
	\(\Wm:\outA\times\inpS\to[0,1]\) is called a transition probability 
	(a stochastic kernel / a Markov kernel) from \((\inpS,\inpA)\) to \((\outS,\outA)\)
	if it satisfies the following two 
	conditions.
	\begin{enumerate}[(i)]
		\item\label{TPdef-c} For all \(\dinp\in\inpS\), the function \(\Wm(\cdot|\dinp):\outA\to[0,1]\) 
		is a probability measure on \((\outS,\outA)\).
		\item\label{TPdef-m} For all \(\oev\in\outA\), the function \(\Wm(\oev|\cdot):\inpS\to[0,1]\) 
		is a \((\inpA,\rborel{[0,1]})\)-measurable function. 
	\end{enumerate}
\end{definition} 
We denote the set of all transition probabilities from \((\inpS,\inpA)\) to \((\outS,\outA)\)
by \(\pmea{\outA|\inpA}\) with the tacit understanding that \(\inpS\) and \(\outS\) will
be clear from the context. 
If \(\Wm\) satisfies (\ref{TPdef-c}), then \(\Wm:\inpS\to\pmea{\outA}\) 
is a channel, i.e. \(\Wm\) is a member of \(\pmea{\outA|\inpS}\), even if \(\Wm\) does not 
satisfy (\ref{TPdef-m}).
Hence \(\pmea{\outA|\inpA}\subset \pmea{\outA|\inpS}\).
Inspired by this observation, we denote the probability measure \(\Wm(\cdot|\dinp)\) 
by \(\Wm(\dinp)\) whenever it is notationally convenient and unambiguous.

In order to extend the definition of the conditional \renyi divergence 
from \(\pdis{\inpS}\) to \(\pmea{\inpA}\), we ensure the 
\(\inpA\)-measurability of \(\RD{\rno}{\Wm(\dinp)}{\Qm(\dinp)}\) on 
\(\inpS\) and replace the sum in \eqref{eq:def:conditionaldivergence}
with an integral.
If \((\inpS,\tau)\) is a topological space and \(\inpA\) is the associated Borel 
\(\sigma\)-algebra, then one can establish the measurability by first establishing 
the continuity. 
Such a continuity result holds if both \(\der{\Wm(\dinp)}{\rfm}\) and 
\(\der{\Qm(\dinp)}{\rfm}\) are continuous in \(\dinp\) for 
\(\rfm\)-almost every \(\dout\) for some probability measure \(\rfm\)
for which \((\Wm(\dinp)+\Qm(\dinp))\AC\rfm\) for all \(\dinp\in\inpS\).
At times this hypothesis on \(\Wm\) and \(\Qm\) might not be easy to confirm.
If, on the other hand, \(\Wm\) and \(\Qm\) are transition probabilities from 
\((\inpS,\inpA)\) to \((\outS,\outA)\) for a countably generated \(\outA\),
then the desired measurability follows from the elementary properties of 
the measurable functions and Lemma \ref{lem:divergence:ISAlimit}, as we demonstrate 
in the following.
\begin{lemma}\label{lem:divergence-measurability}
	For any \(\rno\in\reals{+}\), countable generated \(\sigma\)-algebra \(\outA\) 
	of subsets of \(\outS\), and \(\Wm, \Qm \in \pmea{\outA|\inpA}\) the function
	\(\RD{\rno}{\Wm(\cdot)}{\Qm(\cdot)}:\inpS\to[0,\infty]\) is \(\inpA\)-measurable.
\end{lemma}

\begin{proof}[Proof of Lemma \ref{lem:divergence-measurability}]
	There exists \(\{\oev_{\ind}\}_{\ind\in\integers{+}}\subset\outA\) such that
	\(\outA=\gsal{\{\oev_{\ind}:\ind\in\integers{+}\}}\)
	because \(\outA\) is countably generated \(\sigma\)-algebra.
	Let \(\outA_{\ind}\) be
	\begin{align}
	\notag
	\outA_{\ind}
	&\DEF\gsal{\{\oev_{1},\ldots,\oev_{\ind}\}}
	&
	&\ind\in\integers{+}.
	\end{align}
	Then \(\outA_{1}\subset\outA_{2}\subset\cdots\subset \outA\),
	\(\outA=\gsal{\cup_{\ind=1}^{\infty}\outA_{\ind}}\),
	and Lemma \ref{lem:divergence:ISAlimit} implies that
	\begin{align}
	\label{eq:divergence-measurability}	
	\RD{\rno}{\Wm(\dinp)}{\Qm(\dinp)}
	&=\lim\nolimits_{\ind \to \infty}\RDF{\rno}{\outA_{\ind}}{\Wm(\dinp)}{\Qm(\dinp)}
	&
	&\forall \dinp\in\inpS.
	\end{align}
	On the other hand \(\outA_{\ind}\) is finite set for all \(\ind\in\integers{+}\).
	Thus for all \(\ind\in\integers{+}\) there exists a \(\outA_{\ind}\)-measurable 
	finite partition
%%%%%%%%%%%%	\footnote{A collection \(\opa\) of subsets  of \(\outS\) 
%%%%%%%%%%%%		is called a \(\outA\)-measurable partition of \(\outS\) iff 
%%%%%%%%%%%%		\(\cup_{\oev\in \opa}=\outS\), \(\emptyset\notin\opa\), \(\oev\cap\tilde{\oev}=\emptyset\) 
%%%%%%%%%%%%		for all \(\oev,\tilde{\oev} \in \opa\) and \(\opa\subset\outA\).}
	\(\opa_{\ind}\) of \(\outS\).
	Thus as a result of the definition of the \renyi divergence given in \eqref{eq:def:divergence} 
	we have
	\begin{align}
	\notag
	\RDF{\rno}{\outA_{\ind}}{\Wm(\dinp)}{\Qm(\dinp)}
	&=	\begin{cases}
	\tfrac{1}{\rno-1} 
	\ln \sum_{\oev \in \opa_{\ind}} \left(\Wm(\oev|\dinp)\right)^{\rno} \left(\Qm(\oev|\dinp)\right)^{1-\rno} 
	& \rno\in \reals{+}\setminus\{1\}
	\\
	\sum_{\oev \in \opa_{\ind}} \Wm(\oev|\dinp) \ln \tfrac{\Wm(\oev|\dinp)}{\Qm(\oev|\dinp)}
	&\rno=1
	\end{cases}.
	\end{align}
	Then \(\RDF{\rno}{\outA_{\ind}}{\Wm(\dinp)}{\Qm(\dinp)}\) is a \(\inpA\)-measurable 
	function for any \(\ind\in\integers{+}\) by \cite[Thm. 2.1.5-(i-iv)]{bogachev} and 
	\cite[Remark 2.1.6]{bogachev} because \(\Wm(\oev|\dinp)\) and \(\Qm(\oev|\dinp)\) are 
	\(\inpA\)-measurable for all \(\oev\in\opa_{\ind}\) by the hypothesis of the lemma.
	Then \(\RD{\rno}{\Wm(\dinp)}{\Qm(\dinp)}\) is \(\inpA\)-measurable as a result of 
	\eqref{eq:divergence-measurability}
	by \cite[Thm. 2.1.5-(v)]{bogachev} and 
	\cite[Remark 2.1.6]{bogachev}.
\end{proof} 

\begin{definition}\label{def:general-conditionaldivergence} 
	For any \(\rno\in\reals{+}\), countable generated \(\sigma\)-algebra \(\outA\) 
	of subsets of \(\outS\), \(\Wm\in \pmea{\outA|\inpA}\), and 
	\(\mP\in\pmea{\inpA}\)
	\emph{the order \(\rno\) conditional \renyi divergence for 
		the input distribution \(\mP\)} is
	\begin{align}
	\label{eq:def:general-conditionaldivergence}
	\CRD{\rno}{\Wm}{\Qm}{\mP}
	&\DEF \int \RD{\rno}{\Wm(\dinp)}{\Qm(\dinp)}\mP(\dif{\dinp}).
	\end{align}
	If \(\exists\mQ\in\pmea{\outA}\) such that  \(\Qm(\dinp)=\mQ\) for 
	\(\mP\)-a.s.,
	then we denote \(\CRD{\rno}{\Wm}{\Qm}{\mP}\) by \(\CRD{\rno}{\Wm}{\mQ}{\mP}\).
\end{definition}
Then one can define the Augustin information and the A-L information for all 
\(\mP\) in \(\pmea{\inpA}\),
provided that \(\Wm\) is in \(\pmea{\outA|\inpA}\) for a countably generated 
\(\outA\) and \(\costf\) is a \(\inpA\)-measurable function. 
\begin{definition}\label{def:general-information}
	For any \(\rno\in\reals{+}\), countable generated \(\sigma\)-algebra \(\outA\) 
	of subsets of \(\outS\), \(\Wm\in\pmea{\outA|\inpA}\), and  \(\mP\in\pmea{\inpA}\)
	\emph{the order \(\rno\) Augustin information for the input distribution \(\mP\)} is
	\begin{align}
	\label{eq:def:general-information}
	\RMI{\rno}{\mP}{\Wm}
	&\DEF \inf\nolimits_{\mQ\in \pmea{\outA}} \CRD{\rno}{\Wm}{\mQ}{\mP}.
	\end{align}
	Furthermore, for any \(\inpA\)-measurable cost function
	\(\costf:\inpS\to \reals{\geq0}^{\ell}\) and \(\lgm \in \reals{\geq0}^{\ell}\)
	\emph{the order \(\rno\) Augustin-Legendre information for the input distribution \(\mP\) 
		and the Lagrange multiplier \(\lgm\)} 
	is defined as 
	\begin{align}
	\label{eq:def:general-Linformation} 
	\RMIL{\rno}{\mP}{\Wm}{\lgm} 
	&\DEF \RMI{\rno}{\mP}{\Wm}-\lgm\cdot \EXS{\mP}{\costf}
	\end{align}
	with the understanding that if \(\lgm\cdot \EXS{\mP}{\costf}=\infty\), then 
	\(\RMIL{\rno}{\mP}{\Wm}{\lgm}=-\infty\).
\end{definition}
Although we have included \(\lgm\cdot \EXS{\mP}{\costf}=\infty\) case in
the formal definition of the A-L information, we will only be interested in \(\mP\)'s for
which 	\(\lgm \cdot \EXS{\mP}{\costf}\) is finite. We define \(\csetA^{\lgm}\)
to be the set of all such \(\mP\)'s:
\begin{align}
\label{eq:def:finiteexpectedcost}
\csetA^{\lgm}
&\DEF\{\mP\in\pmea{\inpA}: \lgm\cdot\EXS{\mP}{\costf}<\infty\}.
\end{align}

For an arbitrary \(\sigma\)-algebra \(\inpA\), 
the singletons (i.e. sets with only one element) 
are not necessarily measurable sets; thus \(\pdis{\inpS}\) is not necessarily 
a subset of \(\csetA^{\lgm}\). 
If \(\inpA\) is countably separated, then the singletons
are in \(\inpA\) by \cite[Thm. 6.5.7]{bogachev},
\(\pdis{\inpS}\subset\csetA^{\lgm}\)
 and 
\(\sup_{\mP\in\csetA^{\lgm}}\RMIL{\rno}{\mP}{\Wm}{\lgm}
\geq \RCL{\rno}{\Wm}{\lgm}\). 
The reverse inequality follows from Theorem \ref{thm:Lminimax}
and we have 
\(\sup_{\mP\in\csetA^{\lgm}}\RMIL{\rno}{\mP}{\Wm}{\lgm}=\RCL{\rno}{\Wm}{\lgm}\).
Theorem \ref{thm:general-Lminimax} states these observations formally
together with the ones about the A-L center
through a minimax theorem.
\begin{theorem}\label{thm:general-Lminimax}
	Let \(\inpA\) be a countably separated \(\sigma\)-algebra,
	\(\outA\) be a countably generated \(\sigma\)-algebra,
	\(\Wm\) be a transition probability from \((\inpS,\inpA)\) to \((\outS,\outA)\),
	\(\costf:\inpS\to\reals{\geq0}^{\ell}\) be a \(\inpA\)-measurable cost function, 
	and \(\rno\in\reals{+}\).
	Then for all \(\lgm\in\reals{\geq0}^{\ell}\) we have
	\begin{align}
	\label{eq:thm:general-Lminimax}
	\sup\nolimits_{\mP\in \csetA^{\lgm}
}\inf\nolimits_{\mQ \in \pmea{\outA}} 
	\CRD{\rno}{\Wm}{\mQ}{\mP }-\lgm\cdot\EXS{\mP}{\costf}
	&=
	\inf\nolimits_{\mQ \in \pmea{\outA}} \sup\nolimits_{\mP\in \csetA^{\lgm}
}
	\CRD{\rno}{\Wm}{\mQ}{\mP }-\lgm\cdot\EXS{\mP}{\costf}
	\\
	\label{eq:thm:general-Lminimaxradius}
	&=\inf\nolimits_{\mQ\in\pmea{\outA}}\sup\nolimits_{\dinp \in \inpS} 
	\RD{\rno}{\Wm(\dinp)}{\mQ}-\lgm\cdot \costf(\dinp)
\\
\label{eq:thm:general-Lminimaxradius-pdis}
&=\RCL{\rno}{\Wm}{\lgm}
	\end{align}
	where \(\csetA^{\lgm}\) is defined in \eqref{eq:def:finiteexpectedcost}.
	If \(\RCL{\rno}{\Wm}{\lgm}\) is finite,
	then \(\exists!\qma{\rno,\Wm}{\lgm}\!\in\!\pmea{\outA}\), called the order 
	\(\rno\) Augustin-Legendre center of \(\Wm\) for the Lagrange multiplier \(\lgm\),
	satisfying
	\begin{align}
	\label{eq:thm:general-Lminimaxcenter}
\RCL{\rno}{\Wm}{\lgm}	
&=\sup\nolimits_{\mP \in \csetA^{\lgm}} 
	\CRD{\rno}{\Wm}{\qma{\rno,\Wm}{\lgm}}{\mP}-\lgm\cdot\EXS{\mP}{\costf}
	\\
	\label{eq:thm:general-Lminimaxradiuscenter}
	&=\sup\nolimits_{\dinp \in \inpS} \RD{\rno}{\Wm(\dinp)}{\qma{\rno,\Wm}{\lgm}}-\lgm\cdot\costf(\dinp).
	\end{align}
\end{theorem} 

\begin{proof}[Proof of Theorem \ref{thm:general-Lminimax}]	
Since \(\pdis{\inpS}\subset\csetA^{\lgm}\), the max-min inequality implies
	\begin{align}
	\notag
	\sup\nolimits_{\mP\in \pdis{\inpS}}\inf\nolimits_{\mQ \in \pmea{\outA}} 
	\CRD{\rno}{\Wm}{\mQ}{\mP }-\lgm\cdot\EXS{\mP}{\costf}
	&\leq \sup\nolimits_{\mP\in \csetA^{\lgm}
	}\inf\nolimits_{\mQ \in \pmea{\outA}} 
	\CRD{\rno}{\Wm}{\mQ}{\mP }-\lgm\cdot\EXS{\mP}{\costf}
	\\
	\notag
	&\leq \inf\nolimits_{\mQ \in \pmea{\outA}} \sup\nolimits_{\mP\in \csetA^{\lgm}
	}
	\CRD{\rno}{\Wm}{\mQ}{\mP }-\lgm\cdot\EXS{\mP}{\costf}
	\\
	\notag
	&=\inf\nolimits_{\mQ\in\pmea{\outA}}\sup\nolimits_{\dinp \in \inpS} 
	\RD{\rno}{\Wm(\dinp)}{\mQ}-\lgm\cdot \costf(\dinp).
	\end{align}
	Thus \eqref{eq:thm:general-Lminimax} and \eqref{eq:thm:general-Lminimaxradius} hold 
	as a result of \eqref{eq:thm:Lminimax} and  \eqref{eq:thm:Lminimaxradius} of 
	Theorem \ref{thm:Lminimax}
	and 
	\eqref{eq:thm:general-Lminimaxradius-pdis} follows 
by \eqref{eq:thm:Lminimaxradius} of Theorem \ref{thm:Lminimax}
and \eqref{eq:Lcapacity}.

	If \(\RCL{\rno}{\Wm}{\lgm}\) is finite, then 
	as a result of Theorem \ref{thm:Lminimax}
	there exist a unique \(\qma{\rno,\Wm}{\lgm}\in\pdis{\outA}\) satisfying
	\begin{align}
	\notag
	\sup\nolimits_{\dinp \in \inpS}\RD{\rno}{\Wm(\dinp)}{\qma{\rno,\Wm}{\lgm}}-\lgm\cdot \costf(\dinp)
	&=\RCL{\rno}{\Wm}{\lgm}.
	\end{align}
Then \eqref{eq:thm:general-Lminimaxcenter} and
\eqref{eq:thm:general-Lminimaxradiuscenter} hold 	
because 
\(\sup\nolimits_{\mP\in \csetA^{\lgm}}\CRD{\rno}{\Wm}{\mQ}{\mP }-\lgm\cdot\EXS{\mP}{\costf}
=\sup\nolimits_{\dinp \in \inpS} \RD{\rno}{\Wm(\dinp)}{\mQ}-\lgm\cdot \costf(\dinp)
\) for any \(\mQ\in\pdis{\outA}\).
\end{proof}

Let \(\csetA(\costc)\) be the subset \(\pmea{\inpA}\) 
composed of the probability measures satisfying the cost constraint \(\costc\), 
\begin{align}
\notag %\label{eq:def:general-costconstraint}
\csetA(\costc)
&\DEF\{\mP\in\pmea{\inpA}:\EXS{\mP}{\costf} \leq \costc\}.
\end{align}
Then
\(\cset(\costc)\subset\csetA(\costc)\) 
and
\(\sup_{\mP\in\csetA(\costc)}\RMI{\rno}{\mP}{\Wm}\geq \CRC{\rno}{\Wm}{\costc}\)
whenever \(\inpS\) is countably separated.
For the cost constraints in \(\inte{\fcc{\costf}}\) reverse inequality 
holds as a result of Lemma \ref{lem:Lcapacity}-(\ref{Lcapacity:interior})
and Theorem \ref{thm:general-Lminimax}
and we have \(\sup_{\mP\in\csetA(\costc)}\RMI{\rno}{\mP}{\Wm}=\CRC{\rno}{\Wm}{\costc}\).
Theorem \ref{thm:general-minimax} states these observations formally
together with the ones about the Augustin center
through a minimax theorem.

%%%%%%This inequality is an equality whenever \(\CRC{\rno}{\Wm}{\costc}\) is infinite. 
%%%%%%When \(\CRC{\rno}{\Wm}{\costc}\) is finite, then there exists at least one 
%%%%%%\(\lgm\in\reals{\geq0}^{\ell}\) satisfying 
%%%%%%\(\CRC{\rno}{\Wm}{\costc}=\RCL{\rno}{\Wm}{\lgm}+\lgm \cdot \costc\)
%%%%%%by Lemma \ref{lem:Lcapacity}-(\ref{Lcapacity:interior}).
%%%%%%As a result
%%%%%%\(\sup_{\mP\in\csetA(\costc)}\CRD{\rno}{\Wm}{\qma{\rno,\Wm}{\lgm}}{\mP}\leq \CRC{\rno}{\Wm}{\costc}\)
%%%%%%by  Theorem \ref{thm:general-Lminimax}.
%%%%%%Then  \(\sup_{\mP\in\csetA(\costc)}\RMI{\rno}{\mP}{\Wm}=\CRC{\rno}{\Wm}{\costc}\).
%%%%%%Theorem \ref{thm:general-minimax} states these observations formally
%%%%%%together with the ones about the Augustin center
%%%%%%through a minimax theorem.

\begin{theorem}\label{thm:general-minimax}
	Let \(\inpA\) be a countably separated \(\sigma\)-algebra,
\(\outA\) be a countably generated \(\sigma\)-algebra,
\(\Wm\) be a transition probability from \((\inpS,\inpA)\) to \((\outS,\outA)\),
\(\costf:\inpS\to\reals{\geq0}^{\ell}\) be a \(\inpA\)-measurable cost function, 
and \(\rno\in\reals{+}\).
For any \(\costc\in\inte{\fcc{\costf}}\) we have
\begin{align}
\label{eq:thm:general-minimax}
\sup\nolimits_{\mP \in \csetA(\costc)}   \inf\nolimits_{\mQ \in \pmea{\outA}} \CRD{\rno}{\Wm}{\mQ}{\mP}
&=
\inf\nolimits_{\mQ \in \pmea{\outA}} \sup\nolimits_{\mP \in \csetA(\costc)}   \CRD{\rno}{\Wm}{\mQ}{\mP}
\\
\label{eq:thm:general-minimax-pdis}
&=\CRC{\rno}{\Wm}{\costc}
\end{align}
where \(\CRC{\rno}{\Wm}{\costc}\) is defined in \eqref{eq:def:costcapacity}.
If \(\CRC{\rno}{\Wm}{\costc}\!\in\!\reals{\geq0}\), then
\(\exists!\qmn{\rno,\Wm,\costc}\!\in\!\pmea{\outA}\), called the order 
\(\rno\) Augustin center of \(\Wm\) for the cost constraint \(\costc\),
satisfying
\begin{align}
\label{eq:thm:general-minimaxcenter}
\CRC{\rno}{\Wm}{\costc}
&=\sup\nolimits_{\mP \in \csetA(\costc)} 
\CRD{\rno}{\Wm}{\qmn{\rno,\Wm,\costc}}{\mP}
\\
\label{eq:thm:general-minimaxcenter-pdis}
&=\sup\nolimits_{\mP \in \cset(\costc)} 
\CRD{\rno}{\Wm}{\qmn{\rno,\Wm,\costc}}{\mP}.
\end{align}	
Furthermore, \(\qmn{\rno,\Wm,\costc}=\qma{\rno,\Wm}{\lgm}\) for all \(\lgm\in\reals{\geq0}^{\ell}\)
satisfying \(\CRC{\rno}{\Wm}{\costc}=\RCL{\rno}{\Wm}{\lgm}+\lgm \cdot \costc\). 
\end{theorem}

\begin{proof}[Proof of Theorem \ref{thm:general-minimax}]
Since \(\cset(\costc)\subset\csetA(\costc)\),  the max-min inequality implies
	\begin{align}
	\notag
	\sup\nolimits_{\mP \in \cset(\costc)}   \inf\nolimits_{\mQ \in \pmea{\outA}} \CRD{\rno}{\Wm}{\mQ}{\mP}
	&\leq \sup\nolimits_{\mP \in \csetA(\costc)}   \inf\nolimits_{\mQ \in \pmea{\outA}} \CRD{\rno}{\Wm}{\mQ}{\mP}
	\\
	\notag
	&\leq
	\inf\nolimits_{\mQ \in \pmea{\outA}} \sup\nolimits_{\mP \in \csetA(\costc)}   \CRD{\rno}{\Wm}{\mQ}{\mP}.
	\end{align}
	Thus both \eqref{eq:thm:general-minimax} and \eqref{eq:thm:general-minimax-pdis} hold whenever 
	\(\CRC{\rno}{\Wm}{\costc}=\infty\) by \eqref{eq:capacity}.
	On the other hand, as a result of Theorem \ref{thm:general-Lminimax} 
	for any \(\lgm\) with finite \(\RCL{\rno}{\Wm}{\lgm}\) there exists a unique
	\(\qma{\rno,\Wm}{\lgm}\) satisfying \eqref{eq:thm:general-Lminimaxradiuscenter}. Thus we have,
	\begin{align}
	\notag
	\inf\nolimits_{\mQ \in \pmea{\outA}} \sup\nolimits_{\mP \in \csetA(\costc)}   \CRD{\rno}{\Wm}{\mQ}{\mP}
	&\leq \sup\nolimits_{\mP \in \csetA(\costc)}   \CRD{\rno}{\Wm}{\qma{\rno,\Wm}{\lgm}}{\mP}
	\\
	\notag
	&\leq \sup\nolimits_{\mP \in \csetA(\costc)}   
	\CRD{\rno}{\Wm}{\qma{\rno,\Wm}{\lgm}}{\mP}-\lgm\cdot\EXS{\mP}{\costf}+\lgm\cdot\costc
	\\
	\notag
	&\leq \RCL{\rno}{\Wm}{\lgm}+\lgm\cdot\costc.
	\end{align}
	Furthermore, if \(\CRC{\rno}{\Wm}{\costc}\in\reals{}\), then there exists at least one
	\(\lgm\in\reals{\geq0}^{\ell}\) satisfying 
	\(\CRC{\rno}{\Wm}{\costc}=\RCL{\rno}{\Wm}{\lgm}+\lgm \cdot \costc\)
	by Lemma \ref{lem:Lcapacity}-(\ref{Lcapacity:interior}).
	Then \eqref{eq:thm:general-minimax} and \eqref{eq:thm:general-minimax-pdis} hold when
	\(\CRC{\rno}{\Wm}{\costc}\in\reals{}\) and 
	\eqref{eq:thm:general-minimaxcenter} holds for \(\qmn{\rno,\Wm,\costc}=\qma{\rno,\Wm}{\lgm}\)
	provided that \(\CRC{\rno}{\Wm}{\costc}=\RCL{\rno}{\Wm}{\lgm}+\lgm \cdot \costc\).
	On the other hand 
	\(\qmn{\rno,\Wm,\costc}\) is a probability measure satisfying
	\eqref{eq:thm:general-minimaxcenter-pdis} by Theorem \ref{thm:minimax}
	and
	\(\qmn{\rno,\Wm,\costc}=\qma{\rno,\Wm}{\lgm}\)
	for all \(\lgm\) satisfying \(\CRC{\rno}{\Wm}{\costc}=\RCL{\rno}{\Wm}{\lgm}+\lgm \cdot \costc\)
	by Lemma \ref{lem:Lcenter}.
\end{proof}

The countable separability of \(\inpA\) and countable generatedness of \(\outA\) 
are fairly mild assumptions satisfied by most transition probabilities considered 
in practice.
Hence, Theorems \ref{thm:general-Lminimax} and \ref{thm:general-minimax} provide
further justification for studying the relatively simple case of probability 
mass functions, first.

The existence of an input distribution \(\mP\) 
satisfying both \(\EXS{\mP}{\costf}\leq\costc\) and 
\(\RMI{\rno}{\mP}{\Wm}=\CRC{\rno}{\Wm}{\costc}\)
is immaterial to the existence of a unique \(\qmn{\rno,\Wm,\costc}\)
or its characterization through \(\qma{\rno,\Wm}{\lgm}\)
for \(\lgm\)'s satisfying 
\(\CRC{\rno}{\Wm}{\costc}=\RCL{\rno}{\Wm}{\lgm}+\lgm\cdot\costc\)
by Lemma \ref{lem:Lcapacity}-(\ref{Lcapacity:interior},\ref{Lcapacity:optimal})
and
Theorem \ref{thm:general-minimax}.
Although one can prove the existence of such a \(\mP\) for 
certain special cases
such an input distribution does not exist in general. 
Thus, we believe, it is preferable to separate the issue of the existence
of an optimal input distribution from the discussion 
of \(\CRC{\rno}{\Wm}{\costc}\) and \(\qmn{\rno,\Wm,\costc}\) 
and their characterization via \(\RCL{\rno}{\Wm}{\lgm}\) and
\(\qma{\rno,\Wm}{\lgm}\). 
That, however, is not the standard practice, \cite[Thm. 1]{kostinaV15}.

We have assumed \(\outA\) to be countably generated in order to ensure that 
the conditional \renyi divergence used in \eqref{eq:def:general-information}
is well-defined. 
In order to define the \renyi information, however, we do not need to assume 
\(\outA\) to be countably generated; the transition probability structure 
is sufficient.
Recall that if \(\Wm\in\pmea{\outA|\inpA}\), then for any 
\(\mP\in\pmea{\inpA}\) there exists a unique probability measure 
\(\mP \mtimes \Wm\) on \((\inpS\times\outS,\inpA\otimes\outA)\) 
such that 
\begin{align}
\notag
\mP \mtimes \Wm(\oev_{\dinp}\times\oev_{\dout})
&=\int_{\oev_{\dinp}} \Wm(\oev_{\dout}|\dinp) \mP(\dif{\dinp}).
&
&\forall \oev_{\dinp}\in \inpA,\oev_{\dout}\in \outA
\end{align}
by \cite[Thm. 10.7.2.]{bogachev}.
Thus \(\GMI{\rno}{\mP}{\Wm}\) is well defined for any 
\(\Wm\in\pmea{\outA|\inpA}\) and \(\mP\in\pmea{\inpA}\).

Unfortunately, the situation is not nearly as simple for the R-G information. 
In order to define the R-G information using a similar approach one first
shows that \(\Wm e^{\frac{1-\rno}{\rno}\lgm\cdot\costf}\) is
a transition kernel ---rather than a transition probability (i.e. Markov kernel)--- 
and then proceeds with establishing the existence a unique measure 
\(\mP\mtimes \Wm e^{\frac{1-\rno}{\rno}\lgm\cdot\costf}\) for 
all \(\mP\) in \(\pmea{\outA}\).
For orders greater than one, resulting measure is a sub-probability measure 
and one can use \eqref{eq:def:Ginformation} as the definition of 
the R-G information. 
For orders between zero and one, on the other hand,
\(\mP\mtimes \Wm e^{\frac{1-\rno}{\rno}\lgm\cdot\costf}\) is a 
\(\sigma\)-finite measure for all \(\mP\)'s in \(\pmea{\inpA}\),
but it is not necessarily a finite measure for all \(\mP\)'s in 
\(\pmea{\inpA}\).
Thus for orders between zero and one,
one can use \eqref{eq:def:Ginformation} as the definition of 
the R-G information, only after extending the definition of 
the \renyi divergence to \(\sigma\)-finite measures.
%!TEX root=../main-C.tex
\section{Examples}\label{sec:examples}
In this section, we will first demonstrate certain subtleties that we have pointed out 
in the earlier sections.
After that we will study Gaussian channels and obtain closed form expressions
for their Augustin capacity and center.

\subsection{Shift Invariant Families}\label{sec:examples-invariant}
\begin{example}[A Channel with an Affine Capacity]\label{eg:affine}
	Let the channel \(\Wm:\reals{\geq0}\to \pmea{\rborel{[0,1)}}\) 
	and the associated cost function \(\costf:\reals{\geq0}\to\reals{\geq0}\)
	be 
	\begin{align}
	\notag
	\der{\Wm(\dinp)}{\rfm}
	&=\fX_{\lfloor\dinp\rfloor}(\dout-\dinp-\lfloor\dout-\dinp\rfloor),
	\\
	\notag
	\costf(\dinp)&=\lfloor \dinp \rfloor
	\end{align} 
	where \(\rfm\) is the Lebesgue measure on \([0,1)\) and \(\fX_{\ind}\)'s are given by 
	\begin{align}
	\notag
	\fX_{\ind}(\dout)
	&=e^{\ind+1}\IND{\dout\in [0,e^{-\ind-1})}
	&
	&\forall \ind\in\integers{\geq0}.
	\end{align}
	Let \(\umn{\ind}\) be uniform distribution on \([\ind,\ind+1)\); then 
	one can confirm by substitution that \(\Aop{\rno}{\umn{\ind}}{\umn{0}}=\umn{0}\).
	Then using the Jensen's inequality together with the fixed point property we 
	get\footnote{See the derivation of \eqref{eq:lem:information-zto:EHB} and 
		\eqref{eq:lem:information-oti:EHB} of 
		Lemma \ref{lem:information}-(\ref{information:zto},\ref{information:oti})
		given in 
		Appendix \ref{sec:informationproofs}.}
	\begin{align}
	\notag
	\CRD{\rno}{\!\Wm\!}{\mQ}{\umn{\ind}}
	&\geq \CRD{\rno}{\!\Wm\!}{\umn{0}}{\umn{\ind}}+\RD{\rno\wedge1}{\umn{0}}{\mQ}.
	\end{align}
	Thus \(\umn{0}\) is the unique order \(\rno\) Augustin mean for 
	the input distribution \(\umn{\ind}\), i.e. \(\qmn{\rno,\umn{\ind}}=\umn{0}\),
	and 
	\(\RMI{\rno}{\umn{\ind}}{\!\Wm}=\CRD{\rno}{\!\Wm\!}{\umn{0}}{\umn{\ind}}\)
	---and hence \(\RMI{\rno}{\umn{\ind}}{\!\Wm}=\ind+1\)---
		for all \(\ind\in\integers{+}\) and \(\rno\in\reals{+}\).
	Then using  \(\EXS{\umn{\ind}}{\costf}=\ind\), we can conclude that 
	\(\CRC{\rno}{\!\Wm\!}{\costc}\geq (\costc+1)\) 
	not only for \(\costc\in\integers{\geq0}\) but also for
	\(\costc\in\reals{\geq0}\) because \(\CRC{\rno}{\!\Wm\!}{\costc}\)
	is concave in \(\costc\) by Lemma \ref{lem:CCcapacity}-(\ref{CCcapacity:function}).
	One the other hand, one can confirm by substitution that 
	\begin{align}
	\label{eq:eg:affine-A}
	\CRD{\rno}{\!\Wm\!}{\umn{0}}{\mP}
	&=\EXS{\mP}{\costf}+1.
	\end{align}
	Thus \(\RMI{\rno}{\mP}{\!\Wm}\leq(\costc+1)\) for any \(\mP\) satisfying the 
	cost constraint \(\costc\). Hence,
	\begin{align}
	\notag
	\CRC{\rno}{\!\Wm\!}{\costc}
	&=\costc+1,
	\\
	\notag
	\qmn{\rno,\!\Wm\!,\costc}
	&=\umn{0}.
	\end{align}
	Then as a result of \eqref{eq:Lcapacity-astheconjugate} we have
	\begin{align}
	\notag	
	\RCL{\rno}{\!\Wm\!}{\lgm}
	&=\begin{cases}
	\infty 
	& \lgm\in [0,1)
	\\
	1
	& \lgm\in [1,\infty)
	\end{cases}.
	\end{align}
	Then using \eqref{eq:eg:affine-A} and Theorem \ref{thm:general-Lminimax}, we can conclude that 
	\(\qma{\rno,\Wm}{\lgm}=\umn{0}\) for all \(\lgm\in[1,\infty)\).
\end{example}

\begin{example}[A Channel with a Non-Upper Semicontinuous Capacity]\label{eg:non-uppersemicontinuous}
	Let the channel \(\Wm:\reals{}\to \pmea{\rborel{[0,1)}}\) 
	and the associated cost function \(\costf:\reals{}\to\reals{\geq0}\)
	be 
	\begin{align}
	\notag
	\der{\Wm(\dinp)}{\rfm}
	&=\fX_{\lfloor\dinp\rfloor}(\dout-\dinp-\lfloor\dout-\dinp\rfloor)
	\\
	\notag
	\costf(\dinp)
	&=\begin{cases}
	\lfloor \dinp \rfloor
	&\dinp\geq 0
	\\
	2^{\lfloor \dinp \rfloor}
	&\dinp <0 
	\end{cases}.
	\end{align}
	where \(\rfm\) is the Lebesgue measure on \([0,1)\) and \(\fX_{\ind}:\in[0,1)\to\reals{\geq 0}\) are given by 
	\begin{align}
	\notag
	\fX_{\ind}(\dout)
	&=
	\begin{cases}
	2^{\ind+1}\IND{\dout\in [0,2^{-\ind-1})}
	&\ind>0
	\\
	\sfrac{3}{2}\IND{\dout\in [0,\sfrac{2}{3})}
	&\ind=0
	\\
	2\IND{\dout\in [0,\sfrac{1}{2})}
	&\ind<0	
	\end{cases}.
	\end{align}
	Following an analysis similar to the one described above we can conclude that	
	\begin{align}
	\notag
	\CRC{\rno}{\!\Wm\!}{\costc}
	&=\begin{cases}
	(\costc+1)\ln 2
	&\costc>0
	\\
	\ln\sfrac{3}{2}
	&\costc=0
	\end{cases}	,
	\\
	\notag
	\RCL{\rno}{\!\Wm\!}{\lgm}
	&=\begin{cases}
	\infty 
	& \lgm\in [0,\ln 2)
	\\
	\ln 2
	& \lgm\in [\ln2,\infty)
	\end{cases}.
	\end{align}
	Hence \(\CRC{\rno}{\!\Wm\!}{\costc}\neq \inf_{\lgm\geq0} \RCL{\rno}{\!\Wm\!}{\lgm}+\lgm\cdot\costc\) for 
	\(\costc=0\).
\end{example}

\begin{example}[A Product Channel without an Optimal Cost Allocation]\label{eg:non-additive-product}
	Let \(\Wmn{1}\) and \(\Wmn{2}\) be the channels described in Examples 
	\ref{eg:affine} and \ref{eg:non-uppersemicontinuous} and \(\costf_{1}\)
	and \(\costf_{2}\) be the associated cost functions. 
	Let \(\Wmn{[1,2]}\) be the product of these two channels with the additive 
	cost function \(\costc_{[1,2]}\), i.e.
	\begin{align}
	\notag
	\Wmn{[1,2]}(\dinp_1,\dinp_{2})
	&=\Wmn{1}(\dinp_1) \otimes \Wmn{2}(\dinp_{2}), 
	\\
	\notag
	\costf_{[1,2]}(\dinp_1,\dinp_{2})
	&=\costf_{1}(\dinp_1)+\costf_{2}(\dinp_{2}).
	\end{align}
	Then Lemma \ref{lem:CCcapacityproduct} implies
	\begin{align}
	\notag
	\CRC{\rno}{\Wmn{[1,2]}}{\costc}
	&=\begin{cases}
	\costc+1+\ln 2
	&\costc>0
	\\
	1+\ln\tfrac{3}{2}
	&
	\costc=0
	\end{cases}.
	\end{align}
	Note that for positive values of \(\costc\) there does not exist any 	
	\((\costc_{1},\costc_{2})\) pair satisfying both
	\(\CRC{\rno}{\Wmn{[1,2]}}{\costc}=\CRC{\rno}{\Wmn{1}}{\costc_{1}}+\CRC{\rno}{\Wmn{2}}{\costc_{2}}\)
	and the cost constraint \(\costc_{1}+\costc_{2}\leq \costc\) at the same time.
\end{example}

\subsection{Gaussian Channels}\label{sec:examples-gaussian}
In the following,  we denote the zero mean Gaussian probability measure on 
\(\rborel{\reals{}}\) with variance \(\sigma^{2}\) by  \(\GausDen{\sigma^{2}}\).
With a slight abuse of notation, we denote the corresponding probability density 
function by the same symbol:
\begin{align}
\notag %\label{eq:GaussianDensity}
\GausDen{\sigma^{2}}(\dinp)
&=\tfrac{1}{\sqrt{2 \pi} \sigma} e^{-\frac{\dinp^{2}}{2\sigma^{2}}}
&
&\forall \dinp\in\reals{}.
\end{align}
We use the Gaussian channels and the corresponding 
transition probabilities interchangeably;
they have the same cost constrained Augustin capacity and center 
by Theorems \ref{thm:general-Lminimax} and \ref{thm:general-minimax}. 

\begin{example}[Scalar Gaussian Channel]\label{eg:SGauss}
	Let \(\Wm\) be the scalar Gaussian channel with noise variance \(\sigma^2\) 
	and the associated cost function \(\costf\) be the quadratic one, i.e.
	\begin{align}
	\notag
	\Wm(\oev|\dinp)
	&=\int_{\oev} \GausDen{\sigma^{2}}(\dout-\dinp) \dif{\dout}
	&
	&\forall \oev\in\rborel{\reals{}},
	\\
	\notag
	\costf(\dinp)
	&=\dinp^{2}
	&
	&\forall \dinp\in\reals{}.
	\end{align}
The Augustin capacity and center of this channel are given by the following expressions:
\begin{align}
\label{eq:eg:SGauss-capacity}
\CRC{\rno}{\!\Wm\!}{\costc}
&=\begin{cases}
\tfrac{\rno \costc}{2(\rno \theta_{\rno,\sigma,\costc}+(1-\rno)\sigma^{2})}
+\tfrac{1}{\rno-1}\ln\tfrac{(\theta_{\rno,\sigma,\costc})^{\sfrac{\rno}{2}}\sigma^{(1-\rno)}}{\sqrt{\rno\theta_{\rno,\sigma,\costc}+(1-\rno) \sigma^{2}}}
&\rno\in\reals{+}\setminus\{1\}
\\
\tfrac{1}{2}\ln \left(1+\tfrac{\costc}{\sigma^{2}}\right)
&\rno=1
\end{cases},
\\
\label{eq:eg:SGauss-center}
\qmn{\rno,\!\Wm\!,\costc}
&=\GausDen{\theta_{\rno,\sigma,\costc}},
\\
\label{eq:eg:SGauss-center-variance}
\theta_{\rno,\sigma,\costc}
&\DEF\sigma^{2}+\tfrac{\costc}{2}-\tfrac{\sigma^{2}}{2\rno}+\sqrt{(\tfrac{\costc}{2}-\tfrac{\sigma^{2}}{2\rno})^{2}+ \costc\sigma^2}.
\end{align}	
Furthermore, \(\CRC{\rno}{\!\Wm\!}{\costc}\)  is continuously differentiable in \(\costc\) and 
its derivative is a continuous, decreasing, and bijective function of \(\costc\) 
from \(\reals{+}\) to \([0,\sfrac{\rno}{2\sigma^{2}})\)
given by
\begin{align}
\label{eq:eg:SGauss-capacity-derivative}
\der{}{\costc}\CRC{\rno}{\!\Wm\!}{\costc}
&=\tfrac{\rno}{2(\rno \theta_{\rno,\sigma,\costc}+(1-\rno)\sigma^{2})}
\\
\label{eq:eg:SGauss-capacity-derivative-alt}
&=\tfrac{\rno}{\rno\costc+\sigma^2+\sqrt{(\rno\costc-\sigma^2)^{2}+4\costc\rno^2\sigma^{2}}}.
\end{align}
In order to prove these, we first demonstrate that the Augustin mean 
for the zero mean Gaussian distribution with variance \(\costc\) is
the zero mean Gaussian distribution with variance \(\theta_{\rno,\sigma,\costc}\),
i.e. \(\qmn{\rno,\GausDen{\costc}}=\GausDen{\theta_{\rno,\sigma,\costc}}\).
This will imply \(\RMI{\rno}{\GausDen{\costc}}{\!\Wm}=\CRD{\rno}{\!\Wm\!}{\GausDen{\theta_{\rno,\sigma,\costc}}}{\GausDen{\costc}}\).
\(\CRD{\rno}{\!\Wm\!}{\GausDen{\theta_{\rno,\sigma,\costc}}}{\GausDen{\costc}}\) is equal to 
the expression on the right hand side of \eqref{eq:eg:SGauss-capacity}.
In order to establish \eqref{eq:eg:SGauss-capacity} and \eqref{eq:eg:SGauss-center},
we demonstrate that this value is the greatest value for the Augustin information 
among all input distributions satisfying the cost constraint \(\costc\).
Consequently, we have	
\(\CRC{\rno}{\!\Wm\!}{\costc}=\RMI{\rno}{\GausDen{\costc}}{\!\Wm}\) and 
\(\qmn{\rno,\!\Wm\!,\costc}=\qmn{\rno,\GausDen{\costc}}\).
Then we confirm \eqref{eq:eg:SGauss-capacity-derivative} using 
an identity, i.e. \eqref{eq:eg:SGauss-necessarycondition},
obtained while establishing \(\qmn{\rno,\GausDen{\costc}}=\GausDen{\theta_{\rno,\sigma,\costc}}\).

	One can confirm by substitution that 
	\begin{align}
	\RD{\rno}{\Wm(\dinp)}{\GausDen{\theta}}
	%	&=\begin{cases}
	%	\tfrac{1}{\rno-1}\ln \int \tfrac{1}{\sqrt{2\pi\sigma^{2\rno}\theta^{(1-\rno)}}}
	%	e^{-\frac{\rno(\dinp-\dout)^{2}}{2\sigma^{2}}-\frac{(1-\rno)\dout^{2}}{2\theta}} 
	%	\dif{\dout}
	%	&\rno\in\reals{+}\setminus\{1\}
	%	\notag
	%	\\
	%	\int 
	%	\tfrac{1}{\sqrt{2\pi}\sigma} e^{-\frac{(\dinp-\dout)^{2}}{2\sigma^{2}}} 
	%	\left(\tfrac{\dout^{2}}{2\theta}-\tfrac{(\dinp-\dout)^{2}}{2\sigma^{2}}
	%	+\tfrac{1}{2}\ln \tfrac{\theta}{\sigma^{2}} \right)
	%	\dif{\dout}
	%	&\rno=1
	%	\end{cases}
	%	\\
	%	\notag
	%	&=
	%	\begin{cases}
	%	\frac{\rno\dinp^{2}}{2(\rno \theta+(1-\rno)\sigma^{2})}
	%	+\tfrac{1}{\rno-1}\ln \int \tfrac{1}{\sqrt{2\pi\sigma^{2\rno}\theta^{(1-\rno)}}}
	%	e^{-\frac{1}{2}(\frac{\sqrt{\rno\theta+(1-\rno)\sigma^{2}}}{\sigma \theta^{\sfrac{1}{2}}}\dout
	%		-\frac{\rno \theta^{\sfrac{1}{2}} }{\sigma\sqrt{\rno\theta+(1-\rno)\sigma^{2}}}\dinp)^2}
	%	\dif{\dout}
	%	&\rno\in\reals{+}\setminus\{1\}
	%	\\
	%	\tfrac{\sigma^{2}+\dinp^{2}}{2\theta}-\tfrac{\sigma^{2}}{2\sigma^{2}}+\tfrac{1}{2}\ln \tfrac{\theta}{\sigma^{2}}
	%	&\rno=1
	%	\end{cases}
	%	\\
	\label{eq:eg:SGauss-divergence}
	&=\begin{cases}
	\tfrac{\rno \dinp^{2}}{2(\rno \theta+(1-\rno)\sigma^{2})}
	+\tfrac{1}{\rno-1}\ln \tfrac{\theta^{\sfrac{\rno}{2}}\sigma^{(1-\rno)}}{\sqrt{\rno \theta+(1-\rno) \sigma^{2}}}
	&\rno\in\reals{+}\setminus\{1\}
	\\
	\tfrac{\sigma^{2}+\dinp^{2}-\theta}{2\theta}
	+\tfrac{1}{2}\ln \tfrac{\theta}{\sigma^{2}}
	&\rno=1
	\end{cases}.
	\end{align}
	Then the order \(\rno\) tilted channel \(\Wma{\rno}{\GausDen{\theta}}\), defined in 
	\eqref{eq:def:tiltedchannel}, is a Gaussian channel as well:
	\begin{align}
	\notag
	\Wma{\rno}{\GausDen{\theta}}(\oev|\dinp)
	&=\int_{\oev} \GausDen{\frac{\sigma^{2} \theta}{\rno\theta+(1-\rno)\sigma^{2}}}
	\left(\dout-\tfrac{\rno \theta}{\rno\theta+(1-\rno)\sigma^{2}}\dinp\right)\dif{\dout}.
	\end{align}	
	Then \(\Aop{\rno}{\mP}{\mQ}\) is a zero mean Gaussian probability measure whenever 
	both \(\mP\) and \(\mQ\) are so. In particular,
	\begin{align}
	\label{eq:eg:SGauss-Augustinoperator}
	\Aop{\rno}{\GausDen{\costc}}{\GausDen{\theta}}
	&=\GausDen{(\tfrac{\rno \theta}{\rno \theta+(1-\rno)\sigma^{2}})^2 \costc+\tfrac{\sigma^{2} \theta}{\rno \theta+(1-\rno)\sigma^{2}}}.
	\end{align} 	
	Consequently, if \(\GausDen{\theta}\) is a fixed point of \(\Aop{\rno}{\GausDen{\costc}}{\cdot}\),
	then \(\theta\) satisfies the following equality
	\begin{align}
	\label{eq:eg:SGauss-necessarycondition}
	\theta \left[\theta^{2}-\theta \left(\costc+\left(2-\tfrac{1}{\rno}\right)\sigma^{2}\right)+\left(1-\tfrac{1}{\rno}\right)\sigma^{4}\right]
	&=0.
	\end{align}
\(\theta_{\rno,\sigma,\costc}\), defined in \eqref{eq:eg:SGauss-center-variance}, 
is the only root of the equality given in \eqref{eq:eg:SGauss-necessarycondition} 
that is greater than \(\sigma^{2}\) for \(\rno\)'s in \(\reals{+}\);
it is the only positive root for \(\rno\)'s in \((0,1)\),
as well. 	
	Furthermore, using \eqref{eq:eg:SGauss-Augustinoperator} one can confirm 
	that \(\Aop{\rno}{\GausDen{\costc}}{\GausDen{\theta_{\rno,\sigma,\costc}^{2}}}=
	\GausDen{\theta_{\rno,\sigma,\costc}}\), i.e.
	\(\GausDen{\theta_{\rno,\sigma,\costc}}\) is a fixed point of 
	\(\Aop{\rno}{\GausDen{\costc}}{\cdot}\).
	Then using the Jensen's inequality together with this fixed point property 	we 
	get\footnote{Derivation of this inequality is analogous to the derivation of 
		\eqref{eq:lem:information-zto:EHB} and \eqref{eq:lem:information-oti:EHB} of 
		Lemma \ref{lem:information}-(\ref{information:zto},\ref{information:oti}), presented in 
		Appendix \ref{sec:informationproofs}.}
	\begin{align}
	\notag
	\CRD{\rno}{\!\Wm\!}{\mQ}{\GausDen{\costc}}
	&\geq
	\CRD{\rno}{\!\Wm\!}{\GausDen{\theta_{\rno,\sigma,\costc}}}{\GausDen{\costc}}
	+\RD{1\wedge\rno}{\GausDen{\theta_{\rno,\sigma,\costc}}}{\mQ}
	&
	&\forall \mQ\in \pmea{\rborel{\reals{}}}.
	\end{align}
	Thus \(\GausDen{\theta_{\rno,\sigma,\costc}}\) is the order \(\rno\) 
	Augustin mean for the input distribution \(\GausDen{\costc}\), i.e. 
	\(\qmn{\rno,\GausDen{\costc}}=\GausDen{\theta_{\rno,\sigma,\costc}}\) 
	and \(\RMI{\rno}{\GausDen{\costc}}{\!\Wm}
	=\CRD{\rno}{\!\Wm\!}{\GausDen{\theta_{\rno,\sigma,\costc}}}{\GausDen{\costc}}\).	
	On the other hand, \eqref{eq:eg:SGauss-divergence} implies 
	\begin{align}
	\label{eq:eg:SGauss-information}
	\CRD{\rno}{\!\Wm\!}{\GausDen{\theta_{\rno,\sigma,\costc}}}{\mP}
	&=\tfrac{\rno (\EXS{\mP}{\costf}-\costc)}{2(\rno \theta_{\rno,\sigma,\costc}+(1-\rno)\sigma^{2})}+\RMI{\rno}{\GausDen{\costc}}{\!\Wm} 
	&
	&\forall \mP\in \pmea{\rborel{\reals{}}}.
	\end{align}
	Then \(\RMI{\rno}{\mP}{\!\Wm} \leq \RMI{\rno}{\GausDen{\costc}}{\!\Wm}\) for all \(\mP\)
	satisfying \(\EXS{\mP}{\costf}\leq\costc\).
	Consequently, \(\CRC{\rno}{\!\Wm\!}{\costc}=\RMI{\rno}{\GausDen{\costc}}{\!\Wm}\)
	and \(\qmn{\rno,\!\Wm\!,\costc}=\qmn{\rno,\GausDen{\costc}}\).

For \(\rno=1\) case \eqref{eq:eg:SGauss-capacity-derivative} is evident.
In order to establish \eqref{eq:eg:SGauss-capacity-derivative}
for \(\rno\in\reals{+}\!\setminus\!\{1\}\) case, note that
\begin{align}
\notag
\der{}{\costc}\CRC{\rno}{\!\Wm\!}{\costc}
%&=\tfrac{\rno}{2(\rno \theta_{\rno,\sigma,\costc}+(1-\rno)\sigma^{2})}
%+\left[
%\tfrac{-\rno^{2}\costc}{2(\rno \theta_{\rno,\sigma,\costc}+(1-\rno)\sigma^{2})^{2}}
%+\tfrac{\rno}{2(\rno-1)\theta_{\rno,\sigma,\costc}}
%-\tfrac{\rno}{2(\rno-1)(\rno \theta_{\rno,\sigma,\costc}+(1-\rno)\sigma^{2})}\right]\der{}{\costc}\theta_{\rno,\sigma,\costc}
%\\
%\notag
&=\tfrac{\rno}{2(\rno \theta_{\rno,\sigma,\costc}+(1-\rno)\sigma^{2})}
+\left[
\tfrac{-\rno^{2}\costc}{2(\rno \theta_{\rno,\sigma,\costc}+(1-\rno)\sigma^{2})^{2}}
+\tfrac{\rno(\theta_{\rno,\sigma,\costc}-\sigma^{2})}{2(\rno \theta_{\rno,\sigma,\costc}+(1-\rno)\sigma^{2})\theta_{\rno,\sigma,\costc}}\right]\der{}{\costc}\theta_{\rno,\sigma,\costc}
\\
\notag
&=\tfrac{\rno}{2(\rno \theta_{\rno,\sigma,\costc}+(1-\rno)\sigma^{2})}
+\tfrac{\rno^{2}}{2(\rno \theta_{\rno,\sigma,\costc}+(1-\rno)\sigma^{2})^{2}\theta_{\rno,\sigma,\costc}}
\left[\theta_{\rno,\sigma,\costc}^{2}-\theta_{\rno,\sigma,\costc}\left(\costc+(2-\tfrac{1}{\rno})\sigma^{2}\right)
+(1-\tfrac{1}{\rno})\sigma^{4}
\right]
\der{}{\costc}\theta_{\rno,\sigma,\costc}.
\end{align}
Then  \eqref{eq:eg:SGauss-capacity-derivative} holds
for \(\rno\in\reals{+}\!\setminus\!\{1\}\)  
because \(\theta_{\rno,\sigma,\costc}\) is a root of the equality 
in \eqref{eq:eg:SGauss-necessarycondition}.

The A-L capacity and center of this channel are given by the following expressions:
	\begin{align}
\label{eq:eg:SGauss-ALcapacity}
\RCL{\rno}{\!\Wm\!}{\lgm}
&=\begin{cases}
\left(
\tfrac{\rno}{\rno-1}\ln \sqrt{\tfrac{1}{\rno}+\tfrac{\rno-1}{\rno}\tfrac{2\sigma^{2}\lgm}{\rno}}
-\ln \sqrt{\tfrac{2\sigma^{2}\lgm}{\rno}}
\right)
\IND{\lgm\in (0,\frac{\rno}{2\sigma^{2}})}
&\rno\in\reals{+}\setminus\{1\}
\\
\left(
\sigma^{2}\lgm-\ln \sqrt{2e\sigma^{2}\lgm}
\right)
\IND{\lgm\in (0,\tfrac{1}{2\sigma^{2}})} 
&\rno=1
\end{cases}
\\
\label{eq:eg:SGauss-ALcenter}
\qma{\rno,\Wm}{\lgm}
&=\GausDen{\theta_{\rno,\sigma}^{\lgm}}
\\
\label{eq:eg:SGauss-ALcenter-variance}
\theta_{\rno,\sigma}^{\lgm}
&\DEF
\sigma^{2}+\abp{\tfrac{1}{2\lgm}-\tfrac{\sigma^{2}}{\rno}}
\end{align} 
Then \(\RCL{\rno}{\!\Wm\!}{\lgm}\) is a continuously differentiable function of \(\lgm\) 
and its derivative is a continuous, increasing, and bijective function of \(\lgm\) 
from \(\reals{+}\) to \((-\infty,0]\)
given by
\begin{align}
\label{eq:eg:SGauss-ALcapacity-derivative}
\der{}{\lgm}\RCL{\rno}{\!\Wm\!}{\lgm}
&=-\tfrac{\rno-2\sigma^{2}\lgm}{2\lgm(\rno+(\rno-1)2\sigma^{2}\lgm)}\IND{\lgm\leq\frac{\rno}{2\sigma^{2}}}.
\end{align}
The expressions for the A-L capacity and center given in 
\eqref{eq:eg:SGauss-ALcapacity} and \eqref{eq:eg:SGauss-ALcenter}
are derived using
the expressions for Augustin capacity and center,
\eqref{eq:Lcapacity-astheconjugate}, 
\eqref{eq:eg:SGauss-capacity-derivative},  
\eqref{eq:eg:SGauss-capacity-derivative-alt},
\eqref{eq:eg:SGauss-divergence},
and Lemma \ref{lem:Lcenter}.
\begin{itemize}
\item If \(\lgm\in(0,\sfrac{\rno}{2\sigma^{2}})\), then there exists 
a unique \(\costc_{\rno,\!\Wm\!}^{\lgm}\) satisfying
\(\der{}{\costc}\CRC{\rno}{\Wm}{\costc}\vert_{\costc=\costc_{\rno,\Wm}^{\lgm}}\!=\!\lgm\)
by \eqref{eq:eg:SGauss-capacity-derivative-alt}.
Furthermore, \(\costc_{\rno,\!\Wm\!}^{\lgm}\) satisfies
\(\RCL{\rno}{\!\Wm\!}{\lgm}\!=\!\CRC{\rno}{\!\Wm\!}{\costc_{\rno,\Wm}^{\lgm}}\!-\!\lgm\costc_{\rno,\Wm}^{\lgm}\)
by \eqref{eq:Lcapacity-astheconjugate} because
\(\der{}{\costc}\CRC{\rno}{\Wm}{\costc}\) is decreasing in \(\costc\).
Then \eqref{eq:eg:SGauss-ALcapacity}
follows from \eqref{eq:eg:SGauss-capacity} and \eqref{eq:eg:SGauss-capacity-derivative}.
On the other hand \(\qma{\rno,\Wm}{\lgm}\!=\!\qmn{\rno,\Wm,\costc_{\rno,\Wm}^{\lgm}}\) by 
Lemma \ref{lem:Lcenter} because 
\(\CRC{\rno}{\!\Wm\!}{\costc_{\rno,\Wm}^{\lgm}}\!=\!\RCL{\rno}{\!\Wm\!}{\lgm}\!+\!\lgm \costc_{\rno,\Wm}^{\lgm}\).
Then \eqref{eq:eg:SGauss-ALcenter} follows from
\eqref{eq:eg:SGauss-center},
\eqref{eq:eg:SGauss-center-variance},
\eqref{eq:eg:SGauss-capacity-derivative},
and  \eqref{eq:eg:SGauss-ALcenter-variance}.
In addition one can confirm that
\(\costc_{\rno,\Wm}^{\lgm}=-\der{}{\lgm}\RCL{\rno}{\!\Wm\!}{\lgm}\)
by solving 
\(\der{}{\costc}\CRC{\rno}{\Wm}{\costc}\vert_{\costc=\costc_{\rno,\Wm}^{\lgm}}=\lgm\)
explicitly for \(\costc_{\rno,\Wm}^{\lgm}\).
We, however, do not need to obtain that explicit solution to confirm 
\eqref{eq:eg:SGauss-ALcapacity} and \eqref{eq:eg:SGauss-ALcenter}.
\item If \(\lgm\in[\sfrac{\rno}{2\sigma^{2}},\infty)\), then
\(\CRD{\rno}{\Wm}{\GausDen{\sigma^{2}}}{\mP}-\lgm\EXS{\mP}{\costc}\leq0\)
by \eqref{eq:eg:SGauss-divergence}.
On the other hand, \(\RCL{\rno}{\Wm}{\lgm}\geq0\) because A-L information is zero for
the probability measure that puts all its probability mass to \(\dinp=0\). 
Hence \(\RCL{\rno}{\!\Wm\!}{\lgm}=0\) and \(\qma{\rno,\!\Wm\!}{\lgm}=\GausDen{\sigma^{2}}\).
Thus, both \eqref{eq:eg:SGauss-ALcapacity} and \eqref{eq:eg:SGauss-ALcenter} hold.
\end{itemize}
\end{example}

\begin{example}[Parallel Gaussian Channels]\label{eg:PGauss}
Let \(\Wmn{[1,\blx]}\) be the product of scalar Gaussian channels \(\Wmn{\ind}\)
with noise variance  \(\sigma_{\ind}\) for \(\ind\in\{1,\ldots,\blx\}\)
and \(\costf_{[1,\blx]}\) be the additive cost function, i.e. 
	\begin{align}
	\notag
	\Wmn{[1,\blx]}(\oev|\dinp_{1}^{\blx})
	&=\int_{\oev}\left[
	\prod\nolimits_{\ind=1}^{\blx}\GausDen{\sigma_{\ind}^{2}}(\dout_{\ind}-\dinp_{\ind})
	\right]\dif{\dout_{1}^{\blx}}
	&
	&\forall \oev\in\rborel{\reals{}^{\blx}},
	\\
	\notag
	\costf_{[1,\blx]}(\dinp_{1}^{\blx})
	&=\sum\nolimits_{\ind=1}^{\blx} \dinp_{\ind}^{2}
	&
	&\forall \dinp_{1}^{\blx}\in\reals{}^{\blx}.
	\end{align}
	As a result of Lemma \ref{lem:CCcapacityproduct}, 
	the cost constrained Augustin capacity of \(\Wmn{[1,\blx]}\) 
	satisfies
	\begin{align}
	\notag
	\CRC{\rno}{\!\Wmn{[1,\blx]}\!}{\costc}
	&=\sup\nolimits_{\costc_{1},\ldots,\costc_{\blx}:\sum_{\ind} \costc_{\ind}\leq\costc} 
	\CRC{\rno}{\Wmn{\ind}}{\costc_{\ind}}.
	\end{align}
Since \(\CRC{\rno}{\Wmn{\ind}}{\costc_{\ind}}\!\)'s are 
continuous, strictly concave, and increasing in \(\costc_{\ind}\)  
the supremum is achieved at a  unique \((\costc_{\rno,1},\ldots,\costc_{\rno,\blx})\).
Then 
\(\qmn{\rno,\!\Wmn{[1,\blx]},\costc}=
\qmn{\rno,\!\Wmn{1},\costc_{\rno,1}}
\otimes
\cdots
\otimes
\qmn{\rno,\!\Wmn{\blx},\costc_{\rno,\blx}}
\)
by Lemma \ref{lem:CCcapacityproduct}.
Furthermore, since \(\CRC{\rno}{\Wmn{\ind}}{\costc_{\ind}}\!\)'s are 
continuously
differentiable in \(\costc_{\ind}\), 
the unique point \((\costc_{\rno,1},\ldots,\costc_{\rno,\blx})\) can be determined via the derivative test:
\(\der{}{\costc_{\ind}}\CRC{\rno}{\!\Wmn{\ind}}{\costc_{\ind}}\vert_{\costc_{\ind}=\costc_{\rno,\ind}}\!=\!\lgm_{\rno}\)
for all \(\ind\)'s with a positive \(\costc_{\rno,\ind}\) and
\(\der{}{\costc_{\ind}}\CRC{\rno}{\!\Wmn{\ind}}{\costc_{\ind}}\vert_{\costc_{\ind}=\costc_{\rno,\ind}}\!\leq\!\lgm_{\rno}\)
for all \(\ind\)'s with a zero \(\costc_{\rno,\ind}\) for some \(\lgm_{\rno}\in\reals{+}\).
Thus using  \eqref{eq:eg:SGauss-capacity-derivative-alt}, we can conclude that 
the optimal cost allocation, i.e.\((\costc_{\rno,1},\ldots,\costc_{\rno,\blx})\), 
satisfies 
\begin{align}
\label{eq:eg:PGauss-input-variance}
\costc_{\rno,\ind}
&=
\tfrac{\abp{\rno-2\sigma_{\ind}^{2}\lgm_{\rno}}}{2\lgm_{\rno}(\rno +2(\rno-1)\sigma_{\ind}^2\lgm_{\rno})}
\end{align}
for some \(\lgm_{\rno}\) that is uniquely determined by constraint \(\sum_{\ind=1}^{\blx} \costc_{\rno,\ind}=\costc\)
because the expression on the right hand side of \eqref{eq:eg:PGauss-input-variance} 
is nonincreasing in \(\lgm_{\rno}\) for each \(\ind\).
Consequently,
\begin{align}
\label{eq:eg:PGauss-capacity}
\CRC{\rno}{\!\Wmn{[1,\blx]}\!}{\costc}
&=\sum\nolimits_{\ind=1}^{\blx} \CRC{\rno}{\Wmn{\ind}}{\costc_{\rno,\ind}}
\\
\label{eq:eg:PGauss-capacitycenter}
\qmn{\rno,\!\Wmn{[1,\blx]}\!,\costc}
&=\bigotimes\nolimits_{\ind=1}^{\blx}\GausDen{\theta_{\rno,\sigma_{\ind},\costc_{\rno,\ind}}}
\end{align}
where \(\theta_{\rno,\sigma,\costc}\) is defined in \eqref{eq:eg:SGauss-center-variance}.
Using  the constraints for the optimality of a cost allocation we obtained via the derivative test, 
i.e. \(\der{}{\costc_{\ind}}\CRC{\rno}{\!\Wmn{\ind}}{\costc_{\ind}}\vert_{\costc_{\ind}=\costc_{\rno,\ind}}\!=\!\lgm_{\rno}\)
for all \(\ind\)'s with a positive \(\costc_{\rno,\ind}\) and
\(\der{}{\costc_{\ind}}\CRC{\rno}{\!\Wmn{\ind}}{\costc_{\ind}}\vert_{\costc_{\ind}=\costc_{\rno,\ind}}\!\leq\!\lgm_{\rno}\)
for all \(\ind\)'s with a zero \(\costc_{\rno,\ind}\),
together with \eqref{eq:eg:SGauss-capacity-derivative} 
---instead of \eqref{eq:eg:SGauss-capacity-derivative-alt}--- 
we obtain the following alternative characterization of \(\theta_{\rno,\sigma_{\ind},\costc_{\rno,\ind}}\)
in terms of \(\sigma_{\ind}\) and \(\lgm_{\rno}\) that does not depend on \(\costc_{\rno,\ind}\!\)'s
explicitly
\begin{align}
\label{eq:eg:PGauss-center-variance}
\theta_{\rno,\sigma_{\ind},\costc_{\rno,\ind}}
&=\sigma_{\ind}^{2}+\abp{\tfrac{1}{2\lgm_{\rno}}-\tfrac{\sigma_{\ind}^{2}}{\rno}}.
\end{align}
	
The A-L capacity and center of \(\Wmn{[1,\blx]}\) can be written in terms of the
corresponding quantities for the component channels using Lemma \ref{lem:Lcapacityproduct}
as follows:
	\begin{align}
	\notag
	\RCL{\rno}{\!\Wmn{[1,\blx]}\!}{\lgm}
	&=\sum\nolimits_{\ind=1}^{\blx} \RCL{\rno}{\Wmn{\ind}}{\lgm},
	\\
	\notag
	\qma{\rno,\Wmn{[1,\blx]}}{\lgm}
	&=\bigotimes\nolimits_{\ind=1}^{\blx}\qma{\rno,\Wmn{\ind}}{\lgm}.
	\end{align}
\end{example}
The cost constrained Augustin capacity and center
and A-L capacity and center of
vector Gaussian channels with multiple 
input and output antennas can be analyzed with a similar approach with the help of 
singular value decomposition.   
%!TEX root=../main-C.tex
\section{Discussion}\label{sec:conclusion}
Similar to the \renyi information, the Augustin information is a generalization 
of the mutual information defined in terms of the \renyi divergence.   
Unlike the order \(\rno\) \renyi information, however, 
the order \(\rno\) Augustin information 
does not have a closed form expression, except for the order one case.
This makes it harder to prove certain properties of the Augustin 
information such as its continuous differentiability as a function 
of the order \(\rno\),
the existence of a unique order \(\rno\) Augustin mean \(\qmn{\rno,\mP}\),
or the bounds given in \eqref{eq:augustinslaw}.
However, once these fundamental properties of the Augustin information 
are established, the analysis of the Augustin capacity is rather 
straightforward and very similar to the analogous analysis for the \renyi 
capacity, presented in \cite{nakiboglu19A}.

Previously, the convex conjugation techniques have been applied to 
the calculation of the cost constrained Augustin capacity 
through the quantity \(\GMIL{\rno}{\mP}{\Wm}{\lgm}\),
which we have called the R-G information.
Although such an approach can successfully characterize
the cost constrained Augustin capacity via the R-G capacity;
it is non-standard and somewhat convoluted. 
A more standard approach, based on the concept of 
A-L information \(\RMIL{\rno}{\mP}{\Wm}{\lgm}\), 
is presented in \S\ref{sec:cost-AL}.
The A-L information has not been used or studied before to the best of
our knowledge; nevertheless the resulting capacity is identical 
to the one associated with the R-G information. 
The optimality of the approach based 
on the R-G information seems more intuitive,
in the light of this observation.

Our analysis of the Augustin information and capacity was 
primarily motivated by their operational significance in 
the channel coding problem, \cite{augustin78}. 
We investigate that operational significance more closely and
derive sphere packing bounds with polynomial prefactors
for two families of memoryless channels ---composition constrained 
and cost constrained--- in \cite{nakiboglu18D}.
Broadly speaking, the derivation of the sphere packing bound 
for memoryless channels in \cite{nakiboglu18D}
is similar to the derivation of the sphere packing bound 
for product channels in \cite{nakiboglu19B},
except for the use of the Augustin capacity and center
instead of the \renyi capacity and center.
\appendix
\numberwithin{equation}{subsection}
\def\thesubsection{\Alph{subsection}}
\def\thesubsectiondis{\Alph{subsection}.} 
%!TEX root=../main-C.tex
\subsection{Proofs of Lemmas on the Analyticity of the \renyi Divergence}\label{sec:analyticityproofs}
\begin{proof}[Proof of Lemma \ref{lem:analyticity}]
	Let	\(\gX(\rno)\) and \(\fX(\rno,\dout)\) be
	\begin{align}
	\label{eq:analyticity-1}
	\gX(\rno)
	&\DEF \int\left(\der{\mW}{\rfm}\right)^{\rno} \left(\der{\mQ}{\rfm}\right)^{1-\rno}\rfm(\dif{\dout}),
	\\
	\label{eq:analyticity-2}
	\fX(\rno,\dout)
	&\DEF \left(\der{\mW}{\rfm}\right)^{\rno}\left(\der{\mQ}{\rfm}\right)^{1-\rno}
	\end{align}
	where \(\rfm\) is any reference measure satisfying \(\mW\AC\rfm\) and \(\mQ\AC\rfm\).
	Note that
	\begin{align}
	\label{eq:analyticity-3}
	\RD{\rno}{\mW}{\mQ}
	&=\tfrac{1}{\rno-1}\ln \gX(\rno)
	&
	&\rno\in \reals{+}\setminus\{1\}.
	\end{align}
	Furthermore \(\gX(\rno)\) does not depend on the choice of \(\rfm\), but \(\fX(\rno,\dout)\) does. 
	\begin{align}
	\label{eq:analyticity-4}
	\pder{^{\knd}}{\rno^{\knd}}\fX(\rno,\dout)
	&=\left(\ln\der{\mW}{\rfm}-\ln\der{\mQ}{\rfm}\right)^{\knd} \fX(\rno,\dout)
	&
	&\forall \knd\in\integers{\geq0}.
	\end{align}
	Then using the inequality \(\dsta \ln \dsta \geq -\sfrac{1}{e}\) we get
	\begin{align}
	\notag
	\abs{\pder{^{\knd}}{\rno^{\knd}}\fX(\rno,\dout)}
	&\leq
	\left(\tfrac{\knd}{\rno e}\right)^{\knd}\der{\mQ}{\rfm} 
	\IND{\frac{\dif{\mW}}{\dif{\rfm}}\leq \frac{\dif{\mQ}}{\dif{\rfm}}}
	+\left(\tfrac{\knd}{(\rnf-\rno)e}\right)^{\knd}\fX(\rnf,\dout)
	\IND{\frac{\dif{\mW}}{\dif{\rfm}}>\frac{\dif{\mQ}}{\dif{\rfm}}}
	&
	&\forall \knd\in\integers{\geq0},\rnf\in(\rno,\infty).
	\end{align}
	Invoking the Stirling's approximation for the factorial function, i.e.
	\(\sqrt{2\pi \knd} (\sfrac{\knd}{e})^{\knd}\leq  \knd!  \leq  e \sqrt{\knd} (\sfrac{\knd}{e})^{\knd}\),
	we get
	\begin{align}
	\label{eq:analyticity-5}
	\abs{\pder{^{\knd}}{\rno^{\knd}}\fX(\rno,\dout)}
	&\leq
	\tfrac{\knd!}{\sqrt{2 \pi \knd}} \left(
	\tfrac{1}{\rno^{\knd}}\der{\mQ}{\rfm}
	\IND{\frac{\dif{\mW}}{\dif{\rfm}}\leq \frac{\dif{\mQ}}{\dif{\rfm}}}
	+\tfrac{\fX(\rnf,\dout)}{(\rnf-\rno)^{\knd}}
	\IND{\frac{\dif{\mW}}{\dif{\rfm}}>\frac{\dif{\mQ}}{\dif{\rfm}}}  
	\right)
	&
	&\forall \knd\in\integers{\geq0},\rnf\in(\rno,\infty).
	\end{align}	
	On the other hand 
	\(\int \fX(\rnf,\dout) \rfm(\dif{\dout})=e^{(\rnf-1)\RD{\rnf}{\mW}{\mQ}}\)
	and
	for all \(\rno\) in \((0,\chi_{\mW,\mQ})\) there exists a 
	\(\rnf\) in \((\rno,\chi_{\mW,\mQ})\) with 
	finite \(\RD{\rnf}{\mW}{\mQ}\). 
	Then as a result of \cite[Corollary 2.8.7-(ii)]{bogachev},
	\(\gX(\rno)\) is an infinitely differentiable function of \(\rno\) on \((0,\chi_{\mW,\mQ})\) such that
	\begin{align}
	\label{eq:analyticity-6}
	\pder{^{\knd}}{\rno^{\knd}}\gX(\rno)
	&=\int \left[\pder{^{\knd}}{\rno^{\knd}}\fX(\rno,\dout)\right] \rfm(\dif{\dout})
	&
	&\forall \knd\in\integers{\geq0}.
	\end{align}
	Consequently, if \(\chi_{\mW,\mQ}>1\), then
	\begin{align}
	\label{eq:analyticity-7}
	\RD{1}{\mW}{\mQ}
	&=\left.\pder{}{\rno}\ln \gX(\rno) \right\vert_{\rno=1}.
	\end{align}
	Using \eqref{eq:analyticity-5} and \eqref{eq:analyticity-6} we get
	\begin{align}
	\label{eq:analyticity-8}
	\abs{\pder{^{\knd}}{\rno^{\knd}}\gX(\rno)}
	&\leq \tfrac{\knd!}{\sqrt{2 \pi \knd}} \left(\tfrac{1}{\rno^{\knd}}+\tfrac{\gX(\rnf)}{(\rnf-\rno)^{\knd}}\right)
	&
	&\forall \knd\in\integers{+},\rnf\in(\rno,\chi_{\mW,\mQ}).
	\end{align}
	Thus \(\gX(\rno)\) is not only infinitely differentiable but also analytic in \(\rno\) 
	on \((0,\chi_{\mW,\mQ})\) by \cite[Proposition 1.2.12]{krantzparks}.
	On the other hand \(\gX(\rno)\in\reals{+}\) for all \(\rno\in(0,\chi_{\mW,\mQ})\) 
	because \(\gX(\rno)=e^{(\rno-1)\RD{\rno}{\mW}{\mQ}}\) by \eqref{eq:analyticity-3}
	and \(\RD{\rno}{\mW}{\mQ}\in\reals{\geq0}\) by  Lemmas \ref{lem:divergence-pinsker}
	and \ref{lem:divergence-order} and the definition of \(\chi_{\mW,\mQ}\). 
	Thus \(\ln \gX(\rno)\) is analytic in \(\rno\) on \((0,\chi_{\mW,\mQ})\)
	because composition of analytic functions is analytic by \cite[Proposition 1.4.2]{krantzparks}. 
	Then \(\RD{\rno}{\mW}{\mQ}\) is analytic in \(\rno\) on	\((0,\chi_{\mW,\mQ})\setminus\{1\}\) 
	because the quotient of analytic functions is analytic at points with open neighborhoods on which 
	the function in the denominator is non-zero by \cite[Proposition 1.1.12]{krantzparks}.
	
	Now we proceed with establishing the analyticity of \(\RD{\rno}{\mW}{\mQ}\) at \(\rno=1\) for 
	\(\chi_{\mW,\mQ}>1\) case. 
	Since \(\ln \gX(\rno)\) is analytic in \(\rno\) on \((0,\chi_{\mW,\mQ})\) we can write 
	\(\ln \gX(\rno)\) as a convergent power series around any point in \((0,\chi_{\mW,\mQ})\) 
	for some neighborhood.
	Thus, there exists a \(\delta>0\) for which the following two identities hold for all
	\(\rnt\in(1-\delta,1+\delta)\)
	\begin{align}
	\notag
	\sum\nolimits_{\ind=0}^{\infty}\tfrac{\abs{\rnt-1}^{\ind}}{\ind !}
	\abs{\left.\pder{^{\ind}}{\rno^{\ind}}\ln \gX(\rno)\right\vert_{\rno=1}} 
	&<\infty,
	\\
	\notag
	\sum\nolimits_{\ind=0}^{\infty}\tfrac{(\rnt-1)^{\ind}}{\ind !}
	\left.\pder{^{\ind}}{\rno^{\ind}}\ln \gX(\rno)\right\vert_{\rno=1}
	&=\ln \gX(\rnt).
	\end{align}
	Then using \(\ln \gX(1)=0\) together with \eqref{eq:analyticity-3} and \eqref{eq:analyticity-7} we get
	\begin{align}
	\label{eq:analyticity-9}
	\RD{\rnt}{\mW}{\mQ}
	&=\RD{1}{\mW}{\mQ}+\sum\nolimits_{\ind=2}^{\infty}\tfrac{(\rnt-1)^{\ind-1}}{\ind !}
	\left.\pder{^{\ind}}{\rno^{\ind}}\ln \gX(\rno)\right\vert_{\rno=1}
	&
	&\forall\rnt\in(1-\delta,1+\delta).
	\end{align}
	Then \(\RD{\rnt}{\mW}{\mQ}\) is analytic on \((1-\delta,1+\delta)\) by \cite[Corollary 1.2.4]{krantzparks}
	because it is equal to a function defined by a convergent power series.
	
	The convergent power series given in \eqref{eq:analyticity-9}
	determines the derivatives of \(\RD{\rno}{\mW}{\mQ}\) at \(\rno=1\) by \cite[Corollary 1.1.16]{krantzparks}:
	\begin{align}
	\label{eq:analyticity-10}
	\left.\pder{^{\knd}}{\rno^{\knd}}\RD{\rno}{\mW}{\mQ}\right\vert_{\rno=1}
	&=\tfrac{1}{\knd+1}\left.\pder{^{\knd+1}}{\rno^{\knd+1}}\ln \gX(\rno)\right\vert_{\rno=1}
	&
	&\knd\in\integers{+}.
	\end{align}
	Using \eqref{eq:analyticity-3} together with the elementary rules of differentiation
	we can express the derivatives of \(\RD{\rno}{\mW}{\mQ}\) in terms of the derivatives
	of \(\ln \gX(\rno)\) for other orders in \((0,\chi_{\mW,\mQ})\), as well
	\begin{align}
	\notag
	\left.\pder{^{\knd}}{\rno^{\knd}}\RD{\rno}{\mW}{\mQ}\right\vert_{\rno=\rnf}
	&=\sum\nolimits_{\tin=0}^{\knd} 
	\tfrac{\knd!}{\tin! (\knd-\tin)!}
	\left(\left.\pder{^{\knd-\tin}}{\rno^{\knd-\tin}} \tfrac{1}{\rno-1}\right\vert_{\rno=\rnf}\right)
	\left(\left.\pder{^{\tin}}{\rno^{\tin}}\ln \gX(\rno)\right\vert_{\rno=\rnf}\right)
	\\
	\label{eq:analyticity-11}
	&= \sum\nolimits_{\tin=0}^{\knd} 
	\tfrac{\knd!}{\tin!} \tfrac{(-1)^{\knd-\tin}}{(\rnf-1)^{\knd-\tin+1}}
	\left.\pder{^{\tin}}{\rno^{\tin}}\ln \gX(\rno)\right\vert_{\rno=\rnf}
	&
	&\knd\in\integers{+},\rnf\in(0,\chi_{\mW,\mQ})\setminus\{1\}.
	\end{align}
	On the other hand by Fa\`{a} di Bruno formula for derivatives of the composition of smooth functions
	\cite[Thm. 1.3.2]{krantzparks} we have
	\begin{align}
	\notag
	\pder{^{\tin}}{\rno^{\tin}}\ln \gX(\rno)
	&=\sum\nolimits_{\set{J}_{\tin}} \tfrac{\tin!}{\jnd_{1}! \jnd_{2}!\ldots \jnd_{\tin}!}
	\left(\left.\pder{^{\jnd_{1}+\jnd_{2}+\cdots +\jnd_{\tin}}}{\tau^{\jnd_{1}+\jnd_{2}+\cdots +\jnd_{\tin}}} \ln \tau \right\vert_{\tau=\gX(\rno)}\right)
	\left(\tfrac{1}{1!}\pder{^{1}}{\rno^{1}}\gX(\rno) \right)^{\jnd_{1}}
	\left(\tfrac{1}{2!}\pder{^{2}}{\rno^{2}}\gX(\rno) \right)^{\jnd_{2}}
	\cdots
	\left(\tfrac{1}{\tin!} \pder{^{\tin}}{\rno^{\tin}}\gX(\rno)\right)^{\jnd_{\tin}}
	\\
	\notag
	&=\sum\nolimits_{\set{J}_{\tin}} \tfrac{\tin!}{\jnd_{1}! \jnd_{2}!\ldots \jnd_{\tin}!}
	\tfrac{(-1)(\jnd_{1}+\jnd_{2}+\cdots +\jnd_{\tin}-1)!}{(-\gX(\rno))^{\jnd_{1}+\jnd_{2}+\cdots +\jnd_{\tin}}}
	\left(\tfrac{1}{1!}\pder{^{1}}{\rno^{1}}\gX(\rno) \right)^{\jnd_{1}}
	\left(\tfrac{1}{2!}\pder{^{2}}{\rno^{2}}\gX(\rno) \right)^{\jnd_{2}}
	\cdots
	\left(\tfrac{1}{\tin!} \pder{^{\tin}}{\rno^{\tin}}\gX(\rno)\right)^{\jnd_{\tin}}
	&
	&\forall \tin\in\integers{+}.
	\end{align}
	Then using \eqref{eq:analyticity-1}, \eqref{eq:analyticity-2}, \eqref{eq:analyticity-4}, 
	and \eqref{eq:analyticity-6} we get
	\begin{align}
	\label{eq:analyticity-13}
	\pder{^{\tin}}{\rno^{\tin}}\ln \gX(\rno)
	&=\tin!\sum\nolimits_{\set{J_{\tin}}}\tfrac{(-1)(\jnd_{1}+\jnd_{2}+\cdots +\jnd_{\tin}-1)!}{\jnd_{1}! \jnd_{2}!\ldots \jnd_{\tin}!}
	\prod\nolimits_{\ind=1}^{\tin}
	\left(\tfrac{(-1)}{\ind!}\EXS{\wma{\rno}{\mQ}}{\left(\ln\der{\mW}{\rfm}-\ln\der{\mQ}{\rfm}\right)^{\ind}}\right)^{\jnd_{\ind}}
	&
	&\forall \tin\in\integers{+}.
	\end{align}
	The expression given in \eqref{eq:lem:analyticity-derivative} for \(\knd^{{th}}\) derivative of \(\RD{\rno}{\mW}{\mQ}\) with respect to \(\rno\) follows from 
	the identity \(\ln \gX(1)=0\) and equations
	\eqref{eq:analyticity-3},
	\eqref{eq:analyticity-10}, \eqref{eq:analyticity-11},
	and \eqref{eq:analyticity-13}.
	
	In order to prove the analyticity of \(\RD{1}{\wma{\rno}{\mQ}}{\mW}\) and \(\RD{1}{\wma{\rno}{\mQ}}{\mQ}\), 
	first note that as a result of \eqref{eq:divergence-derivative-first}, which follows from
	\eqref{eq:lem:analyticity-derivative}, we have
	\begin{align}
	\label{eq:analyticity-14}
	\RD{1}{\wma{\rnf}{\mQ}}{\mW}
	&=(\rnf-1)^2 \left.\pder{}{\rno}\RD{\rno}{\mW}{\mQ}\right\vert_{\rno=\rnf}
	&
	&\forall \rnf\in(0,\chi_{\mW,\mQ}).
	\end{align}
	Since \(\RD{\rno}{\mW}{\mQ}\) is analytic in \(\rno\) on \((0,\chi_{\mW,\mQ})\), 
	so is \(\pder{}{\rno}\RD{\rno}{\mW}{\mQ}\).
	Hence, \(\RD{1}{\wma{\rno}{\mQ}}{\mW}\) is analytic in \(\rno\) on \((0,\chi_{\mW,\mQ})\). 
	Since \(\RD{1}{\wma{\rnf}{\mQ}}{\mW}\) is analytic in \(\rnf\) on \((0,\chi_{\mW,\mQ})\), 
	it is finite on \((0,\chi_{\mW,\mQ})\). 
	Thus \eqref{eq:lem:variational-tilted} holds for all \(\rno\) in \((0,\chi_{\mW,\mQ})\) and 
	\eqref{eq:analyticity-14} implies
	\begin{align}
	\notag
	\RD{1}{\wma{\rnf}{\mQ}}{\mQ}
	&=\RD{\rnf}{\mW}{\mQ}-\rnf(1-\rnf)\left.\pder{}{\rno}\RD{\rno}{\mW}{\mQ}\right\vert_{\rno=\rnf}
	&
	&\forall \rnf\in(0,\chi_{\mW,\mQ}).
	\end{align}
	Thus \(\RD{1}{\wma{\rno}{\mQ}}{\mQ}\) is an analytic function of \(\rno\) on \((0,\chi_{\mW,\mQ})\), 
	as well.
\end{proof}

\begin{proof}[Proof of Lemma \ref{lem:uniform-analyticity}]
As results of 
\eqref{eq:analyticity-2},
\eqref{eq:analyticity-4},
\eqref{eq:analyticity-5},
and
Definition \ref{def:tiltedprobabilitymeasure}
we have
	\begin{align}
	\notag
	\EXS{\wma{\rno}{\mQ}}{\abs{\ln \der{\mW}{\rfm}-\ln \der{\mQ}{\rfm}}^{\ind}}
	&\leq \tfrac{e^{(1-\rno)\RD{\rno}{\mW}{\mQ}}\ind!}{\sqrt{2 \pi \ind}}\left(\tfrac{1}{\rno^{\ind}}+\tfrac{e^{(\rnb-1)\RD{\rnb}{\mW}{\mQ}}}{(\rnb-\rno)^{\ind}}\right)
	&
	&\forall \ind\in\integers{+}.
	\end{align}
	Then using \(\RD{\rnb}{\mW}{\mQ}\leq \gamma\) together with Lemma \ref{lem:divergence-order} we get
	\begin{align}
	\notag
	\EXS{\wma{\rno}{\mQ}}{\abs{\ln \der{\mW}{\rfm}-\ln \der{\mQ}{\rfm}}^{\ind}}
	&\leq  \tfrac{\ind!}{\sqrt{\ind}} \tfrac{e^{(1\vee\rnb)\gamma}}{(\rno \wedge(\rnb-\rno))^{\ind}}
	&
	&\forall \ind\in\integers{+}.
	\end{align}
	Then using \eqref{eq:lem:analyticity-derivative} and \eqref{eq:lem:def:G} we get,
	\begin{align}
	\notag
	\abs{\left.\pder{^{\knd}\RD{\rno}{\mW}{\mQ}}{\rno^{\knd}}\right\vert_{\rno=\rnf}} 
	&\leq\begin{cases}
	\knd!\tfrac{\gamma}{\abs{\rnf-1}^{\knd}}+
	\knd!\sum\limits_{\tin=1}^{\knd}
	\tfrac{1}{\abs{\rnf-1}^{\knd-\tin+1}}
	\tfrac{1}{(\rnf \wedge(\rnb-\rnf))^{\tin}}
	\sum\nolimits_{\set{J}_{\tin}}
	\tfrac{(\jnd_{1}+\jnd_{2}+\cdots +\jnd_{\tin}-1)!}{\jnd_{1}! \jnd_{2}!\ldots \jnd_{\tin}!}
	\left(e^{(1\vee\rnb)\gamma}\right)^{\jnd_{1}+\jnd_{2}+\cdots +\jnd_{\tin}}
	&\rnf\neq 1
	\\
	\knd!
	\tfrac{1}{(1 \wedge(\rnb-1))^{\knd+1}}
	\sum\nolimits_{\set{J}_{\knd+1}}
	\tfrac{(\jnd_{1}+\jnd_{2}+\cdots +\jnd_{\knd+1}-1)!}{\jnd_{1}! \jnd_{2}!\ldots \jnd_{\knd+1}!}
	\left(e^{(1\vee\rnb)\gamma}\right)^{\jnd_{1}+\jnd_{2}+\cdots +\jnd_{\knd+1}}
	&\rnf=1
	\end{cases}.
	\end{align}
	On the other hand \(\sum_{\set{J}_{\tin}} \tfrac{(\jnd_{1}+\jnd_{2}+\cdots +\jnd_{\tin})!}{\jnd_{1}! \jnd_{2}!\ldots \jnd_{\tin}!}
	\xi^{\jnd_{1}+\jnd_{2}+\cdots +\jnd_{\tin}}=\xi(1+\xi)^{\tin-1}\)
	by  \cite[Thm. 1.4.1]{krantzparks}. Thus we get the following inequality,
	which implies \eqref{eq:lem:uniform-analyticity-derivativebound}
	for the \(\tau\) defined in	\eqref{eq:lem:uniform-analyticity-tau}.
	\begin{align}
	\notag
	\abs{\left.\pder{^{\knd}\RD{\rno}{\mW}{\mQ}}{\rno^{\knd}}\right\vert_{\rno=\rnf}} 
	&\leq\begin{cases}
	\knd! \sum\limits_{\tin=1}^{\knd}
	\tfrac{1}{\abs{\rnf-1}^{\knd-\tin+1}}
	\left(\tfrac{1+e^{(1\vee\rnb)\gamma}}{\rnf \wedge(\rnb-\rnf)}
	+\gamma\IND{\tin=1}\right)^{\tin}
	&\rnf\neq 1
	\\
	\knd! \left(\tfrac{1+e^{\rnb\gamma}}{1 \wedge(\rnb-1)}\right)^{\knd+1}
	&\rnf=1
	\end{cases}.
	\end{align}
	As a result of \cite[Corollaries 1.2.4 and 1.2.5]{krantzparks} the following equality 
	holds on the open interval in which the power series on the right hand side is convergent,
	\begin{align}
	\label{eq:uniform-analyticity-1}
	\RD{\rnt}{\mW}{\mQ}
	&=\sum\nolimits_{\jnd=0}^{\infty} \tfrac{(\rnt-\rnf)^{\jnd}}{\jnd!} 
	\left.\pder{^{\jnd}}{\rno^{\jnd}}\RD{\rno}{\mW}{\mQ}\right\vert_{\rno=\rnf}.
	\end{align}	
	Note that as a result of \eqref{eq:lem:uniform-analyticity-derivativebound} we have
	\begin{align}
	\notag
	\limsup\nolimits_{\knd\to\infty} \sqrt[\knd]{\tfrac{1}{\knd!}\abs{\left.\pder{^{\knd}}{\rno^{\knd}}\RD{\rno}{\mW}{\mQ}\right\vert_{\rno=\rnf}}}
	&\leq \tau.
	\end{align}
	Thus radius of convergence of the power series on the right hand side of \eqref{eq:uniform-analyticity-1} 
	is at least \(\tfrac{1}{\tau}\) by \cite[Lemma 1.1.8]{krantzparks}, i.e. by Hadamard formula.
	Thus for all \(\rnt\in(\rnf-\tfrac{1}{\tau},\rnf+\tfrac{1}{\tau})\) 
	using \eqref{eq:lem:uniform-analyticity-derivativebound} and \eqref{eq:uniform-analyticity-1} we get	
	\begin{align}
	\notag
	\abs{\RD{\rnt}{\mW}{\mQ}-\sum\nolimits_{\ind=0}^{\knd-1}\tfrac{(\rnt-\rnf)^{\ind}}{\ind!}
		\left.\pder{^{\ind}\RD{\rno}{\mW}{\mQ}}{\rno^{\ind}}\right\vert_{\rno=\rnf}}
	&\leq \sum\nolimits_{\ind=\knd}^{\infty} \abs{\rnt-\rnf}^{\ind}\tau^{\ind+1}\left(\IND{\rnf= 1}+\ind\IND{\rnf\neq 1}\right).
	\end{align}
	Using identities \(\sum_{\ind=0}^{\infty}\dsta^{\ind}=\tfrac{1}{1-\dsta}\) and
	\(\sum_{\ind=0}^{\infty}(\ind +1)\dsta^{\ind}=\tfrac{1}{(1-\dsta)^2}\) for \(\abs{\dsta}<1\) we get,
	\begin{align}
	\notag
	\abs{\RD{\rnt}{\mW}{\mQ}-\sum\nolimits_{\ind=0}^{\knd-1}\tfrac{(\rnt-\rnf)^{\ind}}{\ind!}
		\left.\pder{^{\ind}\RD{\rno}{\mW}{\mQ}}{\rno^{\ind}}\right\vert_{\rno=\rnf}}
	&\leq \tfrac{\tau^{\knd+1}\abs{\rnt-\rnf}^{\knd}}{1-\abs{\rnt-\rnf}\tau}
	\left[\IND{\rnf= 1}+\left(\knd-1+\tfrac{1}{1-\abs{\rnt-\rnf}\tau}\right)\IND{\rnf\neq 1}\right].
	\end{align}
\end{proof}

\subsection{Proofs of Lemmas on the Augustin Information}\label{sec:informationproofs}
\begin{proof}[Proof of Lemma \ref{lem:information}]~
	\begin{enumerate}[(a)]
		\item [(\ref{lem:information}-\ref{information:bounded})]
		\(\RMI{\rno}{\mP}{\Wm}\leq \CRD{\rno}{\Wm}{\mQ}{\mP}\) for all \(\mQ\in\pmea{\outA}\) by definition.
		On the other hand, \(\RD{\rno}{\Wm(\dinp)}{\qmn{1,\mP}}\leq -\ln \mP(\dinp)\) for all \(\dinp\) with positive 
		\(\mP(\dinp)\) by Lemma \ref{lem:divergence-RM} because
		\(\mP(\dinp) \Wm(\dinp)\leq \qmn{1,\mP}\). 
		Hence, \(\RMI{\rno}{\mP}{\Wm}\leq -\sum_{\dinp} \mP(\dinp)\ln \mP(\dinp)\).
		
		\item[(\ref{lem:information}-\ref{information:one})]
		Note that as a result of Lemma \ref{lem:divergence-pinsker} and \eqref{eq:topsoe},
		\begin{align}
		\notag
		\CRD{1}{\Wm}{\mQ}{\mP}
		&\geq\CRD{1}{\Wm}{\qmn{1,\mP}}{\mP}+\tfrac{1}{2}\lon{\qmn{1,\mP}-\mQ}^2 
		&
		&\forall \mQ\in\pmea{\outA}.
		\end{align}
		Then \(\qmn{1,\mP}\) is the unique probability measure satisfying 
		\(\RMI{1}{\mP}{\Wm}=\CRD{1}{\Wm}{\qmn{1,\mP}}{\mP}\).
Then \eqref{eq:lem:information-one:EHB} follows from \eqref{eq:topsoe}.

		\item[(\ref{lem:information}-\ref{information:zto})]
		Let  \(\set{S}\)  and \(\varsigma\) be
		\begin{align}
		\notag
		\varsigma
		&\DEF\min\nolimits_{\dinp:\mP(\dinp)>0} \mP(\dinp),
		\\
		\notag
		\set{S}
		&\DEF \left\{\mS\in \fdis{\inpS}:\varsigma \IND{\mP(\dinp)>0}\leq \mS(\dinp)\leq \left(e^{\frac{1-\rno}{\varsigma}\CRD{\rno}{\Wm}{\qgn{\rno,\mP}}{\mP}}
		\right)\IND{\mP(\dinp)>0}~\forall \dinp\in \inpS\right\}.
		\end{align}
		The statements proved in 
		(\ref{information:zto}-\ref{information-zto:sufficientcondition}),
		(\ref{information:zto}-\ref{information:zto-equivalence}),
		(\ref{information:zto}-\ref{information:zto-fixedpoint}),
		and
		(\ref{information:zto}-\ref	{information:zto-iteration}) 
		collectively imply part (\ref{information:zto}).
		\begin{enumerate}[(\ref{information:zto}-i)]
			\item\label{information-zto:sufficientcondition}
			\emph{If \(\qmn{1,\mP}\AC\mU\) and \(\Aop{\rno}{\mP}{\mU}=\mU\), 
				then \(\CRD{\rno}{\Wm}{\mU}{\mP}=\RMI{\rno}{\mP}{\Wm}\),
				\eqref{eq:lem:information-zto:fixedpoint} 
				and \eqref{eq:lem:information-zto:EHB} 
				hold for \(\qmn{\rno,\mP}=\mU\), and \(\qmn{\rno,\mP}\) is unique:}
			Note that \(\Aop{\rno}{\mP}{\mU}=\mU\) and \(\qmn{1,\mP}\AC\mU\) imply 
			\begin{align}
			\notag
			\der{\mU}{\rfm}
			&=\left[\sum\nolimits_{\dinp}\mP(\dinp) \left(\der{\Wm(\dinp)}{\rfm}\right)^{\rno} e^{(1-\rno)\RD{\rno}{\Wm(\dinp)}{\mU}}\right]^{\frac{1}{\rno}}
			&
			&\forall \rfm:\qmn{1,\mP}\AC\rfm.
			\end{align} 
			Then one can confirm by substitution that
			\begin{align}
			\notag
			\RD{\rno}{\mU}{\mQ}
			&=\tfrac{1}{\rno-1}\ln 
			\sum\nolimits_{\dinp}\mP(\dinp) e^{(\rno-1)(\RD{\rno}{\Wm(\dinp)}{\mQ}-\RD{\rno}{\Wm(\dinp)}{\mU})}
			&
			&\forall \mQ\in\pmea{\outA}.
			\end{align}
			Then Jensen's inequality and convexity of the exponential function imply
			\begin{align}
			\label{eq:information-zto-1}
			\RD{\rno}{\mU}{\mQ}
			&\leq \CRD{\rno}{\Wm}{\mQ}{\mP}-\CRD{\rno}{\Wm}{\mU}{\mP}
			&
			&\forall \mQ\in\pmea{\outA}.
			\end{align}
			Then \(\mU\) is the unique probability measure satisfying 
			\(\RMI{\rno}{\mP}{\Wm}=\CRD{\rno}{\Wm}{\mU}{\mP}\)
			by Lemma \ref{lem:divergence-pinsker}.
			Consequently \eqref{eq:lem:information-zto:fixedpoint}
			and  the lower bound given in \eqref{eq:lem:information-zto:EHB} hold. 
			
In order to establish the upper bound given in \eqref{eq:lem:information-zto:EHB}
for \(\mQ\in \domtr{\rno,\mP}\), first note that
\(\Wm(\dinp)\AC\mU\) for all \(\dinp\) with a positive
\(\mP(\dinp)\) because  \(\qmn{1,\mP}\AC\mU\).
Thus for all \(\dinp\) with positive \(\mP(\dinp)\)
we have
\begin{align}
\notag
\RD{\rno}{\Wm(\dinp)}{\mQ}-\RD{\rno}{\Wm(\dinp)}{\mU}
&=\tfrac{1}{\rno-1}\left[\ln\!\int\! 
(\der{\Wm(\dinp)}{\mU})^{\rno}
(\der{\qmn{\sim}}{\mU})^{1-\rno}\mU(\dif{\dout})
-(\rno-1)\RD{\rno}{\Wm(\dinp)}{\mU}
\right]
\\
\notag
&=\tfrac{1}{\rno-1}\ln\!\int\! 
(\der{\qmn{\sim}}{\mU})^{1-\rno}
(\der{\Wm(\dinp)}{\mU})^{\rno}e^{(1-\rno)\RD{\rno}{\Wm(\dinp)}{\mU}}
\mU(\dif{\dout})
\\
\notag
&=\tfrac{1}{\rno-1}\ln\!\int\! 
(\der{\qmn{\sim}}{\mU})^{1-\rno}
\der{\Wma{\rno}{\mU}(\dinp)}{\mU}
\mU(\dif{\dout})
\\
\notag
&=\tfrac{1}{\rno-1}\ln\!\int\! 
(\der{\qmn{\sim}}{\mU})^{1-\rno} \Wma{\rno}{\mU}(\dinp)(\dif{\dout})
&
&\forall\mQ\in \domtr{\rno,\mP}
\end{align}
where \(\qmn{\sim}\) is the component of \(\mQ\) that is absolutely continuous in \(\mU\).
Consequently,
\begin{align}
\label{eq:information-zto-2}
\CRD{\rno}{\Wm}{\mQ}{\mP}-\CRD{\rno}{\Wm}{\mU}{\mP}
&=\tfrac{1}{\rno-1}\sum\nolimits_{\dinp}\!\mP(\dinp)\ln\!\int\! 
(\der{\qmn{\sim}}{\mU})^{1-\rno} \Wma{\rno}{\mU}(\dinp)(\dif{\dout})
&
&\forall\mQ\in \domtr{\rno,\mP}.
\end{align}			
			
On the other hand using the Jensen's inequality and concavity of the natural logarithm function we get
			\begin{align}
			\notag
			\tfrac{1}{\rno-1}\sum\nolimits_{\dinp} \mP(\dinp) \ln 
			\int (\der{\qmn{\sim}}{\mU})^{1-\rno} \Wma{\rno}{\mU}(\dinp)(\dif{\dout})
			&\leq \tfrac{1}{\rno-1}\sum\nolimits_{\dinp}\mP(\dinp) \int
			\left[\ln(\der{\qmn{\sim}}{\mU})^{1-\rno}\right]						\Wma{\rno}{\mU}(\dinp)(\dif{\dout})
			\\
			\label{eq:information-zto-3}
			&= \int  \left[\ln (\der{\mU}{\qmn{\sim}})\right] 
			\Aop{\rno}{\mP}{\mU}(\dif{\dout}).
			\end{align}
			Since \(\Aop{\rno}{\mP}{\mU}=\mU\) by the hypothesis, using
			\eqref{eq:information-zto-2} and \eqref{eq:information-zto-3} we get
			\begin{align}
			\notag
			\RD{1}{\mU}{\mQ}
			&\geq
			\CRD{\rno}{\Wm}{\mQ}{\mP}-\CRD{\rno}{\Wm}{\mU}{\mP}
			&
			&\forall \mQ\in \domtr{\rno,\mP}. 
			\end{align}
			In order to establish the upper bound given in  \eqref{eq:lem:information-zto:EHB} 
			for \(\mQ\notin \domtr{\rno,\mP}\)  we need to make the 
			following additional observation. If \(\mQ\notin \domtr{\rno,\mP}\),  then
			there exists an \(\dinp\) for which \(\mP(\dinp)>0\) and \(\Wm(\dinp)\perp\mQ\)
			because \(\RD{\rno}{\Wm(\dinp)}{\mQ}=\infty\)
			implies  \(\Wm(\dinp)\perp\mQ\) by \eqref{eq:shiryaev-new}.
As a result there exists an event \(\oev\in\outA\)
			such that such that \(\mU(\oev)>0\) and \(\mQ(\oev)=0\)
			because \(\Wm(\dinp)\AC\qmn{1,\mP}\) and \(\qmn{1,\mP}\AC\mU\).
			Consequently \(\RD{1}{\mU}{\mQ}=\infty\) and 
			the upper bound in equation \eqref{eq:lem:information-zto:EHB} holds for 
			\(\mQ\notin \domtr{\rno,\mP}\), as well.

			\item\emph{\(\CRD{\rno}{\Wm}{\mQ}{\mP}\!-\!\CRD{\rno}{\Wm}{\Aop{\rno}{\mP}{\mQ}}{\mP}\!\geq\!\RD{1}{\Aop{\rno}{\mP}{\mQ}}{\mQ}\) 
				for all \(\mQ\in\domtr{\rno,\mP}\):}
			Note that  \(\Aop{\rno}{\mP}{\mQ} \AC \mQ\) for all \(\mQ\in \domtr{\rno,\mP}\) by definition.
			Then
			\begin{align}
			\notag
			\CRD{\rno}{\Wm}{\mQ}{\mP}-
			\CRD{\rno}{\Wm}{\Aop{\rno}{\mP}{\mQ}}{\mP}
			&=\tfrac{1}{1-\rno}\sum\nolimits_{\dinp} \mP(\dinp) 
			\ln\int \left(\der{\Aop{\rno}{\mP}{\mQ}}{\mQ}\right)^{1-\rno} \Wma{\rno}{\mQ}(\dinp)(\dif{\dout}) 
			\\
			\notag
			&\geq \tfrac{1}{1-\rno} \sum\nolimits_{\dinp} \mP(\dinp) 
			\int \left[\ln\left(\der{\Aop{\rno}{\mP}{\mQ}}{\mQ}\right)^{1-\rno}\right] \Wma{\rno}{\mQ}(\dinp)(\dif{\dout}) 
			&
			&
			\\
			\label{eq:information-1}
			&=\RD{1}{\Aop{\rno}{\mP}{\mQ}}{\mQ}.
			\end{align}
			The inequality follows from the Jensen's inequality  and the concavity of the natural logarithm function.
			
			\item\label{information-zto:totallybounded}
			\emph{\(\{\Aopi{\rno}{\mP}{\ind}{\qgn{\rno,\mP}}\}_{\ind\in\integers{+}}\) is totally bounded for total 
				variation metric on \(\fmea{\outA}\):}
			For any \(\mQ\in \domtr{\rno,\mP}\), as a result of definitions of \(\Aop{\rno}{\mP}{\cdot}\) and 
			\(\mmn{\rno,\mP}\) we have
			\begin{align}
			\notag
			\der{\Aop{\rno}{\mP}{\mQ}}{\rfm}
			&=(\der{\mmn{\rno,\mS}}{\rfm})^{\rno} (\der{\mQ}{\rfm})^{1-\rno}
			\end{align}
			where \(\mS(\dinp)=\mP(\dinp)e^{(1-\rno)\RD{\rno}{\Wm(\dinp)}{\mQ}}\).
			Furthermore, if \(\CRD{\rno}{\Wm}{\mQ}{\mP}\leq \CRD{\rno}{\Wm}{\qgn{\rno,\mP}}{\mP}\), then \(\mS\in \set{S}\).
			
			In addition \(\qgn{\rno,\mP}\) is equal to \(\mmn{\rno,\mS}\) for an \(\mS\in \set{S}\). In particular
			\begin{align}
			\notag
			\qgn{\rno,\mP}
			&=\mmn{\rno,\smn{0}} 
			\end{align}
			where \(\smn{0}=\lon{\mmn{\rno,\mP}}^{-\rno} \mP\).
			One can confirm by substitution that 
			\(\lon{\mmn{\rno,\mP}}^{-\rno}=e^{(1-\rno)\RD{\rno}{\mP\mtimes \Wm}{\mP \otimes\qgn{\rno,\mP}}}\).
			Furthermore, \(\RD{\rno}{\mP\mtimes \Wm}{\mP \otimes\qgn{\rno,\mP}}\geq 0\) by  
			Lemma \ref{lem:divergence-pinsker}
			and \(\RD{\rno}{\mP\mtimes \Wm}{\mP \otimes\qgn{\rno,\mP}}\leq \CRD{\rno}{\Wm}{\qgn{\rno,\mP}}{\mP}\)
			by the Jensen's  inequality and the concavity of the natural logarithm function.
			Thus \(\smn{0}\in \set{S}\).

			On the other hand, \(\CRD{\rno}{\Wm}{\Aopi{\rno}{\mP}{\ind}{\qgn{\rno,\mP}}}{\mP}\leq \CRD{\rno}{\Wm}{\qgn{\rno,\mP}}{\mP}\)
			for all \(\ind\geq \integers{+}\).
			Thus we can write \(\Aopi{\rno}{\mP}{\ind}{\qgn{\rno,\mP}}\) in terms of the elements of  \(\mmn{\rno,\set{S}}\) as follows:
			\begin{align}
			\notag
			\der{\Aopi{\rno}{\mP}{\ind}{\qgn{\rno,\mP}}}{\rfm}
			&=(\der{\mmn{\rno,\smn{0}}}{\rfm})^{(1-\rno)^{\ind}}
			\prod\nolimits_{\jnd=1}^{\ind}
			(\der{\mmn{\rno,\smn{\jnd}}}{\rfm})^{\rno (1-\rno)^{\ind-\jnd}}
			\end{align}
			where \(\smn{\jnd}(\dinp)=\mP(\dinp)e^{(1-\rno)\RD{\rno}{\Wm(\dinp)}{\Aopi{\rno}{\mP}{\jnd-1}{\qgn{\rno,\mP}}}}\).

			In order to prove that \(\{\Aopi{\rno}{\mP}{\ind}{\qgn{\rno,\mP}}\}_{\ind\in\integers{+}}\) is totally bounded,
			we prove that a superset of it, i.e. \(\set{B}\) defined in the following, is totally bounded. 
			\begin{align}
			\label{eq:def:setBi}
			\set{B}_{\ind}
			&\DEF
			\left\{\mB\in\fmea{\outA}: 
			\der{\mB}{\rfm}=(\der{\mmn{\rno,\smn{0}}}{\rfm})^{(1-\rno)^{\ind}}
			\prod\nolimits_{\jnd=1}^{\ind}
			(\der{\mmn{\rno,\smn{\jnd}}}{\rfm})^{\rno (1-\rno)^{\ind-\jnd}}
			\mbox{~for some~} \smn{\jnd}\in\set{S}\right\}
			\\
			\label{eq:def:setB}
			\set{B}
			&\DEF \cup_{\ind\in \integers{+} } \set{B}_{\ind}.
			\end{align}
			Let us denote the number of \(\dinp\)'s with \(\mP(\dinp)>0\) by \(\knd\).
			Then \(\set{S}\) is isometric to a cube 
			in\footnote{We assume \(\reals{}^{\knd}\) has the metric \(d:\reals{}^{\knd}\times\reals{}^{\knd}\to\reals{\geq0}\) 
				given by \(d(\dsta,\tilde{\dsta})=\sum_{\tin=1}^{\knd}\abs{\dsta_{\tin}-\tilde{\dsta}_{\tin}}\) 
				for all \(\dsta,\tilde{\dsta}\in \reals{}^{\knd}\).}  \(\reals{}^{\knd}\)\!.
			We divide each side of the cube into \({\blx}\) equal length intervals. 
			Thus \(\set{S}\) is composed of \(\blx^{\knd}\) sub-cubes. 
			Furthermore,  \(\mmn{\rno,\mS}\leq \mmn{\rno,\widetilde{\mS}}\) 
			whenever
			\(\mS\leq \widetilde{\mS}\) by definition. 
			Thus, for any \(\mS\in \set{S}\) we have
			\begin{align}
			\notag
			\mmn{\rno,\lfloor\mS\rfloor_{\blx}}
			\leq \mmn{\rno,\mS}
			&\leq \left[1+\tfrac{e^{\frac{1-\rno}{\varsigma}\CRD{\rno}{\Wm}{\qgn{\rno,\mP}}{\mP}}-\varsigma}{\varsigma \blx}\right]^{\frac{1}{\rno}}
			\mmn{\rno,\lfloor\mS\rfloor_{\blx}}
			\end{align}  
			where \(\lfloor\mS\rfloor_{\blx}\) is the corner point 
			that satisfies \(\lfloor\mS\rfloor_{\blx}\leq \tilde{\mS}\) for all \(\tilde{\mS}\)
			in the sub-cube for the sub-cube that \(\mS\) is in. 
			
			In order to approximate members of \(\set{B}_{\ind}\) one can use the preceding discretization
			on each \(\mS_{\jnd}\) given in definition \(\set{B}_{\ind}\).
			Thus we have \(\blx^{(\ind+1) \knd}\) point set \(\set{K}_{\ind,\blx}\)
			such that:
			\begin{align}
			\notag
			\forall \mB\in \set{B}_{\ind}~\exists \mean\in \set{K}_{\ind,\blx}
			\mbox{~such that~}
			\mean \leq \mB \leq \left[1+\tfrac{e^{\frac{1-\rno}{\varsigma}\CRD{\rno}{\Wm}{\qgn{\rno,\mP}}{\mP}}-\varsigma}{\varsigma \blx}\right]^{\frac{1}{\rno}}
			\mean.
			\end{align} 
			One can use the points of \(\set{K}_{\ind,\blx}\) to approximate the points in 
			\(\cup_{\tin>\ind} \set{B}_{\tin}\), as well. 
			We apply the approximation with the sub-cubes described 
			above for the last \(\ind\) components of \(\mB\), 
			i.e. for \(\ind\) \(\smn{\jnd}\)'s with the largest indices.
			The remaining component of \(\mean\) is set to the minimum element 
			of\footnote{Such a minimum element might not exist for an arbitrary set of measures, but for the image of 
				\(\set{S}\) it exists: the minimum element is the image of the minimum point of \(\set{S}\).} 
			\(\mmn{\rno,\set{S}}\). Then
			\begin{align}
			\notag
			\forall\!\mB\!\in\!\cup_{\tin>\ind}\set{B}_{\tin}~\exists\mean\!\in\!\set{K}_{\ind,\blx}
			\mbox{~such that~}
			\mean\!\leq\!\mB\!\leq\!
			\left[1+\tfrac{e^{\frac{1-\rno}{\varsigma}\CRD{\rno}{\Wm}{\qgn{\rno,\mP}}{\mP}}-\varsigma}{\varsigma \blx}\right]^{\frac{1-(1-\rno)^{\ind}}{\rno}}
			\left[\tfrac{e^{\frac{1-\rno}{\varsigma}\CRD{\rno}{\Wm}{\qgn{\rno,\mP}}{\mP}}}{\varsigma}\right]^{\frac{(1-\rno)^{\ind}}{\rno}}\!
			\mean.
			\end{align} 
			Let \(\set{K}_{\blx}\) be \(\set{K}_{\blx}=\cup_{\jnd\in\{0,\ldots,\blx\} } \set{K}_{\jnd,\blx}\).
			Then
			\begin{align}
			\notag
			\forall\!\mB\!\in\!\set{B}~\exists\mean\!\in\!\set{K}_{\blx}
			\mbox{~such that~}
			&
			&
			\lon{\mB\!-\!\mean}
			&\!\leq\!
			\left(\!
			\left[1+\tfrac{e^{\frac{1-\rno}{\varsigma}\CRD{\rno}{\Wm}{\qgn{\rno,\mP}}{\mP}}-\varsigma}{\varsigma \blx}\right]^{\frac{1}{\rno}}\!
			\left[\tfrac{e^{\frac{1-\rno}{\varsigma}\CRD{\rno}{\Wm}{\qgn{\rno,\mP}}{\mP}}}{\varsigma}\right]^{\frac{(1-\rno)^{\blx}}{\rno}}\!
			-\!1\!\right)\!
			\sup\nolimits_{\mS\in \set{S}} \lon{\mmn{\rno,\mS}}.
			\end{align}
			Note that \(\sup\nolimits_{\mS\in \set{S}} \lon{\mmn{\rno,\mS}}\) is finite 
			and its coefficient 
			converges to zero as \(\blx\) diverges. 
			Furthermore, \(\set{K}_{\blx}\) is a finite set for any \(\blx\). Thus \(\set{B}\) is totally bounded. 
			As a result every subset of \(\set{B}\), and hence  
			\(\{\Aopi{\rno}{\mP}{\ind}{\qgn{\rno,\mP}}\}_{\ind\in\integers{+}}\),
			is totally bounded. 
			
			\item\label{information:zto-equivalence}
			\emph{\(\{\Aopi{\rno}{\mP}{\ind}{\qgn{\rno,\mP}}\}_{\ind\in\integers{+}}\) has a 
				subsequence \(\{\Aopi{\rno}{\mP}{\ind(\jnd)}{\qgn{\rno,\mP}}\}_{\jnd\in\integers{+}}\) 
				satisfying \(\lim_{\jnd\to \infty}\lon{\Aopi{\rno}{\mP}{\ind(\jnd)}{\qgn{\rno,\mP}}\!-\!\mU}\!=\!0\)
				for a \(\mU\!\sim\!\qmn{1,\mP}\):}
			The existence of a limit point \(\mU\) and convergent subsequence follow from 
			the compactness of the completion of \(\{\Aopi{\rno}{\mP}{\ind}{\qgn{\rno,\mP}}\}_{\ind\in\integers{+}}\).
			The completion is compact by \cite[Thm. 45.1]{munkres} because 
			\(\{\Aopi{\rno}{\mP}{\ind}{\qgn{\rno,\mP}}\}_{\ind\in\integers{+}}\) is totally bounded.
			
			Note that \(\Aopi{\rno}{\mP}{\ind}{\qgn{\rno,\mP}}\AC\qmn{1,\mP}\) because \(\qgn{\rno,\mP}\sim \qmn{1,\mP}\).
			Then \(\mU\AC\qmn{1,\mP}\) because any probability measure that is not absolute continuous in 
			\(\qmn{1,\mP}\) is outside the closure of \(\{\Aopi{\rno}{\mP}{\ind}{\qgn{\rno,\mP}}\}_{\ind\in\integers{+}}\).

			On the other hand, \(\mmn{\rno,\mP}\leq \qgn{\rno,\mP}\) by definition because \(\lon{\mmn{\rno,\mP}}\leq 1\).
			Furthermore, for any \(\mQ\in \domtr{\rno,\mP}\) we have
			\begin{align}
			\notag
			\der{\Aop{\rno}{\mP}{\mQ}}{\rfm}
			&\geq \sum\nolimits_{\dinp} \mP(\dinp) (\der{\Wm(\dinp) }{\rfm})^{\rno} (\der{\mQ}{\rfm})^{1-\rno}
			\\
			\notag
			&=(\der{\mmn{\rno,\mP}}{\rfm})^{\rno} (\der{\mQ}{\rfm})^{1-\rno}.
			\end{align} 
			Hence, if \(\mmn{\rno,\mP}\leq\mQ\), then \(\mmn{\rno,\mP}\leq\Aop{\rno}{\mP}{\mQ}\).
			Consequently,  \(\mmn{\rno,\mP}\leq\Aopi{\rno}{\mP}{\ind}{\qgn{\rno,\mP}}\) for all \(\ind\in\integers{+}\).
			Hence \(\mmn{\rno,\mP}\leq \mU\), because otherwise \(\mU\) can not 
			be in the closure of \(\{\Aopi{\rno}{\mP}{\ind}{\qgn{\rno,\mP}}\}_{\ind\in\integers{+}}\).
			Then \(\qmn{1,\mP}\AC \mU\) because \(\qmn{1,\mP}\sim \mmn{\rno,\mP}\).

			\item\label{information:zto-continuity}
			\emph{\(\Aop{\rno}{\mP}{\cdot}:\domtr{\rno,\mP}\to \pmea{\outA}\) is continuous 
				if both \(\domtr{\rno,\mP}\) and \(\pmea{\outA}\) have the total variation topology:}
			First, note that \((\dsta+\tin)^{1-\rno}-\dsta^{1-\rno}\) is a monotonically decreasing function of 
			\(\dsta\) on \(\reals{\geq0}\) for fixed \(\tin\in\reals{\geq0}\) and \(\rno\in(0,1)\).
			Then for any \(\dinp\) with positive \(\mP(\dinp)\) and \(\qmn{1},\qmn{2}\in \domtr{\rno,\mP}\)
			as a result of Holder's inequality we have
			\begin{align}
			\notag
			\int \abs{\left(\der{\Wm(\dinp)}{\rfm}\right)^{\rno}\left(\der{\qmn{1}}{\rfm}\right)^{1-\rno}
				-\left(\der{\Wm(\dinp)}{\rfm}\right)^{\rno}\left(\der{\qmn{2}}{\rfm}\right)^{1-\rno}} \rfm(\dif{\dout})
			&\leq \int \left(\der{\Wm(\dinp)}{\rfm}\right)^{\rno}
			\abs{\der{\qmn{1}}{\rfm}-\der{\qmn{2}}{\rfm}}^{1-\rno} \rfm(\dif{\dout})
			\\
			\notag
			&\leq \lon{\qmn{1}-\qmn{2}}^{1-\rno}. 
			\end{align}
			Hence \(e^{(\rno-1)\RD{\rno}{\Wm(\dinp)}{\mQ}}\Wma{\rno}{\mQ}(\dinp)\) is a continuous 
			function of \(\mQ\) from \(\domtr{\rno,\mP}\) to \(\fmea{\outA}\) 
			for the total variation topology.
			Then \(\Wma{\rno}{\mQ}(\dinp)\) is a continuous function of \(\mQ\) for 
			the total variation topology, as well,
			because \(\RD{\rno}{\Wm(\dinp)}{\mQ}\) is continuous in \(\mQ\) 
			for the total variation topology by Lemma \ref{lem:divergence:uc}.
			Thus \(\Aop{\rno}{\mP}{\cdot}:\domtr{\rno,\mP}\to \pmea{\outA}\) is continuous.
			
			\item\label{information:zto-fixedpoint}
			\emph{The limit point of the convergent subsequence \(\{\Aopi{\rno}{\mP}{\ind(\jnd)}{\qgn{\rno,\mP}}\}_{\jnd\in\integers{+}}\) 
				is a fixed point of \(\Aop{\rno}{\mP}{\cdot}\), i.e. \(\Aop{\rno}{\mP}{\mU}\!=\!\mU\):} 
			Using the non-negativity of the \renyi divergence for probability measures 
			and \eqref{eq:information-1} we get
			\begin{align}
			\notag
			\CRD{\rno}{\Wm}{\qgn{\rno,\mP}}{\mP}
			&\geq 
			\sum\nolimits_{\ind\in \integers{\geq 0}}
			\CRD{\rno}{\Wm}{\Aopi{\rno}{\mP}{\ind}{\qgn{\rno,\mP}}}{\mP}-
			\CRD{\rno}{\Wm}{\Aop{\rno}{\mP}{\Aopi{\rno}{\mP}{\ind}{\qgn{\rno,\mP}}}}{\mP}
			\\
			\notag
			&\geq 
			\sum\nolimits_{\ind\in \integers{\geq 0}}
			\RD{1}{\Aop{\rno}{\mP}{\Aopi{\rno}{\mP}{\ind}{\qgn{\rno,\mP}}}}{\Aopi{\rno}{\mP}{\ind}{\qgn{\rno,\mP}}}.
			\end{align}
			Then 
			\(\lim\nolimits_{\ind \to \infty}\RD{1}{\Aop{\rno}{\mP}{\Aopi{\rno}{\mP}{\ind}{\qgn{\rno,\mP}}}}{\Aopi{\rno}{\mP}{\ind}{\qgn{\rno,\mP}}}=0\).
			Hence
			\(\lim\nolimits_{\jnd \to \infty}\RD{1}{\Aop{\rno}{\mP}{\Aopi{\rno}{\mP}{\ind(\jnd)}{\qgn{\rno,\mP}}}}{\Aopi{\rno}{\mP}{\ind(\jnd)}{\qgn{\rno,\mP}}}=0\).
			
			On the other hand, \(\RD{1}{\Aop{\rno}{\mP}{\mQ}}{\mQ}\) is lower semicontinuous in \(\mQ\) for the total 
			variation topology because the
			\renyi divergence is lower semicontinuous in its arguments 
			for the topology of setwise convergence ---and hence to the total variation 
			topology---
			by Lemma \ref{lem:divergence:lsc}
			and
			\(\Aop{\rno}{\mP}{\cdot}\) is continuous in the total variation topology. 
			Then \(\RD{1}{\Aop{\rno}{\mP}{\mU}}{\mU}=0\)
			because \(\Aopi{\rno}{\mP}{\ind(\jnd)}{\qgn{\rno,\mP}}\) converges to \(\mU\)
			in total variation topology as \(\jnd\) diverges.
			Thus \(\Aop{\rno}{\mP}{\mU}=\mU\) as a result of Lemma \ref{lem:divergence-pinsker}.
			
			\item\label{information:zto-iteration}
			\emph{\(\qmn{\rno,\mP}\) satisfies \eqref{eq:lem:information-zto:iteration}:}
			Recall that \(\RD{\rno}{\mW}{\mQ}\) is continuous in  \(\mQ\) for the total variation topology 
			by Lemma \ref{lem:divergence:uc}. Furthermore,
			\(\lim_{\jnd\to \infty}\lon{\Aopi{\rno}{\mP}{\ind(\jnd)}{\qgn{\rno,\mP}}\!-\!\qmn{\rno,\mP}}\!=\!0\),
			and  \(\CRD{\rno}{\Wm}{\qmn{\rno,\mP}}{\mP}\!=\!\RMI{\rno}{\mP}{\Wm}\).
			Then
			\begin{align}
			\notag
			\lim\nolimits_{\jnd\to \infty}\CRD{\rno}{\Wm}{\Aopi{\rno}{\mP}{\ind(\jnd)}{\qgn{\rno,\mP}}}{\mP}
			&=\RMI{\rno}{\mP}{\Wm}.
			\end{align}
			On the other hand  \(\CRD{\rno}{\Wm}{\Aopi{\rno}{\mP}{\ind}{\qgn{\rno,\mP}}}{\mP}\geq \CRD{\rno}{\Wm}{\Aopi{\rno}{\mP}{\ind+\tin}{\qgn{\rno,\mP}}}{\mP}\geq \RMI{\rno}{\mP}{\Wm}\) 
			for all \(\tin\in \integers{+}\) by \eqref{eq:information-1} and the definition of the
			Augustin information. Thus
			\begin{align}
			\notag
			\lim\nolimits_{\ind\to \infty}\CRD{\rno}{\Wm}{\Aopi{\rno}{\mP}{\ind}{\qgn{\rno,\mP}}}{\mP}
			&=\RMI{\rno}{\mP}{\Wm}.
			\end{align}
			Then as a result of \eqref{eq:lem:information-zto:EHB}, which is implied by the assertions we have
			already established, we have
			\begin{align}
			\notag
			\lim\nolimits_{\ind\to \infty}\RD{\rno}{\qmn{\rno,\mP}}{\Aopi{\rno}{\mP}{\ind}{\qgn{\rno,\mP}}}
			&=0.
			\end{align}
			Then \(\lim_{\ind\to \infty}\lon{\qmn{\rno,\mP}-\Aopi{\rno}{\mP}{\ind}{\qgn{\rno,\mP}}}=0\)
			as a result of Lemma \ref{lem:divergence-pinsker}.
		\end{enumerate}

\begin{remark}\label{remark:information-zto:iteration-general}
For any \(\mQ\) satisfying \(\mQ\sim\mmn{\rno,\mP}\) 
with a finite
	\(\essup\nolimits_{\mmn{\rno,\mP}}\abs{\ln \der{\mQ}{\mmn{\rno,\mP}}}\), 
	we can define the sets \(\set{S}\) and  \(\set{B}_{\ind}\)
	as follows
	\begin{align}
	\notag
	\set{S}
	&\DEF \left\{\mS\in \fdis{\inpS}:\varsigma \IND{\mP(\dinp)>0}\leq \mS(\dinp)\leq \left(e^{\frac{1-\rno}{\varsigma}\CRD{\rno}{\Wm}{\mQ}{\mP}}
	\right)\IND{\mP(\dinp)>0}~\forall \dinp\in \inpS\right\},
	\\
	%\label{eq:def:setBi-new}
	\notag
	\set{B}_{\ind}
	&\DEF
	\left\{\mB\in\fmea{\outA}: 
	\der{\mB}{\rfm}=(e^{\gamma}\der{\mQ}{\rfm})^{(1-\rno)^{\ind}}
	\prod\nolimits_{\jnd=1}^{\ind}
	(\der{\mmn{\rno,\smn{\jnd}}}{\rfm})^{\rno (1-\rno)^{\ind-\jnd}}
	\mbox{~for some~} \gamma\in \{-\Gamma,0,\Gamma\} 
	\mbox{~and~}\smn{\jnd}\in\set{S} \right\},
	\end{align}
	where \(\Gamma=\tfrac{(1-\rno)\CRD{\rno}{\Wm}{\mQ}{\mP}}{\varsigma}
	-\tfrac{\ln \varsigma}{\rno}+
	\essup\nolimits_{\mmn{\rno,\mP}}\abs{\ln \der{\mQ}{\mmn{\rno,\mP}}}\).
	Then one can confirm that 
	\(e^{-\Gamma}\mQ \leq \mmn{\rno,\mS}
	\leq  e^{\Gamma} \mQ\)
	for all \(\mS\in\set{S}\).
	Using this property, we can repeat the rest of the analysis with appropriate modifications to establish the following:
	\begin{align}
	\label{eq:lem:information-zto:iteration-general}
	\lim\nolimits_{\jnd\to\infty}  
	\lon{\qmn{\rno,\mP}\!-\Aopi{\rno}{\mP}{\jnd}{\mQ}}
	&=0
	&
	&\mbox{if~}\mQ\sim\mmn{\rno,\mP}
	\mbox{~and~}
	\essup\nolimits_{\mmn{\rno,\mP}}\abs{\ln \der{\mQ}{\mmn{\rno,\mP}}}<\infty,
	\end{align}
	On the other hand, 
	\(\qmn{1,\mP}\sim\mmn{\rno,\mP}\)
	by \cite[{Lemma \ref*{A-lem:powermeanequivalence}-(\ref*{A-powermeanequivalence-a})}]{nakiboglu19A}
	and
	\(\abs{\ln \der{\qmn{1,\mP}}{\mmn{\rno,\mP}}}\leq \tfrac{\abs{\rno-1}}{\rno}\ln \tfrac{1}{\varsigma}\)
	holds \(\qmn{1,\mP}\)-a.s.
	by \cite[{Lemma \ref*{A-lem:powermeandensityO}-(\ref*{A-powermeandensityO-a})}]{nakiboglu19A}.
	Thus \(\qmn{1,\mP}\) satisfies the condition given in \eqref{eq:lem:information-zto:iteration-general} and 
	the convergence described in \eqref{eq:lem:information-zto:iteration-general}
	is equivalent to the one in  \eqref{eq:lem:information-zto:iteration-general-new}.
\end{remark}

		\item[(\ref{lem:information}-\ref{information:oti})]
		Let the function \(\fX(\cdot)\) and the set of channels \(\set{U}\)
		be
		\begin{align}
		\notag
		\fX(\Vm)
		&\DEF\tfrac{\rno}{1-\rno}\CRD{1}{\Vm}{\Wm}{\mP}+\RMI{1}{\mP}{\Vm}
		&
		&\forall \Vm\in\set{U},
		\\
		\notag
		\set{U}
&\DEF\{\Vm\in \pmea{\outA|\supp{p}}:\CRD{1}{\Vm}{\Wm}{\mP}<\infty\}.
%\\
%\notag
%&\DEF\{\Vm\in \bigtimes\nolimits_{\dinp:\mP(\dinp)>0}\pmea{\outA}:\CRD{1}{\Vm}{\Wm}{\mP}<\infty\}.
	\end{align}
		The statements proved in 
		(\ref{information:oti}-\ref{information-oti:sufficientcondition}),
		(\ref{information:oti}-\ref{information:oti-existence}),
		(\ref{information:oti}-\ref{information:oti-saddlepoint}), 
		and
		(\ref{information:oti}-\ref{information:oti-fixedpoint})
		collectively imply part (\ref{information:oti}).
		
		\begin{enumerate}[(\ref{information:oti}-i)]
			\item\label{information-oti:sufficientcondition}
			\emph{If \(\Aop{\rno}{\mP}{\mU}=\mU\), then \(\CRD{\rno}{\Wm}{\mU}{\mP}=\RMI{\rno}{\mP}{\Wm}\),
				\eqref{eq:lem:information-oti:fixedpoint} and \eqref{eq:lem:information-oti:EHB} 
				hold for \(\qmn{\rno,\mP}=\mU\),
				\(\qmn{\rno,\mP}\) is unique and \(\qmn{\rno,\mP}\sim\qmn{1,\mP}\):}
			\begin{align}
			\label{eq:information-oti-1}
			\CRD{\rno}{\Wm}{\mQ}{\mP}-\CRD{\rno}{\Wm}{\mU}{\mP}
			&=\tfrac{1}{\rno-1}\sum\nolimits_{\dinp} \mP(\dinp) \ln 
			\int (\der{\qmn{\sim}}{\mU})^{1-\rno} \Wma{\rno}{\mU}(\dinp)(\dif{\dout})
			&
			&\forall\mQ\in \domtr{\rno,\mP}
			\end{align}
			where \(\qmn{\sim}\) is the component of \(\mQ\) that is absolutely continuous in \(\mU\).
			
			On the other hand using the Jensen's inequality and concavity of the natural logarithm function we get
			\begin{align}
			\notag
			\tfrac{1}{\rno-1}\sum\nolimits_{\dinp} \mP(\dinp) \ln 
			\int (\der{\qmn{\sim}}{\mU})^{1-\rno} \Wma{\rno}{\mU}(\dinp)(\dif{\dout})
			&\geq \tfrac{1}{\rno-1} \sum\nolimits_{\dinp}\mP(\dinp) 
			\int \left[\ln(\der{\qmn{\sim}}{\mU})^{1-\rno}\right]
			\Wma{\rno}{\mU}(\dinp)(\dif{\dout})
			\\
			\label{eq:information-oti-2}
			&= \int \left[\ln (\der{\mU}{\qmn{\sim}})\right]  
			\Aop{\rno}{\mP}{\mU}(\dif{\dout}).
			\end{align}
			Since \(\Aop{\rno}{\mP}{\mU}=\mU\) by the hypothesis, using   
			\eqref{eq:information-oti-1} and \eqref{eq:information-oti-2} we get
			\begin{align}
			\label{eq:information-oti-3}
			\CRD{\rno}{\Wm}{\mQ}{\mP}-\CRD{\rno}{\Wm}{\mU}{\mP}
			&\geq \RD{1}{\mU}{\mQ}
			&
			&\forall \mQ\in \domtr{\rno,\mP}. 
			\end{align}
			\(\RD{1}{\mU}{\mQ}>0\) for all  \(\mQ\in\pmea{\outA}\setminus\{\mU\}\) 
			by Lemma \ref{lem:divergence-pinsker}
			and \(\CRD{\rno}{\Wm}{\mQ}{\mP}=\infty\) for \(\mQ \notin \domtr{\rno,\mP}\) by definition.
			Then \(\mU\) is the unique probability measure satisfying 
			\(\RMI{\rno}{\mP}{\Wm}=\CRD{\rno}{\Wm}{\mU}{\mP}\)
			and  \eqref{eq:lem:information-oti:fixedpoint} holds.
			In addition \(\qmn{1,\mP}\AC\mU\) because otherwise
			\(\CRD{\rno}{\Wm}{\mU}{\mP}\) would have been infinite.
			Furthermore, \(\mU\AC\qmn{1,\mP}\) because 
			\(\RD{1}{\mU}{\qmn{1,\mP}}\) is finite
			by \eqref{eq:information-oti-3} and part (\ref{information:bounded}).
			
			The lower bound given in \eqref{eq:lem:information-oti:EHB} holds for 
			\(\mQ\in\domtr{\rno,\mP}\) by \eqref{eq:information-oti-3}
			and for  \(\mQ\notin \domtr{\rno,\mP}\) by definition.
			In order to establish the upper bound given in  \eqref{eq:lem:information-zto:EHB},
			note that \(\Aop{\rno}{\mP}{\mU}=\mU\) implies
			\begin{align}
			\notag
			\der{\mU}{\rfm}
			&=\left[\sum\nolimits_{\dinp}\mP(\dinp) \left(\der{\Wm(\dinp)}{\rfm}\right)^{\rno} e^{(1-\rno)\RD{\rno}{\Wm(\dinp)}{\mU}}\right]^{\frac{1}{\rno}}
			&
			&\forall \rfm:\qmn{1,\mP}\AC\rfm.
			\end{align} 
			Then one can confirm by substitution that
			\begin{align}
			\notag
			\RD{\rno}{\mU}{\mQ}
			&=\tfrac{1}{\rno-1}\ln 
			\sum\nolimits_{\dinp}\mP(\dinp) e^{(\rno-1)(\RD{\rno}{\Wm(\dinp)}{\mQ}-\RD{\rno}{\Wm(\dinp)}{\mU})}
			&
			&\forall \mQ\in\pmea{\outA}.
			\end{align}
			Then Jensen's inequality and convexity of the exponential function imply
			\begin{align}
			\notag
			\RD{\rno}{\mU}{\mQ}
			&\geq \CRD{\rno}{\Wm}{\mQ}{\mP}-\CRD{\rno}{\Wm}{\mU}{\mP}
			&
			&\forall \mQ\in\pmea{\outA}.
			\end{align}

			\item \emph{\(\fX(\cdot):\set{U}\to\reals{}\) is concave and upper semicontinuous
				on \(\set{U}\) for the topology of setwise 
				convergence:\footnote{The set \(\set{U}\) is a
					subset of the Cartesian product of a finite number of copies of \(\pmea{\outA}\).
					What we mean by the topology of setwise convergence on \(\set{U}\) is the product topology obtained 
					by  assuming topology of setwise convergence on each component of the Cartesian product.
					We employ this terminology in the rest of the proof without explicitly mentioning it.}}
			Using the definition of the tilted channel given in \eqref{eq:def:tiltedchannel},
			the identity given in \eqref{eq:lem:information-one:EHB},
			and the joint convexity of the order one \renyi divergence in its arguments, i.e. 
			Lemma \ref{lem:divergence-jointconvexity}, 
			we can write \(\fX(\Vm)\) as the sum of three finite terms as follows for all \(\Vm\in\set{U}\):
			\begin{align}
			\label{eq:information-oti-U}
			\fX(\Vm)
			&=
			\tfrac{1}{1-\rno}\CRD{1}{\Vm}{\Wma{\rno}{\qmn{1,\mP}}}{\mP}
			-\RD{1}{\sum\nolimits_{\dinp}\mP(\dinp)\Vm(\dinp)}{\qmn{1,\mP}}
			+\CRD{\rno}{\Wm}{\qmn{1,\mP}}{\mP}.
			\end{align}	
			Then \(\fX(\cdot)\) is a concave because the order one \renyi divergence is convex in its first argument 
			by Lemma \ref{lem:divergence-jointconvexity}.
			Similarly, \(\fX(\cdot)\) is upper semicontinuity for the topology of setwise convergence, 
			because \renyi divergence is lower semicontinuous in its first argument for the topology of setwise convergence
			by Lemma \ref{lem:divergence:lsc}.
			
			\item \emph{\(\set{U}'\DEF\{\Vm\in\set{U}:\max_{\dinp:\mP(\dinp)>0}\!\mP(\dinp)\RD{1}{\!\Vm(\dinp)}{\!\Wm(\dinp)}\leq \tfrac{\rno-1}{\rno}\bent{\mP}\}\)
				is compact for the topology of setwise convergence:}
			For any \(\mV\in\pmea{\outA}\) and \(\mW\in\pmea{\outA}\),
			the identity \(\dsta\ln \dsta\geq \sfrac{-1}{e}\) implies that
			\begin{align}
			\notag
			\int \abp{\der{\mV}{\mW} \ln \der{\mV}{\mW}} \mW(\dif{\dout})
			&\leq \RD{1}{\mV}{\mW}+\sfrac{1}{e}.
			&
			&
			\end{align}
			Then for any \(\gamma\in\reals{+}\) and \(\mW\in\pmea{\outA}\), the set of Radon-Nikodym derivatives 
			\(\{\der{\mV}{\mW}\}_{\mV:\RD{1}{\mV}{\mW}\leq \gamma}\) is uniformly integrable because 
			it satisfies the necessary 
			and sufficient condition for the uniform integrability given by de la Vallee Poussin 
			\cite[Thm. 4.5.9]{bogachev}.
			Hence, \(\{\mV\in\pmea{\outA}:\RD{1}{\mV}{\mW}\leq \gamma\}\UAC\mW\).
			Then\footnote{Note that \(\{\mV\in\pmea{\outA}:\RD{1}{\mV}{\mW}\leq \gamma\}\) is bounded in variation norm 
				by definition.} \(\{\mV\in\pmea{\outA}:\RD{1}{\mV}{\mW}\leq \gamma\}\) has compact closure in the topology 
			of setwise convergence
			by \cite[Thm. 4.7.25]{bogachev}.
			On the other hand the set \(\{\mV\in\pmea{\outA}:\RD{1}{\mV}{\mW}\leq \gamma\}\) is closed, i.e. it is equal to its 
			closure, because \renyi divergence is lower semicontinuous in its arguments for the topology of setwise convergence
			by Lemma \ref{lem:divergence:lsc}.
			Hence \(\{\mV\in\pmea{\outA}:\RD{1}{\mV}{\mW}\leq \gamma\}\) is compact in the topology of setwise convergence
			for any \(\gamma\in\reals{+}\) and \(\mW\in\pmea{\outA}\). 
			Then \(\set{U}'\) is compact in the topology of setwise convergence because product of finite number of compact sets 
			is compact by \cite[Thm. 26.7]{munkres}.
			
			\item\label{information:oti-existence}
			\emph{\(\exists \Umn{*}\in\set{U}'\) s.t. \(\fX(\Umn{*})=\sup_{\Vm\in\set{U}} \fX(\Vm)\):}
			Note that \(\Wm\in\set{U}\) and \(\fX(\Wm)=\RMI{1}{\mP}{\Wm}\).
			Furthermore, \(\RMI{1}{\mP}{\Wm}\geq0\) by Lemma \ref{lem:divergence-pinsker} 
			and part (\ref{information:one}).
			On the other hand, if \(\mP(\dinp)\RD{1}{\Vm(\dinp)}{\Wm(\dinp)}>\tfrac{\rno-1}{\rno}\bent{\mP}\) 
			for an \(\dinp\),  
			then \(\fX(\Vm)<0\) because 
			\(\RD{1}{\Vm(\dinp)}{\Wm(\dinp)}\geq 0\) by Lemma \ref{lem:divergence-pinsker}
			and \(\RMI{1}{\mP}{\Vm}\leq \bent{\mP}\) by part (\ref{information:bounded}).
			Thus,
			\begin{align}
			\notag
			\sup\nolimits_{\Vm\in\set{U}} \fX(\Vm)
			&=\sup\nolimits_{\Vm\in\set{U}'} \fX(\Vm).
			\end{align}
			On the other hand, \(\exists\Umn{*}\) such that \(\fX(\Umn{*})=\sup\nolimits_{\Vm\in\set{U}'} \fX(\Vm)\) 
			by the extreme value theorem for the upper semicontinuous functions \cite[Ch3\S12.2]{kolmogorovfomin75} because 
			\(\set{U}'\) is compact and \(\fX(\cdot) \) is upper semicontinuous for the topology of setwise 
			convergence. 
			
			\item\label{information:oti-saddlepoint}
			\emph{\(\fX(\Umn{*})=\CRD{\rno}{\Wm}{\umn{*}}{\mP}\) where
				\(\umn{*}\DEF\sum_{\dinp}\mP(\dinp)\Umn{*}(\dinp)\):} 
			As a result of Lemma \ref{lem:variational} we have
			\begin{align}
			\label{eq:information-oti-variational}
			\CRD{\rno}{\Wm}{\umn{*}}{\mP}
			&=\sup\nolimits_{\Vm\in\set{U}}\tfrac{\rno}{1-\rno}\CRD{1}{\Vm}{\Wm}{\mP}+\CRD{1}{\Vm}{\umn{*}}{\mP}.
			\end{align}
			On the other hand 
			\(\tfrac{\rno}{1-\rno}\CRD{1}{\Umn{*}}{\Wm}{\mP}+\CRD{1}{\Umn{*}}{\umn{*}}{\mP}=\fX(\Umn{*})\)
			because  \(\RMI{1}{\mP}{\Umn{*}}=\CRD{1}{\Umn{*}}{\umn{*}}{\mP}\) by part (\ref{information:one}). 
			Then \(\CRD{\rno}{\Wm}{\umn{*}}{\mP}\geq \fX(\Umn{*})\) is evident by 
			\eqref{eq:information-oti-variational}.
			
			In order to prove \(\CRD{\rno}{\Wm}{\umn{*}}{\mP}\leq \fX(\Umn{*})\), let us consider a 
			\(\Vm\in\set{U}\) and define \(\Vma{}{(\ind)}\) and \(\qma{}{(\ind)}\) for each \(\ind\in\integers{+}\) 
			as
			\begin{align}
			\notag
			\Vma{}{(\ind)}
			&\DEF\tfrac{\ind-1}{\ind}\Umn{*}+\tfrac{1}{\ind}\Vm,
			\\
			\notag
			\qma{}{(\ind)}
			&\DEF\tfrac{\ind-1}{\ind}\umn{*}+\tfrac{1}{\ind}\sum\nolimits_{\dinp}\mP(\dinp)\Vm(\dinp).
			\end{align}
			As a result of the decomposition given in \eqref{eq:information-oti-U} we have
			\begin{align}
			\notag
			\fX(\Vma{}{(\ind)})
			&=
			\tfrac{1}{1-\rno}\CRD{1}{\Vma{}{(\ind)}}{\Wma{\rno}{\qmn{1,\mP}}}{\mP}
			-\RD{1}{\qma{}{(\ind)}}{\qmn{1,\mP}}
			+\CRD{\rno}{\Wm}{\qmn{1,\mP}}{\mP}.
			\end{align}
			Then using the Jensen's inequality and convexity of the order one \renyi divergence 
			in its first argument established in Lemma \ref{lem:divergence-jointconvexity} we get
			\begin{align}
			\notag
			\fX(\Vma{}{(\ind)})
			&\geq 
			\tfrac{1}{1-\rno}
			\left[\tfrac{\ind-1}{\ind} \CRD{1}{\Umn{*}}{\Wma{\rno}{\qmn{1,\mP}}}{\mP}
			+\tfrac{1}{\ind} \CRD{1}{\Vm}{\Wma{\rno}{\qmn{1,\mP}}}{\mP}\right]
			-\RD{1}{\qma{}{(\ind)}}{\qmn{1,\mP}}
			+\CRD{\rno}{\Wm}{\qmn{1,\mP}}{\mP}
			\\
			\notag
			&=\tfrac{\ind-1}{\ind}
			\left[\fX(\Umn{*})+\RD{1}{\umn{*}}{\qmn{1,\mP}}\right]
			+\tfrac{1}{\ind}
			\left[\tfrac{\rno}{1-\rno}\CRD{1}{\Vm}{\Wm}{\mP}+\CRD{1}{\Vm}{\qmn{1,\mP}}{\mP}\right]
			-\RD{1}{\qma{}{(\ind)}}{\qmn{1,\mP}}
			\\
			\notag
			&=\tfrac{\ind-1}{\ind}
			\left[\fX(\Umn{*})+\RD{1}{\umn{*}}{\qma{}{(\ind)}}\right]
			+\tfrac{1}{\ind}
			\left[\tfrac{\rno}{1-\rno}\CRD{1}{\Vm}{\Wm}{\mP}+\CRD{1}{\Vm}{\qma{}{(\ind)}}{\mP}\right].
			\end{align}
			Then using \(\fX(\Umn{*})=\sup_{\Vm\in\set{U}} \fX(\Vm)\geq \fX(\Vma{}{(\ind)})\) and \(\RD{1}{\umn{*}}{\qma{}{(\ind)}}\geq 0\)
			we get
			\begin{align}
			\notag
			\fX(\Umn{*})
			&\geq \tfrac{\rno}{1-\rno}\CRD{1}{\Vm}{\Wm}{\mP}+\CRD{1}{\Vm}{\qma{}{(\ind)}}{\mP}
			&
			&\forall \ind \in \integers{+}.
			\end{align}
			On the other hand, \(\lim\nolimits_{\ind \to \infty} \CRD{1}{\Vm}{\qma{}{(\ind)}}{\mP}
			\geq  \CRD{1}{\Vm}{\umn{*}}{\mP}\)
			because \(\lim_{\ind \to \infty} \lon{\qma{}{(\ind)}-\umn{*}}=0\) by construction
			and the \renyi divergence is lower semicontinuous in its second argument 
			for the topology of setwise convergence by Lemma \ref{lem:divergence:lsc}.
			Then 
			\begin{align}
			\notag
			\fX(\Umn{*})
			&\geq \tfrac{\rno}{1-\rno}\CRD{1}{\Vm}{\Wm}{\mP}+\CRD{1}{\Vm}{\umn{*}}{\mP}
			&
			&\forall \Vm\in\set{U}.
			\end{align}
			Hence, \(\fX(\Umn{*})\geq\CRD{\rno}{\Wm}{\umn{*}}{\mP}\)
			by \eqref{eq:information-oti-variational}.
			
			\item\label{information:oti-fixedpoint}
			\emph{\(\Aop{\rno}{\mP}{\umn{*}}=\umn{*}\) and 
				\(\Umn{*}(\dinp)=\Wma{\rno}{\umn{*}}(\dinp)\) 
				for all \(\dinp\) such that \(\mP(\dinp)>0\):}
			Note that \(\CRD{\rno}{\Wm}{\umn{*}}{\mP}<\infty\) by
			part (\ref{information:oti}-\ref{information:oti-saddlepoint})
			because \(\fX(\Umn{*})\leq \bent{\mP}<\infty\) by definition. 
			Consequently, we can define the tilted probability measure 
			\(\Wma{\rno}{\umn{*}}(\dinp)\) for each 
			\(\dinp\) such that \(\mP(\dinp)>0\).
			Using the definitions of \(\fX(\cdot)\) and the tilted channel \(\Wma{\rno}{\mQ}\)
			together with the identity
			\(\RMI{1}{\mP}{\Umn{*}}=\CRD{1}{\Umn{*}}{\umn{*}}{\mP}\), 
			which follows from  part (\ref{information:one}), we get
			\begin{align}
			\notag
			\fX(\Umn{*})
			&=\CRD{\rno}{\Wm}{\umn{*}}{\mP}
			+\tfrac{1}{1-\rno}\CRD{1}{\Umn{*}}{\Wma{\rno}{\umn{*}}}{\mP}.
			\end{align}  
			Since \(\fX(\Umn{*})=\CRD{\rno}{\Wm}{\umn{*}}{\mP}\) by 
			part (\ref{information:oti}-\ref{information:oti-saddlepoint})
			we get \(\CRD{1}{\Umn{*}}{\Wma{\rno}{\umn{*}}}{\mP}=0\). 
			Hence \(\Umn{*}(\dinp)=\Wma{\rno}{\umn{*}}(\dinp)\) for all \(\dinp\) 
			such that \(\mP(\dinp)>0\) by Lemma \ref{lem:divergence-pinsker}.
			As a result \(\Aop{\rno}{\mP}{\umn{*}}=\umn{*}\) because 
			\(\Aop{\rno}{\mP}{\umn{*}}=\sum_{\dinp}\mP(\dinp)\Wma{\rno}{\umn{*}}(\dinp)\)
			and
			\(\umn{*}=\sum_{\dinp}\mP(\dinp)\Umn{*}(\dinp)\) 
			by definition.
		\end{enumerate}
				
		\item[(\ref{lem:information}-\ref{information:alternative})]
		We prove the statement for \(\rno\in(0,1)\) and \(\rno\in(1,\infty)\) cases separately,
		\begin{itemize}
			\item Recall that \(\sum_{\dinp}\mP(\dinp)\Wma{\rno}{\qmn{\rno,\mP}}(\dinp)=\qmn{\rno,\mP}\)
			by parts (\ref{information:zto}).
			Then as a result of \eqref{eq:lem:information-one:EHB} we have
			\begin{align}
			\RMI{1}{\mP}{\Wma{\rno}{\qmn{\rno,\mP}}}
			&=\CRD{1}{\Wma{\rno}{\qmn{\rno,\mP}}}{\qmn{\rno,\mP}}{\mP}.
			\end{align}
			Then \eqref{eq:lem:information:alternative:opt} follows from
			Lemma \ref{lem:variational} 
			and  \(\RMI{\rno}{\mP}{\Wm}=\CRD{\rno}{\Wm}{\qmn{\rno,\mP}}{\mP}\)
			for \(\rno\in(0,1)\). 
			
			On the other hand as a result the definition of the Augustin information,
			and Lemma \ref{lem:variational} we have
			\begin{align}
			\notag
			\RMI{\rno}{\mP}{\Wm}
			&=\inf\nolimits_{\mQ\in\pmea{\outA}}\inf\nolimits_{\Vm\in\pmea{\outA|\inpS}} \tfrac{\rno}{1-\rno}\CRD{1}{\Vm}{\Wm}{\mP}+\CRD{1}{\Vm}{\mQ}{\mP}  
			\\
			\notag
			&=\inf\nolimits_{\Vm\in\pmea{\outA|\inpS}}\inf\nolimits_{\mQ\in\pmea{\outA}}
			\tfrac{\rno}{1-\rno}\CRD{1}{\Vm}{\Wm}{\mP}+\CRD{1}{\Vm}{\mQ}{\mP}.  
			\end{align}
			Then \eqref{eq:lem:information:alternative:def} follows from the definition of the order
			one Augustin information.
			
			\item Note that for \(\rno\in(1,\infty)\) identity given in \eqref{eq:lem:information:alternative:opt} 
			is nothing but \(\fX(\Wma{\rno}{\qmn{\rno,\mP}})=\RMI{\rno}{\mP}{\Wm}\) which is already 
			established in the proof of part (\ref{information:oti}). Similarly  
			\eqref{eq:lem:information:alternative:def} is equivalent to 
			\(\sup_{\Vm\in\pmea{\outA|\inpS}} \fX(\Vm)=\RMI{\rno}{\mP}{\Wm}\) which is established in 
			the proof of part (\ref{information:oti}).
		\end{itemize}
	\end{enumerate}
\end{proof}

\begin{proof}[Proof of Lemma \ref{lem:information:product}]
	The following identity can be confirmed by substitution
	\begin{align}
	\notag
	\CRD{\rno}{\Wmn{[1,\blx]}}{\bigotimes\nolimits_{\tin=1}^{\blx}\qmn{\rno,\pmn{\tin}}}{\mP}
	&=\sum\nolimits_{\tin=1}^{\blx}\CRD{\rno}{\Wmn{\tin}}{\qmn{\rno,\pmn{\tin}}}{\pmn{\tin}}.
	\end{align}
	Then using \eqref{eq:lem:information-one:EHB}, \eqref{eq:lem:information-zto:EHB}, \eqref{eq:lem:information-oti:EHB}
	we get
	\begin{align}
	\notag
	\sum\nolimits_{\tin=1}^{\blx} \RMI{\rno}{\pmn{\tin}}{\Wmn{\tin}}
	-\RD{1\wedge \rno}{\qmn{\rno,\mP}}{\bigotimes\nolimits_{\tin=1}^{\blx}\qmn{\rno,\pmn{\tin}}}
	\geq \RMI{\rno}{\mP}{\Wmn{[1,\blx]}}
	&\geq \sum\nolimits_{\tin=1}^{\blx} \RMI{\rno}{\pmn{\tin}}{\Wmn{\tin}}
	-\RD{\rno\vee 1}{\qmn{\rno,\mP}}{\bigotimes\nolimits_{\tin=1}^{\blx}\qmn{\rno,\pmn{\tin}}}.
	\end{align} 
	Thus \eqref{eq:lem:information:product} holds for all \(\mP\in\pdis{\inpS_{1}^{\blx}}\)
	because \renyi divergence between probability measures is non-negative.
	Furthermore, \eqref{eq:lem:information:product} holds as an equality iff
	for an \(\rno\in\reals{+}\)  iff \(\qmn{\rno,\mP}\) satisfies
	\eqref{eq:lem:information:product-mean}
	because the \renyi divergence between distinct probability measures is positive.
	
	If \(\mP\!=\!\bigotimes\nolimits_{\tin=1}^{\blx}\pmn{\tin}\), then
	one can confirm \eqref{eq:lem:information:product-mean} for \(\rno\!=\!1\) case
	by substitution.
	In addition for any \(\rno\!\in\!\reals{+}\!\setminus\!\{1\}\) one can show by substitution 
	that the probability measure 
	\(\mQ=\bigotimes\nolimits_{\tin=1}^{\blx}\qmn{\rno,\pmn{\tin}}\) is a fixed point of \(\Aop{\rno}{\mP}{\cdot}\), 
	i.e. \(\Aop{\rno}{\mP}{\mQ}=\mQ\).
	Furthermore,	
	\(\qmn{1,\mP}\AC \mQ\) because \(\qmn{1,\pmn{\tin}}\AC \qmn{\rno,\pmn{\tin}}\) for each 
	\(\tin\in\{1,\ldots,\blx\}\).
	Thus for \(\rno\!\in\!\reals{+}\!\setminus\!\{1\}\)
	the identity in \eqref{eq:lem:information:product-mean} follows from 
	Lemma \ref{lem:information}-(\ref{information:zto},\ref{information:oti}). 
\end{proof}

\begin{proof}[Proof of Lemma \ref{lem:informationP}]
	Note that \(\CRD{\rno}{\Wm}{\mQ}{\mP}\) is linear and hence concave in \(\mP\) for any \(\mQ\) by 
	definition. Then \(\RMI{\rno}{\mP}{\Wm}\) is concave in \(\mP\) because pointwise infimum of a family 
	of concave functions is concave. 
	Furthermore by Lemma \ref{lem:information}-(\ref{information:one},\ref{information:zto},\ref{information:oti}),
	\(\exists!\qmn{\rno,\pmn{\beta}}\)  such that \(\CRD{\rno}{\Wm}{\qmn{\rno,\pmn{\beta}}}{\pmn{\beta}}=\RMI{\rno}{\pmn{\beta}}{\Wm}\).
	In addition,
	\begin{align}
	\notag
	\CRD{\rno}{\Wm}{\qmn{\rno,\pmn{\beta}}}{\pmn{\beta}}
	&=\beta \CRD{\rno}{\Wm}{\qmn{\rno,\pmn{\beta}}}{\pmn{1}}
	+(1-\beta) \CRD{\rno}{\Wm}{\qmn{\rno,\pmn{\beta}}}{\pmn{0}}.
	\end{align}
	Then equation \eqref{eq:lem:informationP-A} and \eqref{eq:lem:informationP-B} are 
	obtained by bounding
	\(\CRD{\rno}{\Wm}{\qmn{\rno,\pmn{\beta}}}{\pmn{1}}\) and \(\CRD{\rno}{\Wm}{\qmn{\rno,\pmn{\beta}}}{\pmn{0}}\)
	using Lemma \ref{lem:information}-(\ref{information:one},\ref{information:zto},\ref{information:oti}).
	
	On the other hand, Lemma \ref{lem:divergence-RM} implies
	\begin{align}
	\notag
	\CRD{\rno}{\Wm}{\beta\qmn{\rno,\pmn{1}}+(1-\beta)\qmn{\rno,\pmn{0}}}{\pmn{\beta}}
	&=\beta    \CRD{\rno}{\Wm}{\beta\qmn{\rno,\pmn{1}}+(1-\beta)\qmn{\rno,\pmn{0}}}{\pmn{1}}
	+(1-\beta) \CRD{\rno}{\Wm}{\beta\qmn{\rno,\pmn{1}}+(1-\beta)\qmn{\rno,\pmn{0}}}{\pmn{0}}
	\\
	\notag
	&\leq \beta \CRD{\rno}{\Wm}{\qmn{\rno,\pmn{1}}}{\pmn{1}}-\beta\ln \beta
	+(1-\beta)  \CRD{\rno}{\Wm}{\qmn{\rno,\pmn{0}}}{\pmn{0}}-(1-\beta)\ln(1-\beta)
	\\
	\notag
	&=\beta \RMI{\rno}{\pmn{1}}{\Wm}+(1-\beta) \RMI{\rno}{\pmn{0}}{\Wm}+\bent{\beta}.
	\end{align}
	Then \eqref{eq:lem:informationP-C} follows from the lower bound on \(\CRD{\rno}{\Wm}{\mQ}{\mP}\)
	given in Lemma \ref{lem:information}-(\ref{information:one},\ref{information:zto},\ref{information:oti}).
\end{proof}

\begin{proof}[Proof of Lemma \ref{lem:meanO}]
	Note that as result of \eqref{eq:mean} we have,
	\begin{align}
	\label{eq:mean-bound}
	[\mP(\dinp)]^{\frac{1}{\rno}}e^{\frac{1-\rno}{\rno}\RD{\rno}{\Wm(\dinp)}{\qmn{\rno,\mP}}}\Wm(\dinp)
	&\leq \qmn{\rno,\mP}
	&
	&\forall \dinp \in \inpS.
	\end{align}
	
	\begin{enumerate}
		\item[(\ref{lem:meanO}-\ref{meanO:Rdivergencebound})] 
		Using Lemma \ref{lem:divergence-RM} and \eqref{eq:mean-bound} we get
		\begin{align}
		\notag
		\RD{\rno}{\Wm(\dinp)}{\qmn{\rno,\mP}}
		&\leq \RD{\rno}{\Wm(\dinp)}{[\mP(\dinp)]^{\frac{1}{\rno}}e^{\frac{1-\rno}{\rno}\RD{\rno}{\Wm(\dinp)}{\qmn{\rno,\mP}}} \Wm(\dinp)}
		\\
		\notag
		&=\tfrac{1}{\rno}\ln \tfrac{1}{\mP(\dinp)} +\tfrac{\rno-1}{\rno}\RD{\rno}{\Wm(\dinp)}{\qmn{\rno,\mP}}.
		\end{align}
		
		\item[(\ref{lem:meanO}-\ref{meanO:measurebound})]
		We employ \eqref{eq:mean-bound}
		together with  \(\RD{\rno}{\Wm(\dinp)}{\qmn{\rno,\mP}}>0\) for \(\rno\!\in\!(0,1]\) case 
		and  
		together with part (\ref{meanO:Rdivergencebound}) for \(\rno\!\in\!(1,\infty)\) case.
		
		\item[(\ref{lem:meanO}-\ref{meanO:boundedlogRND})]
		We prove \(\rno\in(0,1)\) case and \(\rno\in(1,\infty)\) case separately.
		\begin{itemize}
			\item \emph{\(\rno\in(0,1)\):} First use Jensen's inequality, i.e. \(\EX{\xi^{\rno}}\leq \EX{\xi}^{\rno}\), 
			in \eqref{eq:mean}; then invoke \(\RD{\rno}{\Wm(\dinp)}{\qmn{\rno,\mP}}\leq \ln \tfrac{1}{\mP(\dinp)}\):
			\begin{align}
			\notag
			\der{\qmn{\rno,\mP}}{\qmn{1,\mP}}
			&\leq \sum\nolimits_{\dinp}  \mP(\dinp)\der{\Wm(\dinp)}{\qmn{1,\mP}}e^{\frac{1-\rno}{\rno}\RD{\rno}{\Wm(\dinp)}{\qmn{\rno,\mP}}}
			\\
			\notag
			&\leq(\min\nolimits_{\dinp:\mP(\dinp)>0}\mP(\dinp))^{\frac{\rno-1}{\rno}}.
			\end{align}
			Recall that if \(\xi(\dinp)\geq0\) for all \(\dinp\), then
			\(\sum_{\dinp} [\xi(\dinp)]^{\rno}\geq [\sum_{\dinp} \xi(\dinp)]^{\rno}\).
			Using \(\RD{\rno}{\Wm(\dinp)}{\qmn{\rno,\mP}}\geq 0\) we get 
			\begin{align}
			\notag
			\der{\qmn{\rno,\mP}}{\qmn{1,\mP}}
			&\geq \sum\nolimits_{\dinp} \left(\mP(\dinp)\right)^{\frac{1}{\rno}} \der{\Wm(\dinp)}{\qmn{1,\mP}} e^{(\frac{1-\rno}{\rno})\RD{\rno}{\Wm(\dinp)}{\qmn{\rno,\mP}}}
			\\
			\notag
			&\geq (\min\nolimits_{\dinp:\mP(\dinp)>0}\mP(\dinp))^{\frac{1-\rno}{\rno}}.
			\end{align}
			\item \emph{\(\rno\in(1,\infty)\):} First use Jensen's inequality, i.e. \(\EX{\xi^{\rno}}\geq \EX{\xi}^{\rno}\), 
			in \eqref{eq:mean}, then invoke \(\RD{\rno}{\Wm(\dinp)}{\qmn{\rno,\mP}}\leq \ln \tfrac{1}{\mP(\dinp)}\):
			\begin{align}
			\notag
			\der{\qmn{\rno,\mP}}{\qmn{1,\mP}}
			&\geq \sum\nolimits_{\dinp}  \mP(\dinp)\der{\Wm(\dinp)}{\qmn{1,\mP}}e^{\frac{1-\rno}{\rno}\RD{\rno}{\Wm(\dinp)}{\qmn{\rno,\mP}}}
			\\
			\notag
			&\geq (\min\nolimits_{\dinp:\mP(\dinp)>0}\mP(\dinp))^{\frac{\rno-1}{\rno}}.
			\end{align}
			Recall that if \(\xi(\dinp)\geq0\) for all \(\dinp\), then
			\(\sum_{\dinp} [\xi(\dinp)]^{\rno}\leq [\sum_{\dinp} \xi(\dinp)]^{\rno}\).
			Using \(\RD{\rno}{\Wm(\dinp)}{\qmn{\rno,\mP}}\geq 0\) we get 
			\begin{align}
			\notag
			\der{\qmn{\rno,\mP}}{\qmn{1,\mP}}
			&\leq \sum\nolimits_{\dinp} \left(\mP(\dinp)\right)^{\frac{1}{\rno}} \der{\Wm(\dinp)}{\qmn{1,\mP}} e^{(\frac{1-\rno}{\rno})\RD{\rno}{\Wm(\dinp)}{\qmn{\rno,\mP}}} 			
			\\
			\notag
			&\leq (\min\nolimits_{\dinp:\mP(\dinp)>0}\mP(\dinp))^{\frac{1-\rno}{\rno}}.
			\end{align}
		\end{itemize}
	\end{enumerate}
\end{proof}

\begin{proof}[Proof of Lemma \ref{lem:informationO}]~
\begin{enumerate}[(a)]
\item [(\ref{lem:informationO}-\ref{informationO:strictconvexity})]
For brevity, let us denote \((\rno-1)\RMI{\rno}{\mP}{\Wm}\) by \(\gX(\rno)\)
in this part of the proof.
We first prove the dichotomy about \(\gX(\cdot)\)
on \((0,1)\) and on \((1,\infty)\). 
Then we extend this dichotomy to \(\reals{+}\) assuming 
that \(\gX(\cdot)\) is convex on \(\reals{+}\).
After that we establish the assumed convexity of \(\gX(\cdot)\)
on \(\reals{+}\).

Let \(\rno_{\beta}=\beta\rno_{1}+(1-\beta)\rno_{0}\)
for any \(\rno_{0},\rno_{1}\in(0,1)\) and \(\beta\in(0,1)\).
Then for any \(\rno_{0},\rno_{1}\in(0,1)\) and \(\beta\in(0,1)\) we have
\begin{align}
\notag
\beta \gX(\rno_{1})+(1-\beta)\gX(\rno_{0})
&\geq \rnb(\rno_{1}-1)
\CRD{\rno_{1}}{\Wm}{\qmn{\rno_{\rnb},\mP}}{\mP}
+(1-\rnb)(\rno_{0}-1)
\CRD{\rno_{0}}{\Wm}{\qmn{\rno_{\rnb},\mP}}{\mP}
\\
\notag
&=\sum\nolimits_{\dinp}\mP(\dinp)
\ln 
\left(
\EXS{\qmn{\rno_{\rnb},\mP}}{(\der{\Wm(\dinp)}{\qmn{\rno_{\rnb},\mP}})^{\rno_{1}}}
\right)^{\rnb}
\left(
\EXS{\qmn{\rno_{\rnb},\mP}}{(\der{\Wm(\dinp)}{\qmn{\rno_{\rnb},\mP}})^{\rno_{0}}}
\right)^{1-\rnb}
\\
\notag
&\geq\sum\nolimits_{\dinp}\mP(\dinp)
\ln 
\EXS{\qmn{\rno_{\rnb},\mP}}{(\der{\Wm(\dinp)}{\qmn{\rno_{\rnb},\mP}})^{\rno_{\rnb}}}
\\
\notag
&=\gX(\rno_{\beta})
\end{align}
where the first inequality follows from the definition of the Augustin 
information and the second inequality follows from the H\"{o}lder's 
inequality. 
Furthemore, the first inequality is an equality iff
\(\qmn{\rno_{0},\mP}=\qmn{\rno_{\beta},\mP}=\qmn{\rno_{1},\mP}\) 
by Lemma \ref{lem:information}-(\ref{information:zto})
and the second inequality is an equality iff
\(\der{\Wm(\dinp)}{\qmn{\rno_{\rnb},\mP}}=\gamma(\dinp)\) 
holds \(\Wm(\dinp)\)-a.s. for all \(\dinp\in\supp{\mP}\).
On the other hand, 
if \(\der{\Wm(\dinp)}{\qmn{\rno,\mP}}=\gamma(\dinp)\) 
holds \(\Wm(\dinp)\)-a.s., 
then \(\Wma{\rno}{\qmn{\rno,\mP}}(\dinp)=\Wm(\dinp)\).
Consequently, if
\(\der{\Wm(\dinp)}{\qmn{\rno,\mP}}=\gamma(\dinp)\) 
holds \(\Wm(\dinp)\)-a.s. for all \(\dinp\in\supp{\mP}\)
then  \(\qmn{\rno,\mP}=\qmn{1,\mP}\) by Lemma \ref{lem:information}-(\ref{information:zto})
because \(\Aop{\rno}{\mP}{\qmn{\rno,\mP}}=\qmn{1,\mP}\).
Thus 
either \(\gX(\cdot)\) is strictly convex on \((0,1)\) 
or \(\der{\Wm(\dinp)}{\qmn{1,\mP}}=\gamma(\dinp)\) 
\(\Wm(\dinp)\)-a.s. for all \(\dinp\in\supp{\mP}\)
and \(\RMI{\rno}{\mP}{\Wm}=\sum_{\dinp}\mP(\dinp)\ln \gamma(\dinp)\)
for all \(\rno\in\reals{+}\). 

Let \(\rno_{\beta}=\beta\rno_{1}+(1-\beta)\rno_{0}\) and
\(\der{\mean}{\qmn{1,\mP}}=(\der{\qmn{\rno_{1},\mP}}{\qmn{1,\mP}})^{\frac{(\rno_{1}-1)\rnb}{\rno_{\rnb}-1}} 
(\der{\qmn{\rno_{0},\mP}}{\qmn{1,\mP}})^{\frac{(\rno_{0}-1)(1-\rnb)}{\rno_{\rnb}-1}}\)
for any  \(\rno_{0},\rno_{1}\in(1,\infty)\) and \(\beta\in(0,1)\).
Then
\begin{align}
\notag
\beta \gX(\rno_{1})+(1-\beta)\gX(\rno_{0})
&=\sum\nolimits_{\dinp}\mP(\dinp)
\ln \left(
\EXS{\qmn{1,\mP}}{(\der{\Wm(\dinp)}{\qmn{1,\mP}})^{\rno_{1}}(\der{\qmn{\rno_{1},\mP}}{\qmn{1,\mP}})^{1-\rno_{1}}}
\right)^{\beta}
\left(\EXS{\qmn{1,\mP}}{(\der{\Wm(\dinp)}{\qmn{1,\mP}})^{\rno_{0}}(\der{\qmn{\rno_{0},\mP}}{\qmn{1,\mP}})^{1-\rno_{0}}}
\right)^{1-\beta}
\\
\notag
&\geq \sum\nolimits_{\dinp}\mP(\dinp)
\ln 
\EXS{\qmn{1,\mP}}{(\der{\Wm(\dinp)}{\qmn{1,\mP}})^{\rno_{\beta}}
	(\der{\qmn{\rno_{1},\mP}}{\qmn{1,\mP}})^{\beta(1-\rno_{1})}
	(\der{\qmn{\rno_{0},\mP}}{\qmn{1,\mP}})^{(1-\beta)(1-\rno_{0})}}
\\
\notag
&=(\rno_{\beta}-1)\CRD{\rno_{\beta}}{\Wm}{\tfrac{\mean}{\lon{\mean}}}{\mP}-(\rno_{\beta}-1)\ln \lon{\mean}
\\
\notag
&\geq \gX(\rno_{\beta})
\end{align}
where the first inequality follows from the H\"{o}lder's inequality
and the second inequality follows from 
the definition of Augustin information 
and the fact that \(\lon{\mean}\leq 1 \), 
which is consequence of the H\"{o}lder's inequality.
Furthermore, the first inequality is an equality iff
\(\der{\Wm(\dinp)}{\qmn{1,\mP}}
(\der{\qmn{\rno_{1},\mP}}{\qmn{1,\mP}})^{\frac{1-\rno_{1}}{\rno_{1}-\rno_{0}}}
(\der{\qmn{\rno_{0},\mP}}{\qmn{1,\mP}})^{\frac{\rno_{0}-1}{\rno_{1}-\rno_{0}}}
=\gamma(\dinp)\)
holds \(\Wm(\dinp)\)-a.s. for all \(\dinp\in\supp{\mP}\)
and the second inequality is an equality iff
\(\mean=\qmn{\rno_{\beta},\mP}\)
by Lemma \ref{lem:information}-(\ref{information:oti})
because \(\lon{\mean}\leq 1\) by the H\"{o}lder's inequality.
On the other hand, \(\lon{\mean}=1\) iff 
\(\qmn{\rno_{0},\mP}=\qmn{\rno_{1},\mP}\)
as a result of the H\"{o}lder's inequality.
Thus the second inequality is an equality iff
\(\qmn{\rno_{0},\mP}=\qmn{\rno_{1},\mP}=\qmn{\rno_{\rnb},\mP}\).
Hence both inequalities are equalities, i.e. 
\(\beta \gX(\rno_{1})+(1-\beta)\gX(\rno_{0})=\gX(\rno_{\beta})\),
iff 
\(\qmn{\rno_{0},\mP}=\qmn{\rno_{\beta},\mP}=\qmn{\rno_{1},\mP}\)
and \(\der{\Wm(\dinp)}{\qmn{\rno_{\rnb},\mP}}=\gamma(\dinp)\) 
holds \(\Wm(\dinp)\)-a.s. for all \(\dinp\in\supp{\mP}\).
Following a reasoning similar to the one for \(\rno_{0},\rno_{1}\in(0,1)\)
case and invoking 
Lemma \ref{lem:information}-(\ref{information:oti})
instead of
Lemma \ref{lem:information}-(\ref{information:zto}),
we conclude that
either \(\gX(\cdot)\) is strictly convex on \((1,\infty)\) 
or \(\der{\Wm(\dinp)}{\qmn{1,\mP}}=\gamma(\dinp)\) 
\(\Wm(\dinp)\)-a.s. for all \(\dinp\in\supp{\mP}\)
and \(\RMI{\rno}{\mP}{\Wm}=\sum_{\dinp}\mP(\dinp)\ln \gamma(\dinp)\)
for all \(\rno\in\reals{+}\).

We assume the convexity of \(\gX(\cdot)\) on \(\reals{+}\), 
in order to establish the strict convexity of \(\gX(\cdot)\) on \(\reals{+}\)
using the strict convexity on \((0,1)\) and \((1,\infty)\).
Note that if \(\rno_{0}\in(0,1]\), \(\rno_{1}\in(1,\infty)\)
and \(\rno_{\beta}\in(1,\infty)\), then 
there exists an \(\epsilon\in(0,\beta)\) such that
\(\rno_{\beta-\epsilon}\in(1,\infty)\). 
Thus
\begin{align}
\notag
\beta \gX(\rno_{1})+(1-\beta)\gX(\rno_{0})
&=\tfrac{\epsilon}{1-\beta+\epsilon}\gX(\rno_{1})+
\tfrac{1-\beta}{1-\beta+\epsilon}
\left[(\beta-\epsilon)\gX(\rno_{1})+(1-\beta+\epsilon)\gX(\rno_{0})\right]
\\
\notag
&\geq \tfrac{\epsilon}{1-\beta+\epsilon}\gX(\rno_{1})+
\tfrac{1-\beta}{1-\beta+\epsilon}\gX(\rno_{\beta-\epsilon})
\\
\notag
&>\gX(\tfrac{\epsilon}{1-\beta+\epsilon}\rno_{1}+\tfrac{1-\beta}{1-\beta+\epsilon}\rno_{\beta-\epsilon})
\\
\notag
&=\gX(\rno_{\beta}).
\end{align}
Similar manipulations can be used to prove the strict inequality for
\(\rno_{0}\in(0,1)\), \(\rno_{1}\in[1,\infty)\),\(\rno_{\beta}\in(0,1)\)
case
and
\(\rno_{0}\in(0,1)\), \(\rno_{1}\in(1,\infty)\), \(\rno_{\beta}=1\)
case.

Now we are left with establishing the convexity of \(\gX(\cdot)\)
on \(\reals{+}\) that we have assumed. 
Invoking Lemma \ref{lem:information}-(\ref{information:alternative}) 
for \(\rno\in\reals{+}\setminus\{1\}\) case and
using \(\CRD{1}{\Wm}{\Wm}{\mP}\!=\!0\) for \(\rno\!=\!1\) case we get
\begin{align}
\notag
\gX(\rno)
&=\sup\nolimits_{\Vm\in\pmea{\outA|\inpS}} (\rno-1)\RMI{1}{\mP}{\Vm}-\rno\CRD{1}{\Vm}{\Wm}{\mP}.
\end{align}
Then \(\gX(\rno)\) is convex in \(\rno\) 
because pointwise supremum of a family of linear/convex functions is convex. 

On the other hand, using \(\Vm=\Wm\) we can deduce that,
\(\gX(\rno)\geq (\rno-1)\RMI{1}{\mP}{\Wm}\).
and \(\RMI{1}{\mP}{\Wm}\in[0,\bent{\mP}]\) by Lemma \ref{lem:information}-(\ref{information:bounded}).
Thus \(\gX(\rno)\geq -\bent{\mP}\).

		\item[(\ref{lem:informationO}-\ref{informationO:decreasing})]
		Since \((\rno-1)\RMI{\rno}{\mP}{\Wm}\) is finite and convex in \(\rno\) on \(\reals{+}\),
		it is  continuous on \(\reals{+}\) by \cite[Thm. 6.3.3]{dudley}.
		Then \(\frac{1-\rno}{\rno}\RMI{\rno}{\mP}{\Wm}\) is continuous in \(\rno\) on \(\reals{+}\), 
		as well.
		Furthermore,
		\begin{align}
		\notag
		\tfrac{1-\rno}{\rno}\RMI{\rno}{\mP}{\Wm}
		&=\inf\nolimits_{\Vm\in\pmea{\outA|\inpS}} \tfrac{1-\rno}{\rno}\RMI{1}{\mP}{\Vm}+\CRD{1}{\Vm}{\Wm}{\mP}
		\end{align}
		by Lemma \ref{lem:information}-(\ref{information:bounded},\ref{information:alternative})
		and \(\CRD{1}{\Wm}{\Wm}{\mP}\!=\!0\).
		Then \(\tfrac{1-\rno}{\rno}\RMI{\rno}{\mP}{\Wm}\) is nonincreasing in \(\rno\) because 
		infimum of a family of nonincreasing functions is nonincreasing.
		Note that \(\tfrac{1-\rno}{\rno}\RMI{1}{\mP}{\Vm}+\CRD{1}{\Vm}{\Wm}{\mP}\)
		is nonincreasing in \(\rno\) because \(\RMI{1}{\mP}{\Vm}\) is nonnegative. 
		
		\item [(\ref{lem:informationO}-\ref{informationO:continuity})] 
		\(\RMI{\rno}{\mP}{\Wm}\) is nondecreasing in \(\rno\) because the pointwise infimum of a family 
		of nondecreasing functions is nondecreasing and 
		the \renyi divergence is nondecreasing in its order by 
		Lemma \ref{lem:divergence-order}.
		
		Since \((\rno-1)\RMI{\rno}{\mP}{\Wm}\) is finite and convex in \(\rno\) on \(\reals{+}\),
		it is  continuous on \(\reals{+}\) by \cite[Thm. 6.3.3]{dudley}.
		Then \(\RMI{\rno}{\mP}{\Wm}\) is continuous on \((0,1)\) and \((1,\infty)\). 
		In order to extend the continuity to \(\reals{+}\) we need to prove that 
		\(\RMI{\rno}{\mP}{\Wm}\) is continuous at \(\rno=1\).
		Note that as a result of the definition of the Augustin information we have
		\(\RMI{\rno}{\mP}{\Wm}\leq \CRD{\rno}{\Wm}{\qmn{1,\mP}}{\mP}\) for all \(\rno\in\reals{+}\).
		Since \(\RMI{\rno}{\mP}{\Wm}\) is nondecreasing in \(\rno\) we have
		\begin{align}
		\label{eq:informationO:continuity-A}
		\RMI{1}{\mP}{\Wm}&\leq \RMI{\rno}{\mP}{\Wm}\leq \CRD{\rno}{\Wm}{\qmn{1,\mP}}{\mP}
		&
		&\forall \rno\in(1,\infty).
		\end{align}
		Since
		\(\qmn{\rno,\mP}\leq (\min\nolimits_{\dinp:\mP(\dinp)>0} \mP(\dinp))^{-\frac{\abs{1-\rno}}{\rno}}\qmn{1,\mP}\)
		by Lemma \ref{lem:meanO}-(\ref{meanO:boundedlogRND}),
		using Lemma \ref{lem:divergence-RM} we get
		\begin{align}
		\notag
		\CRD{\rno}{\Wm}{\qmn{\rno,\mP}}{\mP}
		&\geq \CRD{\rno}{\Wm}{\qmn{1,\mP}}{\mP}+\tfrac{1-\rno}{\rno}\ln (\min\nolimits_{\dinp:\mP(\dinp)>0} \mP(\dinp))
		&
		&\forall \rno\in(0,1).
		\end{align}
		Recall that \(\CRD{\rno}{\Wm}{\qmn{\rno,\mP}}{\mP}=\RMI{\rno}{\mP}{\Wm}\) by Lemma \ref{lem:information}-(\ref{information:zto})
		and \(\RMI{\rno}{\mP}{\Wm}\) is nondecreasing in \(\rno\). Thus we have
		\begin{align}
		\label{eq:informationO:continuity-B}
		\CRD{\rno}{\Wm}{\qmn{1,\mP}}{\mP}+\tfrac{1-\rno}{\rno}\ln (\min\nolimits_{\dinp:\mP(\dinp)>0} \mP(\dinp))
		&\leq \RMI{\rno}{\mP}{\Wm} \leq \RMI{1}{\mP}{\Wm}
		&
		&\forall \rno\in(0,1).
		\end{align}
		On the other hand, \(\CRD{\rnf}{\Wm}{\qmn{1,\mP}}{\mP}
		\leq \bent{\mP}<\infty\) for any \(\rnf\in\reals{+}\) 
		by Lemma \ref{lem:information}-(\ref{information:bounded}).
		Then \(\CRD{\rno}{\Wm}{\qmn{1,\mP}}{\mP}\) is continuous in \(\rno\) by Lemma \ref{lem:divergence-order}.
		Furthermore, \(\CRD{1}{\Wm}{\qmn{1,\mP}}{\mP}=\RMI{1}{\mP}{\Wm}\) by Lemma \ref{lem:information}-(\ref{information:one}).
		Then
		\begin{align}
		\notag
		\lim\nolimits_{\rno \to 1}\CRD{\rno}{\Wm}{\qmn{1,\mP}}{\mP}
		&=\RMI{1}{\mP}{\Wm}.
		\end{align}
		Then the continuity of  \(\RMI{\rno}{\mP}{\Wm}\) at \(\rno=1\) follows from 
		\eqref{eq:informationO:continuity-A} and \eqref{eq:informationO:continuity-B}.

		\item [(\ref{lem:informationO}-\ref{informationO:mean-uec})]
		Let \(\tau_{\dinp}(\rno)\) be \(\tau_{\dinp}(\rno)\DEF \tfrac{\rno-1}{\rno}\RD{\rno}{\Wm(\dinp)}{\qmn{\rno,\mP}}\).
		Then we can rewrite \eqref{eq:mean} as follows:
		\begin{align}
		\notag
		\ln\der{\qmn{\rnt,\mP}}{\qmn{1,\mP}}
		&=\tfrac{1}{\rnt}\ln \sum\nolimits_{\dinp} \mP(\dinp) 
		\left(\left(\der{\Wm(\dinp)}{\qmn{1,\mP}}\right)^{\rnf}e^{(1-\rnf)\RD{\rnf}{\Wm(\dinp)}{\qmn{\rnf,\mP}}}\right)^\frac{\rnt}{\rnf}
		e^{\rnt(\tau_{\dinp}(\rnf)-\tau_{\dinp}(\rnt))}.
		\end{align}
		Let us assume without loss of generality that \(\rnf>\rnt\). Then using the Jensen's inequality we get
		\begin{align}
		\label{eq:informationO:mean-uec-1}
		\ln\der{\qmn{\rnt,\mP}}{\qmn{1,\mP}}
		&\leq \ln\der{\qmn{\rnf,\mP}}{\qmn{1,\mP}}+\max\nolimits_{\dinp:\mP(\dinp)>0}(\tau_{\dinp}(\rnf)-\tau_{\dinp}(\rnt)).
		\end{align}
		On the other hand, using the fact that 
		\(\sum_{\dinp} [\xi(\dinp)]^{\frac{\rnt}{\rnf}}\geq [\sum_{\dinp} \xi(\dinp)]^{\frac{\rnt}{\rnf}}\)
		for non-negative \(\xi(\dinp)\) we get
		\begin{align}
		\label{eq:informationO:mean-uec-2}
		\ln\der{\qmn{\rnt,\mP}}{\qmn{1,\mP}}
		&\geq \ln\der{\qmn{\rnf,\mP}}{\qmn{1,\mP}}+(1-\tfrac{\rnt}{\rnf})\ln(\min\nolimits_{\dinp:\mP(\dinp)>0}\mP(\dinp))
		+\min\nolimits_{\dinp:\mP(\dinp)>0}(\tau_{\dinp}(\rnf)-\tau_{\dinp}(\rnt)).
		\end{align}
		If \(\{\tau_{\dinp}(\rno)\}_{\dinp:\mP(\dinp)>0}\) is equicontinuous in \(\rno\), 
		then \(\{\ln\der{\qmn{\rno,\mP}}{\qmn{1,\mP}}\}_{\dout\in\outS}\) is equicontinuous 
		in \(\rno\) by \eqref{eq:informationO:mean-uec-1} and \eqref{eq:informationO:mean-uec-2}.
		On the other hand, there are only finitely many \(\dinp\)'s with positive \(\mP(\dinp)\).
		Thus \(\{\tau_{\dinp}(\rno)\}_{\dinp:\mP(\dinp)>0}\) is equicontinuous if each \(\tau_{\dinp}(\rno)\)
		is continuous. 
		We are left with establishing the continuity of \(\tau_{\dinp}(\rno)\).
		
		Let \(\gX_{\dinp}(\cdot):[\rnt,\rnf]\to\reals{+}\),
		\(\fX_{\dinp}(\cdot,\cdot):[\rnt,\rnf]\times\outS\to\reals{\geq0}\),
		and \(\smn{\cdot}:[\rnt,\rnf]\to\pmea{\outA}\) be
		\begin{align}
		\notag
		\gX_{\dinp}(\rno)
		&\DEF \int \fX_{\dinp}(\rno,\dout) \rfm(\dif{\dout}),
		\\
		\notag
		\fX_{\dinp}(\rno,\dout)
		&\DEF \left(\der{\Wm(\dinp)}{\rfm}\right)^{\rno} \left(\der{\smn{\rno}}{\rfm}\right)^{1-\rno},
		\\
		\notag
		\smn{\rno}
		&\DEF\tfrac{\rnf-\rno}{\rnf-\rnt}\qmn{\rnt,\mP}+\tfrac{\rno-\rnt}{\rnf-\rnt}\qmn{\rnf,\mP}.
		\end{align}
		Then \(\fX_{\dinp}(\rno,\dout)\) is differentiable in \(\rno\) 
		and its derivative can be bounded using  Lemma \ref{lem:meanO}-(\ref{meanO:measurebound}) 
		and the identity \(\tau \ln \tfrac{1}{\tau}\leq \tfrac{1}{e}\):
		\begin{align}
		\notag
		\pder{}{\rno}\fX_{\dinp}(\rno,\dout)
		&=\left(\der{\Wm(\dinp)}{\rfm}\right)^{\rno} \left(\der{\smn{\rno}}{\rfm}\right)^{1-\rno}\ln \der{\Wm(\dinp)}{\smn{\rno}}
		+\left(\der{\Wm(\dinp)}{\rfm}\right)^{\rno} \left(\der{\smn{\rno}}{\rfm}\right)^{-\rno} \tfrac{1-\rno}{\rnf-\rnt} 
		\left[\der{\qmn{\rnf,\mP}}{\rfm}-\der{\qmn{\rnt,\mP}}{\rfm}\right]
		\\
		\notag
		\abs{\left.\pder{}{\rno}\fX_{\dinp}(\rno,\dout)\right\vert_{\rno=\rnb}} 
		&\leq \der{\smn{\rnb}}{\rfm}\left[\tfrac{1}{\rnb e}+ (\tfrac{1}{\mP(\dinp)^{\frac{1}{\rnt \wedge  1}}})^{\rnb}\ln \tfrac{1}{\mP(\dinp)^{\frac{1}{\rnt \wedge 1}}}\right]
		+(\tfrac{1}{\mP(\dinp)^{\frac{1}{\rnt \wedge  1}}})^{\rnb} 
		\tfrac{\abs{1-\rnb}}{\rnf-\rnt}
		\left[\der{\qmn{\rnf,\mP}}{\rfm}+\der{\qmn{\rnt,\mP}}{\rfm}\right]
		\\
		\notag
		&\leq \left(\der{\qmn{\rnf,\mP}}{\rfm}+\der{\qmn{\rnt,\mP}}{\rfm}\right)
		\left[\tfrac{1}{\rnt e}+ (\tfrac{1}{\mP(\dinp)^{\frac{1}{\rnt \wedge 1}}})^{\rnf}
		\left(\ln \tfrac{1}{\mP(\dinp)^{\frac{1}{\rnt \wedge 1}}}+\tfrac{1+\rnf}{\rnf-\rnt}\right)
		\right]
		&
		&\forall \rnb\in[\rnt,\rnf].
		\end{align}
		The expression on the right hand side is \(\rfm-\)integrable. 
		Thus as a result of \cite[Corollary 2.8.7]{bogachev} we have
		\begin{align}
		\notag
		\pder{}{\rno}\gX_{\dinp}(\rno)
		&=\int \pder{}{\rno}\fX_{\dinp}(\rno,\dout)\rfm(\dif{\dout}).
		\end{align}
		Furthermore, \(\pder{}{\rno}\gX_{\dinp}(\rno)\) is continuous by \cite[Corollary 2.8.7]{bogachev}
		because \(\pder{}{\rno}\fX_{\dinp}(\rno,\dout)\) is continuous in \(\rno\).
		Then \(\tfrac{\ln\gX_{\dinp}(\rno)}{\rno}\) is a continuous function on \([\rnt,\rnf]\) that 
		is continuously  differentiable on \((\rnt,\rnf)\).
		Then, as a result of mean value theorem \cite[Thm. 5.10]{rudin}
		we have 
		\begin{align}
		\label{eq:informationO:mean-uec-3}
		\abs{\tfrac{\ln\gX_{\dinp}(\rnf)}{\rnf}-\tfrac{\ln\gX_{\dinp}(\rnt)}{\rnt}} 
		&\leq (\rnf-\rnt) \sup\nolimits_{\rnb\in(\rnt,\rnf)} \abs{\left.\pder{}{\rno}\tfrac{\ln\gX_{\dinp}(\rno)}{\rno}\right\vert_{\rno=\rnb}}.
		\end{align}
		Using Lemma \ref{lem:meanO}-(\ref{meanO:measurebound}) and  the identity 
		\(\tau \ln \tfrac{1}{\tau}\leq \tfrac{1}{e}\) we get
		\begin{align}
		\notag
		\abs{\left.\pder{}{\rno}\tfrac{\ln\gX_{\dinp}(\rno)}{\rno}\right\vert_{\rno=\rnb}}
		&=\abs{\tfrac{-\ln\gX_{\dinp}(\rnb)}{\rnb^{2}}+\tfrac{1}{\rnb \gX_{\dinp}(\rnb)}
			\int \left.\pder{}{\rno}\fX_{\dinp}(\rno,\dout)\right\vert_{\rno=\rnb} \rfm(\dif{\dout})} 
		\\
		\label{eq:informationO:mean-uec-4}
		&\leq \abs{\tfrac{\ln \gX_{\dinp}(\rnb)}{\rnb^{2}}}
		+\tfrac{1}{\rno \gX_{\dinp}(\rnb)}
		\left[\tfrac{1}{\rnb e}+(\tfrac{1}{\mP(\dinp)^{\frac{1}{\rnt \wedge 1}}})^{\rnb} 
		\left(\ln \tfrac{1}{\mP(\dinp)^{\frac{1}{\rnt \wedge 1}}}
		+ \tfrac{\abs{1-\rnb}}{\rnf-\rnt}\lon{\qmn{\rnf,\mP}-\qmn{\rnt,\mP}}\right)\right].
		\end{align}		
We bound \(\ln \gX_{\dinp}(\rnb)\) using	
the definition of \(\gX_{\dinp}(\rnb)\)
together with  
Lemmas
\ref{lem:meanO}-(\ref{meanO:measurebound}) 
and \ref{lem:divergence-RM}:
\begin{align}
\notag
\abs{\ln \gX_{\dinp}(\rnb)} 
&= \abs{\rnb-1}
\RD{\rnb}{\Wm(\dinp)}{\tfrac{\rnf-\rnb}{\rnf-\rnt}\qmn{\rnt,\mP}+\tfrac{\rnb-\rnt}{\rnf-\rnt}\qmn{\rnf,\mP}}
\\
\notag
&\leq \abs{\rnb-1}
\RD{\rnb}{\Wm(\dinp)}{
\left[
\tfrac{\rnf-\rnb}{\rnf-\rnt}[\mP(\dinp)]^{\frac{1}{1\wedge \rnt}}
+\tfrac{\rnb-\rnt}{\rnf-\rnt}[\mP(\dinp)]^{\frac{1}{1\wedge \rnf}}
\right]
\Wm(\dinp)}
\\
\notag
&\leq (\rnb\vee 1)
\RD{\rnb}{\Wm(\dinp)}{[\mP(\dinp)]^{\frac{1}{1\wedge \rnt}}\Wm(\dinp)}
\\
\label{eq:informationO:mean-uec-5}
&\leq \tfrac{\rnf\vee 1}{\rnt \wedge  1} \ln \tfrac{1}{\mP(\dinp)}
&
&\forall \rnb\in[\rnt,\rnf].
\end{align}

		The Augustin information is nondecreasing in its order by part (\ref{informationO:continuity}).
		Thus Lemmas \ref{lem:divergence-pinsker} and \ref{lem:information} imply that
		\begin{align}
		\label{eq:informationO:mean-uec-6}
		\lon{\qmn{\rnf,\mP}-\qmn{\rnt,\mP}}
		&\leq \sqrt{\tfrac{2}{\rnt \wedge 1} (\RMI{\rnf}{\mP}{\Wm}-\RMI{\rnt}{\mP}{\Wm})}.
		\end{align}
		On the other hand, one can confirm by substitution that
		\begin{align}
		\label{eq:informationO:mean-uec-7}
		\tau_{\dinp}(\rnf)-\tau_{\dinp}(\rnt)  
		&=\tfrac{\ln\gX_{\dinp}(\rnf)}{\rnf}-\tfrac{\ln\gX_{\dinp}(\rnt)}{\rnt}.
		\end{align}
		Then the continuity of \(\tau_{\dinp}(\rno)\) in \(\rno\) is implied by 
		\eqref{eq:informationO:mean-uec-3},
		\eqref{eq:informationO:mean-uec-4},
		\eqref{eq:informationO:mean-uec-5},
		\eqref{eq:informationO:mean-uec-6},
		\eqref{eq:informationO:mean-uec-7}
		and
		the continuity of the Augustin information in the order established in part (\ref{informationO:continuity}).

		\item [(\ref{lem:informationO}-\ref{informationO:differentiability})]
		\(\CRD{\rnt}{\Wm}{\qmn{\rnf,\mP}}{\mP}\geq \RMI{\rnt}{\mP}{\Wm}\) for any 
		\(\rnf\in\reals{+}\) and \(\rnt\in\reals{+}\) by the definition of the Augustin 
		information. Then the differentiability of \(\CRD{\rno}{\Wm}{\mQ}{\mP}\) in \(\rno\)
		established in Lemma \ref{lem:analyticity}
		 implies that
		\begin{align}
		\notag
		\lim\nolimits_{\rnt\downarrow\rnf}
		\tfrac{\RMI{\rnt}{\mP}{\Wm}-\RMI{\rnf}{\mP}{\Wm}}{\rnt-\rnf}
		&\leq 
		\lim\nolimits_{\rnt\downarrow\rnf}
		\tfrac{\CRD{\rnt}{\Wm}{\qmn{\rnf,\mP}}{\mP}-\CRD{\rnf}{\Wm}{\qmn{\rnf,\mP}}{\mP}}{\rnt-\rnf}
		\\
		\label{eq:informationO:differentiability-1}
		&=\left.\pder{}{\rno}\CRD{\rno}{\Wm}{\qmn{\rnf,\mP}}{\mP}\right\vert_{\rno=\rnf} 
		\\
		\notag
		\lim\nolimits_{\rnt\uparrow\rnf}
		\tfrac{\RMI{\rnt}{\mP}{\Wm}-\RMI{\rnf}{\mP}{\Wm}}{\rnt-\rnf}
		&\geq 
		\lim\nolimits_{\rnt\downarrow\rnf}
		\tfrac{\CRD{\rnt}{\Wm}{\qmn{\rnf,\mP}}{\mP}-\CRD{\rnf}{\Wm}{\qmn{\rnf,\mP}}{\mP}}{\rnt-\rnf}
		\\
		\label{eq:informationO:differentiability-2}
		&=\left.\pder{}{\rno}\CRD{\rno}{\Wm}{\qmn{\rnf,\mP}}{\mP}\right\vert_{\rno=\rnf}.
		\end{align}
		Similarly,\footnote{For \(\rnf\in(1,\infty)\), we can also use \eqref{eq:lem:information:alternative:def} 
			of  Lemma \ref{lem:information}-(\ref{information:alternative}) to establish reverse inequalities for 
			\eqref{eq:informationO:differentiability-1}  and \eqref{eq:informationO:differentiability-2}.}
		\(\CRD{\rnf}{\Wm}{\qmn{\rnt,\mP}}{\mP}\geq \RMI{\rnf}{\mP}{\Wm}\) for any 
		\(\rnf\in\reals{+}\) and \(\rnt\in\reals{+}\) by the definition of the Augustin 
		information. Hence,
		\begin{align}
		\label{eq:informationO:differentiability-3}
		\lim\nolimits_{\rnt\downarrow\rnf}
		\tfrac{\RMI{\rnt}{\mP}{\Wm}-\RMI{\rnf}{\mP}{\Wm}}{\rnt-\rnf}
		&\geq 
		\lim\nolimits_{\rnt\downarrow\rnf}
		\tfrac{\CRD{\rnt}{\Wm}{\qmn{\rnt,\mP}}{\mP}-\CRD{\rnf}{\Wm}{\qmn{\rnt,\mP}}{\mP}}{\rnt-\rnf}
		\\
		\label{eq:informationO:differentiability-4}
		\lim\nolimits_{\rnt\uparrow\rnf}
		\tfrac{\RMI{\rnt}{\mP}{\Wm}-\RMI{\rnf}{\mP}{\Wm}}{\rnt-\rnf}
		&\leq 
		\lim\nolimits_{\rnt\uparrow\rnf}
		\tfrac{\CRD{\rnt}{\Wm}{\qmn{\rnt,\mP}}{\mP}-\CRD{\rnf}{\Wm}{\qmn{\rnt,\mP}}{\mP}}{\rnt-\rnf}.
		\end{align}
		For any \(\delta\in(0,\rnf)\) by Lemma \ref{lem:meanO}-(\ref{meanO:measurebound}) and 
		Lemma \ref{lem:divergence-RM} we have 
		\begin{align}
		\label{eq:informationO:differentiability-5}
		\RD{\rno}{\Wm(\dinp)}{\qmn{\rnt,\mP}}
		&\leq \tfrac{1}{(\rnf-\delta) \wedge 1}\ln \tfrac{1}{\mP(\dinp)}
		&
		&\forall \rnt:\abs{\rnt-\rnf}<\delta, \rno\in\reals{+}. 
		\end{align}
		Then as a result of Lemma \ref{lem:uniform-analyticity}, there 
		exists\footnote{Note that, this is not just the Taylor expansion of \(\RD{\rno}{\mW}{\mQ}\)
			around \(\rno=\rnf\) for a given \((\mW,\mQ)\) pair. 
			Lemma \ref{lem:analyticity} allows us to apply the Taylor expansion for a family of 
			\((\mW,\mQ)\) pairs around \(\rno=\rnf\) simultaneously if we can bound 
			\(\RD{\rnb}{\mW}{\mQ}\) uniformly for all \((\mW,\mQ)\)'s for some \(\rnb>\rnf\).}
		a \(K_{\rnf,\mP}>0\) such that for \(\rnt\) close enough to \(\rnf\) we have
		\begin{align}
		\label{eq:informationO:differentiability-6}
		\abs{\CRD{\rnt}{\Wm}{\qmn{\rnt,\mP}}{\mP}-\CRD{\rnf}{\Wm}{\qmn{\rnt,\mP}}{\mP}-
			(\rnt-\rnf)\left.\pder{}{\rno}\CRD{\rno}{\Wm}{\qmn{\rnt,\mP}}{\mP}\right\vert_{\rno=\rnf}} 
		&\leq K_{\rnf,\mP} \abs{\rnt-\rnf}^{2}.
		\end{align}
		We show in the following that \(\left.\pder{}{\rno}\CRD{\rno}{\Wm}{\qmn{\rnt,\mP}}{\mP}\right\vert_{\rno=\rnf}\) 
		is a continuous function of \(\rnt\), i.e.
		\begin{align}
		\label{eq:informationO:differentiability-7}
		\lim\nolimits_{\rnt\to\rnb}
		\left.\pder{}{\rno}\CRD{\rno}{\Wm}{\qmn{\rnt,\mP}}{\mP}\right\vert_{\rno=\rnf}
		&=\left.\pder{}{\rno}\CRD{\rno}{\Wm}{\qmn{\rnb,\mP}}{\mP}\right\vert_{\rno=\rnf}
		&
		&\forall\rnf,\rnb\in\reals{+}.
		\end{align}
		Using \eqref{eq:informationO:differentiability-6} and \eqref{eq:informationO:differentiability-7}
		we get
		\begin{align}
		\label{eq:informationO:differentiability-8}
		\lim\nolimits_{\rnt\to\rnf}
		\tfrac{\CRD{\rnt}{\Wm}{\qmn{\rnt,\mP}}{\mP}-\CRD{\rnf}{\Wm}{\qmn{\rnt,\mP}}{\mP}}{\rnt-\rnf}
		&= 
		\left.\pder{}{\rno}\CRD{\rno}{\Wm}{\qmn{\rnf,\mP}}{\mP}\right\vert_{\rno=\rnf}.
		\end{align}
		Differentiability of the Augustin information and \eqref{eq:lem:informationO:differentiability}
		follow from
		\eqref{eq:informationO:differentiability-1},
		\eqref{eq:informationO:differentiability-2},
		\eqref{eq:informationO:differentiability-3}, 
		\eqref{eq:informationO:differentiability-4}, 
		and
		\eqref{eq:informationO:differentiability-8}.
		
		In order to establish the continuity of 
		\(\left.\pder{}{\rno}\CRD{\rno}{\Wm}{\qmn{\rnt,\mP}}{\mP}\right\vert_{\rno=\rnf}\) in \(\rnt\),
		i.e. \eqref{eq:informationO:differentiability-7},
		let us first recall that the expression for the derivative of the \renyi divergence given in \eqref{eq:divergence-derivative-first}:
		\begin{align}
		\notag
		\left.\pder{}{\rno}\CRD{\rno}{\Wm}{\qmn{\rnt,\mP}}{\mP}\right\vert_{\rno=\rnf}
		&=\begin{cases}
		\tfrac{1}{(\rnf-1)^2}
		\sum\nolimits_{\dinp}\mP(\dinp)
		\int \der{\Wma{\rnf}{\qmn{\rnt,\mP}}(\dinp)}{\qmn{1,\mP}} 
		\left(\ln \der{\Wma{\rnf}{\qmn{\rnt,\mP}}(\dinp)}{\Wm(\dinp)}\right) \qmn{1,\mP}(\dif{\dout})
		&\rnf\neq 1
		\\
		\sum\nolimits_{\dinp}\tfrac{\mP(\dinp)}{2}
		\left(\int \der{\Wm(\dinp)}{\qmn{1,\mP}}
		\left(\ln\der{\Wm(\dinp)}{\qmn{\rnt,\mP}}
		\right)^{2} \qmn{1,\mP}(\dif{\dout})-
		[\RD{1}{\Wm(\dinp)}{\qmn{\rnt,\mP}}]^2
		\right)
		&\rnf= 1
		\end{cases}.
		\end{align}
		Recall that, 
		\begin{align}
		\notag
		\RD{\rnf}{\Wm(\dinp)}{\qmn{\rnt,\mP}}
		&=
		\begin{cases}
		\tfrac{1}{\rnf-1}\ln \int 
		\left(\der{\Wm(\dinp)}{\qmn{1,\mP}}\right)^{\rnf}
		\left(\der{\qmn{\rnt,\mP}}{\qmn{1,\mP}}\right)^{1-\rnf}
		\qmn{1,\mP}(\dif{\dout})
		&\rnf\neq 1
		\\
		\int\der{\Wm(\dinp)}{\qmn{1,\mP}} 
		\left(\ln \der{\Wm(\dinp)}{\qmn{1,\mP}}-\ln \der{\qmn{\rnt,\mP}}{\qmn{1,\mP}}\right)
		\qmn{1,\mP}(\dif{\dout})
		&\rnf=1
		\end{cases}.
		\end{align}
		Then \(\RD{\rnf}{\Wm(\dinp)}{\qmn{\rnt,\mP}}\) is continuous in \(\rnt\) for any \(\rnf\in\reals{+}\) 
		by \cite[Corollary 2.8.7-(i)]{bogachev} because 
		\(\{\ln\der{\qmn{\rnt,\mP}}{\qmn{1,\mP}}\}_{\dout\in\outS}\) is equicontinuous function of \(\rnt\)
		by part (\ref{informationO:mean-uec})
		and \(\abs{\ln\der{\qmn{\rnt,\mP}}{\qmn{1,\mP}}}\leq \tfrac{\abs{\rnt-1}}{\rnt}\ln \tfrac{1}{\min\nolimits_{\dinp:\mP(\dinp)>0}\mP(\dinp)}\)
		by Lemma \ref{lem:meanO}-(\ref{meanO:boundedlogRND}).
		On the other hand,
		\begin{align}
		\notag
		\ln \der{\Wma{\rnf}{\qmn{\rnt,\mP}}(\dinp)}{\qmn{1,\mP}}
		&=\rnf \ln \der{\Wm(\dinp)}{\qmn{1,\mP}}+(1-\rnf)\ln \der{\qmn{\rnt,\mP}}{\qmn{1,\mP}}
		+(1-\rnf)\RD{\rnf}{\Wm(\dinp)}{\qmn{\rnt,\mP}}.
		\end{align}
		Thus \(\{\der{\Wma{\rnf}{\qmn{\rnt,\mP}}(\dinp)}{\qmn{1,\mP}}\}_{\dout\in\outS}\) is 
		equicontinuous in \(\rnt\) because \(\{\ln \der{\qmn{\rnt,\mP}}{\qmn{1,\mP}}\}_{\dout\in\outS}\)
		is equicontinuous and \(\RD{\rnf}{\Wm(\dinp)}{\qmn{\rnt,\mP}}\) is continuous in 
		\(\rnt\).
		Furthermore, using Lemma \ref{lem:divergence-RM}, Lemma \ref{lem:meanO}-(\ref{meanO:measurebound}),
		and the identity \(\tau \ln \tfrac{1}{\tau}\leq \tfrac{1}{e}\) we obtain the following bounds
		\begin{align}
		\label{eq:informationO:differentiability-9}
		\der{\Wma{\rnf}{\qmn{\rnt,\mP}}(\dinp)}{\qmn{1,\mP}} 
		\abs{\ln \der{\Wma{\rnf}{\qmn{\rnt,\mP}}(\dinp)}{\Wm(\dinp)}} 
		&\leq \der{\Wm(\dinp)}{\qmn{1,\mP}} \left(
		\IND{\frac{\dif{\Wma{\rnf}{\qmn{\rnt,\mP}}}(\dinp)}{\dif{\Wm}(\dinp)}\leq 1}\tfrac{1}{e}
		+
		\IND{\frac{\dif{\Wma{\rnf}{\qmn{\rnt,\mP}}}(\dinp)}{\dif{\Wm}(\dinp)}> 1}
		\left[\tfrac{1}{\mP(\dinp)}\right]^{\frac{\rnf-1}{\rnt \wedge 1}}
		\ln
		\left[\tfrac{1}{\mP(\dinp)}\right]^{\frac{\rnf-1}{\rnt \wedge 1}}		
		\right)
		&
		&\mbox{if~}\rnf\in[1,\infty),
		\\
		\label{eq:informationO:differentiability-10}
		\der{\Wma{\rnf}{\qmn{\rnt,\mP}}(\dinp)}{\qmn{1,\mP}} 
		\abs{\ln \der{\Wma{\rnf}{\qmn{\rnt,\mP}}(\dinp)}{\Wm(\dinp)}} 
		&\leq \der{\Wm(\dinp)}{\qmn{1,\mP}}
		\IND{\frac{\dif{\Wma{\rnf}{\qmn{\rnt,\mP}}}(\dinp)}{\dif{\Wm}(\dinp)}\leq 1}\tfrac{1}{e}
		+
		\der{\qmn{\rnt,\mP}}{\qmn{1,\mP}}
		\IND{\frac{\dif{\Wma{\rnf}{\qmn{\rnt,\mP}}}(\dinp)}{\dif{\Wm}(\dinp)}> 1}
		\tfrac{1-\rnf}{\rnf e}
		\left[\tfrac{1}{\mP(\dinp)}\right]^{\frac{1}{\rnt \wedge 1}}
		&
		&\mbox{if~}\rnf\in(0,1).
		\end{align}
		Using \(\abs{\ln\der{\qmn{\rnt,\mP}}{\qmn{1,\mP}}}\leq \tfrac{\abs{\rnt-1}}{\rnt}\ln \tfrac{1}{\min\nolimits_{\dinp:\mP(\dinp)>0}\mP(\dinp)}\)
		---i.e. Lemma \ref{lem:meanO}-(\ref{meanO:boundedlogRND})---
		in \eqref{eq:informationO:differentiability-10},
		we get
\begin{align}
\label{eq:informationO:differentiability-10-new}
\der{\Wma{\rnf}{\qmn{\rnt,\mP}}(\dinp)}{\qmn{1,\mP}} 
\abs{\ln \der{\Wma{\rnf}{\qmn{\rnt,\mP}}(\dinp)}{\Wm(\dinp)}} 
&\leq \der{\Wm(\dinp)}{\qmn{1,\mP}}\tfrac{1}{e}
+
\left[\min\limits_{\dinp:\mP(\dinp)>0}\mP(\dinp)
\right]^{\frac{-\abs{\rnt-1}}{\rnt}}
\tfrac{1}{\rnf e}
\left[\tfrac{1}{\mP(\dinp)}\right]^{\frac{1}{\rnt \wedge 1}}
&
&\mbox{if~}\rnf\in(0,1).
\end{align}
Using \eqref{eq:informationO:differentiability-9} and 
\eqref{eq:informationO:differentiability-10-new}, we get
the following bound for all 
\(\rnf\in\reals{+}\) and  \(\rnt\in[\mA,\mB]\)
\begin{align}
\notag
\der{\Wma{\rnf}{\qmn{\rnt,\mP}}(\dinp)}{\qmn{1,\mP}} 
\abs{\ln \der{\Wma{\rnf}{\qmn{\rnt,\mP}}(\dinp)}{\Wm(\dinp)}} 
&\leq \der{\Wm(\dinp)}{\qmn{1,\mP}} \left(
\tfrac{1}{e}
+\left[\tfrac{1}{\mP(\dinp)}\right]^{\frac{\rnf}{\mA \wedge 1}}
\ln
\left[\tfrac{1}{\mP(\dinp)}\right]^{\frac{\rnf}{\mA \wedge 1}}		
\right)
+
\left[\min\limits_{\dinp:\mP(\dinp)>0}\mP(\dinp)
\right]^{-1-\frac{1}{\mA}}
\tfrac{1}{\rnf e}
\left[\tfrac{1}{\mP(\dinp)}\right]^{\frac{1}{\rnt \wedge 1}}
.
\end{align}
Then
		\(\left\{\der{\Wma{\rnf}{\qmn{\rnt,\mP}}(\dinp)}{\qmn{1,\mP}} \abs{\ln \der{\Wma{\rnf}{\qmn{\rnt,\mP}}(\dinp)}{\Wm(\dinp)}}\right\}_{\rnt\in[\mA,\mB]}\)
		is bounded from above by a \(\qmn{1,\mP}\)-integrable function 
		for any closed interval \([\mA,\mB]\subset\reals{+}\). 
		Thus \(\CRD{1}{\Wma{\rnf}{\qmn{\rnt,\mP}}}{\Wm}{\mP}\)
		is a continuous function of \(\rnt\) by  \cite[Corollary 2.8.7-(i)]{bogachev}
		for all \(\rnf\in\reals{+}\).
		Then  \(\left.\pder{}{\rno}\CRD{\rno}{\Wm}{\qmn{\rnt,\mP}}{\mP}\right\vert_{\rno=\rnf}\)
		is continuous in \(\rnt\) for \(\rnf\in\reals{+}\setminus{1}\).
		The continuity of 
		\(\left.\pder{}{\rno}\CRD{\rno}{\Wm}{\qmn{\rnt,\mP}}{\mP}\right\vert_{\rno=1}\) in \(\rnt\) 
		follows from  the continuity of \(\RD{1}{\Wm(\dinp)}{\qmn{\rnt,\mP}}\) and 
		\cite[Corollary 2.8.7-(i)]{bogachev} via the following bound, which can be established using
		the identity \(\tau (\ln \tau)^{2}\IND{\tau\in (0,1]}\leq \tfrac{4}{e^{2}}\) and
		Lemma \ref{lem:meanO}-(\ref{meanO:measurebound}),
		\begin{align}
\label{eq:informationO:differentiability-11-new}		
\der{\Wm(\dinp)}{\qmn{1,\mP}}
		\left(\ln\der{\Wm(\dinp)}{\qmn{\rnt,\mP}}
		\right)^{2}
		&\leq \der{\qmn{\rnt,\mP}}{\qmn{1,\mP}} \tfrac{4}{e^{2}}+
		\der{\Wm(\dinp)}{\qmn{1,\mP}}
		\left(\tfrac{\ln \mP(\dinp)}{\rnt \wedge 1}\right)^{2}.
		\end{align}
		
		Now we are left with establishing the continuity of the derivative of 
		the Augustin information.
		Since  \(\{\ln \der{\qmn{\rno,\mP}}{\qmn{1,\mP}}\}_{\dout\in\outS}\) is equicontinuous in \(\rno\)
		by part (\ref{informationO:mean-uec}),
		for any \(\epsilon>0\) there exists a \(\delta\) such that 
		\begin{align}
		\notag
		e^{-\epsilon}\qmn{\rnf,\mP}\leq\qmn{\rnt,\mP}
		&\leq e^{\epsilon}\qmn{\rnf,\mP}
		&
		&\forall \rnt:\abs{\rnt-\rnf}<\delta. 
		\end{align}
		On the other hand \((\mP(\dinp))^{\frac{1}{\rnf \wedge 1}}\Wm(\dinp)\leq \qmn{\rnf,\mP}\) by Lemma \ref{lem:meanO}-(\ref{meanO:measurebound}).
		Then as a result of Lemma \ref{lem:divergence-RM}
		\begin{align}
		\notag
		\RD{\rno}{\Wm(\dinp)}{\qmn{\rnt,\mP}}
		&\leq \tfrac{1}{\rnf \wedge 1} \ln \tfrac{1}{\mP(\dinp)}+\epsilon 
		&
		&\forall \rnt:\abs{\rnf-\rnt}<\delta,~\forall \rno\in\reals{+}. 
		\end{align}
		Hence, Lemma \ref{lem:uniform-analyticity} implies the existence of  
		a \(\tau\in\reals{+}\) that does not depend on \(\rnt\)
		such that
		\begin{align}
		\label{eq:informationO:differentiability-11}
		\abs{\left.\pder{^{\knd}\CRD{\rno}{\Wm}{\qmn{\rnt,\mP}}{\mP}}{\rno^{\knd}}\right\vert_{\rno=\rnf}} 
		&\leq \knd!\tau^{\knd+1} \knd
		&
		&\forall \rnt:\abs{\rnf-\rnt}<\delta.
		\end{align}
		Then \(\limsup\nolimits_{\knd\to\infty}\abs{\tfrac{1}{\knd!}\left.\pder{^{\knd}}{\rno^{\knd}}
			\left(\pder{}{\rno} \CRD{\rno}{\Wm}{\qmn{\rnt,\mP}}{\mP}\right)\right\vert_{\rno=\rnf}}^{1/\knd}\leq \tau\).
		Thus the radius of convergence of the Taylor's expansion of  \(\pder{\CRD{\rno}{\Wm}{\qmn{\rnt,\mP}}{\mP}}{\rno}\)
		around \(\rno=\rnf\) is  at least \(\tfrac{1}{\tau}\) for all \(\rnt\in [\rnf-\delta,\rnf+\delta]\)
		by Hadamard's formula \cite[Lemma 1.1.8]{krantzparks}.
		Furthermore, we can use \eqref{eq:informationO:differentiability-11} to bound higher order derivatives:
		\begin{align}
		\notag
		\abs{\left. \pder{\CRD{\rno}{\Wm}{\qmn{\rnt,\mP}}{\mP}}{\rno}\right\vert_{\rno=\rnb}
			-\left.\pder{\CRD{\rno}{\Wm}{\qmn{\rnt,\mP}}{\mP}}{\rno}\right\vert_{\rno=\rnf}}
		&\leq \sum\nolimits_{\ind=1}^{\infty} \tfrac{\abs{\rnb-\rnf}^{\ind}}{\ind!}
		\abs{\left.
			\pder{^{\ind+1}\CRD{\rno}{\Wm}{\qmn{\rnt,\mP}}{\mP}}{\rno^{\ind+1}}\right\vert_{\rno=\rnf}} 
		\\
		\notag
		&\leq \tau^{2} \sum\nolimits_{\ind=1}^{\infty} \left(\ind^{2}+2\ind+1\right)\abs{\rnb-\rnf}^{\ind}\tau^{\ind}
		&
		&\forall \rnt:\abs{\rnt-\rnf}\leq\delta,~\forall\rnb:\abs{\rnb-\rnf}<\tfrac{1}{\tau}.
		\end{align}
		Using identities 
		\(\sum_{\ind=1}^{\infty}\xi^{\ind}\leq \sum_{\ind=1}^{\infty}\ind \xi^{\ind}\leq 
		\sum_{\ind=1}^{\infty}\ind^{2}\xi^{\ind}\)
		for \(\xi\geq 0\)
		and
		\(\sum_{\ind=1}^{\infty}\ind^{2} \xi^{\ind}=\tfrac{(1+\xi)\xi}{(1-\xi)^3}\) for \(\abs{\xi}<1\) we get,
		\begin{align}
		\notag
		\abs{\left. \pder{\CRD{\rno}{\Wm}{\qmn{\rnt,\mP}}{\mP}}{\rno}\right\vert_{\rno=\rnb}
			-\left.\pder{\CRD{\rno}{\Wm}{\qmn{\rnt,\mP}}{\mP}}{\rno}\right\vert_{\rno=\rnf}}
		&\leq 4\tau^{3}\abs{\rnb-\rnf}\tfrac{1+\tau\abs{\rnb-\rnf}}{(1-\tau\abs{\rnb-\rnf})^{3}}
		&
		&\forall \rnt:\abs{\rnt-\rnf}\leq\delta,~\forall\rnb:\abs{\rnb-\rnf}<\tfrac{1}{\tau}.
		\end{align}
		Then using \eqref{eq:lem:informationO:differentiability}  we get
		\begin{align}
		\notag
		\abs{\left.\pder{}{\rno}\RMI{\rno}{\mP}{\Wm}\right\vert_{\rno=\rnt}
			-\left.\pder{\CRD{\rno}{\Wm}{\qmn{\rnt,\mP}}{\mP}}{\rno}\right\vert_{\rno=\rnf}}
		&\leq 4\tau^{3}\abs{\rnt-\rnf}\tfrac{1+\tau\abs{\rnt-\rnf}}{(1-\tau\abs{\rnt-\rnf})^{3}}
		&
		&\forall \rnt:\abs{\rnt-\rnf}\leq\delta\wedge\tfrac{1}{\tau}.
		\end{align}
		Hence,
		\(\lim\nolimits_{\rnt\to\rnf} \left.\pder{}{\rno}\RMI{\rno}{\mP}{\Wm}\right\vert_{\rno=\rnt}
		=\lim\nolimits_{\rnt\to\rnf} \left.\pder{\CRD{\rno}{\Wm}{\qmn{\rnt,\mP}}{\mP}}{\rno}\right\vert_{\rno=\rnf}\),
		if the latter limit exists.
		However, we have already established the existence of that limit in order to calculate 
		the derivative of the Augustin information: it is equal to 
		\(\left.\pder{}{\rno}\RMI{\rno}{\mP}{\Wm}\right\vert_{\rno=\rnf}\).
		Thus the Augustin information is continuously differentiable in the order.

\item [(\ref{lem:informationO}-\ref{informationO:monotonicityofharoutunianinformation})]
Let us start with analyzing the case when \((\rno-1)\RMI{\rno}{\mP}{\Wm}\) is strictly convex
in \(\rno\).
The chain rule for derivatives implies
\begin{align}
\notag
\pder{}{\mS} \mS \RMI{\frac{1}{1+\mS}}{\mP}{\Wm}
&=\RMI{\frac{1}{1+\mS}}{\mP}{\Wm}
+\mS\tfrac{(-1)}{(1+\mS)^{2}}
\left.\pder{}{\rno} \RMI{\rno}{\mP}{\Wm}\right\vert_{\rno=\frac{1}{1+\mS}}.
\end{align} 
Using  \eqref{eq:lem:informationO:differentiability-alt},
\eqref{eq:informationO:differentiability-11-new},		
and the fact that
 \(\EXS{\Wm(\dinp)}{\left(\ln \der{\Wm(\dinp)}{\qmn{1,\mP}}-\RD{1}{\Wm(\dinp)}{\qmn{1,\mP}}\right)^{2}}
\leq \EXS{\Wm(\dinp)}{\left(\ln \der{\Wm(\dinp)}{\qmn{1,\mP}}\right)^{2}}
\),
we get
\begin{align}
\notag
\left.\pder{}{\mS} \mS \RMI{\frac{1}{1+\mS}}{\mP}{\Wm}\right\vert_{\mS=0}
&=
\begin{cases}
\RMI{1}{\mP}{\Wm}
&\mS=0
\\
\RMI{\frac{1}{1+\mS}}{\mP}{\Wm}
-\tfrac{1}{\mS}\CRD{1}{\Wma{\frac{1}{1+\mS}}{\qmn{\frac{1}{1+\mS},\mP}}}{\Wm}{\mP}
&\mS\in(-1,0)\cup(0,\infty)
\end{cases}.
\end{align}
Then as a result of \eqref{eq:lem:information:alternative:opt}, 
we can assert the following expression for all \(\mS\in(-1,\infty)\)
\begin{align}
\label{eq:informationO:monotonicityofharoutunianinformation-1}
\pder{}{\mS} \mS \RMI{\frac{1}{1+\mS}}{\mP}{\Wm}
&=\RMI{1}{\mP}{\Wma{\frac{1}{1+\mS}}{\qmn{\frac{1}{1+\mS},\mP}}}.
\end{align} 

Then the continuous differentiability of \(\RMI{\rno}{\mP}{\Wm}\) in \(\rno\) on \(\reals{+}\),
established in part (\ref{informationO:differentiability}),
implies the continuity of \(\RMI{1}{\mP}{\Wma{\rno}{\qmn{\rno,\mP}}}\) in
\(\rno\) on \(\reals{+}\).

In order to prove that \(\RMI{1}{\mP}{\Wma{\rno}{\qmn{\rno,\mP}}}\) is monotonically increasing 
in \(\rno\) on \(\reals{+}\),
note that the strict convexity of \((\rno-1)\RMI{\rno}{\mP}{\Wm}\) in \(\rno\) on \(\reals{+}\)
is equivalent to the strict concavity of \(\mS \RMI{\frac{1}{1+\mS}}{\mP}{\Wm}\) in \(\mS\) on \((-1,\infty)\)
because the inequality 
\begin{align}
\notag
(\rno_{\beta}-1)\RMI{\rno_{\beta}}{\mP}{\Wm}
<\beta(\rno_{1}-1)\RMI{\rno_{1}}{\mP}{\Wm}
+(1-\beta)(\rno_{0}-1)\RMI{\rno_{0}}{\mP}{\Wm}
\end{align}
holds for \(\rno_{0},\rno_{1}\in\reals{+}\), \(\beta\in(0,1)\),
and \(\rno_{\beta}=\beta\rno_{1}+(1-\beta)\rno_{0}\)
iff the inequality
\begin{align}
\notag
\smn{\mean}\RMI{\frac{1}{1+\smn{\mean}}}{\mP}{\Wm}
>\mean\smn{1}\RMI{\frac{1}{1+\smn{1}}}{\mP}{\Wm}
+(1-\mean)\smn{0}\RMI{\frac{1}{1+\smn{0}}}{\mP}{\Wm}.
\end{align}
holds for \(\smn{0}=\frac{1-\rno_{0}}{\rno_{0}}\), \(\smn{1}=\frac{1-\rno_{1}}{\rno_{1}}\)
\(\smn{\mean}=\mean\smn{1}+(1-\mean)\smn{0}\)
and \(\mean=\tfrac{\beta \rno_{1}}{\rno_{\beta}}\).

On the other hand for any strictly concave function \(\fX(\cdot)\)
and \(\smn{1},\smn{2},\smn{3},\smn{4}\) satisfying 
\(\smn{1}<\smn{2}<\smn{3}<\smn{4}\) we 
have\footnote{Note that \(\fX(\smn{2})>\tfrac{\smn{3}-\smn{2}}{\smn{3}-\smn{1}} \fX(\smn{1}) 
	+\tfrac{\smn{2}-\smn{1}}{\smn{3}-\smn{1}}\fX(\smn{3})\)
	implies \(\tfrac{\fX(\smn{2})-\fX(\smn{1})}{\smn{2}-\smn{1}}>\tfrac{\fX(\smn{3})-\fX(\smn{2})}{\smn{3}-\smn{2}}\).}
\begin{align}
\notag
\tfrac{\fX(\smn{2})-\fX(\smn{1})}{\smn{2}-\smn{1}}
>\tfrac{\fX(\smn{4})-\fX(\smn{3})}{\smn{4}-\smn{3}}.
\end{align}
Thus \(\RMI{1}{\mP}{\Wma{\frac{1}{1+\mS}}{\qmn{\frac{1}{1+\mS},\mP}}}\) is a decreasing 
function of \(\mS\)  on \((-1,\infty)\)
by \eqref{eq:informationO:monotonicityofharoutunianinformation-1}
and the definition of the derivative
because \(\mS \RMI{\frac{1}{1+\mS}}{\mP}{\Wm}\) is strictly concave in \(\mS\).
Hence \(\RMI{1}{\mP}{\Wma{\rno}{\qmn{\rno,\mP}}}\) is an increasing 
function of \(\rno\)  on \(\reals{+}\).

If \((\rno-1)\RMI{\rno}{\mP}{\Wm}\) is not strictly convex
in \(\rno\) then there exists a \(\gamma:\inpS\to[1,\infty)\) satisfying
\(\der{\Wm(\dinp)}{\qmn{1,\mP}}=\gamma(\dinp)\)  \(\Wm(\dinp)\)-a.s. for all 
\(\dinp\in\supp{\mP}\) and \(\qmn{\rno,\mP}=\qmn{1,\mP}\) for all \(\rno\in\reals{+}\)
by part (\ref{informationO:strictconvexity}).
Thus \(\Wma{\rno}{\qmn{\rno,\mP}}(\dinp)=\Wm(\dinp)\) and \(\der{\Wma{\rno}{\qmn{\rno,\mP}}(\dinp)}{\qmn{\rno,\mP}}=\gamma(\dinp)\)
for all \(\dinp\in\supp{\mP}\). Consequently \(\RMI{1}{\mP}{\Wma{\rno}{\qmn{\rno,\mP}}}=\sum_{\dinp}\mP(\dinp)\ln \gamma(\dinp)\).

\item [(\ref{lem:informationO}-\ref{informationO:limitofharoutunianinformation})]
	Let us define \(\RMI{0}{\mP}{\Wm}\) to be \(\lim_{\rno\downarrow0}\RMI{\rno}{\mP}{\Wm}\),
	such a limit exists because \(\RMI{\rno}{\mP}{\Wm}\) is non-decreasing
	function of  \(\rno\) on \(\reals{+}\).
	Then \((\rno\!-\!1)\RMI{\rno}{\mP}{\Wm}\) is convex in \(\rno\) on \([0,\infty)\),
	as well.
	Thus for any \(\rno\in\reals{+}\) and \(\epsilon\in\reals{+}\)
	\begin{align}
	\notag
	\tfrac{(\rno-1)\RMI{\rno}{\mP}{\Wm}+\RMI{0}{\mP}{\Wm}}{\rno}
	\leq \tfrac{(\rno+\epsilon-1)\RMI{\rno+\epsilon}{\mP}{\Wm}-(\rno-1)\RMI{\rno}{\mP}{\Wm}}{\epsilon}
	\end{align}
	by \cite[Proposition 6.3.2]{dudley}.
	Taking the limits as \(\epsilon\downarrow0\) 
	and invoking \eqref{eq:lem:informationO:differentiability-alt}
	we get
	\begin{align}
	\notag
	\tfrac{(\rno-1)\RMI{\rno}{\mP}{\Wm}+\RMI{0}{\mP}{\Wm}}{\rno}
	\leq \RMI{\rno}{\mP}{\Wm}+\tfrac{1}{\rno-1}\CRD{1}{\Wma{\rno}{\qmn{\rno,\mP}}}{\Wm}{\mP}.
	\end{align}
	Thus \eqref{eq:lem:information:alternative:opt} implies
	\begin{align}
	\notag
	\RMI{0}{\mP}{\Wm}
	&\leq \RMI{1}{\mP}{\Wma{\rno}{\qmn{\rno,\mP}}}.
	\end{align}
	On the other hand
	for all \(\rno\in(0,1)\)
	the non-negativity of the \renyi divergence and \eqref{eq:lem:information:alternative:opt}
	implies 
	\begin{align}
	\notag
	\RMI{1}{\mP}{\Wma{\rno}{\qmn{\rno,\mP}}}
	&\leq \RMI{\rno}{\mP}{\Wm}.
	\end{align}
	Hence \(\lim_{\rno\downarrow0}\RMI{1}{\mP}{\Wma{\rno}{\qmn{\rno,\mP}}}=\RMI{0}{\mP}{\Wm}\),
	i.e.  \(\lim_{\rno\downarrow0}\RMI{1}{\mP}{\Wma{\rno}{\qmn{\rno,\mP}}}=\lim_{\rno\downarrow0}\RMI{\rno}{\mP}{\Wm}\).
\end{enumerate}
\end{proof}

\begin{proof}[Proof of Lemma \ref{lem:polsha}]~
	Lemma \ref{lem:polsha} is nothing but Lemma \ref{lem:Lpolsha} for the case 
	when \(\lgm\) is a vector of zeros. Thus we do not present a separate proof
	for Lemma \ref{lem:polsha},	see the proof of Lemma \ref{lem:Lpolsha}.	 
\end{proof}

\subsection{Augustin's Proof of Lemma \ref{lem:information}-(\ref{information:zto})}\label{sec:augustinsproof}
We have employed the relative compactness in the total variation topology for proving 
Lemma \ref{lem:information}-(\ref{information:zto}) because we wanted to assert 
\(\qmn{\rno,\mP}\sim\qmn{1,\mP}\), the convergence described in \eqref{eq:lem:information-zto:iteration},
and the inequality given in \eqref{eq:lem:information-zto:EHB}. 
Establishing the existence of a unique Augustin mean together with
the fixed point property described in 
\eqref{eq:lem:information-zto:fixedpoint} is considerably easier. 
It can be done using the concept of relative compactness in the 
topology of setwise convergence, as demonstrated by Augustin in \cite[\S34]{augustin78}. 
Augustin claims to establish other assertions of 
Lemma \ref{lem:information}-(\ref{information:zto}), as well.
In the following, we discuss why we think there are caveats in 
Augustin's argument in \cite[\S34]{augustin78}.

Let us first establish the existence of a unique Augustin mean.
First, we establish \eqref{eq:information-1} as we have done in the current proof. 
Then we consider the set 
\(\set{Q}'\DEF 
\{\mQ\in \pmea{\outA}:\CRD{\rno}{\Wm}{\mQ}{\mP}
<\CRD{\rno}{\Wm}{\qgn{\rno,\mP}}{\mP}\}\).
Note that
\(\RMI{\rno}{\mP}{\Wm}=\inf\nolimits_{\mQ\in \Aop{\rno}{\mP}{\set{Q}'}} \CRD{\rno}{\Wm}{\mQ}{\mP}\)
because of the definition of \(\set{Q}'\) and \eqref{eq:information-1}.
Furthermore for all \(\mQ\in \set{Q}'\) and
\(\oev \in \outA\),
\begin{align}
\notag
\Aop{\rno}{\mP}{\mQ}(\oev)
&= 
\sum\nolimits_{\dinp} \mP(\dinp)
e^{(1-\rno)\RD{\rno}{\Wm(\dinp)}{\mQ}}
 \int_{\oev} (\der{\Wm(\dinp)}{\rfm})^{\rno}(\der{\mQ}{\rfm})^{1-\rno} \rfm(\dif{\dout})  
&
&
\\
\notag
&\leq 
e^{\frac{(1-\rno)\CRD{\rno}{\Wm(\dinp)}{\qgn{\rno,\mP}}{\mP}}{\min_{\dinp:\mP(\dinp)>0}}}
\sum\nolimits_{\dinp} \mP(\dinp)
\int_{\oev} (\der{\Wm(\dinp)}{\rfm})^{\rno}(\der{\mQ}{\rfm})^{1-\rno} \rfm(\dif{\dout})  
&
&\mbox{by the definition of \(\set{Q}'\) and Lemma \ref{lem:divergence-pinsker}}
\\
\notag
&= 
e^{\frac{(1-\rno)\CRD{\rno}{\Wm(\dinp)}{\qgn{\rno,\mP}}{\mP}}{\min_{\dinp:\mP(\dinp)>0}}}
\int_{\oev} (\der{\mmn{\rno,\mP}}{\rfm})^{\rno}(\der{\mQ}{\rfm})^{1-\rno} \rfm(\dif{\dout})  
&
&\mbox{by \eqref{eq:def:meanmeasure}}
\\
\notag
&\leq 
e^{\frac{(1-\rno)\CRD{\rno}{\Wm(\dinp)}{\qgn{\rno,\mP}}{\mP}}{\min_{\dinp:\mP(\dinp)>0}}}
\left[\mmn{\rno,\mP}(\oev)\right]^{\rno}
\left[\mQ(\oev)\right]^{1-\rno}
&
&\mbox{by Holder's inequality}
\\
\notag
&\leq 
e^{\frac{(1-\rno)\CRD{\rno}{\Wm(\dinp)}{\qgn{\rno,\mP}}{\mP}}{\min_{\dinp:\mP(\dinp)>0}}}
\left[\mmn{\rno,\mP}(\oev)\right]^{\rno}
&
&\mbox{because \(\mQ(\oev)\leq 1\)}
\end{align}

Thus \(\Aop{\rno}{\mP}{\set{Q}'}\UAC \qgn{\rno,\mP}\) and
\(\Aop{\rno}{\mP}{\set{Q}'}\) has compact closure in the topology of setwise convergence 
by a version of Dunford-Pettis theorem \cite[4.7.25]{bogachev}.
On the other hand, \(\CRD{\rno}{\Wm}{\mQ}{\mP}\) is lower-semicontinuous in \(\mQ\) for the 
topology of 
setwise convergence because \(\RD{\rno}{\Wm(\dinp)}{\mQ}\) is, 
by Lemma \ref{lem:divergence:lsc}.
Then there exists a \(\qmn{\rno,\mP}\) in the closure of \({\Aop{\rno}{\mP}{\set{Q}'}}\) 
for the topology of setwise convergence such that 
\(\CRD{\rno}{\Wm}{\qmn{\rno,\mP}}{\mP}=\inf_{\mQ\in \pmea{\outA}} \CRD{\rno}{\Wm}{\mQ}{\mP}\)
by the extreme value theorem  for lower semicontinuous functions \cite[Ch3\S12.2]{kolmogorovfomin75}.
The uniqueness of \(\qmn{\rno,\mP}\) follows from the strict convexity of the \renyi divergence 
in its second argument described in Lemma \ref{lem:divergence-convexity}.

This construction establishes certain additional properties of the Augustin mean, as well.
Note that \(\qmn{\rno,\mP}\AC\qmn{1,\mP}\) 
because \(\qmn{\rno,\mP}\) is in the closure of \(\Aop{\rno}{\mP}{\set{Q}'}\)
for the topology of setwise convergence.
In addition,
\(\Aop{\rno}{\mP}{\qmn{\rno,\mP}}=\qmn{\rno,\mP}\)
because of Lemma \ref{lem:divergence-pinsker} and \eqref{eq:information-1}. 
Furthermore, any \(\mQ\) satisfying 
\(\Aop{\rno}{\mP}{\mQ}=\mQ\) and \(\qmn{1,\mP}\AC \mQ\) is equal to \(\qmn{\rno,\mP}\)
because of the argument presented in step (\ref*{information:zto}-\ref{information-zto:sufficientcondition})
of the proof of Lemma \ref{lem:information}-(\ref{information:zto}).
These observations, with minor differences, exist in 
Augustin's proof of \cite[Lemma 34.2]{augustin78}.

Above discussion establishes Lemma \ref{lem:information}-(\ref{information:zto}) except 
for the following three assertions:
\begin{enumerate}[(i)]
	\item \(\qmn{1,\mP}\AC\qmn{\rno,\mP}\),
	\item the identity given in  \eqref{eq:lem:information-zto:iteration},
	\item the inequality  given in \eqref{eq:lem:information-zto:EHB}. 	
\end{enumerate}
Note that \(\Aop{\rno}{\mP}{\set{Q}'}\UAC \qgn{\rno,\mP}\) and 
\(\{\Aopi{\rno}{\mP}{\ind}{\qgn{\rno,\mP}}\}_{\ind\in\integers{+}} \subset\Aop{\rno}{\mP}{\set{Q}'}\). 
Then by \cite[Thm. 4.7.25]{bogachev},
\(\{\Aopi{\rno}{\mP}{\ind}{\qgn{\rno,\mP}}\}_{\ind\in\integers{+}}\)
has a subsequence \(\{\Aopi{\rno}{\mP}{\ind(\jnd)}{\qgn{\rno,\mP}}\}_{\jnd\in\integers{+}}\)
converging to a \(\mQ\in \clos{(\{\Aopi{\rno}{\mP}{\ind}{\qgn{\rno,\mP}}\}_{\ind\in\integers{+}})}\) 
where both the convergence and the closure are for the topology of setwise convergence.
Furthermore, \(\mQ\sim\qmn{1,\mP}\) because of the arguments used in step 
(\ref*{information:zto}-\ref{information:zto-equivalence}) 
of the proof of Lemma \ref{lem:information}-(\ref{information:zto}). 
There are two ways one can prove remaining assertions of 
Lemma \ref{lem:information}-(\ref{information:zto}) without using 
the totally boundedness of 
\(\{\Aopi{\rno}{\mP}{\ind}{\qgn{\rno,\mP}}\}_{\ind\in\integers{+}}\)
established in step (\ref*{information:zto}-\ref{information-zto:totallybounded})
of the proof of Lemma \ref{lem:information}-(\ref{information:zto}):
\begin{itemize}
\item If one can show that \(\der{\Aopi{\rno}{\mP}{\ind(\jnd)}{\qgn{\rno,\mP}}}{\qmn{1,\mP}}\)
converges to \(\der{\mQ}{\qmn{1,\mP}}\) in measure \(\qmn{1,\mP}\), then
because of the Lebesgue-Vitali convergence theorem \cite[4.5.4]{bogachev} 
one would have
\(\lim_{\jnd\to \infty}\lon{\Aopi{\rno}{\mP}{\ind(\jnd)}{\qgn{\rno,\mP}}-\mQ}=0\),
established step  (\ref*{information:zto}-\ref{information:zto-equivalence}).
Thus one can skip steps 
(\ref*{information:zto}-\ref{information-zto:totallybounded})
and
(\ref*{information:zto}-\ref{information:zto-equivalence})
and 
proceed with the step (\ref*{information:zto}-\ref{information:zto-continuity})
of the proof. 

\item If one can show that the limit point \(\mQ\) of the subsequence 
\(\der{\Aopi{\rno}{\mP}{\ind(\jnd)}{\qgn{\rno,\mP}}}{\qmn{1,\mP}}\) is a fixed point 
of \(\Aop{\rno}{\mP}{\cdot}\), then one would have a statement equivalent to step
(\ref*{information:zto}-\ref{information:zto-fixedpoint}).
Thus  one can skip steps 
(\ref*{information:zto}-\ref{information-zto:totallybounded}), 
(\ref*{information:zto}-\ref{information:zto-equivalence}),
and
(\ref*{information:zto}-\ref{information:zto-fixedpoint})
and 
proceed with the step 
(\ref*{information:zto}-\ref{information:zto-iteration})
after deriving
(\ref*{information:zto}-\ref{information:zto-continuity}).
\end{itemize}
We proved Lemma \ref{lem:information}-(\ref{information:zto}) using the concept of 
totally boundedness
because we could not find an easy way to establish either the convergence in measure property 
or the fixed point property mentioned in the preceding discussion. 
However, we do know that both properties hold.
The convergence in measure holds because of the only if part of the 
Lebesgue-Vitali convergence theorem \cite[4.5.4]{bogachev}.
The fixed point property holds because \(\{\Aopi{\rno}{\mP}{\ind}{\qgn{\rno,\mP}}\}_{\ind\in\integers{+}}\)
has a unique limit point both in total variation topology and in topology of setwise convergence
by \eqref{eq:lem:information-zto:iteration}. 

While proving \cite[Lemma 34.2]{augustin78}, 
after establishing the weak convergence
of \(\{\Aopi{\rno}{\mP}{\ind(\jnd)}{\qmn{1,\mP}}\}_{\jnd\in \integers{+}}\) 
to \(\mQ\), Augustin asserts that
\(\lim_{\jnd\to \infty}\lon{\Aopi{\rno}{\mP}{\ind(\jnd)}{\qmn{1,\mP}}-\mQ}=0\).
This is the first one of our two major reservations for Augustin's proof 
of \cite[Lemma 34.2]{augustin78}.
Note that convergence in the topology of setwise 
convergence and weak convergence are one and the same thing for sequences of measures
by \cite[Corollary 4.7.26]{bogachev}. 
But convergence in the topology of setwise
convergence does not imply convergence in total variation 
topology.\footnote{Consider, for example, the 
sequence of measure on the unit interval whose Radon-Nikodym derivatives with respect to
the Lebesgue measure is given by \(\{(1+\cos(\pi \jnd \dsta)) \}_{\jnd\in\integers{+}}\).
This set of probability measures converges to the Lebesgue measure on every measurable set, 
but not in total variation.}
Thus we don't know how one can justify such an assertion.

In order to prove \cite[Lemma 34.2]{augustin78}, Augustin establishes the totally 
boundedness of \(\{\Aopi{\rno}{\mP}{\ind}{\qmn{1,\mP}}\}_{\ind\in\integers{+}}\) 
for the total variation metric.
In that proof Augustin asserts that \(\Aopi{\rno}{\mP}{\ind}{\qmn{1,\mP}}\) is in
\(\set{B}_{\jnd}\), defined in equation \eqref{eq:def:setBi},  for some \(\jnd\in\integers{+}\).
We don't know whether such an assertion is correct or not. 
But we know that \(\Aopi{\rno}{\mP}{\ind}{\qgn{\rno,\mP}}\) is in \(\set{B}_{\ind}\).
Thus one can fix this problem easily. 
A more important problem stems from Augustin's obliviousness about the 
infiniteness of the set of positive integers. 
Either in his discussion or in his equations there is 
no evidence suggesting that he makes a distinction of cases for approximating 
\(\{\Aopi{\rno}{\mP}{\tin}{\qgn{\rno,\mP}}\}_{\tin\leq\ind}\)
and
\(\{\Aopi{\rno}{\mP}{\tin}{\qgn{\rno,\mP}}\}_{\tin>\ind}\).
This is our other major reservation for Augustin's proof 
of \cite[Lemma 34.2]{augustin78}.

\subsection{Proofs of Lemmas on the Augustin Capacity}\label{sec:capacityproofs}
\begin{proof}[Proof of Lemma \ref{lem:capacityFLB}]~
	\begin{enumerate}[(i)]
		\item {\it \(\forall\rno\in\reals{+}\exists \widetilde{\mP}\in\pdis{\inpS}\) s.t. \(\RMI{\rno}{\widetilde{\mP}}{\Wm}=\CRC{\rno}{\!\Wm\!}{\cset}\):}
		Since \(\bent{\mP}\leq \ln \abs{\inpS}\) for all \(\mP\in \pdis{\inpS}\),
		\eqref{eq:concavitybound} and Lemma \ref{lem:information}-(\ref{information:bounded})
		imply that
		\begin{align}
		\label{eq:capacityFLB-1}
		\abs{\RMI{\rno}{\pmn{2}}{\Wm}-\RMI{\rno}{\pmn{1}}{\Wm}}
		&\leq \bent{\tfrac{\lon{\pmn{1}-\pmn{2}}}{2}}+\tfrac{\lon{\pmn{1}-\pmn{2}}}{2} \ln \abs{\inpS}.
		\end{align}
		Hence, \(\RMI{\rno}{\mP}{\Wm}\) is continuous in \(\mP\) on \(\pdis{\inpS}\).
		On the other hand, \(\pdis{\inpS}\) is compact because \(\inpS\) is a finite set.
		Then \(\cset\) is compact because any closed subset of a compact set is compact, 
		\cite[Thm. 26.2]{munkres}.
		Then there exists a \(\widetilde{\mP}\in \cset\) such that
		\(\RMI{\rno}{\widetilde{\mP}}{\Wm}=\sup_{\mP \in \cset} \RMI{\rno}{\mP}{\Wm}\)
		by the extreme value
		theorem,\footnote{We do not need to establish the continuity of 
			\(\RMI{\rno}{\mP}{\Wm}\) in \(\mP\); the upper semicontinuity is sufficient
			as a result of \cite[Ch3\S12.2]{kolmogorovfomin75}.
			 Note that 
			\(\RMI{\rno}{\mP}{\Wm}\) is upper semicontinuous in \(\mP\) because it 
			is the infimum of a family of linear functions.}
		\cite[27.4]{munkres}.
		
		\item {\it If \(\rno\in\reals{+}\) and \(\RMI{\rno}{\widetilde{\mP}}{\Wm}=\CRC{\rno}{\!\Wm\!}{\cset}\), 
			then \(\CRD{\rno}{\Wm}{\qmn{\rno,\widetilde{\mP}}}{\mP}\leq\CRC{\rno}{\!\Wm\!}{\cset}\)
			for all \(\mP\in \cset\):}
		Let \(\mP\) be any member of \(\cset\) and \(\pma{}{(\ind)}\) be \(\tfrac{\ind-1}{\ind}\widetilde{\mP}+\tfrac{1}{\ind}{\mP}\) 
		for \(\ind \in \integers{+}\). Then \(\pma{}{(\ind)}\in \cset\) because \(\cset\) is convex. 
		Furthermore, by Lemma \ref{lem:information}-(\ref{information:one},\ref{information:zto},\ref{information:oti}) we have
		\begin{align}
		\notag
		\RMI{\rno}{\pma{}{(\ind)}}{\Wm}
		&\!=\!
		\tfrac{\ind-1}{\ind} 
		\CRD{\rno}{\Wm}{\qmn{\rno,\pma{}{(\ind)}}}{\widetilde{\mP}}
		\!+\!
		\tfrac{1}{\ind}\CRD{\rno}{\Wm}{\qmn{\rno,\pma{}{(\ind)}}}{\mP}
		\\
		\notag
		&\!\geq \!
		\tfrac{\ind-1}{\ind} 
		\left[\RMI{\rno}{\widetilde{\mP}}{\Wm}+\RD{\rno\wedge 1}{\qmn{\rno,\widetilde{\mP}}}{\qmn{\rno,\pma{}{(\ind)}}}\right]
		\!+\!\tfrac{1}{\ind}\CRD{\rno}{\Wm}{\qmn{\rno,\pma{}{(\ind)}}}{\mP}
		&
		&\forall \ind \in \integers{+}.
		\end{align}
		Using \(\RMI{\rno}{\pma{}{(\ind)}}{\Wm}\leq \CRC{\rno}{\!\Wm\!}{\cset}\), 
		\(\RMI{\rno}{\widetilde{\mP}}{\Wm}=\CRC{\rno}{\!\Wm\!}{\cset}\), and 
		\(\RD{\rno\wedge 1}{\qmn{\rno,\widetilde{\mP}}}{\qmn{\rno,\pma{}{(\ind)}}}\geq0\)
		we get
		\begin{align}
		\label{eq:capacityFLB-2}
		\CRC{\rno}{\!\Wm\!}{\cset}
		&\geq \CRD{\rno}{\Wm}{\qmn{\rno,\pma{}{(\ind)}}}{\mP}
		&
		&\forall \ind \in \integers{+}.
		\end{align}
		On the other hand, using \(\RMI{\rno}{\pma{}{(\ind)}}{\Wm}\leq \CRC{\rno}{\!\Wm\!}{\cset}\),
		\(\RMI{\rno}{\widetilde{\mP}}{\Wm}=\CRC{\rno}{\!\Wm\!}{\cset}\), and 
		\(\CRD{\rno}{\Wm}{\qmn{\rno,\pma{}{(\ind)}}}{\mP}\geq 0\) we get
		\begin{align}
		\notag
		\tfrac{1}{\ind}\CRC{\rno}{\!\Wm\!}{\cset}
		&\geq \tfrac{\ind-1}{\ind}\RD{\rno \wedge 1}{\qmn{\rno,\widetilde{\mP}}}{\qmn{\rno,\pma{}{(\ind)}}}
		&
		&\forall \ind \in \integers{+}.
		\end{align}
		Then using Lemma \ref{lem:divergence-pinsker} we get
		\begin{align}
		\notag
		\sqrt{\tfrac{2}{\rno\wedge 1}\tfrac{1}{\ind-1}\CRC{\rno}{\!\Wm\!}{\cset}}
		&\geq \lon{\qmn{\rno,\widetilde{\mP}}-\qmn{\rno,\pma{}{(\ind)}}}
		&
		&\forall \ind \in \integers{+}.
		\end{align}
		Thus \(\qmn{\rno,\pma{}{(\ind)}}\) converges to \(\qmn{\rno,\widetilde{\mP}}\) in the total variation topology and hence
		in the topology of setwise convergence. 
		Since the \renyi divergence is lower semicontinuous in the topology of setwise convergence by
		Lemma \ref{lem:divergence:lsc}, we have
		\begin{align}
		\label{eq:capacityFLB-3}
		\liminf\nolimits_{\ind \to \infty} 
		\CRD{\rno}{\Wm}{\qmn{\rno,\pma{}{(\ind)}}}{\mP}
		&\geq  \CRD{\rno}{\Wm}{\qmn{\rno,\widetilde{\mP}}}{\mP}
		\end{align}
		Then the inequality \(\CRD{\rno}{\Wm}{\qmn{\rno,\widetilde{\mP}}}{\mP}\leq\CRC{\rno}{\!\Wm\!}{\cset}\) 
		follows from \eqref{eq:capacityFLB-2} and \eqref{eq:capacityFLB-3}.
		
		\item {\it If \(\rno\in\reals{+}\), then \(\exists!\qmn{\rno,\!\Wm\!,\cset}\in\pmea{\outA}\) satisfying \eqref{eq:lem:capacityFLB}
			such that \(\qmn{\rno,\mP}=\qmn{\rno,\!\Wm\!,\cset}\) for all \(\mP\in \cset\) 
			satisfying \(\RMI{\rno}{\mP}{\Wm}=\CRC{\rno}{\!\Wm\!}{\cset}\):}
		If \(\RMI{\rno}{\mP}{\Wm}=\CRC{\rno}{\!\Wm\!}{\cset}\) for a 
		\(\mP\in\cset\), then 
		Lemma \ref{lem:information}-(\ref{information:one},\ref{information:zto},\ref{information:oti})
and Lemma \ref{lem:divergence-pinsker}
imply 		
\begin{align}
\label{eq:capacityFLB-4}
\CRD{\rno}{\Wm}{\qmn{\rno,\widetilde{\mP}}}{\mP}
\geq \CRC{\rno}{\!\Wm\!}{\cset}+
\tfrac{\rno\wedge 1}{2}\lon{\qmn{\rno,\mP}-\qmn{\rno,\widetilde{\mP}}}^{2}.
\end{align}
Since we have already established that
\(\CRD{\rno}{\Wm}{\qmn{\rno,\widetilde{\mP}}}{\mP}\leq\CRC{\rno}{\!\Wm\!}{\cset}\)
for all \(\mP\in\cset\), \eqref{eq:capacityFLB-4} implies that 		
\(\qmn{\rno,\mP}=\qmn{\rno,\widetilde{\mP}}\) for any \(\mP\in\cset\) 
satisfying \(\RMI{\rno}{\mP}{\Wm}=\CRC{\rno}{\!\Wm\!}{\cset}\).		
\end{enumerate}
\end{proof}

\begin{proof}[Proof of Theorem \ref{thm:minimax}]
	The right hand side of \eqref{eq:thm:minimax} is an upper bound on the left hand side
	because of the max-min inequality.
	Furthermore, the left hand side of \eqref{eq:thm:minimax} is equal to \(\CRC{\rno}{\!\Wm\!}{\cset}\) by \eqref{eq:capacity}.
	Thus when \(\CRC{\rno}{\!\Wm\!}{\cset}\) is infinite, \eqref{eq:thm:minimax} holds trivially.
	When \(\CRC{\rno}{\!\Wm\!}{\cset}\) is finite, \eqref{eq:thm:minimax} follows from \eqref{eq:thm:minimaxcenter}
	and the max-min inequality.
	Thus we can assume \(\CRC{\rno}{\!\Wm\!}{\cset}\) to be finite and prove  the claims about \(\qmn{\rno,\!\Wm\!,\cset}\)
	in order to prove the theorem.
	\begin{enumerate}[(i)]
		\item {\it If \(\CRC{\rno}{\!\Wm\!}{\cset}\in\reals{\geq0}\) and 
			\(\lim\nolimits_{\ind \to \infty}\!\RMI{\rno}{\pma{}{(\ind)}}{\Wm}\!=\!\CRC{\rno}{\!\Wm\!}{\cset}\)
			for a \(\{\pma{}{(\ind)}\}_{\ind\in\integers{+}} \subset \cset\),
			then \(\{\qmn{\rno,\pma{}{(\ind)}}\}_{\ind\in\integers{+}}\) is a Cauchy sequence in \(\pmea{\outA}\) for the total variation metric:}
		For any sequence of members of \(\cset\) satisfying 
		\(\lim\nolimits_{\ind \to \infty} \RMI{\rno}{\pma{}{(\ind)}}{\Wm}=\CRC{\rno}{\!\Wm\!}{\cset}\), let \(\{\cset^{(\ind)}\}_{\ind\in\integers{+}}\) 
		be  a nested sequence of closed, convex, subsets of \(\cset\) defined as follows, 
		\begin{align}
		\notag
		\cset^{(\ind)} 
		&\DEF \conv{(\cup_{\jnd=1}^{\ind}  \{\pma{}{(\jnd)}\})}. 
		\end{align}
		Furthermore, each \(\cset^{(\ind)}\) can be interpreted as a constraint set for a \(\Wm^{(\ind)}\)
		with a finite input set \(\inpS^{(\ind)}\) defined as follows,
		\begin{align}
		\notag
		\inpS^{(\ind)} 
		&\DEF\{\dinp\in\inpS:\exists \jnd\in\{1,\ldots,\ind\} \mbox{~such that~} \pma{}{(\jnd)}(\dinp)>0\}.
		\end{align}
		With a slight abuse of notation we use the symbol \(\cset^{(\ind)}\) not only for a subset of \(\pdis{\inpS}\) 
		but also for the corresponding subset of \(\pdis{\inpS^{(\ind)}}\).
		For any \(\ind\in\integers{+}\), there exists a unique \(\qmn{\rno,\!\Wm^{(\ind)}\!,\cset^{(\ind)}}\) 
		satisfying inequality \eqref{eq:lem:capacityFLB} by Lemma \ref{lem:capacityFLB}.
		Furthermore, \(\cset^{(\jnd)}\subset\cset^{(\ind)}\) for any \(\ind,\jnd\in\integers{+}\) such that \(\jnd\leq \ind\).
		In order to bound
		\(\lon{\qmn{\rno,\pma{}{(\jnd)}}-\qmn{\rno,\pma{}{(\ind)}}}\) for positive integers \(\jnd<\ind\), we use the triangle
		inequality for \(\qmn{\rno,\pma{}{(\jnd)}}\), \(\qmn{\rno,\pma{}{(\ind)}}\), and \(\qmn{\rno,\!\Wm^{(\ind)}\!,\cset^{(\ind)}}\)
		\begin{align}
		\label{eq:minimax-A}
		\lon{\qmn{\rno,\pma{}{(\jnd)}}-\qmn{\rno,\pma{}{(\ind)}}}
		\leq \lon{\qmn{\rno,\pma{}{(\jnd)}}-\qmn{\rno,\!\Wm^{(\ind)}\!,\cset^{(\ind)}}}
		+\lon{\qmn{\rno,\pma{}{(\ind)}}-\qmn{\rno,\!\Wm^{(\ind)}\!,\cset^{(\ind)}}}.
		\end{align} 
		
		Let us proceed with bounding \(\lon{\qmn{\rno,\pma{}{(\jnd)}}-\qmn{\rno,\!\Wm^{(\ind)}\!,\cset^{(\ind)}}}\)
		and \(\lon{\qmn{\rno,\pma{}{(\ind)}}-\qmn{\rno,\!\Wm^{(\ind)}\!,\cset^{(\ind)}}}\) from above.
		\begin{align}
		\notag
		\lon{\qmn{\rno,\pma{}{(\jnd)}}-\qmn{\rno,\!\Wm^{(\ind)}\!,\cset^{(\ind)}}}
		&\mathop{\leq}^{(a)} 
		\sqrt{\tfrac{2}{\rno \wedge 1}
			\RD{\rno\wedge 1}{\qmn{\rno,\pma{}{(\jnd)}}}{\qmn{\rno,\!\Wm^{(\ind)}\!,\cset^{(\ind)}}}}
		\\
		\notag
		&\mathop{\leq}^{(b)} 
		\sqrt{\tfrac{2}{\rno \wedge 1}}
		\sqrt{\CRD{\rno}{\Wm}{\qmn{\rno,\!\Wm^{(\ind)}\!,\cset^{(\ind)}}}{\pma{}{(\jnd)}}-\RMI{\rno}{\pma{}{(\jnd)}}{\Wm^{(\ind)}}}
		\\
		\notag
		&\mathop{\leq}^{(c)} 
		\sqrt{\tfrac{2}{\rno \wedge 1}}
		\sqrt{\CRC{\rno}{\Wm^{(\ind)}}{\cset^{(\ind)}}-\RMI{\rno}{\pma{}{(\jnd)}}{\Wm^{(\ind)}}}
		\\
		\notag
		&\mathop{\leq}^{(d)} 
		\sqrt{\tfrac{2}{\rno \wedge 1}}
		\sqrt{\CRC{\rno}{\!\Wm\!}{\cset}-\RMI{\rno}{\pma{}{(\jnd)}}{\Wm}}
		&
		&
		\end{align}
		where 
		\((a)\) follows from Lemma \ref{lem:divergence-pinsker}, 
		\((b)\) follows from %		\eqref{eq:lem:information-one:EHB}, \eqref{eq:lem:information-zto:EHB}, \eqref{eq:lem:information-oti:EHB} of 
		Lemma \ref{lem:information}-(\ref{information:one},\ref{information:zto},\ref{information:oti}),
		\((c)\) follows from Lemma \ref{lem:capacityFLB} because \(\pma{}{(\jnd)}\!\in\!\cset^{(\ind)}\),
		and \((d)\)   follows from the identities 
		\(\CRC{\rno}{\Wm^{(\ind)}}{\cset^{(\ind)}}\!=\!\CRC{\rno}{\!\Wm\!}{\cset^{(\ind)}}\!\leq\!\CRC{\rno}{\!\Wm\!}{\cset}\)
		and
		\(\RMI{\rno}{\pma{}{(\jnd)}}{\Wm^{(\ind)}}\!=\!\RMI{\rno}{\pma{}{(\jnd)}}{\Wm}\).
		We can obtain a similar bound on \(\lon{\qmn{\rno,\pma{}{(\ind)}}-\qmn{\rno,\!\Wm^{(\ind)}\!,\cset^{(\ind)}}}\).
		Then \(\{\qmn{\rno,\pma{}{(\ind)}}\}\) is a Cauchy sequence by  
		\eqref{eq:minimax-A} because \(\lim\limits_{\jnd \to \infty} \RMI{\rno}{\pma{}{(\jnd)}}{\Wm}\!=\!\CRC{\rno}{\!\Wm\!}{\cset}\).
		
		\item {\it If \(\CRC{\rno}{\!\Wm\!}{\cset}\in\reals{\geq0}\), 
			then \(\exists!\qmn{\rno,\!\Wm\!,\cset}\!\in\!\pmea{\outA}\) 
			satisfying  \(\lim\nolimits_{\ind\to \infty}\!\lon{\qmn{\rno,\!\Wm\!,\cset}\!-\!\qmn{\rno,\pma{}{(\ind)}}}\!=\!0\)
			for all \(\{\pma{}{(\ind)}\}_{\ind\in\integers{+}}\subset \cset\) such that 
			\(\lim\nolimits_{\ind \to \infty} \RMI{\rno}{\pma{}{(\ind)}}{\Wm}\!=\!\CRC{\rno}{\!\Wm\!}{\cset}\):}
		Note that \(\smea{\outA}\) is a complete metric space for the total variation metric, i.e. every Cauchy sequence has 
		a unique limit point in \(\smea{\outA}\),
		because \(\smea{\outA}\) is a Banach space for the total variation topology \cite[Thm. 4.6.1]{bogachev}.
		Then \(\{\qmn{\rno,\pma{}{(\ind)}}\}_{\ind\in\integers{+}}\) has a unique limit point \(\qmn{\rno,\pma{}{*}}\) in \(\smea{\outA}\). 
		Since \(\pmea{\outA}\) is a closed set for the total variation topology and 
		\(\cup_{\ind\in\integers{+}}\qmn{\rno,\pma{}{(\ind)}}\subset \pmea{\outA}\),
		then \(\qmn{\rno,\pma{}{*}}\in\pmea{\outA}\) by \cite[Thm. 2.1.3]{munkres}.
		
		We have established the existence of a unique limit point \(\qmn{\rno,\pma{}{*}}\),
		for any \(\{\pma{}{(\ind)}\}_{\ind\in\integers{+}}\subset \cset\) satisfying 
		\(\lim\nolimits_{\ind \to \infty}\RMI{\rno}{\pma{}{(\ind)}}{\Wm}=\CRC{\rno}{\!\Wm\!}{\cset}\).
		This, however, implies
		\(\lim_{\ind\to\infty}\lon{\qmn{\rno,\widetilde{\mP}^{(\ind)}}-\qmn{\rno,\pma{}{*}}}=0\)
		for any \(\{\widetilde{\mP}^{(\ind)}\}_{\ind\in\integers{+}}\subset \cset\) 
		satisfying
		\(\lim\nolimits_{\ind \to \infty}\RMI{\rno}{\widetilde{\mP}^{(\ind)}}{\Wm}=\CRC{\rno}{\!\Wm\!}{\cset}\)
		because we can interleave the elements of 
		\(\{\pma{}{(\ind)}\}_{\ind\in\integers{+}}\) and 
		\(\{\widetilde{\mP}^{(\ind)}\}_{\ind\in\integers{+}}\)
		to obtain a new sequence \(\{\widehat{\mP}^{(\ind)}\}_{\ind\in\integers{+}}\subset \cset\)
		satisfying \(\lim\nolimits_{\ind \to \infty} \RMI{\rno}{\widehat{\mP}^{(\ind)}}{\Wm}=\CRC{\rno}{\!\Wm\!}{\cset}\)
		for which \(\{\qmn{\rno,\widehat{\mP}^{(\ind)}}\}\) is a Cauchy sequence. 
		Then \(\qmn{\rno,\Wm,\cset}=\qmn{\rno,\pma{}{*}}\).

		\item {\it \(\qmn{\rno,\!\Wm\!,\cset}\) satisfies the equality given in \eqref{eq:thm:minimaxcenter}:}
		For any \(\mP \in \cset\), let us consider any sequence \(\{\pma{}{(\ind)}\}_{\ind\in\integers{+}}\subset\cset\) 
		satisfying \(\pma{}{(1)}=\mP\) and 
		\(\lim\nolimits_{\ind \to \infty} \RMI{\rno}{\pma{}{(\ind)}}{\Wm}=\CRC{\rno}{\!\Wm\!}{\cset}\). 
		Then \(\mP\in \cset^{(\ind)}\) for all \(\ind\in\integers{+}\). Using Lemma \ref{lem:capacityFLB} we get
		\begin{align}
		\label{eq:minimax-D}
		\CRD{\rno}{\Wm}{\qmn{\rno,\!\Wm^{(\ind)}\!,\cset^{(\ind)}}}{\mP}
		&\leq  \CRC{\rno}{\Wm^{(\ind)}}{\cset^{(\ind)}}
		&
		&\forall \ind\in \integers{+}.
		\end{align}
		Since \(\Wm^{(\ind)}\) has a finite input set, \(\exists\widetilde{\mP}^{(\ind)}\in \cset^{(\ind)}\) 
		satisfying \(\RMI{\rno}{\widetilde{\mP}^{(\ind)}}{\Wm^{(\ind)}}=\CRC{\rno}{\Wm^{(\ind)}}{\cset^{(\ind)}}\)
		and \(\qmn{\rno,\widetilde{\mP}^{(\ind)}}=\qmn{\rno,\!\Wm^{(\ind)}\!,\cset^{(\ind)}}\) by Lemma \ref{lem:capacityFLB}.
		Then \(\RMI{\rno}{\widetilde{\mP}^{(\ind)}}{\Wm^{(\ind)}}\geq \RMI{\rno}{\pma{}{(\ind)}}{\Wm^{(\ind)}}\) 
		and consequently \(\lim\nolimits_{\ind \to \infty} \RMI{\rno}{\widetilde{\mP}^{(\ind)}}{\Wm}=\CRC{\rno}{\!\Wm\!}{\cset}\). 
		We have already established that for such a sequence
		\(\qmn{\rno,\widetilde{\mP}^{(\ind)}}\rightarrow \qmn{\rno,\!\Wm\!,\cset}\) in the total variation topology, 
		and hence in the topology of setwise convergence. Then
		the lower semicontinuity  of the \renyi divergence in its arguments for the topology of setwise convergence, i.e. Lemma \ref{lem:divergence:lsc}, 
		the identity \(\CRC{\rno}{\Wm^{(\ind)}}{\cset^{(\ind)}}=\CRC{\rno}{\!\Wm\!}{\cset^{(\ind)}}\leq\CRC{\rno}{\!\Wm\!}{\cset}\),
		and \eqref{eq:minimax-D} imply 
		\begin{align}
		\notag
		\CRD{\rno}{\Wm}{\qmn{\rno,\!\Wm\!,\cset}}{\mP}
		&\leq \CRC{\rno}{\!\Wm\!}{\cset}
		&
		&\forall \mP\in\cset.
		\end{align}
		On the other hand \(\CRD{\rno}{\Wm}{\qmn{\rno,\!\Wm\!,\cset}}{\mP}\geq \RMI{\rno}{\mP}{\Wm}\) 
		and \(\sup_{\mP\in\cset}\RMI{\rno}{\mP}{\Wm}=\CRC{\rno}{\!\Wm\!}{\cset}\) by definition.
		Thus \eqref{eq:thm:minimaxcenter} holds.
	\end{enumerate} 
\end{proof}

\begin{proof}[Proof of Lemma \ref{lem:center}]
Lemma \ref{lem:information}-(\ref{information:one},\ref{information:zto},\ref{information:oti})
and 
the hypothesis given in \eqref{eq:def:center}
imply
	\begin{align}
	\notag
	\CRC{\rno}{\!\Wm\!}{\cset}-\RMI{\rno}{\mP}{\Wm}
	&\geq 	\RD{\rno\wedge 1}{\qmn{\rno,\mP}}{\qmn{\rno,\!\Wm\!,\cset}}
	&
	&\forall \mP\in\cset.
	\end{align}
	Then as a result of Lemma \ref{lem:divergence-pinsker}, 
	\begin{align}
	\notag
	\sqrt{\tfrac{2(\CRC{\rno}{\!\Wm\!}{\cset}-\RMI{\rno}{\mP}{\Wm})}{\rno\wedge 1}}
	&\geq 	\lon{\qmn{\rno,\mP}-\qmn{\rno,\!\Wm\!,\cset}}
	&
	&\forall \mP\in\cset.
	\end{align}
	Thus \(\{\qmn{\rno,\pma{}{(\ind)}}\}_{\ind\in\integers{+}}\) is a Cauchy sequence 
	with the limit point \(\qmn{\rno,\!\Wm\!,\cset}\) for any sequence of input distributions 
	\(\{\pma{}{(\ind)}\}_{\ind\in\integers{+}}\!\subset\!\cset\) satisfying
	\(\lim_{\ind \to \infty} \RMI{\rno}{\pma{}{(\ind)}}{\Wm}\!=\!\CRC{\rno}{\!\Wm\!}{\cset}\).
\end{proof}

\begin{proof}[Proof of Lemma \ref{lem:EHB}]
As a result of 
Lemma \ref{lem:information}-(\ref{information:one},\ref{information:zto},\ref{information:oti}) 
we have
	\begin{align}
	\notag
	\sup\nolimits_{\tilde{\mP}\in\cset}  \CRD{\rno}{\Wm}{\mQ}{\tilde{\mP}} 
	&\geq \CRD{\rno}{\Wm}{\mQ}{\mP} 
	&
	&\forall \mP \in \cset
	\\
	\label{eq:EHB-B}
	&\geq \RMI{\rno}{\mP}{\Wm}+\RD{\rno\wedge 1}{\qmn{\rno,\mP}}{\mQ}
	&
	&\forall \mP \in \cset.
	\end{align}
	Let \(\{\pma{}{(\ind)}\}_{\ind\in\integers{+}}\subset\cset\) be a sequence such that
	\(\lim\nolimits_{\ind\to\infty}\RMI{\rno}{\pma{}{(\ind)}}{\Wm}=\CRC{\rno}{\!\Wm\!}{\cset}\).
	Then \(\{\qmn{\rno,\pma{}{(\ind)}}\}\to \qmn{\rno,\!\Wm\!,\cset}\) in total variation topology
	and hence in the topology of set wise convergence by Lemma \ref{lem:center}.
	On the other hand, the \renyi divergence is lower semicontinuous in its arguments for 
	the topology of setwise convergence by Lemma \ref{lem:divergence:lsc}. 
	Then
	\begin{align}
	\label{eq:EHB-C}
	\liminf\nolimits_{\ind \to \infty} \left[\RMI{\rno}{\pma{}{(\ind)}}{\Wm}+\RD{\rno \wedge 1}{\qmn{\rno,\pma{}{(\ind)}}}{\mQ}\right]
	&\geq \CRC{\rno}{\!\Wm\!}{\cset}+\RD{\rno \wedge 1}{\qmn{\rno,\!\Wm\!,\cset}}{\mQ}.
	\end{align}
	\eqref{eq:lem:EHB} follows from  \eqref{eq:EHB-B} and \eqref{eq:EHB-C}.
\end{proof}

\begin{proof}[Proof of Lemma \ref{lem:blahutidentity}]
	Note that as a result of \eqref{eq:capacity:altdef} and the max-min inequality we have
	\begin{align}
	\label{eq:blahutidentity-maxmin}
	\CRC{\rno}{\!\Wm\!}{\cset}
	&\leq \inf\nolimits_{\Vm\in\pmea{\outA|\inpS}} \sup\nolimits_{\mP \in \cset} 
	\tfrac{\rno}{1-\rno}\CRD{1}{\Vm}{\Wm}{\mP}+\RMI{1}{\mP}{\Vm}.
	\end{align}
	Hence, \eqref{eq:lem:blahutminimax} holds trivially whenever \(\CRC{\rno}{\!\Wm\!}{\cset}=\infty\)
	and \eqref{eq:lem:blahutidentity} implies \eqref{eq:lem:blahutminimax}
	whenever \(\CRC{\rno}{\!\Wm\!}{\cset}\in\reals{\geq0}\).
	
	In order to establish \eqref{eq:lem:blahutidentity} assuming \(\CRC{\rno}{\!\Wm\!}{\cset}\in\reals{\geq0}\), 
	first note that whenever \(\CRC{\rno}{\!\Wm\!}{\cset}\in\reals{\geq0}\) there exists a unique \(\qmn{\rno,\Wm,\cset}\)
	satisfying \eqref{eq:def:center} by Lemma \ref{lem:center}. 
	Then as a result of 
	Definitions \ref{def:divergence}, \ref{def:tiltedprobabilitymeasure}, \ref{def:conditionaldivergence}, \ref{def:tiltedchannel} 
	and Lemma \ref{lem:variational}  we have
	\begin{align}
	\notag
	\CRD{\rno}{\Wm}{\qmn{\rno,\!\Wm\!,\cset}}{\mP}
	&=\tfrac{\rno}{1-\rno}\CRD{1}{\Wma{\rno}{\qmn{\rno,\!\Wm\!,\cset}}}{\Wm}{\mP}
	+\CRD{1}{\Wma{\rno}{\qmn{\rno,\!\Wm\!,\cset}}}{\qmn{\rno,\!\Wm\!,\cset}}{\mP}.
	\end{align}
	Then using Lemma \ref{lem:information}-(\ref{information:one}) and Lemma \ref{lem:divergence-pinsker}, we get
	\begin{align}
	\notag
	\CRD{\rno}{\Wm}{\qmn{\rno,\!\Wm\!,\cset}}{\mP}
	&=\tfrac{\rno}{1-\rno}\CRD{1}{\Wma{\rno}{\qmn{\rno,\!\Wm\!,\cset}}}{\Wm}{\mP}+\RMI{1}{\mP}{\Wma{\rno}{\qmn{\rno,\!\Wm\!,\cset}}}
	+\RD{1}{\sum\nolimits_{\dinp}\mP(\dinp)\Wma{\rno}{\qmn{\rno,\!\Wm\!,\cset}}(\dinp)}{\qmn{\rno,\!\Wm\!,\cset}}
	\\
	\notag
	&\geq \tfrac{\rno}{1-\rno}\CRD{1}{\Wma{\rno}{\qmn{\rno,\!\Wm\!,\cset}}}{\Wm}{\mP}+\RMI{1}{\mP}{\Wma{\rno}{\qmn{\rno,\!\Wm\!,\cset}}}
	&
	&\forall \mP\in\cset.
	\end{align}
	Thus \eqref{eq:def:center} and Lemma \ref{lem:center} implies that
	\begin{align}
	\label{eq:blahutidentity-A}
	\CRC{\rno}{\!\Wm\!}{\cset}
	&\geq \sup\nolimits_{\mP\in\cset}
	\tfrac{\rno}{1-\rno}\CRD{1}{\Wma{\rno}{\qmn{\rno,\!\Wm\!,\cset}}}{\Wm}{\mP}+\RMI{1}{\mP}{\Wma{\rno}{\qmn{\rno,\!\Wm\!,\cset}}}
	\\
	\label{eq:blahutidentity-B}
	&\geq \inf\nolimits_{\Vm} \sup\nolimits_{\mP\in\cset}\tfrac{\rno}{1-\rno}\CRD{1}{\Vm}{\Wm}{\mP}+\RMI{1}{\mP}{\Vm}.
	\end{align}
Note that \eqref{eq:lem:blahutminimax} follows from \eqref{eq:blahutidentity-maxmin}
and \eqref{eq:blahutidentity-B}.
On the other hand, using the \csiszar\!\!'s form for 
	the Augustin information, given in \eqref{eq:lem:information:alternative:def}, 
	we get
	\begin{align}
	\notag
	\tfrac{\rno}{1-\rno}\CRD{1}{\Wma{\rno}{\qmn{\rno,\!\Wm\!,\cset}}}{\Wm}{\mP}+\RMI{1}{\mP}{\Wma{\rno}{\qmn{\rno,\!\Wm\!,\cset}}}
	&\geq \inf\nolimits_{\Vm} \tfrac{\rno}{1-\rno}\CRD{1}{\Vm}{\Wm}{\mP}
	+\RMI{1}{\mP}{\Vm}
	&
	&
	\\
	\notag
	&=\RMI{\rno}{\mP}{\Wm}
	&
	&\forall \mP\in\pdis{\inpS}.
	\end{align}

	Then \eqref{eq:lem:blahutidentity} follows from 
	the definition of \(\CRC{\rno}{\!\Wm\!}{\cset}\) 
	and \eqref{eq:blahutidentity-A}.
\end{proof}

\begin{proof}[Proof of Lemma \ref{lem:capacityO}]~
	\begin{enumerate}
		\item[(\ref{lem:capacityO}-\ref{capacityO-ilsc})]
		\(\CRC{\rno}{\!\Wm\!}{\cset}\) is nondecreasing and lower semicontinuous because
		it is the pointwise supremum of \(\RMI{\rno}{\mP}{\Wm}\) for \(\mP\in\cset\) and 
		\(\RMI{\rno}{\mP}{\Wm}\) is nondecreasing and continuous in \(\rno\) by 
		Lemma \ref{lem:informationO}-(\ref{informationO:continuity}).

		\item[(\ref{lem:capacityO}-\ref{capacityO-decreasing})] 
		\(\tfrac{1-\rno}{\rno}\RMI{\rno}{\mP}{\Wm}\) is nonincreasing and continuous in \(\rno\) on \(\reals{+}\) 
		for all \(\mP\in\pdis{\inpS}\) by Lemma \ref{lem:informationO}-(\ref{informationO:decreasing}). 
		Furthermore,
		\begin{align}
		\notag
		\tfrac{1-\rno}{\rno}\CRC{\rno}{\!\Wm\!}{\cset}
		&=\sup\nolimits_{\mP\in \cset} \tfrac{1-\rno}{\rno}\RMI{\rno}{\mP}{\Wm}
		&
		&\forall \rno\in (0,1).
		\end{align} 
		Then \(\tfrac{1-\rno}{\rno}\CRC{\rno}{\!\Wm\!}{\cset}\) is nonincreasing and lower semicontinuous 
		in \(\rno\) on \((0,1)\) because the pointwise supremum of a family of nonincreasing 
		(lower semicontinuous) functions is nonincreasing (lower semicontinuous). 
		Thus \(\tfrac{1-\rno}{\rno}\CRC{\rno}{\!\Wm\!}{\cset}\) and \(\CRC{\rno}{\!\Wm\!}{\cset}\) are both continuous from the right on \((0,1)\). 
		On the other hand \(\CRC{\rno}{\!\Wm\!}{\cset}\) and \(\tfrac{1-\rno}{\rno}\CRC{\rno}{\!\Wm\!}{\cset}\) are both continuous 
		from the left on \((0,1)\) because \(\CRC{\rno}{\!\Wm\!}{\cset}\) is nondecreasing 
		and lower semicontinuous on \((0,1)\) by part (\ref{capacityO-ilsc}).
		Consequently, \(\CRC{\rno}{\!\Wm\!}{\cset}\) and 
		\(\tfrac{1-\rno}{\rno}\CRC{\rno}{\!\Wm\!}{\cset}\) are both continuous on \((0,1)\). 
		
		\item[(\ref{lem:capacityO}-\ref{capacityO-convexity})] 
		\((\rno-1)\RMI{\rno}{\mP}{\Wm}\) is convex in \(\rno\) on \(\reals{+}\) by Lemma \ref{lem:informationO}-(\ref{informationO:strictconvexity}). 
		Furthermore,
		\begin{align}
		\notag
		(\rno-1)\CRC{\rno}{\!\Wm\!}{\cset}
		&=\sup\nolimits_{\mP\in \cset} (\rno-1)\RMI{\rno}{\mP}{\Wm}
		&
		&\forall \rno\in(1,\infty).
		\end{align}
		Then \((\rno\!-\!1)\!\CRC{\rno}{\!\Wm\!}{\cset}\) is convex in \(\rno\) because
		the pointwise supremum of a family of convex functions is convex.
		
		\item[(\ref{lem:capacityO}-\ref{capacityO-continuity})]
		\(\CRC{\rno}{\!\Wm\!}{\cset}\) is continuous in \(\rno\) on \((0,1)\) by 
		part (\ref{capacityO-decreasing}).
		Furthermore, \(\CRC{\rno}{\!\Wm\!}{\cset}\) is continuous from the left because 
		it is nondecreasing and lower semicontinuous. 
		Thus \(\CRC{\rno}{\!\Wm\!}{\cset}\) is continuous in \(\rno\) on \((0,1]\).
		If \(\chi_{\Wm,\cset}=1\) we are done. 
		
		If \(\chi_{\Wm,\cset}>1\), then  \((\rno-1)\CRC{\rno}{\!\Wm\!}{\cset}\)  is finite and convex in \(\rno\) on 
		\([1,\chi_{\Wm,\cset})\) by part (\ref{capacityO-convexity}) and the definition of \(\chi_{\Wm,\cset}\). 
		Then \((\rno-1)\CRC{\rno}{\!\Wm\!}{\cset}\) is continuous in \(\rno\) on \((1,\chi_{\Wm,\cset})\) by \cite[Thm. 6.3.3]{dudley}. 
		The continuity of \((\rno-1)\CRC{\rno}{\!\Wm\!}{\cset}\) on \((1,\chi_{\Wm,\cset})\) implies the continuity of 
		\(\CRC{\rno}{\!\Wm\!}{\cset}\) on \((1,\chi_{\Wm,\cset})\). 
		Furthermore, \(\CRC{\rno}{\!\Wm\!}{\cset}\) is continuous from the left because \(\CRC{\rno}{\!\Wm\!}{\cset}\) 
		is nondecreasing and lower semicontinuous. Hence, \(\CRC{\rno}{\!\Wm\!}{\cset}\) is continuous in \(\rno\) 
		on \((1,\chi_{\Wm,\cset}]\), as well.

		\item[(\ref{lem:capacityO}-\ref{capacityO-continuity-extension})]
		As a result of part (\ref{capacityO-continuity}), we only need to prove the 
		continuity of \(\CRC{\rno}{\!\Wm\!}{\cset}\) from the right at \(\rno=1\) 
		when \(\chi_{\Wm,\cset}>1\).
		As a result of \cite[(\ref*{A-eq:UECatone})]{nakiboglu19A} we have
		\begin{align}
		\notag
		\GMI{\rno}{\mP}{\Wm} 
		&\leq \GMI{1}{\mP}{\Wm} + \tfrac{8(\rno-1)}{\epsilon^{2}e^{2}} 
		e^{\frac{\rnt-1}{\rnt}\GMI{\rnt}{\mP}{\Wm}}
		&
		&\forall \epsilon\in (0,\tfrac{\rnt-1}{\rnt}), \rno\in [1,(1-\epsilon)\rnt].
		\end{align}
		On the other hand \(\RMI{1}{\mP}{\Wm}=\GMI{1}{\mP}{\Wm}\) and \(\RMI{\rno}{\mP}{\Wm}\leq \GMI{\rno}{\mP}{\Wm}\)
		for \(\rno>1\). Then,
		\begin{align}
		\notag
		\CRC{\rno}{\!\Wm\!}{\cset}
		&\leq \CRC{1}{\!\Wm\!}{\cset} 
		+ \tfrac{8(\rno-1)}{\epsilon^{2}e^{2}} e^{\frac{\rnt-1}{\rnt}\sup_{\mP\in\cset}\GMI{\rnt}{\mP}{\Wm}}
		&
		&\forall \epsilon\in (0,\tfrac{\rnt-1}{\rnt}), \rno\in [1,(1-\epsilon)\rnt].
		\end{align}
		Thus \(\CRC{\rno}{\!\Wm\!}{\cset}\) is continuous at \(\rno=1\) from the right because 
		\(\CRC{\rno}{\!\Wm\!}{\cset}\geq \CRC{1}{\!\Wm\!}{\cset}\).
	\end{enumerate}
\end{proof}

\begin{proof}[Proof of Lemma \ref{lem:centercontinuity}]
	Note that \(\CRC{\rno}{\!\Wm\!}{\cset}\) is finite for all \(\rno\in(0,\rnt]\) by Lemma \ref{lem:capacityO}.
	Then there exists a unique order \(\rno\) Augustin center, \(\qmn{\rno,\!\Wm\!,\cset}\), 
	for all \(\rno\in(0,\rnt]\) by Theorem \ref{thm:minimax}.
	We apply Lemma \ref{lem:EHB} for \(\mQ=\qmn{\rnf,\Wm,\cset}\) to get
	\begin{align}
	\label{eq:centercontinuity-A}
	\sup\nolimits_{\mP\in\cset}\CRD{\rno}{\Wm}{\qmn{\rnf,\Wm,\cset}}{\mP}
	&\geq \CRC{\rno}{\!\Wm\!}{\cset}+\RD{\rno \wedge 1}{\qmn{\rno,\!\Wm\!,\cset}}{\qmn{\rnf,\Wm,\cset}}.  
	\end{align}
	Note that \(\CRD{\rno}{\Wm}{\qmn{\rnf,\Wm,\cset}}{\mP}\) is nondecreasing in \(\rno\)  for all \(\mP\in \cset\), 
	because \(\RD{\rno}{\Wm(\dinp)}{\qmn{\rnf,\Wm,\cset}}\) is,
	by Lemma \ref{lem:divergence-order}. 
	Then,
	\begin{align}
	\label{eq:centercontinuity-B}
	\CRD{\rnf}{\Wm}{\qmn{\rnf,\Wm,\cset}}{\mP}
	&\geq \CRD{\rno}{\Wm}{\qmn{\rnf,\Wm,\cset}}{\mP}
	&
	&\forall \mP\in \cset,~\rnf\in[\rno,\rnt].
	\end{align}
	On the other hand, by \eqref{eq:thm:minimaxcenter} of Theorem \ref{thm:minimax} we have
	\begin{align}
	\label{eq:centercontinuity-C}
	\sup\nolimits_{\mP\in\cset}\CRD{\rnf}{\Wm}{\qmn{\rnf,\Wm,\cset}}{\mP}
	&=\CRC{\rnf}{\Wm}{\cset}
	&
	&\forall \rnf\in(0,\rnt].
	\end{align}
	\eqref{eq:lem:centercontinuity} follows from 
	\eqref{eq:centercontinuity-A}, \eqref{eq:centercontinuity-B}, and \eqref{eq:centercontinuity-C}.
	
	Using Lemma \ref{lem:divergence-pinsker} together with 
	\eqref{eq:lem:centercontinuity} we get
	\begin{align}
	\label{eq:centercontinuity-D}
	\lon{\qmn{\rnf,\Wm,\cset}-\qmn{\rno,\!\Wm\!,\cset}}
	&\leq \sqrt{\tfrac{2}{\rno \wedge 1}(\CRC{\rnf}{\!\Wm\!}{\cset}-\CRC{\rno}{\!\Wm\!}{\cset})}
	&
	&\forall \rno,\rnf \mbox{~such that~}0<\rno<\rnf\leq\rnt.
	\end{align}  
	Then the continuity \(\qmn{\rno,\!\Wm\!,\cset}\) in \(\rno\) for the total  variation topology 
	on \(\pmea{\outA}\) follows 
	from the continuity \(\CRC{\rno}{\!\Wm\!}{\cset}\) in \(\rno\) on \(\alg{I}\).
\end{proof}

\begin{proof}[Proof of Lemma \ref{lem:capacityunion}]
	We analyze the upper bound on \(\CRC{\rno}{\!\Wm\!}{\cset}\) 
	and the lower bound on \(\CRC{\rno}{\!\Wm\!}{\cset}\) separately. 
	\begin{itemize}
		\item \(\sup\nolimits_{\ind \in \tinS} \CRC{\rno}{\!\Wm\!}{\cset^{(\ind)}}\leq \CRC{\rno}{\!\Wm\!}{\cset}\):
		Note that \(\CRC{\rno}{\!\Wm\!}{\cset^{(\ind)}}\leq \CRC{\rno}{\!\Wm\!}{\cset}\) by definition because \(\cset^{(\ind)}\subset\cset\).
		Thus \(\CRC{\rno}{\!\Wm\!}{\cset}\) is bounded from below by \(\sup\nolimits_{\ind \in \tinS} \CRC{\rno}{\!\Wm\!}{\cset^{(\ind)}}\), as well.
		\begin{itemize}
			\item If \(\CRC{\rno}{\!\Wm\!}{\cset^{(\ind)}}=\CRC{\rno}{\!\Wm\!}{\cset}<\infty\), 
			then \(\qmn{\rno,\!\Wm\!,\cset}=\qmn{\rno,\!\Wm\!,\cset^{(\ind)}}\) because using 
			Theorem \ref{thm:minimax} for \((a)\), Lemma \ref{lem:EHB} for \((b)\)
			and Lemma \ref{lem:divergence-pinsker} for \((c)\) we get
			\begin{align}
			\notag
			\CRC{\rno}{\!\Wm\!}{\cset}
			&\mathop{\geq}^{(a)} \sup\nolimits_{\mP\in \cset^{(\ind)}} \CRD{\rno}{\Wm}{\qmn{\rno,\!\Wm\!,\cset}}{\mP}
			\\
			\notag
			&\mathop{\geq}^{(b)} \CRC{\rno}{\!\Wm\!}{\cset^{(\ind)}}+\RD{\rno \wedge 1}{\qmn{\rno,\!\Wm\!,\cset^{(\ind)}}}{\qmn{\rno,\!\Wm\!,\cset}}
			\\
			\notag
			&\mathop{\geq}^{(c)} \CRC{\rno}{\!\Wm\!}{\cset^{(\ind)}}+\tfrac{\rno\wedge 1}{2} 
			\lon{\qmn{\rno,\!\Wm\!,\cset^{(\ind)}}-\qmn{\rno,\!\Wm\!,\cset}}^2.
			\end{align}
			If \(\CRC{\rno}{\!\Wm\!}{\cset^{(\ind)}}=\CRC{\rno}{\!\Wm\!}{\cset}\) and \(\qmn{\rno,\!\Wm\!,\cset^{(\ind)}}=\qmn{\rno,\!\Wm\!,\cset}\),
			then \(\sup\nolimits_{\mP\in \cset} \CRD{\rno}{\Wm}{\qmn{\rno,\!\Wm\!,\cset^{(\ind)}}}{\mP}\leq \CRC{\rno}{\!\Wm\!}{\cset^{(\ind)}}\)
			by Theorem \ref{thm:minimax}.
			\item If \(\sup_{\mP\in\cset} \CRD{\rno}{\Wm}{\qmn{\rno,\cset^{(\ind)}}}{\mP}\leq \CRC{\rno}{\!\Wm\!}{\cset^{(\ind)}}\), 
			then \(\CRC{\rno}{\!\Wm\!}{\cset}\leq \CRC{\rno}{\!\Wm\!}{\cset^{(\ind)}}\) 
			by \eqref{eq:capacity} and Theorem \ref{thm:minimax}.
			On the other hand,  \(\CRC{\rno}{\!\Wm\!}{\cset}\geq \CRC{\rno}{\!\Wm\!}{\cset^{(\ind)}}\) 
			because \(\cset^{(\ind)}\subset \cset\). 
			Hence, \(\CRC{\rno}{\!\Wm\!}{\cset}=\CRC{\rno}{\!\Wm\!}{\cset^{(\ind)}}<\infty\).
		\end{itemize}
		
		\item \(\CRC{\rno}{\!\Wm\!}{\cset}\leq \ln \sum\nolimits_{\ind\in \tinS} e^{\CRC{\rno}{\!\Wm\!}{\cset^{(\ind)}}}\):
		If \(\tinS\) is an infinite set, then the inequality holds trivially because 
		\(\ln \sum\nolimits_{\ind\in \tinS} e^{\CRC{\rno}{\!\Wm\!}{\cset^{(\ind)}}}=\infty\). 
		Thus we assume \(\tinS\) to be a finite set for the rest of the proof.
		Let \(\mean\) be \(\mean=\bigvee_{\ind\in \tinS} e^{\CRC{\rno}{\!\Wm\!}{\cset^{(\ind)}}}\qmn{\rno,\!\Wm\!,\cset^{(\ind)}}\).
		Then as a result of Lemma \ref{lem:divergence-RM} we have
		\begin{align}
		\notag
		\sup\nolimits_{\mP\in \cset^{(\ind)}}\CRD{\rno}{\Wm}{\tfrac{\mean}{\lon{\mean}}}{\mP}
		&=\sup\nolimits_{\mP\in \cset^{(\ind)}}\CRD{\rno}{\Wm}{\mean}{\mP}+\ln \lon{\mean} 
		\\
		\notag
		&\leq\sup\nolimits_{\mP\in \cset^{(\ind)}}
		\CRD{\rno}{\Wm}{\qmn{\rno,\!\Wm\!,\cset^{(\ind)}}}{\mP}-\CRC{\rno}{\!\Wm\!}{\cset^{(\ind)}}+\ln \lon{\mean}
		&
		&\forall \ind\in\tinS. 
		\end{align}
		Since
		\(\sup\nolimits_{\mP\in\cset^{(\ind)}}\CRD{\rno}{\Wm}{\qmn{\rno,\!\Wm\!,\cset^{(\ind)}}}{\mP}=\CRC{\rno}{\!\Wm\!}{\cset^{(\ind)}}\)
		by hypothesis, we have
		\begin{align}
		\notag
		\sup\nolimits_{\mP\in \cset^{(\ind)}}\CRD{\rno}{\Wm}{\tfrac{\mean}{\lon{\mean}}}{\mP}
		&\leq \ln \lon{\mean}
		&
		&\forall \ind\in\tinS. 
		\end{align}
		Then using \eqref{eq:capacity} and Theorem \ref{thm:minimax} we get
		\begin{align}
		\notag
		\CRC{\rno}{\!\Wm\!}{\cset}
		&\leq \sup\nolimits_{\mP\in \cset}\CRD{\rno}{\Wm}{\tfrac{\mean}{\lon{\mean}}}{\mP}
		\\
		\notag
		&=
		\sup\nolimits_{\ind\in\tinS}\sup\nolimits_{\mP\in \cset^{(\ind)}}\CRD{\rno}{\Wm}{\tfrac{\mean}{\lon{\mean}}}{\mP}
		\\
		\notag
		&\leq \ln \lon{\mean}
		\\
		\notag
		&\leq 
		\ln \sum\nolimits_{\ind\in \tinS} e^{\CRC{\rno}{\!\Wm\!}{\cset^{(\ind)}}}. 
		\end{align}
		\begin{itemize}
			\item If \(\qmn{\rno,\!\Wm\!,\cset^{(\ind)}}\) and \(\qmn{\rno,\!\Wm\!,\cset^{(\jnd)}}\) are not singular for some 
			\(\ind \neq \jnd\), 
			then \(\lon{\mean}<\sum\nolimits_{\ind\in\tinS} e^{\CRC{\rno}{\!\Wm\!}{\cset^{(\ind)}}}\).
			Thus \(\CRC{\rno}{\!\Wm\!}{\cset}<\ln \sum\nolimits_{\ind\in\tinS} e^{\CRC{\rno}{\!\Wm\!}{\cset^{(\ind)}}}\).
			Consequently, if \(\CRC{\rno}{\!\Wm\!}{\cset}=\ln \sum\nolimits_{\ind\in\tinS} e^{\CRC{\rno}{\!\Wm\!}{\cset^{(\ind)}}}\), 
			then \(\qmn{\rno,\!\Wm\!,\cset^{(\ind)}}\perp\qmn{\rno,\!\Wm\!,\cset^{(\jnd)}}\) for all \(\ind \neq \jnd\).
			
			\item If \(\qmn{\rno,\!\Wm\!,\cset^{(\ind)}}\perp\qmn{\rno,\!\Wm\!,\cset^{(\jnd)}}\) 
			for all \(\ind \neq \jnd\), 
			then any \(\mS\in \pmea{\outA}\) can be written as \(\mS=\sum_{\ind=1}^{\abs{\tinS} +1} \mS_{\ind}\) 
			where \(\mS_{\ind}\)'s are finite measures  such that \(\mS_{\ind}\AC \qmn{\rno,\!\Wm\!,\cset^{(\ind)}}\) 
			for \(\ind\in \tinS\) and \(\mS_{\abs{\tinS}+1}\perp  (\sum_{\ind\in \tinS} \qmn{\rno,\!\Wm\!,\cset^{(\ind)}})\)
			by the Lebesgue decomposition theorem \cite[5.5.3]{dudley}.
			Using Lemmas  \ref{lem:divergence-RM}, \ref{lem:divergence-pinsker}, and \ref{lem:EHB} we get
			\begin{align}
			\notag
			\sup\nolimits_{\mP\in \cset^{(\ind)}}\!\CRD{\rno}{\Wm}{\mS}{\mP}
			&\geq \CRC{\rno}{\!\Wm\!}{\cset^{(\ind)}}+\RD{\rno \wedge 1 }{\qmn{\rno,\!\Wm\!,\cset^{(\ind)}}}{\mS}
			\\
			\notag
			&=\CRC{\rno}{\!\Wm\!}{\cset^{(\ind)}}+\RD{\rno \wedge 1 }{\qmn{\rno,\!\Wm\!,\cset^{(\ind)}}}{\tfrac{\mS_{\ind}}{\lon{\mS_{\ind}}}}
			-\ln \lon{\mS_{\ind}}
			\\
			\notag
			&\geq \CRC{\rno}{\!\Wm\!}{\cset^{(\ind)}}-\ln \lon{\mS_{\ind}}.
			\end{align}
			\(\sup\nolimits_{\mP\in \cset}\CRD{\rno}{\Wm}{\mS}{\mP}= \max\nolimits_{\ind\in \tinS}\sup\nolimits_{\mP\in \cset^{(\ind)}}\!\CRD{\rno}{\Wm}{\mS}{\mP}\)
			because \(\CRD{\rno}{\Wm}{\mS}{\mP}\) is linear in \(\mP\) and \(\cset=\conv{(\cup_{\ind\in \tinS} \cset^{(\ind)})}\).
			Then  using \(\sum_{\ind=1}^{\abs{\tinS}} \lon{\mS_{\ind}}\leq \lon{\mS} =1\) we get,
			\begin{align}
			\notag
			\sup\nolimits_{\mP\in \cset}\CRD{\rno}{\Wm}{\mS}{\mP}
			&\geq \max\nolimits_{\ind\in \tinS}  \ln \tfrac{e^{\CRC{\rno}{\!\Wm\!}{\cset^{(\ind)}}}}{\lon{\mS_{\ind}}} 
			\\
			\notag
			&\geq \ln \sum\nolimits_{\ind\in\tinS} e^{\CRC{\rno}{\!\Wm\!}{\cset^{(\ind)}}}
			&
			&\forall \mS\in\pmea{\outA}.
			\end{align}
			Then \(\CRC{\rno}{\!\Wm\!}{\cset}\geq \ln \sum\nolimits_{\ind\in\tinS} e^{\CRC{\rno}{\!\Wm\!}{\cset^{(\ind)}}}\)
			by \eqref{eq:capacity} and Theorem \ref{thm:minimax}.
			Since we have already established the reverse inequality, we have
			\(\CRC{\rno}{\!\Wm\!}{\cset}=\ln \sum\nolimits_{\ind\in\tinS} e^{\CRC{\rno}{\!\Wm\!}{\cset^{(\ind)}}}\).
		\end{itemize}
		We have proved that if 
		\(\qmn{\rno,\!\Wm\!,\cset^{(\ind)}}\perp\qmn{\rno,\!\Wm\!,\cset^{(\jnd)}}\) 
		for all \(\ind \neq \jnd\), 
		then 
		\(\CRC{\rno}{\Wm}{\cset}=\sum_{\ind\in \tinS} e^{\CRC{\rno}{\!\Wm\!}{\cset^{(\ind)}}}\).
		One can confirm by substitution that
		\(\sup\nolimits_{\mP\in \cset^{(\ind)}}\!\CRD{\rno}{\Wm}{\mS}{\mP}\leq \CRC{\rno}{\Wm}{\cset}\)
		for all \(\ind\in\tinS\)
		for \(\mS=\sum_{\ind\in \tinS} e^{-\CRC{\rno}{\Wm}{\cset}+\CRC{\rno}{\!\Wm\!}{\cset^{(\ind)}}}\qmn{\rno,\!\Wm\!,\cset^{(\ind)}}\).
		On the other hand, 
		\(\sup\nolimits_{\mP\in \cset}\CRD{\rno}{\Wm}{\mS}{\mP}= \max\nolimits_{\ind\in \tinS}\sup\nolimits_{\mP\in \cset^{(\ind)}}\!\CRD{\rno}{\Wm}{\mS}{\mP}\)
		because \(\CRD{\rno}{\Wm}{\mS}{\mP}\) is linear in \(\mP\).
		Then \(\mS\) is the unique order \(\rno\) Augustin center by Theorem \ref{thm:minimax}.
	\end{itemize}
\end{proof}

\begin{proof}[Proof of Lemma \ref{lem:capacityproduct}]
	Let \(\rno\) be any fixed positive real order.
	Then as a result of Lemma \ref{lem:information:product} we have
	\begin{align}
	\label{eq:capacityproduct-1}
	\CRC{\rno}{\!\Wmn{[1,\blx]}}{\cset_{1}^{\blx}}
	&= \sum\nolimits_{\tin=1}^{\blx}\CRC{\rno}{\Wmn{\tin}}{\cset_{\tin}}.
	\end{align}
	On the other hand, \(\CRC{\rno}{\!\Wmn{[1,\blx]}}{\cset_{1}^{\blx}}\leq \CRC{\rno}{\!\Wmn{[1,\blx]}}{\cset}\)
	because \(\cset_{1}^{\blx}\subset \cset\). Then
	\begin{align}
	\label{eq:capacityproduct-2}
	\sum\nolimits_{\tin=1}^{\blx}\CRC{\rno}{\Wmn{\tin}}{\cset_{\tin}}
	&\leq \CRC{\rno}{\!\Wmn{[1,\blx]}}{\cset}.
	\end{align}
	We proceed to prove \(\CRC{\rno}{\!\Wmn{[1,\blx]}}{\cset}\leq\sum\nolimits_{\tin=1}^{\blx}\CRC{\rno}{\Wmn{\tin}}{\cset_{\tin}}\). 
	If there exists a \(\tin\in\{1,\ldots,\blx\}\) such that \(\CRC{\rno}{\Wmn{\tin}}{\cset_{\tin}}=\infty\),
	then the inequality holds trivially.
	Else \(\CRC{\rno}{\Wmn{\tin}}{\cset_{\tin}}\) is finite for all \(\tin\in\{1,\ldots,\blx\}\) and by
	Lemma \ref{lem:center} there exists a unique 
	\(\qmn{\rno,\Wmn{\tin},\cset_{\tin}}\) for each \(\tin\in\{1,\ldots,\blx\}\) such that
	\begin{align}
	\notag
	\CRD{\rno}{\Wmn{\tin}}{\qmn{\rno,\Wmn{\tin},\cset_{\tin}}}{\widetilde{\pmn{\tin}}}
	&\leq \CRC{\rno}{\Wmn{\tin}}{\cset_{\tin}}
	&
	&\forall \widetilde{\pmn{\tin}}\in \cset_{\tin}.
	\end{align}
	Since the conditional \renyi divergence \(\CRD{\rno}{\Wmn{\tin}}{\qmn{\rno,\Wmn{\tin},\cset_{\tin}}}{\widetilde{\pmn{\tin}}}\)
	is linear in the input distribution \(\widetilde{\pmn{\tin}}\),
	this implies
	\begin{align}
	\label{eq:capacityproduct-3}
	\CRD{\rno}{\Wmn{\tin}}{\qmn{\rno,\Wmn{\tin},\cset_{\tin}}}{\widetilde{\pmn{\tin}}}
	&\leq \CRC{\rno}{\Wmn{\tin}}{\cset_{\tin}}
	&
	&\forall \widetilde{\pmn{\tin}}\in \conv{\cset_{\tin}}.
	\end{align}
	Let \(\mQ\) be \(\mQ=\bigotimes\nolimits_{\tin=1}^{\blx}\qmn{\rno,\Wmn{\tin},\cset_{\tin}}\).
	Then as a result of  Tonelli-Fubini theorem \cite[4.4.5]{dudley} we have
	\begin{align}
	\notag
	\RD{\rno}{\Wmn{[1,\blx]}(\dinp_{1}^{\blx})}{\mQ}
	&=\sum\nolimits_{\tin=1}^{\blx} \RD{\rno}{\Wmn{\tin}(\dinp_{\tin})}{\qmn{\rno,\Wmn{\tin},\cset_{\tin}}}
	&
	&\forall \dinp_{1}^{\blx}\in\inpS_{1}^{\blx}.
	\end{align}
	Hence,
	\begin{align}
	\notag
	\CRD{\rno}{\!\Wmn{[1,\blx]}}{\mQ}{\mP}
	&=\sum\nolimits_{\tin=1}^{\blx} \CRD{\rno}{\Wmn{\tin}}{\qmn{\rno,\Wmn{\tin},\cset_{\tin}}}{\pmn{\tin}}
	&
	&\forall \mP\in\pdis{\inpS_{1}^{\blx}},
	\end{align}
	where \(\pmn{\tin}\in\pdis{\inpS_{\tin}}\) is the \(\inpS_{\tin}\) marginal of \(\mP\)
	for each \(\tin\in\{1,\ldots,\blx\}\).
	Note that \(\pmn{\tin}\in\conv{\cset_{\tin}}\) for all  \(\tin\in\{1,\ldots,\blx\}\)  
	by the definition constraint set \(\cset\).
	Thus \eqref{eq:capacityproduct-3} implies
	\begin{align}
	\label{eq:capacityproduct-4}
	\CRD{\rno}{\!\Wmn{[1,\blx]}}{\mQ}{\mP}
	&\leq \sum\nolimits_{\tin=1}^{\blx}\CRC{\rno}{\Wmn{\tin}}{\cset_{\tin}}
	&
	&\forall \mP\in \cset.
	\end{align}
	On the other hand \(\CRD{\rno}{\!\Wmn{[1,\blx]}}{\mQ}{\mP}\geq\RMI{\rno}{\mP}{\Wmn{[1,\blx]}~\!}\) by definition.
	Thus \eqref{eq:capacityproduct-2} and \eqref{eq:capacityproduct-4} imply 
	\(\CRC{\rno}{\!\Wmn{[1,\blx]}}{\cset}=\sum\nolimits_{\tin=1}^{\blx}\CRC{\rno}{\Wmn{\tin}}{\cset_{\tin}}\)
	and \(\qmn{\rno,\!\Wmn{[1,\blx]}\!,\cset}=\mQ\).
	Then \(\qmn{\rno,\!\Wmn{[1,\blx]}\!,\cset_{1}^{\blx}}=\mQ\) by Lemma \ref{lem:capacityunion}, as well,
	because \(\cset_{1}^{\blx}\subset \cset\) and  \(\CRC{\rno}{\!\Wmn{[1,\blx]}}{\cset_{1}^{\blx}}=\CRC{\rno}{\!\Wmn{[1,\blx]}}{\cset}\).
\end{proof}

\subsection{Proofs of Lemmas on the Cost Constrained Problem}\label{sec:costproofs}
%\subsubsection{Proof of Lemmas on Cost Constrained Capacity}\label{sec:CCproofs}
\begin{proof}[Proof of Lemma \ref{lem:CCcapacity}]~
	\begin{enumerate}[(a)]
		\item[(\ref{lem:CCcapacity}-\ref{CCcapacity:function})]
		If \(\costc_{1}\leq \costc_{2}\), then \(\CRC{\rno}{\!\Wm\!}{\costc_{1}}\leq\CRC{\rno}{\!\Wm\!}{\costc_{2}}\)
		because \(\cset(\costc_{1})\subset \cset(\costc_{2})\).
		Thus \(\CRC{\rno}{\!\Wm\!}{\costc}\) is nondecreasing in \(\costc\).
		
		Let \(\costc_{\rnb}=\rnb \costc_{1}+(1-\rnb)\costc_{0}\); then
		\((\rnb \pmn{1}+(1-\rnb) \pmn{0})\in\cset(\costc_{\rnb})\) for any \(\pmn{1}\in\cset(\costc_{1})\)
		and \(\pmn{0}\in\cset(\costc_{0})\).
		Hence, using the concavity of the Augustin information in its input distribution established in 
		Lemma \ref{lem:informationP} we get
		\begin{align}
		\notag
		\CRC{\rno}{\!\Wm\!}{\costc_{\rnb}}
		&\geq \sup\nolimits_{\pmn{1}\in\cset(\costc_{1}),\pmn{0}\in\cset(\costc_{0})}\RMI{\rno}{\rnb \pmn{1}+(1-\rnb) \pmn{0}}{\Wm}
		\\
		\notag
		&\geq \sup\nolimits_{\pmn{1}\in\cset(\costc_{1}),\pmn{0}\in\cset(\costc_{0})}
		\rnb\RMI{\rno}{\pmn{1}}{\Wm} + (1-\rnb)\RMI{\rno}{\pmn{0}}{\Wm}
		\\
		\notag
		&=\rnb \CRC{\rno}{\!\Wm\!}{\costc_{1}}+(1-\rnb)\CRC{\rno}{\!\Wm\!}{\costc_{0}}.
		\end{align} 
		Thus \(\CRC{\rno}{\!\Wm\!}{\costc}\) is concave in \(\costc\).
		
		Now let us proceed by proving that if \(\CRC{\rno}{\!\Wm\!}{\costc_{0}}=\infty\) for a \(\costc_{0}\in\inte{\fcc{\costf}}\); 
		then \(\CRC{\rno}{\!\Wm\!}{\costc}=\infty\) for all  \(\costc\in\inte{\fcc{\costf}}\). 
		Note that any point \(\costc\) in \(\inte{\fcc{\costf}}\) can be written as
		\(\costc=\rnb\costc_{1}+(1-\rnb)\costc_{0}\) for some \(\rnb\in(0,1)\) and \(\costc_{1}\in\inte{\fcc{\costf}}\) 
		because \(\inte{\fcc{\costf}}\) is a convex open subset \(\reals{}^{\ell}\).
		Then \(\CRC{\rno}{\!\Wm\!}{\costc}=\infty\) follows from the concavity of \(\CRC{\rno}{\!\Wm\!}{\costc}\).
		
		If \(\CRC{\rno}{\!\Wm\!}{\costc}\) is finite on \(\inte{\fcc{\costf}}\), 
		then \(\CRC{\rno}{\!\Wm\!}{\costc}\) is continuous on 
		\(\inte{\fcc{\costf}}\) by \cite[Thm. 6.3.4]{dudley}
		because \(\inte{\fcc{\costf}}\) is a convex  open subset \(\reals{}^{\ell}\) and \((-\CRC{\rno}{\!\Wm\!}{\costc})\)
		is a convex function of \(\costc\) on \(\inte{\fcc{\costf}}\).
		
		\item[(\ref{lem:CCcapacity}-\ref{CCcapacity:interior})]
		Let us extend the definition of \(\CRC{\rno}{\!\Wm\!}{\costc}\) from \(\reals{\geq0}^{\ell}\)
		to \(\reals{}^{\ell}\) by setting \(\CRC{\rno}{\!\Wm\!}{\costc}\) to \(-\infty\) for all \(\costc\in\reals{}^{\ell}\setminus\reals{\geq0}^{\ell}\).
		Then \((\!-\CRC{\rno}{\!\Wm\!}{\costc}\!)\) is a proper convex function, i.e. 
		\((\!-\CRC{\rno}{\!\Wm\!}{\costc}\!)\!:\!\reals{}^{\ell}\!\to\!(-\infty,\infty]\) is a convex function and \(\exists\costc\) such that
		\((\!-\CRC{\rno}{\!\Wm\!}{\costc}\!)\!<\!\infty\).
		Furthermore, \(\inte{\fcc{\costf}}\) is also the interior of the effective domain of the extended function. 
		Hence the sub-differential \(\partial(-\CRC{\rno}{\!\Wm\!}{\costc})\) is non-empty and compact by 
		\cite[Proposition 4.4.2]{bertsekas}. 
		Then \eqref{eq:CCcapacity:interior} follows from the fact that 
		the epigraph of a convex function lies above the tangent planes drawn at any point.
		The non-negativity of the components of \(\lgm_{\rno,\!\Wm\!,\costc}\) follows from the monotonicity of 
		\(\CRC{\rno}{\!\Wm\!}{\costc}\) in \(\costc\).

		\item[(\ref{lem:CCcapacity}-\ref{CCcapacity:zeroonedichotomy})]
		If \(\CRC{\tilde{\rno}}{\Wm}{\tilde{\costc}}=\infty\) for a \(\tilde{\rno}\in(0,1)\) and
		\(\tilde{\costc}\in \inte{\fcc{\costf}}\), 
		then \(\CRC{\rno}{\!\Wm\!}{\tilde{\costc}}=\infty\) for all 
		\(\rno\in(0,1)\) by Lemma \ref{lem:capacityO}-(\ref{capacityO-ilsc},\ref{capacityO-decreasing}). 
		Therefore, \(\CRC{\rno}{\!\Wm\!}{\costc}=\infty\) for all 
		\(\rno\in(0,1)\) and \(\costc\in\inte{\fcc{\costf}}\) by part (\ref{CCcapacity:function}).
		
		In order to prove the continuity when \(\CRC{\rno}{\!\Wm\!}{\costc}\) is finite, note that 
		as a result of the triangle inequality we have 
		\begin{align}
		%\label{eq:CCcapacity:zeroonedichotomy-0}
		\notag
		\abs{\CRC{\rno_{1}}{\!\Wm\!}{\costc_{1}}-\CRC{\rno_{2}}{\!\Wm\!}{\costc_{2}}} 
		&\leq \abs{\CRC{\rno_{1}}{\!\Wm\!}{\costc_{1}}-\CRC{\rno_{1}}{\!\Wm\!}{\costc_{2}}} 
		+\abs{\CRC{\rno_{1}}{\!\Wm\!}{\costc_{2}}-\CRC{\rno_{2}}{\!\Wm\!}{\costc_{2}}}. 
		\end{align}
		The first term converges to zero as \(\costc_{2}\to\costc_{1}\) as a result of the continuity of 
		the Augustin capacity 
		in the constraint established in part (\ref{CCcapacity:function}).
		The second term converges to zero as \(\rno_{2}\to\rno_{1}\) because of 
		\eqref{eq:CCcapacity:zeroonedichotomy-1} established in the following. 
		Thus \(\CRC{\rno}{\!\Wm\!}{\costc}\) is continuous in the pair \((\rno,\costc)\).
		
		Using the monotonicity of \(\CRC{\rno}{\!\Wm\!}{\costc}\) and \(\tfrac{1-\rno}{\rno}\CRC{\rno}{\!\Wm\!}{\costc}\)
		established in Lemma \ref{lem:capacityO}-(\ref{capacityO-ilsc}) and Lemma \ref{lem:capacityO}-(\ref{capacityO-decreasing})  
		we get
		\begin{align}
		\notag
		\abs{\CRC{\rno_{1}}{\!\Wm\!}{\costc_{2}}-\CRC{\rno_{2}}{\!\Wm\!}{\costc_{2}}} 
		&\leq \tfrac{\abs{\rno_{2}-\rno_{1}}}{(\rno_{1}\wedge \rno_{2}) (1-\rno_{1}\vee \rno_{2})}\CRC{\rno_{1}}{\!\Wm\!}{\costc_{2}}.
		\end{align}
		Thus using \eqref{eq:CCcapacity:interior} to bound \(\CRC{\rno_{1}}{\!\Wm\!}{\costc_{2}}\) we get
		\begin{align}
		\label{eq:CCcapacity:zeroonedichotomy-1}
		\abs{\CRC{\rno_{1}}{\!\Wm\!}{\costc_{2}}-\CRC{\rno_{2}}{\!\Wm\!}{\costc_{2}}} 
		&\leq \tfrac{\abs{\rno_{2}-\rno_{1}}}{(\rno_{1}\wedge\rno_{2}) (1-\rno_{1} \vee \rno_{2})}
		(\CRC{\rno_{1}}{\!\Wm\!}{\costc_{1}}+\abp{\lgm_{\rno_{1},\!\Wm\!,\costc_{1}} \cdot (\costc_{2}-\costc_{1})}).
		\end{align}
		
		In order to prove the continuity of the Augustin center, note that by the triangle inequality we have
		\begin{align}
		\label{eq:CCcapacity:zeroonedichotomy-2}
		\lon{\qmn{\rno_{1},\!\Wm\!,\costc_{1}}-\qmn{\rno_{2},\!\Wm\!,\costc_{2}}} 
		&\leq 
		\lon{\qmn{\rno_{1},\!\Wm\!,\costc_{1}}-\qmn{\rno_{1},\!\Wm\!,\costc_{2}}} 
		+\lon{\qmn{\rno_{1},\!\Wm\!,\costc_{2}}-\qmn{\rno_{2},\!\Wm\!,\costc_{2}}}.
		\end{align}
		Using first Lemmas \ref{lem:divergence-pinsker} and \ref{lem:centercontinuity}, 
		and then \eqref{eq:CCcapacity:zeroonedichotomy-1} we get
		\begin{align}
		\notag
		\lon{\qmn{\rno_{1},\!\Wm\!,\costc_{2}}-\qmn{\rno_{2},\!\Wm\!,\costc_{2}}}
		&\leq\sqrt{\tfrac{2\abs{\CRC{\rno_{1}}{\!\Wm\!}{\costc_{2}}-\CRC{\rno_{2}}{\!\Wm\!}{\costc_{2}}}}{\rno_{1}\wedge \rno_{2}}}
		\\
		\label{eq:CCcapacity:zeroonedichotomy-3}
		&\leq\sqrt{\tfrac{2\abs{\rno_{2}-\rno_{1}}}{(\rno_{1}\wedge \rno_{2})^{2} (1-\rno_{1}\vee \rno_{2})}(\CRC{\rno_{1}}{\!\Wm\!}{\costc_{1}}+\abp{\lgm_{\rno_{1},\!\Wm\!,\costc_{1}} \cdot (\costc_{2}-\costc_{1})})}.
		\end{align}
		In order to bound \(\lon{\qmn{\rno_{1},\!\Wm\!,\costc_{1}}-\qmn{\rno_{1},\!\Wm\!,\costc_{2}}}\), 
		we use triangle inequality once more
		\begin{align}
		\notag
		\lon{\qmn{\rno_{1},\!\Wm\!,\costc_{1}}-\qmn{\rno_{1},\!\Wm\!,\costc_{2}}} 
		&\leq 
		\lon{\qmn{\rno_{1},\!\Wm\!,\costc_{1}}-\qmn{\rno_{1},\Wm,\costc_{\vee}}} 
		+\lon{\qmn{\rno_{1},\Wm,\costc_{\vee}}-\qmn{\rno_{1},\!\Wm\!,\costc_{2}}} 
		\end{align}
		where \(\costc_{\vee}=\costc_{1}\vee\costc_{2}\), i.e. \(\costc_{\vee}^{\ind}=\costc_{1}^{\ind}\vee\costc_{2}^{\ind}\) for each \(\ind\in\{1,\ldots,\ell\}\).

		On the other hand by Lemma \ref{lem:EHB} we have
		\begin{align}
		\notag
		\sup\nolimits_{\mP \in \cset(\costc_{1})} \CRD{\rno_{1}}{\Wm}{\qmn{\rno_{1},\Wm,\costc_{\vee}}}{\mP}
		&\geq  \CRC{\rno_{1}}{\!\Wm\!}{\costc_{1}}
		+\RD{1\wedge\rno_{1}}{\qmn{\rno_{1},\!\Wm\!,\costc_{1}}}{\qmn{\rno_{1},\Wm,\costc_{\vee}}}.
		\end{align}
		Since \(\cset(\costc_{1}) \subset \cset(\costc_{\vee})\), using Theorem \ref{thm:minimax} and
		Lemma \ref{lem:divergence-pinsker}  we get,
		\begin{align}
		\notag
		\CRC{\rno_{1}}{\Wm}{\costc_{\vee}}-\CRC{\rno_{1}}{\!\Wm\!}{\costc_{1}}
		&\geq \tfrac{\rno_{1}}{2} \lon{\qmn{\rno_{1},\!\Wm\!,\costc_{1}}-\qmn{\rno_{1},\Wm,\costc_{\vee}}}^{2}.
		\end{align}
		Repeating the same analysis for \(\lon{\qmn{\rno_{1},\!\Wm\!,\costc_{2}}-\qmn{\rno_{1},\Wm,\costc_{\vee}}}\)
		and bounding \(\CRC{\rno_{1}}{\Wm}{\costc_{\vee}}\) using \eqref{eq:CCcapacity:interior} we get
		\begin{align}
		\label{eq:CCcapacity:zeroonedichotomy-4}
		\lon{\qmn{\rno_{1},\!\Wm\!,\costc_{1}}-\qmn{\rno_{1},\!\Wm\!,\costc_{2}}} 
		\leq \sqrt{\tfrac{2}{\rno_{1}}} \left(
		\sqrt{\CRC{\rno_{1}}{\!\Wm\!}{\costc_{1}}-\CRC{\rno_{1}}{\!\Wm\!}{\costc_{2}}+\lgm_{\rno_{1},\!\Wm\!,\costc_{1}} \cdot (\costc_{\vee}-\costc_{1})}
		+\sqrt{\lgm_{\rno_{1},\!\Wm\!,\costc_{1}} \cdot (\costc_{\vee}-\costc_{1})}
		\right) 
		\end{align}
		The continuity of \(\qmn{\rno,\!\Wm\!,\costc}\) for the total variation topology on \(\pmea{\outA}\)
		follows from \eqref{eq:CCcapacity:zeroonedichotomy-2}, \eqref{eq:CCcapacity:zeroonedichotomy-3},
		\eqref{eq:CCcapacity:zeroonedichotomy-4} and the continuity of Augustin capacity as a function of 
		the constraint established in part (\ref{CCcapacity:function}).
	\end{enumerate}
\end{proof}

\begin{proof}[Proof of Lemma \ref{lem:CCcapacityproduct}]
	Let \(\set{B}(\costc)\) be 
	\begin{align}
	\notag
	\set{B}(\costc)
	&\DEF\left\{(\costc_{1},\ldots,\costc_{\blx}):\sum\nolimits_{\tin=1}^{\blx} \costc_{\tin}\leq \costc,\costc_{\tin}\in\fcc{\costf_{\tin}}\right\}.  
	\end{align}
	Note that if \(\set{B}(\costc)=\emptyset\),  
	then \(\costc\notin\fcc{\costf_{[1,\blx]}}\) and
	\(\CRC{\rno}{\Wmn{[1,\blx]}}{\costc}=-\infty\).	
	On the other hand,
	\(\sum\nolimits_{\tin=1}^{\blx} \CRC{\rno}{\!\Wmn{\tin}}{\costc_{\tin}}\) is minus infinity 
	for \((\costc_{1},\ldots,\costc_{\blx})\)'s that are outside  \(\set{B}(\costc)\)
	by the convention stated in the lemma.
	Thus \eqref{eq:lem:CCcapacityproduct} holds for \(\costc\in\reals{\geq0}^{\ell}\setminus\fcc{\costf_{[1,\blx]}}\) 
	case and the constraints \(\costc_{\tin}\in \reals{\geq0}^{\ell}\) can be replaced by 
	\(\costc_{\tin}\in\fcc{\costf_{\tin}}\) for \(\costc\in\fcc{\costf_{[1,\blx]}}\) case in 
	\eqref{eq:lem:CCcapacityproduct}.
	
	Furthermore, as a result of Lemma \ref{lem:capacityproduct} 
	for any  \((\costc_{1},\ldots,\costc_{\blx})\in \set{B}(\costc)\)
	we have
	\begin{align}
	\notag
	\CRC{\rno}{\Wmn{[1,\blx]}}{\bigtimes\nolimits_{\tin=1}^{\blx} \cset_{\tin}(\costc_{\tin})}
	&=\sum\nolimits_{\tin=1}^{\blx} \CRC{\rno}{\!\Wmn{\tin}}{\costc_{\tin}}.
	\end{align}
	On the other hand \((\bigtimes\nolimits_{\tin=1}^{\blx} \cset_{\tin}(\costc_{\tin}))\subset \cset(\costc)\)
	for any  \((\costc_{1},\ldots,\costc_{\blx})\in \set{B}(\costc)\).
	Thus as a result of Lemma \ref{lem:capacityunion} we have
	\begin{align}
	\notag
	\CRC{\rno}{\Wmn{[1,\blx]}}{\costc}
	&\geq\sup \left\{\sum\nolimits_{\tin=1}^{\blx} \CRC{\rno}{\!\Wmn{\tin}}{\costc_{\tin}}:
	\sum\nolimits_{\tin=1}^{\blx} \costc_{\tin}\leq \costc,~
	\costc_{\tin}\in\fcc{\costf_{\tin}} \right\}.
	\end{align}
	For deriving the reverse inequality,
	first recall that Lemma \ref{lem:information:product} implies
	\(\RMI{\rno}{\mP}{\Wmn{[1,\blx]}}\leq \sum\nolimits_{\tin=1}^{\blx}\RMI{\rno}{\pmn{\tin}}{\Wmn{\tin}}\) 
	for all \(\mP\in\pdis{\inpS_{1}^{\blx}}\) 
	where \(\pmn{\tin}\in\pdis{\inpS_{\tin}}\) is the \(\inpS_{\tin}\)
	marginal of \(\mP\).
	On the other hand, 
	\(\EXS{\mP}{\costf_{[1,\blx]}}=\sum\nolimits_{\tin=1}^{\blx}\EXS{\pmn{\tin}}{\costf_{\tin}}\)
	and \(\EXS{\pmn{\tin}}{\costf_{\tin}}\in \fcc{\costf_{\tin}}\).
	Hence,
	\begin{align}
	\notag
	\sup\nolimits_{\mP: \EXS{\mP}{\costf_{[1,\blx]}}\leq \costc}\RMI{\rno}{\mP}{\Wmn{[1,\blx]}}
	&\leq \sup\nolimits_{\pmn{1},\ldots,\pmn{\blx}:\sum\nolimits_{\tin=1}^{\blx}\EXS{\pmn{\tin}}{\costf_{\tin}}\leq \costc} \sum\nolimits_{\tin=1}^{\blx}\RMI{\rno}{\pmn{\tin}}{\Wmn{\tin}}
	\\
	\notag
	&=\sup \left\{\sum\nolimits_{\tin=1}^{\blx} \CRC{\rno}{\!\Wmn{\tin}}{\costc_{\tin}}:
	\sum\nolimits_{\tin=1}^{\blx} \costc_{\tin}\leq \costc,~
	\costc_{\tin}\in\fcc{\costf_{\tin}} \right\}.
	\end{align}
	Thus \eqref{eq:lem:CCcapacityproduct} holds.
	In addition, \(\cset(\costc)\) can be interpreted as the union of 
	\((\bigtimes\nolimits_{\tin=1}^{\blx} \cset_{\tin}(\costc_{\tin}))\) and
	\(\cset(\costc)\setminus (\bigtimes\nolimits_{\tin=1}^{\blx} \cset_{\tin}(\costc_{\tin}))\).
	Therefore, if there exists a \((\costc_{1},\ldots,\costc_{\blx})\) such that
	\(\CRC{\rno}{\!\Wmn{[1,\blx]}}{\costc}=\sum\nolimits_{\tin=1}^{\blx} \CRC{\rno}{\!\Wmn{\tin}}{\costc_{\tin}}\)
	and \(\CRC{\rno}{\!\Wmn{[1,\blx]}}{\costc}<\infty\),
	then \(\qmn{\rno,\!\Wmn{[1,\blx]},\costc}=\bigotimes\nolimits_{\tin=1}^{\blx} \qmn{\rno,\!\Wmn{\tin}\!,\costc_{\tin}}\)
	because 
	\(\CRC{\rno}{\!\Wmn{[1,\blx]}}{\costc}=\CRC{\rno}{\Wmn{[1,\blx]}}{\bigtimes\nolimits_{\tin=1}^{\blx} \cset_{\tin}(\costc_{\tin})}\)
	and \(\CRC{\rno}{\!\Wmn{[1,\blx]}}{\costc}<\infty\) imply
	\(\qmn{\rno,\!\Wmn{[1,\blx]},\costc}=\qmn{\rno,\Wmn{[1,\blx]},\bigtimes\nolimits_{\tin=1}^{\blx} \cset_{\tin}(\costc_{\tin})}\)
	by Lemma \ref{lem:capacityunion}
	and
	\(\qmn{\rno,\Wmn{[1,\blx]},\bigtimes\nolimits_{\tin=1}^{\blx} \cset_{\tin}(\costc_{\tin})}=\bigotimes\nolimits_{\tin=1}^{\blx} \qmn{\rno,\!\Wmn{\tin}\!,\costc_{\tin}}\)
	by  Lemma \ref{lem:capacityproduct}. 
\end{proof}

%\subsubsection{Proof of Lemmas on A-L Capacity}\label{sec:ALproofs}
\begin{proof}[Proof of Lemma \ref{lem:Lcapacity}]~
	\begin{enumerate}[(a)]
		\item[(\ref{Lcapacity:function})]
		\(\RCL{\rno}{\Wm}{\lgm}\) is convex, nonincreasing, and lower semicontinuous in \(\lgm\) because 
		\(\RCL{\rno}{\Wm}{\lgm}\) is the pointwise supremum of such functions
		as a result of  \eqref{eq:Lcapacity-astheconjugate}.
		
		Since \(\RCL{\rno}{\Wm}{\lgm}\) is convex it is continuous on the interior of \(\{\lgm\in\reals{\geq0}^{\ell}:\RCL{\rno}{\Wm}{\lgm}<\infty\}\)
		by \cite[Thm. 6.3.4]{dudley}. The interior of \(\{\lgm\in\reals{\geq0}^{\ell}:\RCL{\rno}{\Wm}{\lgm}<\infty\}\) is
		\(\{\lgm\!\in\reals{\geq0}^{\ell}\!:\exists\epsilon>0~s.t.~ \RCL{\rno}{\Wm}{\lgm-\epsilon\uc}<\infty\}\)
		because \(\RCL{\rno}{\Wm}{\lgm}\)  is nonincreasing in \(\lgm\).
		
		\item[(\ref{Lcapacity:minimax})]
		Note that \(\CRC{\rno}{\!\Wm\!}{\costc}=\sup\nolimits_{\mP\in\pdis{\inpS}}\xi_{\rno,\mP}(\costc)\)
		as a result of \eqref{eq:def:costcapacity} and \eqref{eq:def:information-constrained}.
		Then as a result of  \eqref{eq:information-constrained},
		we have
		\begin{align}
		\label{eq:Lcostcapacity}
		\CRC{\rno}{\!\Wm\!}{\costc}
		&=\sup\nolimits_{\mP\in\pdis{\inpS}}\inf\nolimits_{\lgm\geq 0}
		\RMIL{\rno}{\mP}{\Wm}{\lgm}+\lgm\cdot\costc
		&
		&\forall \costc\in\reals{\geq0}^{\ell}.
		\end{align}
		If \(\inpS\) is finite, then \(\pdis{\inpS}\) is compact.
		Furthermore, using \eqref{eq:capacityFLB-1} together with triangle inequality we get
		\begin{align}
		\label{eq:Lcapacity:informationContinuity}
		\abs{\RMIL{\rno}{\mP_{2}}{\Wm}{\lgm}-\RMIL{\rno}{\mP_{1}}{\Wm}{\lgm}}
		&\leq \bent{\tfrac{\lon{\pmn{1}-\pmn{2}}}{2}} 
		+\tfrac{\lon{\pmn{1}-\pmn{2}}}{2} \ln \abs{\inpS} 
		+\tfrac{\lon{\pmn{1}-\pmn{2}}}{2} \max\nolimits_{\dinp\in \inpS} \lgm\cdot \costf(\dinp).
		\end{align}
		Then \(\RMIL{\rno}{\mP}{\Wm}{\lgm}+\lgm\cdot\costc\) is continuous in \(\mP\) on \(\pdis{\inpS}\).
		On the other hand, \(\RMIL{\rno}{\mP}{\Wm}{\lgm}+\lgm\cdot\costc\) is concave in \(\mP\) 
		by Lemma \ref{lem:informationP} and convex and continuous in \(\lgm\).
		Thus we can change the order of the infimum and supremum in \eqref{eq:Lcostcapacity}
		---using the Sion's minimax theorem, \cite{komiya88}, \cite{sion58}---
		and \(\CRC{\rno}{\!\Wm\!}{\costc}=\inf\nolimits_{\lgm\geq0} \RCL{\rno}{\Wm}{\lgm}+\lgm\cdot\costc\)
		by \eqref{eq:def:Lcapacity}.
		
		\item[(\ref{Lcapacity:interior})]
		If \(\costc\in\inte{\fcc{\costf}}\) and \(\CRC{\rno}{\!\Wm\!}{\costc}\) is infinite, 
		then
		\(\CRC{\rno}{\!\Wm\!}{\costc}=\inf\nolimits_{\lgm\geq0} \RCL{\rno}{\Wm}{\lgm}+\lgm\cdot\costc\)
		follows from \eqref{eq:CCcapacity:dualitybound} trivially.
		
		If \(\costc\in\inte{\fcc{\costf}}\) and \(\CRC{\rno}{\!\Wm\!}{\costc}\) is finite, 
		then there  exists a non-empty, convex, and compact  set of 
		\(\lgm_{\rno,\!\Wm\!,\costc}\)'s satisfying \eqref{eq:CCcapacity:interior}
		by 
		Lemma \ref{lem:CCcapacity}-(\ref{CCcapacity:interior}). 
		Furthermore, \eqref{eq:Lcapacity-astheconjugate} implies
		for any \(\lgm_{\rno,\!\Wm\!,\costc}\) satisfying  \eqref{eq:CCcapacity:interior}
		the following identity
		\begin{align}
		\notag
		\RCL{\rno}{\Wm}{\lgm_{\rno,\!\Wm\!,\costc}}
		&=\sup\nolimits_{\tilde{\costc}\geq0}
		\CRC{\rno}{\Wm}{\tilde{\costc}}-\lgm_{\rno,\!\Wm\!,\costc}\cdot \tilde{\costc}
		\\
		\notag
		&= \CRC{\rno}{\Wm}{\costc}-\lgm_{\rno,\!\Wm\!,\costc}\cdot \costc.
		\end{align}
		Then \(\CRC{\rno}{\Wm}{\costc}=\inf_{\lgm\geq 0} \RCL{\rno}{\Wm}{\lgm}+\lgm\cdot \costc\)
		by \eqref{eq:CCcapacity:dualitybound}.
		
		\item[(\ref{Lcapacity:optimal})] 
		Note that \(\RMIL{\rno}{\mP}{\Wm}{\lgm}\leq \RCL{\rno}{\Wm}{\lgm}\) by definition. 
		Hence \(\limsup_{\ind \to \infty} \RMIL{\rno}{\pma{}{(\ind)}}{\Wm}{\lgm}
		\leq\RCL{\rno}{\Wm}{\lgm}\). 
		Furthermore,
		\begin{align}
		\notag
		\RMIL{\rno}{\mP}{\Wm}{\lgm}
		&\geq \RMI{\rno}{\mP}{\Wm}-\lgm\cdot\costc
		&
		&\forall \mP\in\cset(\costc).
		\end{align}
		Thus \(\liminf_{\ind \to \infty} \RMIL{\rno}{\pma{}{(\ind)}}{\Wm}{\lgm}
		\geq \CRC{\rno}{\!\Wm\!}{\costc}-\lgm\cdot\costc
		=\RCL{\rno}{\Wm}{\lgm}\), as well.
		Thus \(\lim_{\ind \to \infty} \RMIL{\rno}{\pma{}{(\ind)}}{\Wm}{\lgm}=\RCL{\rno}{\Wm}{\lgm}\).
	\end{enumerate}
		\end{proof}

\begin{proof}[Proof of Lemma \ref{lem:LcapacityFLB}]~
	\begin{enumerate}[(i)]
		\item {\it \(\forall\rno\in\reals{+}\)\(\exists \widetilde{\mP}\in\pdis{\inpS}\) such that \(\RMIL{\rno}{\widetilde{\mP}}{\Wm}{\lgm}=\RCL{\rno}{\Wm}{\lgm}\):}
		Note that
		\(\RMIL{\rno}{\mP}{\Wm}{\lgm}\) is continuous in \(\mP\) on \(\pdis{\inpS}\)
		by 	\eqref{eq:Lcapacity:informationContinuity}.
		On the other hand, \(\pdis{\inpS}\) is compact because \(\inpS\) is a finite set.
		Then there exists a \(\widetilde{\mP}\in \pdis{\inpS}\) such that 
		\(\RMIL{\rno}{\widetilde{\mP}}{\Wm}{\lgm}=\sup_{\mP \in \pdis{\inpS}} \RMIL{\rno}{\mP}{\Wm}{\lgm}\)
		by the extreme value theorem, \cite[27.4]{munkres}.
		
		\item {\it If \(\rno\in\reals{+}\) and \(\RMIL{\rno}{\widetilde{\mP}}{\Wm}{\lgm}=\RCL{\rno}{\Wm}{\lgm}\), 
			then
			\(\CRD{\rno}{\Wm}{\qmn{\rno,\widetilde{\mP}}}{\mP}-\lgm\cdot\EXS{\mP}{\costf}\leq \RCL{\rno}{\Wm}{\lgm}\)
			for all \(\mP\in \pdis{\inpS}\):}
		Let  
		\(\mP\) be any member of \(\pdis{\inpS}\) and \(\pma{}{(\ind)}\) be 
		\(\tfrac{\ind-1}{\ind}\widetilde{\mP} +\tfrac{1}{\ind}{\mP}\) for  \(\ind \in \integers{+}\). 
		Then by Lemma \ref{lem:information}
		\begin{align}
		\notag
		\RMIL{\rno}{\pma{}{(\ind)}}{\Wm}{\lgm}
		&\!=\!
		\tfrac{\ind-1}{\ind} \left[\CRD{\rno}{\Wm}{\qmn{\rno,\pma{}{(\ind)}}}{\widetilde{\mP}}-\lgm\cdot\EXS{\widetilde{\mP}}{\costf}\right]
		\!+\!\tfrac{1}{\ind} 
		\left[\CRD{\rno}{\Wm}{\qmn{\rno,\pma{}{(\ind)}}}{\mP}-\lgm\cdot\EXS{\mP}{\costf}\right]
		\\
		\notag
		&\!\geq \!
		\tfrac{\ind-1}{\ind} 
		\left[\RMIL{\rno}{\widetilde{\mP}}{\Wm}{\lgm}+\RD{\rno \wedge 1}{\qmn{\rno,\widetilde{\mP}}}{\qmn{\rno,\pma{}{(\ind)}}}\right]
		\!+\!\tfrac{1}{\ind}
		\left[\CRD{\rno}{\Wm}{\qmn{\rno,\pma{}{(\ind)}}}{\mP}-\lgm\cdot\EXS{\mP}{\costf}\right]
		&
		&\forall \ind \in \integers{+}.
		\end{align}
		Then using \(\RMIL{\rno}{\pma{}{(\ind)}}{\Wm}{\lgm} \leq \RCL{\rno}{\Wm}{\lgm}\),
		\(\RMIL{\rno}{\widetilde{\mP}}{\Wm}{\lgm}=\RCL{\rno}{\Wm}{\lgm}\), and 
		\(\RD{\rno \wedge 1}{\qmn{\rno,\widetilde{\mP}}}{\qmn{\rno,\pma{}{(\ind)}}}\geq0\)
		we get
		\begin{align}
		\label{eq:LcapacityFLB-b-A}
		\RCL{\rno}{\Wm}{\lgm}
		&\geq \CRD{\rno}{\Wm}{\qmn{\rno,\pma{}{(\ind)}}}{\mP}-\lgm\cdot\EXS{\mP}{\costf}.
		\end{align}
		On the other hand, using \(\RMIL{\rno}{\pma{}{(\ind)}}{\Wm}{\lgm}\leq \RCL{\rno}{\Wm}{\lgm}\),
		\(\RMIL{\rno}{\widetilde{\mP}}{\Wm}{\lgm}=\RCL{\rno}{\Wm}{\lgm}\) and 
		\(\CRD{\rno}{\Wm}{\qmn{\rno,\pma{}{(\ind)}}}{\mP}\geq 0\) we get
		\begin{align}
		\notag
		\tfrac{\RCL{\rno}{\Wm}{\lgm}+\lgm\cdot \EXS{\mP}{\costf}}{\ind}
		&\geq \tfrac{\ind-1}{\ind}\RD{\rno \wedge 1}{\qmn{\rno,\widetilde{\mP}}}{\qmn{\rno,\pma{}{(\ind)}}}
		&
		&\forall \ind \in \integers{+}.
		\end{align}
		Then using, Lemma \ref{lem:divergence-pinsker} we get
		\begin{align}
		\notag
		\sqrt{\tfrac{2}{\rno \wedge 1}\tfrac{\RCL{\rno}{\Wm}{\lgm}+\lgm\cdot \EXS{\mP}{\costf}}{\ind-1}}
		&\geq \lon{\qmn{\rno,\widetilde{\mP}}-\qmn{\rno,\pma{}{(\ind)}}}
		&
		&\forall \ind \in \integers{+}.
		\end{align}
		Thus \(\qmn{\rno,\pma{}{(\ind)}}\) converges to \(\qmn{\rno,\widetilde{\mP}}\) in the total variation topology and hence
		in the topology of setwise convergence. 
		Since the \renyi divergence is lower semicontinuous in the topology of setwise convergence by
		Lemma \ref{lem:divergence:lsc}, we have
		\begin{align}
		\label{eq:LcapacityFLB-b-B}
		\liminf\nolimits_{\ind \to \infty} \CRD{\rno}{\Wm}{\qmn{\rno,\pma{}{(\ind)}}}{\mP}
		&\geq  \CRD{\rno}{\Wm}{\qmn{\rno,\widetilde{\mP}}}{\mP}.
		\end{align}
		Then the inequality \(\CRD{\rno}{\Wm}{\qmn{\rno,\widetilde{\mP}}}{\mP}-\lgm\cdot\EXS{\mP}{\costf}\leq \RCL{\rno}{\Wm}{\lgm}\)
		follows from \eqref{eq:LcapacityFLB-b-A}  and \eqref{eq:LcapacityFLB-b-B}.
		
		\item {\it If \(\rno\in\reals{+}\), 
			then \(\exists!\qma{\rno,\Wm}{\lgm}\!\in\!\pmea{\outA}\) satisfying \eqref{eq:lem:LcapacityFLB}
			such that \(\qmn{\rno,\mP}\!=\!\qma{\rno,\Wm}{\lgm}\) for all \(\mP\in\pdis{\inpS}\) satisfying 
			\(\RMIL{\rno}{\mP}{\Wm}{\lgm}\!=\!\RCL{\rno}{\Wm}{\lgm}\):}
If \(\RMIL{\rno}{\mP}{\Wm}{\lgm}\!=\!\RCL{\rno}{\Wm}{\lgm}\)
for a \(\mP\in\pdis{\inpS}\), 
then Lemma \ref{lem:information}-(\ref{information:one},\ref{information:zto},\ref{information:oti})
and Lemma \ref{lem:divergence-pinsker} imply
\begin{align}
\label{eq:LcapacityFLB-c}
\CRD{\rno}{\Wm}{\qmn{\rno,\widetilde{\mP}}}{\mP}-\lgm\cdot \EXS{\mP}{\costc}\geq \RCL{\rno}{\!\Wm\!}{\lgm}+
\tfrac{\rno\wedge 1}{2}\lon{\qmn{\rno,\mP}-\qmn{\rno,\widetilde{\mP}}}^{2}.
\end{align}
Since we have already established that
\(\CRD{\rno}{\Wm}{\qmn{\rno,\widetilde{\mP}}}{\mP}-\lgm\cdot \EXS{\mP}{\costc}\leq \RCL{\rno}{\!\Wm\!}{\lgm}\)
for any \(\mP\in\pdis{\inpS}\),
\eqref{eq:LcapacityFLB-c} implies that \(\qmn{\rno,\mP}=\qmn{\rno,\widetilde{\mP}}\)
for any \(\mP\in\pdis{\inpS}\) 
satisfying \(\RMIL{\rno}{\mP}{\Wm}{\lgm}=\RCL{\rno}{\!\Wm\!}{\lgm}\).
\end{enumerate}
\end{proof}

\begin{proof}[Proof of Theorem \ref{thm:Lminimax}]
	First note that \eqref{eq:thm:Lminimax} implies \eqref{eq:thm:Lminimaxradius} and 
	\eqref{eq:thm:Lminimaxcenter} implies \eqref{eq:thm:Lminimaxradiuscenter}.
	Furthermore, the left hand side of \eqref{eq:thm:Lminimax} is equal to \(\RCL{\rno}{\Wm}{\lgm}\) by
	\eqref{eq:Lcapacity}.
	Thus when \(\RCL{\rno}{\Wm}{\lgm}\) is infinite, \eqref{eq:thm:Lminimax} holds trivially by the max-min inequality.
	When \(\RCL{\rno}{\Wm}{\lgm}\) is finite, \eqref{eq:thm:Lminimax} follows from \eqref{eq:thm:Lminimaxcenter}
	and the max-min inequality.
	Thus we can assume \(\RCL{\rno}{\Wm}{\lgm}\) to be finite and prove  the claims about \(\qma{\rno,\Wm}{\lgm}\), 
	in order to prove the theorem.
	
	\begin{enumerate}[(i)]		
		\item {\it If \(\RCL{\rno}{\Wm}{\lgm}\!<\!\infty\) and \(\lim\nolimits_{\ind \to \infty}\!\RMIL{\rno}{\pma{}{(\ind)}}{\Wm}{\lgm}\!=\!\RCL{\rno}{\Wm}{\lgm}\),
			then \(\{\qmn{\rno,\pma{}{(\ind)}}\}_{\ind\in\integers{+}}\) is a Cauchy 
			sequence in \(\pmea{\outA}\) for the total variation metric:}
		For any sequence \(\{\pma{}{(\ind)}\}_{\ind\in \integers{+}} \subset \pdis{\inpS}\) satisfying 
		\(\lim\nolimits_{\ind \to \infty} \RMIL{\rno}{\pma{}{(\ind)}}{\Wm}{\lgm}=\RCL{\rno}{\Wm}{\lgm}\), let 
		us consider a sequence of channels \(\{\Wm^{(\ind)}\}_{\ind\in\integers{+}}\) 
		whose input sets \(\{\inpS^{(\ind)}\}_{\ind\in\integers{+}}\) 
		form a nested sequence of finite subsets of \(\inpS\) defined as follows, 
		\begin{align}
		\notag
		\inpS^{(\ind)} 
		&\DEF\{\dinp\in\inpS:\exists \jnd\in\{1,\ldots,\ind\} \mbox{~such that~} \pma{}{(\jnd)}(\dinp)>0\}.
		\end{align}
		Then for any \(\ind\in\integers{+}\), there exists a unique \(\qma{\rno,\Wm^{(\ind)}}{\lgm}\) satisfying 
		\eqref{eq:lem:LcapacityFLB} by Lemma \ref{lem:LcapacityFLB}.
		Furthermore, \(\pdis{\inpS^{(\jnd)}}\subset\pdis{\inpS^{(\ind)}}\) for any \(\ind,\jnd\in\integers{+}\) such that \(\jnd\leq \ind\).
		In order to bound
		\(\lon{\qmn{\rno,\pma{}{(\jnd)}}-\qmn{\rno,\pma{}{(\ind)}}}\) for positive integers \(\jnd<\ind\) we use the triangle
		inequality for \(\qmn{\rno,\pma{}{(\jnd)}}\), \(\qmn{\rno,\pma{}{(\ind)}}\), and \(\qma{\rno,\Wm^{(\ind)}}{\lgm}\)
		\begin{align}
		\label{eq:Lminimax-A}
		\lon{\qmn{\rno,\pma{}{(\jnd)}}-\qmn{\rno,\pma{}{(\ind)}}}
		\leq \lon{\qmn{\rno,\pma{}{(\jnd)}}-\qma{\rno,\Wm^{(\ind)}}{\lgm}}
		+\lon{\qmn{\rno,\pma{}{(\ind)}}-\qma{\rno,\Wm^{(\ind)}}{\lgm}}.
		\end{align} 
		Let us proceed with bounding \(\lon{\qmn{\rno,\pma{}{(\jnd)}}-\qma{\rno,\Wm^{(\ind)}}{\lgm}}\)
		and \(\lon{\qmn{\rno,\pma{}{(\ind)}}-\qma{\rno,\Wm^{(\ind)}}{\lgm}}\) from above.
		\begin{align}
		\notag
		\lon{\qmn{\rno,\pma{}{(\jnd)}}-\qma{\rno,\Wm^{(\ind)}}{\lgm}}
		&\mathop{\leq}^{(a)} 
		\sqrt{\tfrac{2}{\rno \wedge 1}\RD{\rno \wedge 1}{\qmn{\rno,\pma{}{(\jnd)}}}{\qma{\rno,\Wm^{(\ind)}}{\lgm}}}
		\\
		\notag
		&\mathop{\leq}^{(b)} \sqrt{\tfrac{2}{\rno \wedge 1}}
		\sqrt{\CRD{\rno}{\Wm}{\qma{\rno,\Wm^{(\ind)}}{\lgm}}{\pma{}{(\jnd)}} -\RMI{\rno}{\pma{}{(\jnd)}}{\Wm^{(\ind)}}}
		\\
		\notag
		&\mathop{\leq}^{(c)} \sqrt{\tfrac{2}{\rno \wedge 1}}
		\sqrt{\RCL{\rno}{\Wm^{(\ind)}}{\lgm}-\RMIL{\rno}{\pma{}{(\jnd)}}{\Wm^{(\ind)}}{\lgm}}
		\\
		\notag
		&\mathop{\leq}^{(d)} \sqrt{\tfrac{2}{\rno \wedge 1}}
		\sqrt{\RCL{\rno}{\Wm}{\lgm}-\RMIL{\rno}{\pma{}{(\jnd)}}{\Wm}{\lgm}}
		\end{align}
		where 
		\((a)\) follows from Lemma \ref{lem:divergence-pinsker}, 
		\((b)\) follows from Lemma \ref{lem:information}-(\ref{information:one},\ref{information:zto},\ref{information:oti}),
		\((c)\) follows from Lemma \ref{lem:LcapacityFLB} because \(\pma{}{(\jnd)}\in \pdis{\inpS^{(\ind)}}\), 
		and 
		\((d)\) follows from the identities 
		\(\RMIL{\rno}{\pma{}{(\jnd)}}{\Wm^{(\ind)}}{\lgm}=\RMIL{\rno}{\pma{}{(\jnd)}}{\Wm}{\lgm}\) 
		and \(\RCL{\rno}{\Wm^{(\ind)}}{\lgm}\leq \RCL{\rno}{\Wm}{\lgm}\).
		We can obtain a similar bound on \(\lon{\qmn{\rno,\pma{}{(\ind)}}-\qma{\rno,\Wm^{(\ind)}}{\lgm}}\).
		Then \(\{\qmn{\rno,\pma{}{(\ind)}}\}\) is a Cauchy sequence as a result of \eqref{eq:Lminimax-A}
		because \(\lim_{\jnd\to\infty}\RMIL{\rno}{\pma{}{(\jnd)}}{\Wm}{\lgm}=\RCL{\rno}{\Wm}{\lgm}\).

		\item {\it If \(\RCL{\rno}{\Wm}{\lgm}<\infty\), 
			then \(\exists!\qma{\rno,\Wm}{\lgm}\in \pmea{\outA}\)
			satisfying 
			\(\lim\nolimits_{\ind\to \infty}\lon{\qma{\rno,\Wm}{\lgm}-\qmn{\rno,\pma{}{(\ind)}}}=0\) for all 
			\(\{\pma{}{(\ind)}\}_{\ind\in\integers{+}}\subset \pdis{\inpS}\) such that  
			\(\lim\nolimits_{\ind \to \infty} \RMIL{\rno}{\pma{}{(\ind)}}{\Wm}{\lgm}=\RCL{\rno}{\Wm}{\lgm}\):}
		Note that \(\smea{\outA}\) is a complete metric space for the total variation metric
		because \(\smea{\outA}\) is a Banach space for the total variation topology \cite[Thm. 4.6.1]{bogachev}.
		Then \(\{\qmn{\rno,\pma{}{(\ind)}}\}_{\ind\in\integers{+}}\) has a unique limit point \(\qmn{\rno,\pma{}{*}}\) 
		in \(\smea{\outA}\). 
		Since \(\pmea{\outA}\) is a closed set for the total variation topology and 
		\(\cup_{\ind\in\integers{+}}\qmn{\rno,\pma{}{(\ind)}}\subset \pmea{\outA}\),
		then \(\qmn{\rno,\pmn{}{*}}\in\pmea{\outA}\), by \cite[Thm. 2.1.3]{munkres}.

		We have established the existence of a unique limit point \(\qmn{\rno,\pmn{}{*}}\),
		for any sequence  \(\{\pma{}{(\ind)}\}_{\ind\in\integers{+}}\subset\pdis{\inpS}\) satisfying 
		\(\lim\nolimits_{\ind \to \infty}\RMIL{\rno}{\pma{}{(\ind)}}{\Wm}{\lgm}=\RCL{\rno}{\Wm}{\lgm}\).
		This, however, implies \(\lim_{\ind\to\infty}\lon{\qmn{\rno,\widetilde{\mP}^{(\ind)}}-\qmn{\rno,\pmn{}{*}}}=0\)
		for any \(\{\widetilde{\mP}^{(\ind)}\}_{\ind\in\integers{+}}\) satisfying 
		\(\lim\nolimits_{\ind \to \infty}\RMIL{\rno}{\widetilde{\mP}^{(\ind)}}{\Wm}{\lgm}=\RCL{\rno}{\Wm}{\lgm}\)
		because we can interleave the elements of \(\{\pma{}{(\ind)}\}_{\ind\in\integers{+}}\) and 
		\(\{\widetilde{\mP}^{(\ind)}\}_{\ind\in\integers{+}}\) to obtain a new sequence 
		\(\{\widehat{\mP}^{(\ind)}\}_{\ind\in\integers{+}}\)
		satisfying 
		\(\lim\nolimits_{\ind \to \infty}\RMIL{\rno}{\widehat{\mP}^{(\ind)}}{\Wm}{\lgm}=\RCL{\rno}{\Wm}{\lgm}\)
		for which \(\{\qmn{\rno,\widehat{\mP}^{(\ind)}}\}\) is a Cauchy sequence. Then 
		\(\qma{\rno,\Wm}{\lgm}=\qmn{\rno,\pma{}{*}}\)
		
		\item {\it \(\qma{\rno,\Wm}{\lgm}\) satisfies the equality given in \eqref{eq:thm:Lminimaxcenter}:}
		For any \(\mP \in \pdis{\inpS}\), let us consider any sequence \(\{\pma{}{(\ind)}\}_{\ind\in\integers{+}}\) 
		satisfying \(\pma{}{(1)}=\mP\) and
		\(\lim\nolimits_{\ind \to \infty} \RMIL{\rno}{\pma{}{(\ind)}}{\Wm}{\lgm}=\RCL{\rno}{\Wm}{\lgm}\).
		Then \(\mP\in \pdis{\inpS^{(\ind)}}\) for all \(\ind\in\integers{+}\). 
		Using Lemma \ref{lem:LcapacityFLB} we get
		\begin{align}
		\label{eq:LcapacityLB-D}
		\CRD{\rno}{\Wm}{\qmn{\rno,\Wm^{(\ind)}}}{\mP}-\lgm\cdot \EXS{\mP}{\costf}
		&\leq  \RCL{\rno}{\Wm^{(\ind)}}{\lgm}
		&
		&\forall \ind\in\integers{+}.
		\end{align}
		Since \(\inpS^{(\ind)}\) is a finite set, \(\exists\widetilde{\mP}^{(\ind)}\in \pdis{\inpS^{(\ind)}}\) 
		satisfying \(\RMIL{\rno}{\widetilde{\mP}^{(\ind)}}{\Wm^{(\ind)}}{\lgm}=\RCL{\rno}{\Wm^{(\ind)}}{\lgm}\)
		and \(\qmn{\rno,\widetilde{\mP}^{(\ind)}}=\qma{\rno,\Wm^{(\ind)}}{\lgm}\)  by Lemma \ref{lem:LcapacityFLB}.
		Then \(\RMIL{\rno}{\widetilde{\mP}^{(\ind)}}{\Wm^{(\ind)}}{\lgm}\geq \RMIL{\rno}{\pma{}{(\ind)}}{\Wm^{(\ind)}}{\lgm}\) 
		and consequently \(\lim\nolimits_{\ind \to \infty} \RMIL{\rno}{\widetilde{\mP}^{(\ind)}}{\Wm}{\lgm}=\RCL{\rno}{\Wm}{\lgm}\). 
		We have already established that for such a sequence
		\(\qmn{\rno,\widetilde{\mP}^{(\ind)}}\rightarrow \qma{\rno,\Wm}{\lgm}\) in the total variation topology, 
		and hence in the topology of setwise convergence. Then
		the lower semicontinuity  of the \renyi divergence in its arguments for the topology of setwise convergence, 
		i.e. Lemma \ref{lem:divergence:lsc},
		the identity \(\RCL{\rno}{\Wm^{(\ind)}}{\lgm}\leq\RCL{\rno}{\Wm}{\lgm}\), and \eqref{eq:LcapacityLB-D} 
		imply that
		\begin{align}
		\notag
		\CRD{\rno}{\Wm}{\qma{\rno,\Wm}{\lgm}}{\mP}-\lgm\cdot \EXS{\mP}{\costf}
		&\leq  \RCL{\rno}{\Wm}{\lgm}
		&
		&\forall \mP\in\pdis{\inpS}.
		\end{align}
		On the other hand \(\CRD{\rno}{\Wm}{\qma{\rno,\Wm}{\lgm}}{\mP}-\lgm\cdot \EXS{\mP}{\costf}\geq \RMIL{\rno}{\mP}{\Wm}{\lgm}\) 
		and \(\RCL{\rno}{\Wm}{\lgm}=\sup_{\mP\in\pdis{\inpS}} \RMIL{\rno}{\mP}{\Wm}{\lgm}\) 
		by the definitions of \(\RMI{\rno}{\mP}{\Wm}\), \(\RMIL{\rno}{\mP}{\Wm}{\lgm}\),  and
		\(\RCL{\rno}{\Wm}{\lgm}\). Thus \eqref{eq:thm:Lminimaxcenter} holds.
	\end{enumerate} 
\end{proof}

\begin{proof}[Proof of Lemma \ref{lem:Lcenter}]
	Let  \(\{\pma{}{(\ind)}\}_{\ind\in\integers{+}}\subset\cset(\costc)\) be 
	such that \(\lim\nolimits_{\ind \to \infty} \RMI{\rno}{\pma{}{(\ind)}}{\Wm}=\CRC{\rno}{\!\Wm\!}{\costc}\).
	Then \(\{\qmn{\rno,\pma{}{(\ind)}}\}_{\ind\in\integers{+}}\) is a Cauchy sequence with the limit point 
	\(\qmn{\rno,\!\Wm\!,\costc}\) by Theorem \ref{thm:minimax}.
	On the other hand,
	\(\lim\nolimits_{\ind \to \infty} \RMIL{\rno}{\pma{}{(\ind)}}{\Wm}{\lgm}=\RCL{\rno}{\Wm}{\lgm}\)
	by Lemma \ref{lem:Lcapacity}-(\ref{Lcapacity:optimal}).
	Then \(\{\qmn{\rno,\pma{}{(\ind)}}\}_{\ind\in\integers{+}}\) is a Cauchy sequence with the limit point 
	\(\qma{\rno,\Wm}{\lgm}\) by Theorem \ref{thm:Lminimax}.
	Hence \(\qmn{\rno,\!\Wm\!,\costc}=\qma{\rno,\Wm}{\lgm}\).
\end{proof}

\begin{proof}[Proof of Lemma \ref{lem:Lcapacityproduct}]
%	For all \(\mP\in\pdis{\inpS_{1}^{\blx}}\)
%	\begin{align}
%	\notag
%	\EXS{\mP}{\costf_{[1,\blx]}}
%	&=\sum\nolimits_{\tin=1}^{\blx}\EXS{\pmn{\tin}}{\costf_{\tin}}
%	\end{align}
%	where \(\pmn{\tin}\in\pdis{\inpS_{\tin}}\) is  the \(\inpS_{\tin}\) marginal of \(\mP\).
%	Thus Lemma \ref{lem:information:product} and
%	the definition of the A-L information given in \eqref{eq:def:Linformation} imply
%	\begin{align}
%	\notag
%	\RMIL{\rno}{\mP}{\Wmn{[1,\blx]}}{\lgm}
%	&\leq \RMIL{\rno}{\pmn{1}\otimes\cdots\otimes\pmn{\blx}}{\Wmn{[1,\blx]}}{\lgm}
%	\\
%	\notag
%	&=\sum\nolimits_{\tin=1}^{\blx}\RMIL{\rno}{\pmn{\tin}}{\Wmn{\tin}}{\lgm}.
%	\end{align}
%	Hence, 
As a result of \eqref{eq:lem:A-L-information:product} we have
	\begin{align}
	\notag
	\sup\nolimits_{\mP\in\pdis{\inpS_{1}^{\blx}}}\RMIL{\rno}{\mP}{\Wmn{[1,\blx]}}{\lgm}
	&=\sup\nolimits_{\pmn{1}\in\pdis{\inpS_{1}},\ldots,\pmn{\blx}\in\pdis{\inpS_{\blx}}}
	\sum\nolimits_{\tin=1}^{\blx}\RMIL{\rno}{\pmn{\tin}}{\Wmn{\tin}}{\lgm}.
	\end{align}
Thus \eqref{eq:lem:Lcapacityproduct} holds.	
	In order to establish \(\qma{\rno,\Wmn{[1,\blx]}}{\lgm}=\bigotimes\nolimits_{\tin=1}^{\blx}\qma{\rno,\Wmn{\tin}}{\lgm}\), one can confirm 
	by substitution that \(\bigotimes\nolimits_{\tin=1}^{\blx}\qma{\rno,\Wmn{\tin}}{\lgm}\) 
	satisfies \eqref{eq:thm:Lminimaxradiuscenter}.
\end{proof}

%\subsubsection{Proof of Lemmas on R-G Capacity}\label{sec:RGproofs}
\begin{proof}[Proof of Lemma \ref{lem:Lpolsha}]~
\begin{enumerate}
	\item[(\ref{Lpolsha:poltyrev})]
	Note that as a result of Lemma \ref{lem:information}-(\ref{information:zto},\ref{information:oti}) 
	and the definition of \(\RMIL{\rno}{\mP}{\Wm}{\lgm}\) given in \eqref{eq:def:Linformation} we have
	\begin{align}
	\label{eq:Lpolsha:poltyrev-1}
	\RD{1}{\mP}{\uma{\rno,\mP}{\lgm}}
	&=(\rno-1)\RMIL{\rno}{\mP}{\Wm}{\lgm}
	+\ln\sum\nolimits_{\tilde{\dinp}}\mP(\tilde{\dinp})e^{(1-\rno)\RD{\rno}{\Wm(\tilde{\dinp})}{\qmn{\rno,\mP}}
		+(\rno-1)\lgm \cdot \costf(\dinp)}.
	\end{align}
	On the other hand as a result of \eqref{eq:mean}, \eqref{eq:def:Gmeanmeasure}, 
	and \eqref{eq:Ginformation-neq-alternative}
	\begin{align}
	\notag
	\GMIL{\rno}{\uma{\rno,\mP}{\lgm}}{\Wm}{\lgm}
	&=\tfrac{\rno}{\rno-1}\ln 
	\int \left(\sum\nolimits_{\dinp} \uma{\rno,\mP}{\lgm}(\dinp)
	e^{(1-\rno)\lgm\cdot \costf(\dinp)}
	(\der{\Wm(\dinp)}{\rfm})^{\rno} \right)^{\sfrac{1}{\rno}} \rfm(\dif{\dout})
	\\
	\notag
	&=\tfrac{\rno}{\rno-1}\ln \int\der{\qmn{\rno,\mP}}{\rfm}\rfm(\dif{\dout})	
	-\tfrac{1}{\rno-1}
	\ln\sum\nolimits_{\tilde{\dinp}}\mP(\tilde{\dinp})e^{(1-\rno)\RD{\rno}{\Wm(\tilde{\dinp})}{\qmn{\rno,\mP}}
		+(\rno-1)\lgm \cdot \costf(\dinp)}.
	\end{align}
	Then  \eqref{eq:lem:Lpolsha:poltyrev} follows from \eqref{eq:Lpolsha:poltyrev-1}
	\begin{itemize}
		\item 
		In order to prove \eqref{eq:lem:Lpolsha:poltyrev-variational} for \(\rno\in(0,1)\) case,
		we prove the following inequality
		\begin{align}
		\notag
		\GMIL{\rno}{\mU}{\Wm}{\lgm}+\tfrac{1}{\rno-1}\RD{1}{\mP}{\mU}
		&\leq \RMIL{\rno}{\mP}{\Wm}{\lgm}
		&
		&\forall \mU\in\pdis{\inpS}.
		\end{align}
		Preceding inequality together with \eqref{eq:lem:Lpolsha:poltyrev} imply 
		\eqref{eq:lem:Lpolsha:poltyrev-variational} for \(\rno\in(0,1)\).
		Note that the inequality holds trivially when \(\mP\NAC\mU\) because \(\RD{1}{\mP}{\mU}\)
		is infinite in that case. Thus we are  left with \(\mP\AC\mU\) case.
		On the other hand, any \(\mU\in\pdis{\inpS}\) can be written as 
		\(\mU=\umn{ac}+\umn{s}\) where \(\umn{ac}\AC\mP\) and \(\umn{s}\perp\mP\). 
		Then
		\begin{align}
		\notag
		\GMIL{\rno}{\mU}{\Wm}{\lgm}
		&\mathop{\leq}^{(i)}\RD{\rno}{\mU\mtimes \Wm e^{\frac{1-\rno}{\rno}\lgm \cdot \costf}}{\mU\otimes\qmn{\rno,\mP}}
		\\
		\notag
		&\mathop{\leq}^{(ii)} \tfrac{1}{\rno-1}
		\ln \left[\sum\nolimits_{\dinp}\umn{ac}(\dinp)e^{(\rno-1)\RD{\rno}{\Wm(\dinp)}{\qmn{\rno,\mP}}+(1-\rno)\lgm \cdot\costf(\dinp)}\right]
		\\
		\notag
		&\mathop{=}^{(iii)}\tfrac{1}{\rno-1}
		\ln \left[\sum\nolimits_{\dinp}\mP(\dinp)\tfrac{\umn{ac}(\dinp)}{\mP(\dinp)}e^{(\rno-1)\RD{\rno}{\Wm(\dinp)}{\qmn{\rno,\mP}}+(1-\rno)\lgm \cdot\costf(\dinp) }\right]
		\\
		\notag
		&\mathop{\leq}^{(iv)}\tfrac{1}{\rno-1}
		\left[\sum\nolimits_{\dinp} \mP(\dinp)\ln \tfrac{\umn{ac}(\dinp)}{\mP(\dinp)}e^{(\rno-1)\RD{\rno}{\Wm(\dinp)}{\qmn{\rno,\mP}}+(1-\rno)\lgm \cdot\costf(\dinp) }\right]
		\\
		\notag
		&\mathop{=}^{(v)}\RMIL{\rno}{\mP}{\Wm}{\lgm}
		-\tfrac{1}{\rno-1}\RD{1}{\mP}{\umn{ac}}
		&
		&
		\\
\notag
&\mathop{=}^{(vi)}\RMIL{\rno}{\mP}{\Wm}{\lgm}
-\tfrac{1}{\rno-1}\RD{1}{\mP}{\mU}.
&
&
		\end{align}
		where \((i)\) follows from \eqref{eq:def:Ginformation},
		\((ii)\) follows from \eqref{eq:def:divergence} and 
		the monotonicity of the natural logarithm function,
		\((iii)\) follows from \(\umn{ac}\sim\mP\) which holds because \(\mP\AC\mU\),
		\((iv)\) follows from the Jensen's inequality and the concavity of the 
		natural logarithm function,
		\((v)\) follows from Lemma \ref{lem:information}-\eqref{information:zto} and the definition of 
		\(\RMIL{\rno}{\mP}{\Wm}{\lgm}\) given in \eqref{eq:def:Linformation}. 	
		
		\item
		In order to prove \eqref{eq:lem:Lpolsha:poltyrev-variational} for \(\rno\in(1,\infty)\) case,
		we prove the following inequality
		\begin{align}
		\notag
		\GMIL{\rno}{\mU}{\Wm}{\lgm}+\tfrac{1}{\rno-1}\RD{1}{\mP}{\mU}
		&\geq \RMIL{\rno}{\mP}{\Wm}{\lgm}
		&
		&\forall \mU\in\pdis{\inpS}.
		\end{align}
Preceding inequality together with \eqref{eq:lem:Lpolsha:poltyrev} imply 
\eqref{eq:lem:Lpolsha:poltyrev-variational} for \(\rno\in(1,\infty)\).
Note that the inequality holds trivially when \(\mP\NAC\mU\) because \(\RD{1}{\mP}{\mU}\)
is infinite in that case. Thus we are  left with \(\mP\AC\mU\) case.
On the other hand, any \(\mU\in\pdis{\inpS}\) can be written as 
\(\mU=\umn{ac}+\umn{s}\) where \(\umn{ac}\AC\mP\) and \(\umn{s}\perp\mP\). 
Then
\begin{align}
\notag
\GMIL{\rno}{\mU}{\Wm}{\lgm}
&\mathop{=}^{(i)}\RD{\rno}{\mU\mtimes \Wm e^{\frac{1-\rno}{\rno}\lgm \cdot \costf}}{\mU\otimes\qga{\rno,\mU}{\lgm}}
\\
\notag
&\mathop{\geq}^{(ii)} \tfrac{1}{\rno-1}
\ln \left[\sum\nolimits_{\dinp}\umn{ac}(\dinp)e^{(\rno-1)\RD{\rno}{\Wm(\dinp)}{\qga{\rno,\mU}{\lgm}}+(1-\rno)\lgm \cdot\costf(\dinp)}\right]
\\
\notag
&\mathop{=}^{(iii)}\tfrac{1}{\rno-1}
\ln \left[\sum\nolimits_{\dinp}\mP(\dinp)\tfrac{\umn{ac}(\dinp)}{\mP(\dinp)}e^{(\rno-1)\RD{\rno}{\Wm(\dinp)}{\qga{\rno,\mU}{\lgm}}+(1-\rno)\lgm \cdot\costf(\dinp) }\right]
\\
\notag
&\mathop{\geq}^{(iv)}\tfrac{1}{\rno-1}
\left[\sum\nolimits_{\dinp} \mP(\dinp)\ln \tfrac{\umn{ac}(\dinp)}{\mP(\dinp)}e^{(\rno-1)\RD{\rno}{\Wm(\dinp)}{\qga{\rno,\mU}{\lgm}}+(1-\rno)\lgm \cdot\costf(\dinp) }\right]
\\
\notag
&\mathop{\geq}^{(v)}\RMI{\rno}{\mP}{\Wm}+\RD{1}{\qmn{\rno,\mP}}{\qga{\rno,\mU}{\lgm}} -\lgm\cdot\EXS{\mP}{\costf}-\tfrac{1}{\rno-1}\RD{1}{\mP}{\umn{ac}}
&
&
\\
\notag
&\mathop{=}^{(vi)}\RMIL{\rno}{\mP}{\Wm}{\lgm}
-\tfrac{1}{\rno-1}\RD{1}{\mP}{\mU}.
&
&
\end{align}
where \((i)\) follows from \eqref{eq:Ginformation-neq}, 
\((ii)\) follows from \eqref{eq:def:divergence} and 
the monotonicity of the natural logarithm function,
\((iii)\) follows from \(\umn{ac}\sim\mP\) which holds because \(\mP\AC\mU\),
\((iv)\) follows from the Jensen's inequality and the concavity of the 
natural logarithm function,
\((v)\) follows from Lemma \ref{lem:information}-\eqref{information:oti},
\((vi)\) follows from Lemma \ref{lem:divergence-pinsker} and the definition of 
\(\RMIL{\rno}{\mP}{\Wm}{\lgm}\) given in \eqref{eq:def:Linformation}. 	
\end{itemize}
	
	\item[(\ref{Lpolsha:shayevitz})]
	Note that the order \(\rno\) R-G mean for the input distribution \(\mP\) and 
	the Lagrange multiplier \(\lgm\)
	is a fixed point of the order \(\rno\) Augustin operator for the input distribution \(\ama{\rno,\mP}{\lgm}\), i.e.
	\begin{align}
	\notag
	\der{\Aop{\rno}{\ama{\rno,\mP}{\lgm}}{\qga{\rno,\mP}{\lgm}}}{\rfm}
	&=\sum\nolimits_{\dinp} \ama{\rno,\mP}{\lgm}(\dinp) (\der{\Wm(\dinp)}{\rfm})^{\rno} (\der{\qga{\rno,\mP}{\lgm}}{\rfm})^{1-\rno} e^{(1-\rno)\RD{\rno}{\Wm(\dinp)}{\qga{\rno,\mP}{\lgm}}}
	\\
	\notag
	&=\tfrac{1}{e^{(\rno-1)\GMIL{\rno}{\mP}{\Wm}{\lgm}}}
	\sum\nolimits_{\dinp} \mP(\dinp) e^{(1-\rno)\lgm \cdot \costf(\dinp)} (\der{\Wm(\dinp)}{\rfm})^{\rno} (\der{\qga{\rno,\mP}{\lgm}}{\rfm})^{1-\rno} 
	\\
	\notag
	&=\der{\qga{\rno,\mP}{\lgm}}{\rfm}.
	\end{align} 
	Consequently \(\RMI{\rno}{\ama{\rno,\mP}{\lgm}}{\Wm}=\CRD{\rno}{\Wm}{\qga{\rno,\mP}{\lgm}}{\ama{\rno,\mP}{\lgm}}\)
	by Lemma \ref{lem:information}-(\ref{information:zto},\ref{information:oti}).
	Then 
	\begin{align}
	\notag
	\RD{1}{\ama{\rno,\mP}{\lgm}}{\mP}
	&=\sum\nolimits_{\dinp} \ama{\rno,\mP}{\lgm}(\dinp)
	\ln\tfrac{\mP(\dinp)e^{(\rno-1)\RD{\rno}{\Wm(\dinp)}{\qga{\rno,\mP}{\lgm}}+(1-\rno)\lgm \cdot \costf(\dinp)} }{\sum_{\tilde{\dinp}}\mP(\tilde{\dinp})e^{(\rno-1)\RD{\rno}{\Wm(\tilde{\dinp})}{\qga{\rno,\mP}{\lgm}}+(1-\rno)\lgm \cdot \costf(\tilde{\dinp})}} \tfrac{1}{\mP(\dinp)}
	\\
	\notag
	&=(\rno-1)\RMIL{\rno}{\ama{\rno,\mP}{\lgm}}{\Wm}{\lgm}
	-\ln \sum\nolimits_{\tilde{\dinp}}\mP(\tilde{\dinp})e^{(\rno-1)\RD{\rno}{\Wm(\tilde{\dinp})}{\qga{\rno,\mP}{\lgm}}+(1-\rno)\lgm \cdot \costf(\dinp)}
	\\
	\notag
	&=(\rno-1)[\RMIL{\rno}{\ama{\rno,\mP}{\lgm}}{\Wm}{\lgm}-\GMIL{\rno}{\mP}{\Wm}{\lgm}].
	\end{align}
	Thus \eqref{eq:lem:Lpolsha:shayevitz} holds.
	\begin{itemize}
		\item 
		In order to prove \eqref{eq:lem:Lpolsha:shayevitz-variational} for \(\rno\in(0,1)\) case,
		we prove the following inequality 
		\begin{align}
		\notag
		\RMIL{\rno}{\mA}{\Wm}{\lgm}-\tfrac{1}{\rno-1}\RD{1}{\mA}{\mP}
		&\geq \GMIL{\rno}{\mP}{\Wm}{\lgm}
		&
		&\forall \mA\in\pdis{\inpS}.
		\end{align}
		Preceding inequality together with \eqref{eq:lem:Lpolsha:shayevitz} imply 
		\eqref{eq:lem:Lpolsha:shayevitz-variational} for \(\rno\in(0,1)\).
		Note that the inequality holds trivially when \(\mA\NAC\mP\) because \(\RD{1}{\mA}{\mP}\)
		is infinite in that case. Thus we are  left with \(\mA\AC\mP\) case.
		On the other hand, for any \(\mA\in\pdis{\inpS}\), \(\mP\) can be written as 
		\(\mP=\pmn{ac}+\pmn{s}\) where \(\pmn{ac}\AC\mA\) and \(\pmn{s}\perp\mA\). 
		Then
		\begin{align}
		\notag
		\RMIL{\rno}{\mA}{\Wm}{\lgm}-\tfrac{1}{\rno-1}\RD{1}{\mA}{\mP}
		&\mathop{=}^{(i)}\CRD{\rno}{\Wm}{\qmn{\rno,\mA}}{\mA}-\lgm\cdot \EXS{\mA}{\costf}-\tfrac{1}{\rno-1}\RD{1}{\mA}{\pmn{ac}}
		\\
		\notag
		&\mathop{=}^{(ii)}\tfrac{1}{\rno-1}\sum\nolimits_{\dinp}
		\mA(\dinp)\ln \left[\tfrac{\pmn{ac}(\dinp)}{\mA(\dinp)}
		e^{(\rno-1)\RD{\rno}{\Wm(\dinp)}{\qmn{\rno,\mA}}+(1-\rno)\lgm \cdot \costf(\dinp)}\right]
		\\
		\notag
		&\mathop{\geq}^{(iii)}
		\tfrac{1}{\rno-1}\ln \sum\nolimits_{\dinp} \pmn{ac}(\dinp) 
		e^{(\rno-1)\RD{\rno}{\Wm(\dinp)}{\qmn{\rno,\mA}}+(1-\rno)\lgm \cdot \costf(\dinp)}
		\\
		\notag
		&\mathop{\geq}^{(iv)}\tfrac{1}{\rno-1}\ln\sum\nolimits_{\dinp} \mP(\dinp) 
		e^{(\rno-1)\RD{\rno}{\Wm(\dinp)}{\qmn{\rno,\mA}}+(1-\rno)\lgm\cdot\costf(\dinp)}
		\\
		\notag
		&\mathop{\geq}^{(v)} \GMIL{\rno}{\mP}{\Wm}{\lgm}.
		&
		&
		\end{align}
		where \((i)\) follows from \eqref{eq:def:divergence}, \eqref{eq:def:Linformation}, and 
Lemma \ref{lem:information}-\eqref{information:zto},
		\((ii)\) follows from \(\pmn{ac}\sim\mA\) which holds because \(\mA\AC\mP\), 
		\((iii)\) follows from the Jensen's inequality and the concavity of the natural logarithm function, 
		\((iv)\) follows from the monotonicity of the natural logarithm function,
		\((v)\) follows from \eqref{eq:def:divergence} and \eqref{eq:def:Ginformation}.
		\item
		In order to prove \eqref{eq:lem:Lpolsha:shayevitz-variational} for \(\rno\in(1,\infty)\) case,
		we prove the following inequality 
		\begin{align}
		\notag
		\RMIL{\rno}{\mA}{\Wm}{\lgm}-\tfrac{1}{\rno-1}\RD{1}{\mA}{\mP}
		&\leq \GMIL{\rno}{\mP}{\Wm}{\lgm}
		&
		&\forall \mA\in\pdis{\inpS}.
		\end{align}
		Preceding inequality together with \eqref{eq:lem:Lpolsha:shayevitz} imply 
\eqref{eq:lem:Lpolsha:shayevitz-variational} for \(\rno\in(1,\infty)\).
Note that the inequality holds trivially when \(\mA\NAC\mP\) because \(\RD{1}{\mA}{\mP}\)
is infinite in that case. Thus we are  left with \(\mA\AC\mP\) case.
On the other hand, for any \(\mA\in\pdis{\inpS}\), \(\mP\) can be written as 
\(\mP=\pmn{ac}+\pmn{s}\) where \(\pmn{ac}\AC\mA\) and \(\pmn{s}\perp\mA\). 
Then
\begin{align}
\notag
\RMIL{\rno}{\mA}{\Wm}{\lgm}-\tfrac{1}{\rno-1}\RD{1}{\mA}{\mP}
&\mathop{\leq}^{(i)}\CRD{\rno}{\Wm}{\qga{\rno,\mP}{\lgm}}{\mA}-\lgm\cdot \EXS{\mA}{\costf}-\tfrac{1}{\rno-1}\RD{1}{\mA}{\pmn{ac}}
\\
\notag
&\mathop{=}^{(ii)}\tfrac{1}{\rno-1}\sum\nolimits_{\dinp}
\mA(\dinp)\ln \left[\tfrac{\pmn{ac}(\dinp)}{\mA(\dinp)}
e^{(\rno-1)\RD{\rno}{\Wm(\dinp)}{\qga{\rno,\mP}{\lgm}}+(1-\rno)\lgm \cdot \costf(\dinp)}\right]
\\
\notag
&\mathop{\leq}^{(iii)}
\tfrac{1}{\rno-1}\ln \sum\nolimits_{\dinp} \pmn{ac}(\dinp) 
e^{(\rno-1)\RD{\rno}{\Wm(\dinp)}{\qga{\rno,\mP}{\lgm}}+(1-\rno)\lgm \cdot \costf(\dinp)}
\\
\notag
&\mathop{\leq}^{(iv)}\tfrac{1}{\rno-1}\ln\sum\nolimits_{\dinp} \mP(\dinp) 
e^{(\rno-1)\RD{\rno}{\Wm(\dinp)}{\qga{\rno,\mP}{\lgm}}+(1-\rno)\lgm\cdot\costf(\dinp)}
\\
\notag
&\mathop{=}^{(v)} \GMIL{\rno}{\mP}{\Wm}{\lgm}.
&
&
\end{align}
where \((i)\) follows from \eqref{eq:def:divergence}, \eqref{eq:def:information},
  and \eqref{eq:def:Linformation},
\((ii)\) follows from \(\pmn{ac}\sim\mA\) which holds because \(\mA\AC\mP\), 
\((iii)\) follows from the Jensen's inequality and the concavity of the natural logarithm function, 
\((iv)\) follows from the monotonicity of the natural logarithm function,
\((v)\) follows from \eqref{eq:def:divergence} and \eqref{eq:Ginformation-neq}.
	\end{itemize}    
	\item[(\ref{Lpolsha:augustin})]	
\eqref{eq:lem:Lpolsha:augustin} follows from \eqref{eq:mean} by substitution. 
On the other hand, \eqref{eq:Ginformation-neq-alternative}
and \eqref{eq:lem:Lpolsha:poltyrev-variational} imply
\begin{align}
\label{eq:Lpolsha:augustin-1}
\tfrac{\rno-1}{\rno}\RMIL{\rno}{\mP}{\Wm}{\lgm}
&\leq 
\ln \lon{\mma{\rno,\mU}{\lgm}}
+\tfrac{\RD{1}{\mP}{\mU}}{\rno}
&
&\forall \mU\in\pdis{\inpS}.
\end{align}
For any \(\fX\) satisfying \(\fX:\EXS{\mP}{\fX}=0\),
let \(\umn{\fX}\in\pdis{\inpS}\) be \(\umn{\fX}(\dinp)
\DEF\tfrac{\mP(\dinp) e^{(1-\rno)\fX(\dinp)}}{\sum_{\dsta} \mP(\dsta) e^{(1-\rno)\fX(\dsta)}}
\) for all \(\dinp \in\inpS\).
Thus as a result of \eqref{eq:Lpolsha:augustin-1} and \eqref{eq:def:Gmeanmeasure}
we have
\begin{align}
\notag
\tfrac{\rno-1}{\rno}\RMIL{\rno}{\mP}{\Wm}{\lgm}
&\leq 
\ln\EXS{\rfm}{\left(\sum\nolimits_{\dinp}\mP(\dinp)
	e^{(1-\rno)(\fX(\dinp)+\lgm\cdot\costf(\dinp))}
	\left[\der{\Wm(\dinp)}{\rfm}\right]^{\rno}
	\right)^{\sfrac{1}{\rno}}}
&
&\forall \fX:\EXS{\mP}{\fX}=0.
\end{align}
Then \eqref{eq:lem:Lpolsha:augustin-variational} follows from 
\eqref{eq:lem:Lpolsha:augustin}.
\end{enumerate}
\end{proof}

\begin{proof}[Proof of Lemma \ref{lem:GcapacityFLB}]~
	\begin{enumerate}[(i)]
		\item  {\it \(\exists \widetilde{\mP}\in\pdis{\inpS}\) such that \(\GMIL{\rno}{\widetilde{\mP}}{\Wm}{\lgm}=\GCL{\rno}{\Wm}{\lgm}\):}
		Note that \(\pdis{\inpS}\) is compact because \(\inpS\) is a finite set.
		If \(\GMIL{\rno}{\mP}{\Wm}{\lgm}\) is continuous in \(\mP\), then
		the existence of \(\widetilde{\mP}\) follows from 
		the extreme value theorem, \cite[27.4]{munkres}.
		Thus we are left with establishing the continuity
		of \(\GMIL{\rno}{\mP}{\Wm}{\lgm}\) in \(\mP\).

		Note that for any \(\pmn{1}\) and \(\pmn{0}\) there exist probability mass functions  
		\(\smn{1}\), \(\smn{0}\), and \(\smn{\wedge}\) satisfying
		\(\smn{0}\perp\smn{1}\), 
		\(\pmn{1}=(1-\delta)\smn{\wedge}+\delta\smn{1}\),
		and
		\(\pmn{0}=(1-\delta)\smn{\wedge}+\delta\smn{0}\)
		where \(\delta=\tfrac{\lon{\pmn{1}-\pmn{0}}}{2}\).
		Then applying first \eqref{eq:Ginformation-neq-A} and \eqref{eq:Ginformation-neq}
		we get
		\begin{align}
		\label{eq:GcapacityFLB-1}
		\GMIL{\rno}{\pmn{1}}{\Wm}{\lgm}
		&=\tfrac{1}{\rno-1}\ln \left[
		(1-\delta)e^{(\rno-1)[\GMIL{\rno}{\smn{\wedge}}{\Wm}{\lgm}+\RD{\rno}{\qma{\rno,\smn{\wedge}}{\lgm}}{\qma{\rno,\pmn{1}}{\lgm}}]}
		+\delta e^{(\rno-1)[\GMIL{\rno}{\smn{1}}{\Wm}{\lgm}+\RD{\rno}{\qma{\rno,\smn{1}}{\lgm}}{\qma{\rno,\pmn{1}}{\lgm}}]}
		\right].
		\end{align}
		Note that  \(\GMIL{\rno}{\mP}{\Wm}{\lgm}\leq \GMIL{\rno}{\mP}{\Wm}{}\) 		 
		and \(\GMIL{\rno}{\mP}{\Wm}{}\leq \RD{\rno}{\mP\mtimes \Wm}{\mP\otimes \qmn{1,\mU}}\) 
		for all \(\mP\in\pdis{\inpS}\) by definition
		where \(\mU\) is the uniform distribution on \(\inpS\).
		Furthermore, 
		\(\RD{\rno}{\mP\mtimes \Wm}{\mP\otimes\qmn{1,\mU}}=\tfrac{1}{\rno-1}\ln \sum_{\dinp}\mP(\dinp)e^{(\rno-1)\RD{\rno}{\Wm(\dinp)}{\qmn{1,\mU}}} 
		\leq \ln  \abs{\inpS}\) 
		for all \(\mP\in\pdis{\inpS}\) by Lemma \ref{lem:divergence-RM}.
		Thus \(\GMIL{\rno}{\smn{\wedge}}{\Wm}{\lgm}\leq \ln  \abs{\inpS}\)
		and using Lemma \ref{lem:divergence-pinsker} to bound the expression 
		in  \eqref{eq:GcapacityFLB-1}  we get
		\begin{align}
		\label{eq:GcapacityFLB-2}
		\GMIL{\rno}{\pmn{1}}{\Wm}{\lgm}
		&\geq \GMIL{\rno}{\smn{\wedge}}{\Wm}{\lgm}+\tfrac{1}{\rno-1}\ln \left[(1-\delta)+\delta
		e^{(1-\rno)\ln \abs{\inpS}}\right].
		\end{align}
		On the other hand 
		\((1-\delta)^{\frac{1}{\rno}}\mma{\rno,\smn{\wedge}}{\lgm}\leq \mma{\rno,\pmn{1}}{\lgm}\)
		and
		\(\delta^{\frac{1}{\rno}}\mma{\rno,\smn{1}}{\lgm}\leq \mma{\rno,\pmn{1}}{\lgm}\)
		by \eqref{eq:def:Gmeanmeasure}.
		Then using \eqref{eq:def:Gmean} and Lemma \ref{lem:divergence-RM}
		we get
		\begin{align}
		\notag
		\RD{\rno}{\qma{\rno,\smn{\wedge}}{\lgm}}{\qma{\rno,\pmn{1}}{\lgm}}
		&\leq \tfrac{1}{\rno}\ln\tfrac{1}{1-\delta}-\tfrac{\rno-1}{\rno}(\GMIL{\rno}{\smn{\wedge}}{\Wm}{\lgm}-\GMIL{\rno}{\pmn{1}}{\Wm}{\lgm}),
		\\
		\notag
		\RD{\rno}{\qma{\rno,\smn{1}}{\lgm}}{\qma{\rno,\pmn{1}}{\lgm}}
		&\leq \tfrac{1}{\rno}\ln\tfrac{1}{\delta}-\tfrac{\rno-1}{\rno}(\GMIL{\rno}{\smn{1}}{\Wm}{\lgm}-\GMIL{\rno}{\pmn{1}}{\Wm}{\lgm}).
		\end{align}
		Since \(\GMIL{\rno}{\smn{1}}{\Wm}{\lgm}\leq \GMIL{\rno}{\smn{1}}{\Wm}{}\leq \ln \abs{\inpS}\) 
		using \eqref{eq:GcapacityFLB-1} and we get
		\begin{align}
		\label{eq:GcapacityFLB-4}
		\GMIL{\rno}{\pmn{1}}{\Wm}{\lgm}
		&\leq \GMIL{\rno}{\smn{\wedge}}{\Wm}{\lgm}+ 
		\tfrac{\rno}{\rno-1}\ln \left[(1-\delta)^{\frac{1}{\rno}}
		+\delta^{\frac{1}{\rno}}e^{\frac{\rno-1}{\rno}\ln \abs{\inpS}}
		\right].
		\end{align}
		Using \eqref{eq:GcapacityFLB-2} and \eqref{eq:GcapacityFLB-4} we get
		\begin{align}
		\notag
		\abs{\GMIL{\rno}{\pmn{1}}{\Wm}{\lgm}-\GMIL{\rno}{\pmn{2}}{\Wm}{\lgm}}
		&\leq 
		\tfrac{\rno}{\rno-1}\ln \left[(1-\delta)^{\frac{1}{\rno}}+\delta^{\frac{1}{\rno}}e^{\frac{\rno-1}{\rno}\ln \abs{\inpS}}\right]
		-\tfrac{1}{\rno-1}\ln \left[(1-\delta)+\delta e^{(1-\rno)\ln \abs{\inpS}}\right].
		\end{align}
		Then \(\GMIL{\rno}{\mP}{\Wm}{\lgm}\) is continuous in \(\mP\).
		
		\item {\it If \(\GMIL{\rno}{\widetilde{\mP}}{\Wm}{\lgm}=\GCL{\rno}{\Wm}{\lgm}\), 
			then
			\(\RD{\rno}{\mP \mtimes  \Wm e^{\frac{1-\rno}{\rno}\lgm\cdot \costf}}{\mP \otimes \qga{\rno,\widetilde{\mP}}{\lgm}}\leq\GCL{\rno}{\Wm}{\lgm}\)
			for all \(\mP\in \pdis{\inpS}\):}
		Let \(\widetilde{\mP} \in \pdis{\inpS}\) be such that \(\GMIL{\rno}{\widetilde{\mP}}{\Wm}{\lgm}=\GCL{\rno}{\Wm}{\lgm}\), 
		\(\mP\) be any member of \(\pdis{\inpS}\) and \(\pma{}{(\ind)}\) be \(\tfrac{\ind-1}{\ind}\widetilde{\mP} +\tfrac{1}{\ind}{\mP}\) for 
		\(\ind \in \integers{+}\). Then
		\begin{align}
		%\label{eq:rmi-convexcombination}
		\notag
		\GMIL{\rno}{\pma{}{(\ind)}}{\Wm}{\lgm}
		&\!=\!
		\tfrac{1}{\rno-1}\!\ln\!\left[\!
		\tfrac{\ind-1}{\ind} e^{(\rno-1)\left(\GMIL{\rno}{\widetilde{\mP}}{\Wm}{\lgm}+\RD{\rno}{\qga{\rno,\widetilde{\mP}}{\lgm}}{\qga{\rno,\pma{}{(\ind)}}{\lgm}}\right)}
		\!+\!\tfrac{1}{\ind} e^{(\rno-1)\left(\GMIL{\rno}{{\mP}}{\Wm}{\lgm}+\RD{\rno}{\qga{\rno,{\mP}}{\lgm}}{\qga{\rno,\pma{}{(\ind)}}{\lgm}}\right)}
		\right].
		\end{align}
		Then using \(\GMIL{\rno}{\pma{}{(\ind)}}{\Wm}{\lgm} \leq \GCL{\rno}{\Wm}{\lgm}\),
		\(\GMIL{\rno}{\widetilde{\mP}}{\Wm}{\lgm}=\GCL{\rno}{\Wm}{\lgm}\),  
		and 
		\(\RD{\rno}{\qga{\rno,\widetilde{\mP}}{\lgm}}{\qga{\rno,\pma{}{(\ind)}}{\lgm}}\geq0\)
		we get
		\begin{align}
		\label{eq:GcapacityFLB-b-A}
		\GMIL{\rno}{\mP}{\Wm}{\lgm}+\RD{\rno}{\qga{\rno,\mP}{\lgm}}{\qga{\rno,\pma{}{(\ind)}}{\lgm}}
		&\leq \GCL{\rno}{\Wm}{\lgm}
		&
		&\forall \ind \in \integers{+}.
		\end{align}
		On the other hand using \(\GMIL{\rno}{\pma{}{(\ind)}}{\Wm}{\lgm} \leq \GCL{\rno}{\Wm}{\lgm}\),
		\(\GMIL{\rno}{\widetilde{\mP}}{\Wm}{\lgm}=\GCL{\rno}{\Wm}{\lgm}\),  
		\(\GMIL{\rno}{\mP}{\Wm}{\lgm}\geq 0\), and 
		\(\RD{\rno}{\qga{\rno,{\mP}}{\lgm}}{\qga{\rno,\pma{}{(\ind)}}{\lgm}}\geq 0\),
		we get
		\begin{align}
\notag
\RD{\rno}{\qga{\rno,\widetilde{\mP}}{\lgm}}{\qga{\rno,\pma{}{(\ind)}}{\lgm}}
&\leq \tfrac{1}{\rno-1}\ln \tfrac{\ind-e^{(1-\rno)\GCL{\rno}{\Wm}{\lgm}}}{\ind-1}
&
&\forall \ind\in\integers{+}.
\end{align}
Thus Lemma \ref{lem:divergence-pinsker} implies 
\begin{align}
\notag
\limsup\nolimits_{\ind\to\infty}
\lon{\qga{\rno,\widetilde{\mP}}{\lgm}-\qga{\rno,\pma{}{(\ind)}}{\lgm}}
\leq 0.
\end{align}
Then \(\qga{\rno,\pma{}{(\ind)}}{\lgm}\) converges to \(\qga{\rno,\widetilde{\mP}}{\lgm}\) 
		in the total variation topology and hence in the topology of setwise convergence.
		Since the \renyi divergence is lower semicontinuous in the topology of setwise convergence by 
		Lemma \ref{lem:divergence:lsc},  we have
		\begin{align}
		\label{eq:GcapacityFLB-b-B}
		\RD{\rno}{\qga{\rno,\mP}{\lgm}}{\qga{\rno,\widetilde{\mP}}{\lgm}}
		&\leq  \liminf\nolimits_{\ind \to \infty} 
		\RD{\rno}{\qga{\rno,\mP}{\lgm}}{\qga{\rno,\pma{}{(\ind)}}{\lgm}}.
		\end{align}
		Equations 
		\eqref{eq:Ginformation-neq-A},
		\eqref{eq:Ginformation-neq},
		\eqref{eq:GcapacityFLB-b-A}, \eqref{eq:GcapacityFLB-b-B} imply that 
		\(\RD{\rno}{\mP \mtimes  \Wm e^{\frac{1-\rno}{\rno}\lgm\cdot \costf}}{\mP \otimes \qga{\rno,\widetilde{\mP}}{\lgm}}\leq\GCL{\rno}{\Wm}{\lgm}\)
		for all \(\mP\in \pdis{\inpS}\).
		
		\item {\it \(\exists!\qga{\rno,\Wm}{\lgm}\in\pmea{\outA}\) satisfying \eqref{eq:lem:GcapacityFLB}
			such that \(\qga{\rno,\mP}{\lgm}=\qga{\rno,\Wm}{\lgm}\) for all \(\mP\) with 
			\(\GMIL{\rno}{\mP}{\Wm}{\lgm}=\GCL{\rno}{\Wm}{\lgm}\):}
		If \(\GMIL{\rno}{\mP}{\Wm}{\lgm}\!=\!\GCL{\rno}{\Wm}{\lgm}\)
		for a \(\mP\in\pdis{\inpS}\), 
		then
		as a result of 
		\eqref{eq:Ginformation-neq-A},	\eqref{eq:Ginformation-neq}, and Lemma \ref{lem:divergence-pinsker} we have
		\begin{align}
		\label{eq:GcapacityFLB-c}
		\RD{\rno}{\mP \mtimes  \Wm e^{\frac{1-\rno}{\rno}\lgm\cdot \costf}}{\mP \otimes \qga{\rno,\widetilde{\mP}}{\lgm}}
		&\geq \GCL{\rno}{\!\Wm\!}{\lgm}+\tfrac{\rno\wedge 1}{2}\lon{\qga{\rno,\mP}{\lgm}-\qga{\rno,\widetilde{\mP}}{\lgm}}^{2}.
		\end{align}
		Since  we have already established that
		\(\RD{\rno}{\mP \mtimes  \Wm e^{\frac{1-\rno}{\rno}\lgm\cdot \costf}}{\mP \otimes \qga{\rno,\widetilde{\mP}}{\lgm}}\leq \GCL{\rno}{\!\Wm\!}{\lgm}\)
		for any \(\mP\in\pdis{\inpS}\), \eqref{eq:GcapacityFLB-c} 
		implies that \(\qga{\rno,\mP}{\lgm}=\qga{\rno,\tilde{\mP}}{\lgm}\)  
		for any \(\mP\in\pdis{\inpS}\) satisfying 
		\(\GMIL{\rno}{\mP}{\Wm}{\lgm}=\GCL{\rno}{\Wm}{\lgm}\).		
	\end{enumerate}
\end{proof}

\begin{proof}[Proof of Theorem \ref{thm:Gminimax}]
	Note that \eqref{eq:thm:Gminimax} implies \eqref{eq:thm:Gminimaxradius} and 
	\eqref{eq:thm:Gminimaxcenter} implies \eqref{eq:thm:Gminimaxradiuscenter}.
	Furthermore, the left hand side of \eqref{eq:thm:Gminimax} is equal to \(\GCL{\rno}{\Wm}{\lgm}\) by
	\eqref{eq:Gcapacity}.
	Thus when \(\GCL{\rno}{\Wm}{\lgm}\) is infinite, \eqref{eq:thm:Gminimax} holds trivially by 
	the max-min inequality.
	When \(\GCL{\rno}{\Wm}{\lgm}\) is finite, \eqref{eq:thm:Gminimax} follows from \eqref{eq:thm:Gminimaxcenter}
	and the max-min inequality. 
	Thus we can assume \(\GCL{\rno}{\Wm}{\lgm}\) to be finite and prove  the claims about \(\qga{\rno,\Wm}{\lgm}\), 
	in order to prove the theorem.
	\begin{enumerate}[(i)]
		\item {\it If \(\GCL{\rno}{\Wm}{\lgm}\!<\!\infty\) and
			\(\lim\nolimits_{\ind \to \infty}\!\GMIL{\rno}{\pma{}{(\ind)}}{\Wm}{\lgm}\!=\!\GCL{\rno}{\Wm}{\lgm}\),
			then \(\{\qga{\rno,\pma{}{(\ind)}}{\lgm}\}_{\ind\in\integers{+}}\) is a Cauchy sequence in 
			\(\pmea{\outA}\) for the total variation metric:}
		For any sequence \(\{\pma{}{(\ind)}\}_{\ind\in \integers{+}}\subset\pdis{\inpS}\) satisfying 
		\(\lim\nolimits_{\ind \to \infty} \GMIL{\rno}{\pma{}{(\ind)}}{\Wm}{\lgm}=\GCL{\rno}{\Wm}{\lgm}\), let 
		us consider a sequence of channels \(\{\Wm^{(\ind)}\}_{\ind\in\integers{+}}\) 
		whose input sets \(\{\inpS^{(\ind)}\}_{\ind\in\integers{+}}\) 
		form a nested sequence of finite subsets of \(\inpS\) defined as follows, 
		\begin{align}
		\notag
		\inpS^{(\ind)} 
		&\DEF\{\dinp\in\inpS:\exists \jnd\in\{1,\ldots,\ind\} \mbox{~such that~} \pma{}{(\jnd)}(\dinp)>0\}.
		\end{align}
		
		Then for any \(\ind\in\integers{+}\), there exists a unique \(\qga{\rno,\Wm^{(\ind)}}{\lgm}\) satisfying 
		\eqref{eq:lem:GcapacityFLB} by Lemma \ref{lem:GcapacityFLB}.
		Furthermore, \(\pdis{\inpS^{(\jnd)}}\subset\pdis{\inpS^{(\ind)}}\) for any \(\ind,\jnd\in\integers{+}\) such that \(\jnd\leq \ind\).
		In order to bound
		\(\lon{\qga{\rno,\pma{}{(\jnd)}}{\lgm}-\qga{\rno,\pma{}{(\ind)}}{\lgm}}\) for positive integers \(\jnd<\ind\), we use the triangle
		inequality for \(\qma{\rno,\pma{}{(\jnd)}}{\lgm}\), \(\qga{\rno,\pma{}{(\ind)}}{\lgm}\) and \(\qga{\rno,\Wm^{(\ind)}}{\lgm}\)
		\begin{align}
		\label{eq:GcapacityLB-A}
		\lon{\qga{\rno,\pma{}{(\jnd)}}{\lgm}-\qga{\rno,\pma{}{(\ind)}}{\lgm}}
		\leq \lon{\qga{\rno,\pma{}{(\jnd)}}{\lgm}-\qga{\rno,\Wm^{(\ind)}}{\lgm}}
		+\lon{\qga{\rno,\pma{}{(\ind)}}{\lgm}-\qga{\rno,\Wm^{(\ind)}}{\lgm}}.
		\end{align} 
		Let us proceed with bounding \(\lon{\qga{\rno,\pma{}{(\jnd)}}{\lgm}-\qga{\rno,\Wm^{(\ind)}}{\lgm}}\)
		and \(\lon{\qga{\rno,\pma{}{(\ind)}}{\lgm}-\qga{\rno,\Wm^{(\ind)}}{\lgm}}\).
		\begin{align}
		\notag
		\lon{\qga{\rno,\pma{}{(\jnd)}}{\lgm}-\qga{\rno,\Wm^{(\ind)}}{\lgm}}
		&\mathop{\leq}^{(a)} 
		\sqrt{\tfrac{2}{\rno\wedge 1} \RD{\rno}{\qga{\rno,\pma{}{(\jnd)}}{\lgm}}{\qga{\rno,\Wm^{(\ind)}}{\lgm}}}
		\\
		\notag
		&\mathop{\leq}^{(b)} 
		\sqrt{\tfrac{2}{\rno\wedge 1}}
\sqrt{\RD{\rno}{\pma{}{(\jnd)} \mtimes  \Wm^{(\ind)} e^{\frac{1-\rno}{\rno}\lgm\cdot\costf}}{\pma{}{(\jnd)}\otimes {\qga{\rno,\Wm^{(\ind)}}{\lgm}}}
	-\GMIL{\rno}{\pma{}{(\jnd)}}{\Wm^{(\ind)}}{\lgm}}
		\\
		\notag
		&\mathop{\leq}^{(c)} 
		\sqrt{\tfrac{2}{\rno\wedge 1}}
		\sqrt{\GCL{\rno}{\Wm^{(\ind)}}{\lgm}-\GMIL{\rno}{\pma{}{(\jnd)}}{\Wm^{(\ind)}}{\lgm}}
		\\
		\notag
		&\mathop{\leq}^{(d)} 
		\sqrt{\tfrac{2}{\rno\wedge 1}}
\sqrt{\GCL{\rno}{\Wm}{\lgm}-\GMIL{\rno}{\pma{}{(\jnd)}}{\Wm}{\lgm}}
		\end{align}
		where 
		\((a)\) follows from Lemma \ref{lem:divergence-pinsker},
		\((b)\) follows from \eqref{eq:Ginformation-neq-A} and \eqref{eq:Ginformation-neq}, 
		\((c)\) follows Lemma \ref{lem:GcapacityFLB} because \(\pma{}{(\jnd)}\!\in\!\pdis{\inpS^{(\ind)}}\), 
		and 
		\((d)\) follows from  \(\GMIL{\rno}{\pma{}{(\jnd)}}{\Wm^{(\ind)}}{\lgm}\!=\!\GMIL{\rno}{\pma{}{(\jnd)}}{\Wm}{\lgm}\) 
		and \(\GCL{\rno}{\Wm^{(\ind)}}{\lgm}\!\leq\!\GCL{\rno}{\Wm}{\lgm}\).
		We can obtain a similar bound on \(\lon{\qga{\rno,\pma{}{(\ind)}}{\lgm}\!-\!\qga{\rno,\Wm^{(\ind)}}{\lgm}}\).
		Then \(\{\qga{\rno,\pma{}{(\ind)}}{\lgm}\}\) is a Cauchy sequence as a result of \eqref{eq:GcapacityLB-A}
		because \(\lim\nolimits_{\ind \to \infty}\!\GMIL{\rno}{\pma{}{(\ind)}}{\Wm}{\lgm}\!=\!\GCL{\rno}{\Wm}{\lgm}\).
		
		\item {\it If \(\GCL{\rno}{\Wm}{\lgm}<\infty\), then \(\exists!\qga{\rno,\Wm}{\lgm}\in \pmea{\outA}\) 
			satisfying 
			\(\lim\nolimits_{\ind\to \infty}\lon{\qga{\rno,\Wm}{\lgm}-\qga{\rno,\pma{}{(\ind)}}{\lgm}}=0\) 
			for all \(\{\pma{}{(\ind)}\}_{\ind\in\integers{+}}\subset\pdis{\inpS}\) such that
			\(\lim\nolimits_{\ind \to \infty} \GMIL{\rno}{\pma{}{(\ind)}}{\Wm}{\lgm}=\GCL{\rno}{\Wm}{\lgm}\):}
		Note that \(\smea{\outA}\) is a complete metric space for the total variation metric.
		Then \(\{\qga{\rno,\pma{}{(\ind)}}{\lgm}\}_{\ind\in\integers{+}}\) has a unique limit point \(\qga{\rno,\Wm}{\lgm}\) in \(\smea{\outA}\). 
		Since \(\pmea{\outA}\) is a closed set for the total variation topology and 
		\(\cup_{\ind\in\integers{+}}\qga{\rno,\pma{}{(\ind)}}{\lgm}\subset \pmea{\outA}\),
		then \(\qga{\rno,\Wm}{\lgm}\in\pmea{\outA}\), by \cite[Thm. 2.1.3]{munkres}.
		
		We have established the existence of a unique limit point \(\qga{\rno,\Wm}{\lgm}\),
		for any sequence  \(\{\pma{}{(\ind)}\}_{\ind\in\integers{+}}\subset\pdis{\inpS}\) satisfying
		\(\lim\nolimits_{\ind \to \infty}\GMIL{\rno}{\pma{}{(\ind)}}{\Wm}{\lgm}=\GCL{\rno}{\Wm}{\lgm}\).
		This, however, implies 
		\(\lim_{\ind\to\infty}\lon{\qga{\rno,\widetilde{\mP}^{(\ind)}}{\lgm}-\qga{\rno,\pmn{}{*}}{\lgm}}=0\)
		for any \(\{\widetilde{\mP}^{(\ind)}\}_{\ind\in\integers{+}}\) satisfying 
		\(\lim\nolimits_{\ind \to \infty}\GMIL{\rno}{\widetilde{\mP}^{(\ind)}}{\Wm}{\lgm}=\GCL{\rno}{\Wm}{\lgm}\)
		because we can interleave the elements of \(\{\pma{}{(\ind)}\}_{\ind\in\integers{+}}\) and 
		\(\{\widetilde{\mP}^{(\ind)}\}_{\ind\in\integers{+}}\) to obtain a new sequence 
		\(\{\widehat{\mP}^{(\ind)}\}_{\ind\in\integers{+}}\)
		satisfying 
		\(\lim\nolimits_{\ind \to \infty}\GMIL{\rno}{\widehat{\mP}^{(\ind)}}{\Wm}{\lgm}=\GCL{\rno}{\Wm}{\lgm}\)
		for which \(\{\qga{\rno,\widehat{\mP}^{(\ind)}}{\lgm}\}\) is a Cauchy sequence. Then 
		\(\qga{\rno,\Wm}{\lgm}=\qga{\rno,\pma{}{*}}{\lgm}\)
		
		\item {\it \(\qga{\rno,\Wm}{\lgm}\) satisfies the equality given in \eqref{eq:thm:Gminimaxcenter}:}
		For any \(\mP \in \pdis{\inpS}\), let us consider any sequence \(\{\pma{}{(\ind)}\}_{\ind\in\integers{+}}\) 
		satisfying \(\pma{}{(1)}=\mP\) and
		\(\lim\nolimits_{\ind \to \infty} \GMIL{\rno}{\pma{}{(\ind)}}{\Wm}{\lgm}=\GCL{\rno}{\Wm}{\lgm}\).
		Then \(\mP\in \pdis{\inpS^{(\ind)}}\) for all \(\ind\in\integers{+}\). 
		Using Lemma \ref{lem:GcapacityFLB} we get
		\begin{align}
		\label{eq:GcapacityLB-D}
		\RD{\rno}{\mP \mtimes  \Wm e^{\frac{1-\rno}{\rno}\lgm\cdot\costf}}{\mP\otimes\qga{\rno,\Wm^{(\ind)}}{\lgm}}
		&\leq  \GCL{\rno}{\Wm^{(\ind)}}{\lgm}
		&
		&\forall \ind\in\integers{+}.
		\end{align}
		Since \(\inpS^{(\ind)}\) is a finite set, \(\exists\widetilde{\mP}^{(\ind)}\in \pdis{\inpS^{(\ind)}}\) 
		satisfying \(\GMIL{\rno}{\widetilde{\mP}^{(\ind)}}{\Wm^{(\ind)}}{\lgm}=\GCL{\rno}{\Wm^{(\ind)}}{\lgm}\)
		and \(\qga{\rno,\widetilde{\mP}^{(\ind)}}{\lgm}=\qga{\rno,\Wm^{(\ind)}}{\lgm}\)  by Lemma \ref{lem:GcapacityFLB}.
		Then \(\GMIL{\rno}{\widetilde{\mP}^{(\ind)}}{\Wm^{(\ind)}}{\lgm}\geq \GMIL{\rno}{\pma{}{(\ind)}}{\Wm^{(\ind)}}{\lgm}\) 
		and consequently \(\lim\nolimits_{\ind \to \infty} \GMIL{\rno}{\widetilde{\mP}^{(\ind)}}{\Wm}{\lgm}=\GCL{\rno}{\Wm}{\lgm}\). 
		We have already established that for such a sequence
		\(\qga{\rno,\widetilde{\mP}^{(\ind)}}{\lgm}\rightarrow \qga{\rno,\Wm}{\lgm}\) in the total variation topology, 
		and hence in the topology of setwise convergence. Then
		the lower semicontinuity  of the \renyi divergence (i.e. Lemma \ref{lem:divergence:lsc})
		and	the identity \(\GCL{\rno}{\Wm^{(\ind)}}{\lgm}\leq\GCL{\rno}{\Wm}{\lgm}\) imply that 
		\begin{align}
		\notag
		\RD{\rno}{\mP \mtimes  \Wm e^{\frac{1-\rno}{\rno}\lgm\cdot\costf}}{\mP\otimes\qga{\rno,\Wm}{\lgm}}
		&\leq \GCL{\rno}{\Wm}{\lgm}
		&
		&\forall \mP \in \pdis{\inpS}.
		\end{align}
		On the other hand
		\(\RD{\rno}{\mP \mtimes  \Wm e^{\frac{1-\rno}{\rno}\lgm\cdot\costf}}{\mP\otimes\qga{\rno,\Wm}{\lgm}}
		\geq \GMIL{\rno}{\mP}{\Wm}{\lgm}\)
		and
		\(\GCL{\rno}{\Wm}{\lgm}=\sup_{\mP\in\pdis{\inpS}}\GMIL{\rno}{\mP}{\Wm}{\lgm}\)
		by definitions of \(\GMIL{\rno}{\mP}{\Wm}{\lgm}\) and \(\GCL{\rno}{\Wm}{\lgm}\).
		Thus \eqref{eq:thm:Gminimaxcenter} holds.
	\end{enumerate} 
\end{proof}

\begin{proof}[Proof of Lemma \ref{lem:GEHB}]
	Let us first consider the case \(\rno\in\reals{+}\setminus\{1\}\). As a result 
	of \eqref{eq:Ginformation-neq-A} and \eqref{eq:Ginformation-neq} we have,
	\begin{align}
	\notag
	\sup\nolimits_{\dinp \in \inpS} \RD{\rno}{\Wm(\dinp)}{\mQ}-\lgm \cdot \costf(\dinp)
	&=\sup\nolimits_{\dinp \in \inpS} \RD{\rno}{\Wm(\dinp)e^{\frac{1-\rno}{\rno}\lgm \cdot \costf(\dinp)}}{\mQ}
	\\
	\notag
	&\geq \RD{\rno}{\mP \mtimes  \Wm e^{\frac{1-\rno}{\rno}\lgm \cdot \costf}}{\mP \otimes\mQ}
	&
	&
	\\
	\label{eq:LEHB-A}
	&\geq \GMIL{\rno}{\mP}{\Wm}{\lgm}+\RD{\rno}{\qga{\rno,\mP}{\lgm}}{\mQ}
	&
	&\forall \mP \in \pdis{\inpS}.
	\end{align}
	Let \(\{\pma{}{(\ind)}\}_{\ind\in\integers{+}}\) be a sequence of elements of \(\pdis{\inpS}\) such that
	\(\lim\nolimits_{\ind\to\infty}\GMIL{\rno}{\pma{}{(\ind)}}{\Wm}{\lgm}=\GCL{\rno}{\Wm}{\lgm}\).
	Then the sequence \(\{\qga{\rno,\pma{}{(\ind)}}{\lgm}\}_{\ind\in\integers{+}}\) is a Cauchy sequence with the unique 
	limit point \(\qga{\rno,\Wm}{\lgm}\) by Theorem \ref{thm:Gminimax}. 
	Since \(\{\qga{\rno,\pma{}{(\ind)}}{\lgm}\}\to \qga{\rno,\Wm}{\lgm}\) in total variation topology, same convergence holds in 
	the topology of setwise convergence.
	On the other hand, the order \(\rno\) \renyi divergence is lower semicontinuous for the topology of setwise convergence 
	by Lemma \ref{lem:divergence:lsc}. Thus we have 
	\begin{align}
	\label{eq:LEHB-B}
	\liminf\nolimits_{\ind \to \infty} \left[\GMIL{\rno}{\pma{}{(\ind)}}{\Wm}{\lgm}+
	\RD{\rno}{\qga{\rno,\pma{}{(\ind)}}{\lgm}}{\mQ}\right]
	&\geq \GCL{\rno}{\Wm}{\lgm}+\RD{\rno}{\qga{\rno,\Wm}{\lgm}}{\mQ}.
	\end{align}
	\eqref{eq:LEHB-A} and \eqref{eq:LEHB-B} imply  \eqref{eq:lem:GEHB} 
	for \(\rno\!\in\!\reals{+}\!\setminus\!\{1\}\) because 
	\(\GCL{\rno}{\Wm}{\lgm}\!=\!\RCL{\rno}{\Wm}{\lgm}\) by \eqref{eq:equalityofcapacities} 
	and \(\qga{\rno,\Wm}{\lgm}\!=\!\qma{\rno,\Wm}{\lgm}\)
	by \eqref{eq:equalityofcenters}.
	
	For \(\rno=1\) case, as a result of Lemma \ref{lem:information}-\eqref{information:one}
	and the definition of A-L information given in \eqref{eq:def:Linformation}
	 we have,
	\begin{align}
	\label{eq:LEHB-C}
	\sup\nolimits_{\dinp \in \inpS} \RD{1}{\Wm(\dinp)}{\mQ}-\lgm \cdot \costf(\dinp)
	&\geq \RMIL{1}{\mP}{\Wm}{\lgm}+\RD{1}{\qmn{1,\mP}}{\mQ}
	&
	&\forall \mP \in \pdis{\inpS}.
	\end{align}
	Repeating the argument leading to \eqref{eq:LEHB-B} and invoking Theorem \ref{thm:Lminimax}, 
	rather than Theorem \ref{thm:Gminimax}, we get
	\begin{align}
	\label{eq:LEHB-D}
	\liminf\nolimits_{\ind \to \infty} \left[\RMIL{1}{\pma{}{(\ind)}}{\Wm}{\lgm}+
	\RD{1}{\qmn{1,\pma{}{(\ind)}}}{\mQ}\right]
	&\geq \RCL{1}{\Wm}{\lgm}+\RD{1}{\qma{1,\Wm}{\lgm}}{\mQ}.
	\end{align}
	\eqref{eq:LEHB-C} and \eqref{eq:LEHB-D} imply  \eqref{eq:lem:GEHB} 
	for \(\rno=1\) case.
\end{proof}

\begin{proof}[Proof of Lemma \ref{lem:Lcentercontinuity}]
	Since \(\RCL{\rno}{\Wm}{\lgm}\) is nonincreasing in \(\lgm\) by Lemma \ref{lem:Lcapacity}-(\ref{Lcapacity:function}),
	\(\RCL{\rno}{\Wm}{\lgm_{2}}\leq\RCL{\rno}{\Wm}{\lgm_{1}}\leq\RCL{\rno}{\Wm}{\lgm_{0}}<\infty\).
	We apply Lemma \ref{lem:GEHB} for 
	\(\lgm=\lgm_{2}\) and \(\mQ=\qma{\rno,\Wm}{\lgm_{1}}\) 
	and use the fact that \(0\leq \costf(\dinp)\) for all \(\dinp\in\inpS\) 
	to obtain
	\begin{align}
	\notag
	\RD{\rno}{\qma{\rno,\Wm}{\lgm_{2}}}{\qma{\rno,\Wm}{\lgm_{1}}}+\RCL{\rno}{\Wm}{\lgm_{2}}
	&\leq \sup\nolimits_{\dinp\in \inpS}\RD{\rno}{\Wm(\dinp)}{\qma{\rno,\Wm}{\lgm_{1}}}-\lgm_{2}\cdot \costf(\dinp)  
	\\
	\notag
	&\leq \sup\nolimits_{\dinp\in \inpS}\RD{\rno}{\Wm(\dinp)}{\qma{\rno,\Wm}{\lgm_{1}}}-\lgm_{1}\cdot \costf(\dinp).  
	\end{align}
	Then \eqref{eq:lem:Lcentercontinuity} follows from \eqref{eq:thm:Lminimaxradiuscenter} of Theorem \ref{thm:Lminimax}.
	
	For any two point \(\lgm_{1}\) and \(\lgm_{2}\) in 
	\(\{\lgm:\exists\epsilon>0~s.t.~ \RCL{\rno}{\Wm}{\lgm-\epsilon\uc}\!<\!\infty\!\}\),
	not necessarily satisfying \(\lgm_{1}\leq\lgm_{2} \),
	let \(\lgm_{\vee}\) be \(\lgm_{1}\vee \lgm_{2}\), i.e. \(\lgm_{\vee}^{\ind}=\lgm_{1}^{\ind} \vee \lgm_{2}^{\ind}\)
	for all \(\ind\in\{1,\ldots,\ell\}\).
	Then as a result of the triangle inequality we have
	\begin{align}
	\label{eq:Lcentercontinuity-A}
	\lon{\qma{\rno,\Wm}{\lgm_{1}}-\qma{\rno,\Wm}{\lgm_{2}}}
	&\leq 
	\lon{\qma{\rno,\Wm}{\lgm_{1}}-\qma{\rno,\Wm}{\lgm_{\vee}}}+\lon{\qma{\rno,\Wm}{\lgm_{\vee}}-\qma{\rno,\Wm}{\lgm_{2}}}.
	\end{align}
	On the other hand, as a result of Lemma \ref{lem:divergence-pinsker} and \eqref{eq:lem:Lcentercontinuity}
	we have,
	\begin{align}
	\label{eq:Lcentercontinuity-B}
	\lon{\qma{\rno,\Wm}{\lgm_{1}}-\qma{\rno,\Wm}{\lgm_{\vee}}}
	&=\sqrt{\tfrac{2}{\rno \wedge 1}}\sqrt{\RCL{\rno}{\Wm}{\lgm_{1}}-\RCL{\rno}{\Wm}{\lgm_{\vee}}},
	\\
	\label{eq:Lcentercontinuity-C}
	\lon{\qma{\rno,\Wm}{\lgm_{\vee}}-\qma{\rno,\Wm}{\lgm_{2}}}
	&=\sqrt{\tfrac{2}{\rno \wedge 1}}\sqrt{\RCL{\rno}{\Wm}{\lgm_{2}}-\RCL{\rno}{\Wm}{\lgm_{\vee}}}.
	\end{align}
	Then continuity of \(\qma{\rno,\Wm}{\lgm}\) in \(\lgm\) on \(\{\lgm:\exists\epsilon>0~s.t.~ \RCL{\rno}{\Wm}{\lgm-\epsilon\uc}\!<\!\infty\!\}\)
	for the total variation topology on \(\pmea{\outA}\) follows from
	\eqref{eq:Lcentercontinuity-A}, \eqref{eq:Lcentercontinuity-B},
	\eqref{eq:Lcentercontinuity-C},
	and
	the continuity of \(\RCL{\rno}{\Wm}{\lgm}\) in \(\lgm\) on \(\{\lgm:\exists\epsilon>0~s.t.~ \RCL{\rno}{\Wm}{\lgm-\epsilon\uc}\!<\!\infty\!\}\)
	established in Lemma \ref{lem:Lcapacity}-(\ref{Lcapacity:function}).
\end{proof}

%\subsubsection{Proof of Lemmas about Transition Probabilities}\label{sec:TPproofs}

%\clearpage
\section*{Acknowledgment}
The author would like to thank Fatma Nakibo\u{g}lu and Mehmet Nakibo\u{g}lu for their hospitality; 
this work would not have been possible without it.  
The author would like to thank  Marco Dalai for informing him about Fano's implicit assertion of 
the fixed point property in  \cite{fano} and 
Gonzalo Vazquez-Vilar for pointing out Poltyrev's paper \cite{poltyrev82} on the random coding bound. 
Author would also like to thank the reviewer for his meticulous report, which allowed 
the author to correct a number of inaccurate and/or imprecise statements in the original 
manuscript.
\bibliographystyle{unsrt} 
%\bibliography{main}
\newcommand{\noopsort}[1]{} \newcommand{\printfirst}[2]{#1}
  \newcommand{\singleletter}[1]{#1} \newcommand{\switchargs}[2]{#2#1}

\end{document}